\documentclass[12pt]{article}

\newcommand{\blind}{0}

\usepackage{amsthm,amsmath,amsfonts,amssymb}
\usepackage[colorlinks,citecolor=blue,urlcolor=blue]{hyperref}

\numberwithin{equation}{section}
\theoremstyle{plain}
\newtheorem{proposition}{Proposition}
\newtheorem{definition}{Definition}
\newtheorem{corollary}{Corollary}
\newtheorem{lemma}{Lemma}
\newtheorem{remark}{Remark}

\usepackage{subfigure}     
\usepackage{booktabs}
\usepackage{ctable}
\usepackage{colortbl}
\usepackage{longtable}
\usepackage{ragged2e}

\usepackage[ruled,vlined]{algorithm2e}
\SetKwInput{kwInit}{Initialize}

\usepackage{graphicx}
\usepackage{enumerate}
\usepackage{natbib}
\usepackage{url} 

\usepackage{subcaption}
\usepackage[skip=0pt]{caption}
\begin{document}

\def\spacingset#1{\renewcommand{\baselinestretch}%
{#1}\small\normalsize} \spacingset{1}


\if0\blind
{
  \title{\bf Scalable Variational Bayes Inference for Dynamic Variable Selection}
  \author{Nicolas Bianco\\
    Scientific Computing Center, \\ Karlsruhe Institute of Technology, Germany.\\
    and \\
    Mauro Bernardi \\
    Department of Statistical Sciences, \\ University of Padova, Italy\\
    and \\
    Daniele Bianchi \\
    School of Economics and Finance, \\ Queen Mary University of London, United Kingdom}
  \maketitle
} \fi

\if1\blind
{
  \bigskip
  \bigskip
  \bigskip
  \begin{center}
    {\LARGE\bf Scalable Bayesian Inference for Dynamic Variable Selection in High-Dimensional Regressions}

\end{center}
  \medskip
} \fi

\bigskip
\begin{abstract}
We develop a variational Bayes approach for dynamic variable selection in high-dimensional regression models with time-varying parameters and predictors that exhibit a predefined group structure. Through comprehensive simulation studies, we demonstrate that our method yields more accurate parameter estimates than existing Bayesian static and dynamic variable selection approaches while maintaining computational efficiency. We illustrate the performance of our approach within the context of a popular problem in economics: forecasting inflation based on a large set of macroeconomic predictors. Our approach demonstrates significant improvements in out-of-sample point and density forecasting accuracy. A retrospective analysis of the time-varying parameter estimates reveals economically interpretable patterns in inflation dynamics.
\end{abstract}

\noindent%
{\it Keywords:}  High-dimensional regressions, approximate inference, Bayesian variable selection, mean-field approximation, macroeconomic forecasting.
\vfill

\newpage
\spacingset{1.75} 

\section{Introduction}
Predicting the dynamics of economic variables is critical for policymakers and market participants. For instance, inflation forecasts guide central bankers’ monetary and fiscal decisions and help investors hedge against inflation risk. However, the predictive power of individual variables is rarely known \textit{ex ante}. Consequently, variable selection methods have grown in prominence, particularly in predictive regressions \citep[e.g.,][]{rockova_george.2014,huber2021inducing}. Yet, although economic predictability is often transient, conventional regression models typically perform static variable selection, that is the subsets of ``active'' predictors are assumed constant over the training sample.

To address this limitation, we develop a dynamic Bernoulli-Gaussian regression framework that treats variable selection as a latent stochastic process. Our approach specifically tackles the statistical challenge of dynamically incorporating predictors into a time-varying parameter regression model - admitting variables only during periods where they demonstrate significant predictive power. We propose a variational Bayes estimation method that enables a computationally efficient and statistically consistent identification of relevant predictors in high-dimensional, time-varying parameter regression settings.

Our methodological contribution is threefold. First, the optimal variational density for the latent selection process accommodates flexible smoothness assumptions, regularizing the dynamic inclusion probabilities by filtering noisy estimates that could lead to unstable model identification. Second, we sequentially eliminate predictors showing no informational value across the entire training sample, significantly enhancing both estimation accuracy and computational efficiency. Third, our framework incorporates prior information about group correlations among predictors, allowing decision-makers to leverage domain knowledge about structural relationships between variables for individual variable selection. 

We first assess the estimation accuracy of our approach through an extensive simulation study where time-varying parameters exhibit diverse temporal sparsity patterns. Our evaluation proceeds in two stages: First, we benchmark our approach against recent dynamic variable selection methods, including the dynamic spike-and-slab approaches of \cite{koop_korobilis_2020} and \cite{rockova_mcalinn_2021} and static two-component mixture priors \citep[e.g.,][]{rockova_george.2014}. Second, we compare our variational Bayes method against an equivalent MCMC algorithm to establish baseline performance differences. 

The simulation results demonstrate that our variational Bayes method successfully balances accuracy and computational efficiency. While conventional MCMC becomes computationally prohibitive even for moderate dimensions, our approach scales effectively to high-dimensional settings without compromising estimation accuracy. This computational advantage proves particularly crucial for real-time economic forecasting, where recursive predictions require frequent model updates. Moreover, our method maintains efficiency comparable to competing approximate inference algorithms while achieving significantly superior accuracy across all model dimensions.


The central intuition of our paper is that the precise identification of time-varying relevant predictors is crucial for decision-making in data-rich environments. To empirically validate this claim, we examine a fundamental problem in applied macroeconomics: inflation forecasting using a broad set of macroeconomic indicators \citep[e.g.,][]{stock2006forecasting}. Our analysis employs a comprehensive dataset of 220+ quarterly macroeconomic variables from 1967:Q3 to 2022:Q2, sourced from FRED-QD and transformed following \cite{mccracken2020fred}. We evaluate forecasting performance for four key inflation measures—total CPI, core CPI, GDP deflator, and PCE deflator—across multiple horizons.

Our empirical results reaffirm that parsimonious specifications, particularly the local-level model of \cite{stock2007has}, remain strong benchmarks in inflation forecasting. However, our dynamic Bernoulli-Gaussian predictive regression demonstrates consistent and statistically significant outperformance across all competing methods incorporating macroeconomic predictors. This superior performance holds for both point and density forecasts and proves robust across different forecasting horizons.

Unlike conventional time-varying intercept specifications \`{a} la \cite{stock2007has}, our framework enables retrospective analysis of evolving relationships between macroeconomic variables and inflation measures. Analysis of posterior dynamic inclusion probabilities yields three principal findings: First, the set of predictors with meaningful explanatory power is strikingly sparse. Second, the composition of relevant predictors shows substantial time variation. Third, economic theories receive nuanced empirical support, with their validity contingent on specific periods.


This paper contributes to the growing literature on Bayesian variable selection in high dimensions  (e.g., \citealp{rockova_george.2014,huber2021inducing,ray2022variational,mogliani2024bayesian}). Our work extends these approaches to dynamic settings, complementing several existing frameworks, such as the dynamic variable selection approaches in \cite{koop_korobilis_2020,uribe2020dynamic,rockova_mcalinn_2021}, and (iii) dynamic shrinkage methods \citep[e.g.,][]{kowal2019dynamic}.

\section{Model specification and inference}	
\label{sec:model}

Let $y_t$ a scalar response at time $t$ and $\mathbf{x}_{t-1}=\left(x_{1t-1},\ldots,x_{pt-1}\right)^\prime$ a $p$-dimensional vector of known predictors. A predictive regression model can be specified as
\begin{align}\label{eq:tvp}
	y_t &= \sum_{j=1}^p\beta_{jt}x_{jt-1} + \varepsilon_t, \qquad\varepsilon_t\sim\mathsf{N}(0,e^{h_t}),\qquad t=1,\ldots,n,
\end{align}
where $\boldsymbol{\beta}_t=\left(\beta_{1t},\ldots,\beta_{pt}\right)^\prime$ is a vector of regression coefficients and $h_t=\log\sigma_t^2$ is a latent log-volatility process. We assume that predictors can dynamically enter or leave the regression model as time progresses according to
\begin{align}
\beta_{jt}& = b_{jt}\gamma_{jt},\qquad\text{where}\qquad b_{jt}=b_{jt-1}+v_{jt}\qquad   v_{jt}\sim \mathsf{N}\left(0,\eta_j^2\right),\label{eq:betas}   
\end{align}
with $b_{j0}\sim N\left(0,k_0\eta_j^2\right)$ the initial state, and $\gamma_{jt}\in\{0,1\}$ an indicator variable which identifies if the $j$-th predictor is included or not in the model at time $t$.

Leveraging the Markov property of the random walk, the vector $\mathbf{b}_j=\left(b_{j0},\ldots,b_{jn}\right)^\prime$ for $j=1,\ldots,p$ is a Gaussian Markov Random Field (GMRF) $\mathbf{b}_j\sim\mathsf{N}_{n+1}(\mathbf{0},\eta_j^2\mathbf{Q}^{-1})$. $\mathbf{Q}$ is a tridiagonal precision matrix with diagonal elements $q_{1,1}=1+1/k_0$, $q_{n+1,n+1}=1$, and $q_{l,l}=2$ for $l=2,\ldots,n$, and off-diagonal ones $q_{l,m}=-1$ if $|l-m|=1$ and $0$ elsewhere. The same logic applies to the log-volatility process $\mathbf{h}=\left(h_0,\ldots,h_n\right)^\prime$ with initial state $h_0\sim\mathsf{N}\left(0,k_0\nu^2\right)$, such that $h_t=h_{t-1}+e_t$ with $e_t\sim\mathsf{N}\left(0,\nu^2\right)$. Hence, $\mathbf{h}\sim\mathsf{N}_{n+1}(\mathbf{0},\nu^2\mathbf{Q}^{-1})$.

Eq.\eqref{eq:betas} implies that the time-varying regression coefficients $\left\{\beta_{jt}\right\}_{t=1}^n$ are the product of the latent stochastic process $\left\{b_{jt}\right\}_{t=1}^n$ and the indicator $\left\{\gamma_{jt}\right\}_{t=1}^n$. We assume a hierarchical specification for $\gamma_{jt}$ of the form $\gamma_{jt}|\omega_{jt}\sim\mathsf{Bern}(\mathrm{expit}(\omega_{jt}))$ for $j=1,\ldots,p$, where $\mathrm{expit}(\cdot)$ is the inverse of the logit function. The latent $\boldsymbol{\omega}_j=\left(\omega_{j0},\ldots,\omega_{jn}\right)^\prime$ is a GMRF $\boldsymbol{\omega}_j\sim\mathsf{N}_{n+1}(\mathbf{0},\xi_j^2\mathbf{Q}^{-1})$ and it drives the persistence of the inclusion probability $\mathbb{P}(\gamma_{jt}=1)$.\footnote{The marginal distribution for the vector $\boldsymbol{\gamma}_j=(\gamma_{j1},\ldots,\gamma_{jn})^\prime$ can be retrieved by integrating out $\boldsymbol{\omega}_j$ as,
\begin{equation}
	p(\gamma_{j1},\ldots,\gamma_{jn}) = \int p(\boldsymbol{\omega}_j)\prod_{t=1}^n p(\gamma_{jt}|\omega_{jt}) \,d\boldsymbol{\omega}_j,
\end{equation}
so that $\gamma_{j1},\ldots,\gamma_{jn}$ represent a set of autocorrelated latent states for each $j=1,\ldots,p$.} 

\subsection{Variational Bayes inference}
\label{sec:variational}

Variational inference requires to minimize a Kullback-Leibler divergence ($\mathit{KL}$) between the approximating $q(\boldsymbol{\vartheta})$ and the true posterior $p(\boldsymbol{\vartheta}|\mathbf{y})$ densities, \citep[e.g.,][]{Blei.2017}. \cite{ormerod_wand.2010} show that the problem of minimizing this $\mathit{KL}$ divergence can be equivalently stated as maximizing the variational lower bound (ELBO) denoted by $\underline{p}\left(\mathbf{y};q\right)$:
\begin{equation}
\label{eq:vb_elbo_optim}
q^*(\boldsymbol{\vartheta}) = \arg\max_{q(\boldsymbol{\vartheta}) \in \mathcal{Q}}\log\underline{p}\left(\mathbf{y};q\right)\qquad\text{with}\qquad
\underline{p}\left(\mathbf{y};q\right)=\int q(\boldsymbol{\vartheta}) \log\left\{\frac{p(\mathbf{y},\boldsymbol{\vartheta})}{q(\boldsymbol{\vartheta})}\right\}\,d\boldsymbol{\vartheta}
\end{equation}
where $q^*(\boldsymbol{\vartheta})\in\mathcal{Q}$ represents the optimal variational density and $\mathcal{Q}$ is a space of functions. The choice of the family of distributions $\mathcal{Q}$ is critical and leads to different algorithmic approaches. We consider a Mean-Field Variational Bayes (MFVB) approach which is based on a factorization of the variational density $q(\boldsymbol{\vartheta})=\prod_{i=1}^p q_i(\boldsymbol{\vartheta}_i)$ for a partition $\{ \boldsymbol{\vartheta}_1,\dots,\boldsymbol{\vartheta}_p\}$ of the parameter vector $\boldsymbol{\vartheta}$. This factorization implies that the optimal variational density of each component $q(\boldsymbol{\vartheta}_j)$ can be derived in closed form as
\begin{equation}
    q^\ast(\boldsymbol{\vartheta}_j) \propto \exp\left\{\mathbb{E}_{q(\boldsymbol{\vartheta}\setminus\boldsymbol{\vartheta}_j)}\Big[\log p(\mathbf{y},\boldsymbol{\vartheta})\Big] \right\}\qquad\text{with}\qquad
q(\boldsymbol{\vartheta}\setminus\boldsymbol{\vartheta}_j)=\prod_{\substack{i=1\\i\neq j}}^p q_i(\boldsymbol{\vartheta}_i)\label{eq:mfvb}
\end{equation}
where the expectation is taken with respect to the joint approximating density with the $j$-th element of the partition removed $q(\boldsymbol{\vartheta}\setminus\boldsymbol{\vartheta}_j)$. As a result, the joint density $q^*(\boldsymbol{\vartheta})$ can then be estimated based on a coordinate ascent variational inference (CAVI) algorithm. 

Eq.\eqref{eq:mfvb} shows that the factorization of $q(\boldsymbol{\vartheta})$ is crucial. Let $\boldsymbol{\vartheta}=(\mathbf{h}^\prime,\mathbf{b}^\prime,\boldsymbol{\gamma}^\prime,\boldsymbol{\omega}^\prime,\nu^{2},\boldsymbol{\eta}^{2\prime},\boldsymbol{\xi}^{2\prime})^\prime$ be the vector that collects the model parameters. We factorize $p(\mathbf{y},\boldsymbol{\vartheta})$ as:
\begin{align}\label{eq:joint_p}
	p(\mathbf{y},\boldsymbol{\vartheta}) &= p(\mathbf{y}|\boldsymbol{\vartheta})p(\mathbf{h})p(\nu^2)\prod_{j=1}^p p(\mathbf{b}_j|\eta^2_j)p(\boldsymbol{\omega}_j|\xi^2_j)p(\eta_j^2)p(\xi_j^2)\underbrace{\prod_{t=1}^n p(\gamma_{jt}|\omega_{jt})}_{p(\boldsymbol{\gamma}_{j}|\boldsymbol{\omega}_{j})}.
\end{align}
Two comments are in order. First, we exploit a Polya-Gamma representation \citep{polson2013bayesian} $p(\gamma_{jt}|\omega_{jt}) = \int_0^{+\infty}p(\gamma_{jt}|z_{jt},\omega_{jt})p(z_{jt}|\omega_{jt})\,dz_{jt}$, where $p(z_{jt})$ is the probability density function of a Polya-Gamma $\mathsf{PG}(1,0)$ random variable. This allows for computationally tractable updates, as it leads to a known form of the full conditionals. Second, the optimal variational density of the time-varying coefficients \( q(\mathbf{b}) \) presents a computational bottleneck in high dimensions; assuming that all predictors are correlated requires inversion of a \( p(n+1) \times p(n+1) \) matrix. Although this matrix is block-tridiagonal and can be handled with fast algebra, for large \( p \), the computation remains prohibitive.

To address this issue, we follow \cite{mogliani2024bayesian} and introduce the assumption that predictors may present a known group structure based on shared economic fundamentals, such that the joint density $q\left(\mathbf{b}\right)$ can be factorized as $q(\mathbf{b})=\prod_{k=1}^Kq(\mathbf{b}_k)$ with $K$ groups. This specification covers different cases; $K=1$ represents the most general case in which all regression coefficients are correlated, while $K=p$ implies independence between predictors. By leveraging Polya-Gamma augmentation and the a-priori group correlation structure of the predictors, we can factorize $q(\boldsymbol{\vartheta})$ as,
\begin{align}\label{eq:q_fact}
	q(\boldsymbol{\vartheta}) &= q(\mathbf{h})q(\nu^2)\prod_{k=1}^Kq(\mathbf{b}_k)\prod_{j=1}^p q(\boldsymbol{\omega}_j)q(\eta_j^2)q(\xi_j^2)\prod_{t=1}^n q(\gamma_{jt})q(z_{jt}),
\end{align}
where a characterization of the joint distributions $q(\mathbf{h})$, $q(\mathbf{b}_k)$, and $q(\boldsymbol{\omega}_j)$ is required to provide a global approximation for the latent states. 

\paragraph{Optimal variational densities.} We now discuss the optimal variational densities of $\mathbf{b}_k$, the variable selection indicators $\gamma_{jt}$, and the stochastic volatility $\mathbf{h}$. A detailed description of the remaining components in Eq.\eqref{eq:q_fact} and the full set of analytical derivations and proofs are available in the Supplementary Material. In what follows, we define as $\boldsymbol{\mu}_{q(\mathbf{a})}=\mathbb{E}_q[\mathbf{a}]$ and $\mathbf{\Sigma}_{q(\mathbf{a})}=\mathsf{Var}_q[\mathbf{a}]$ (or $\sigma^2_{q({a})}=\mathsf{Var}_q[{a}]$ if $a$ is univariate) the expected value and variance of the random variable $a$, with respect to density $q$, respectively, and with $\mathbf{a}^{-j}$ the vector $\mathbf{a}$ without its $j$-th element. Proposition \ref{prop:up_gamma_main} provides the optimal variational density $q^\ast(\gamma_{jt})$ for all variables $j=1,\ldots,p_k$ that belong to a group $k=1,\ldots,K$ of size $p_k$.

\begin{proposition}\label{prop:up_gamma_main}
The optimal variational density of $\gamma_{jt}$ is a Bernoulli\\ $q^\ast(\gamma_{jt})=\mathsf{Bern}(\mathrm{expit}(\omega_{q(\gamma_{jt})}))$, where $\mathrm{expit}(\cdot)$ is the inverse of the logit function and 
    \begin{align}
	\omega_{q(\gamma_{jt})}&=\mu_{q(\omega_{jt})}-\frac{1}{2}\mu_{q(1/\sigma_t^2)}(x^2_{jt-1}\mathbb{E}_q[b^2_{jt}]-2\mu_{q(b_{jt})} x_{jt-1}\mu_{q(\varepsilon_{-j,t})}) \nonumber\\
    &\qquad -\mu_{q(1/\sigma_t^2)}x_{jt-1}\mathbf{x}^{-j}_{kt-1}\mathrm{diag}(\boldsymbol{\mu}^{-j}_{q(\gamma_{kt})})[\boldsymbol{\Sigma}_{q(\mathbf{b}_{kt})}]_{-j,j},\label{eq:up_main_gamma}
	\end{align}
with $[\boldsymbol{\Sigma}_{q(\mathbf{b}_{kt})}]_{-j,j}$ the $j$-th column withouth the $j$-th element of the $t$-th diagonal block of $\boldsymbol{\Sigma}_{q(\mathbf{b}_{k})}$, and $\mu_{q(\varepsilon_{-j,t})}=y_t - \mathbf{x}^{-j\,\prime}_{k,t-1}\mathrm{diag}(\boldsymbol{\mu}_{q(\gamma^{-j}_{kt})}){\boldsymbol{\mu}}_{q(\mathbf{b}^{-j}_{kt})}- \sum_{m=1,m\neq k}^K \mathbf{x}_{m,t-1}^\prime\mathrm{diag}(\boldsymbol{\mu}_{q(\gamma_{mt})}){\boldsymbol{\mu}}_{q(\mathbf{b}_{mt})}$ with $\boldsymbol{\mu}_{q(\gamma_{kt})}$ being the collection of $\boldsymbol{\mu}_{q(\gamma_{jt})}$ for the predictors in group $k$.
\end{proposition}

We note that $\omega_{q(\gamma_{jt})}$ depends on the covariance $\boldsymbol{\Sigma}_{q(\mathbf{b}_{k})}$. As a result, the probability that a given predictor is selected depends on its correlation with the other predictors in the same group. If predictors are independent, the second line of Eq.\eqref{eq:up_main_gamma} would equal zero. The parameter $\mu_{q(1/\sigma_t^2)}$ is defined in Remark 1 in the Supplementary Material. Proposition \ref{prop:up_beta_main} shows that the mean and variance of the time-varying parameters $\mathbf{b}_k$ depend on the full trajectory $t=1,\ldots,n$ of the estimates $\boldsymbol{\mu}_{q(\gamma_{jt})}$ for all $j=1,\ldots,p_k$ in group $k$. 



\begin{proposition}\label{prop:up_beta_main}
The optimal variational density for $\mathbf{b}_k=\left(\mathbf{b}_{k0},\ldots,\mathbf{b}_{kn}\right)^\prime$, with $\mathbf{b}_{kt}=(b_{1t},\ldots,b_{p_k t})^\prime$, is a multivariate Gaussian  $q^\ast(\mathbf{b}_k)=\mathsf{N}_{p_k(n+1)}(\boldsymbol{\mu}_{q(\mathbf{b}_k)},\boldsymbol{\Sigma}_{q(\mathbf{b}_k)})$:
	\begin{align}
		\boldsymbol{\Sigma}_{q(\mathbf{b}_k)} &= (\mathbf{D}_k+\mathbf{Q}_k)^{-1},\qquad \boldsymbol{\mu}_{q(\mathbf{b}_k)}=\boldsymbol{\Sigma}_{q(\mathbf{b}_k)}{\mathbf{\Lambda}_k},
	\end{align}
where $\mathbf{D}_k$ is a block-diagonal matrix with $n+1$ blocks of size $p_k$ having generic block equal to $[\mathbf{D}_k]_t = \mu_{q(1/\sigma^2_t)}\mathbf{x}_{k,t-1}\mathbf{x}_{k,t-1}^\prime \odot \{\boldsymbol{\mu}_{q(\gamma_{kt})}\boldsymbol{\mu}_{q(\gamma_{kt})}^\prime+\mathrm{diag}(\boldsymbol{\mu}_{q(\gamma_{kt})}(1-\boldsymbol{\mu}_{q(\gamma_{kt})})\}$, $\mathbf{Q}_k$ is a tridiagonal block matrix $\mathbf{Q}_k=\mathbf{Q}\otimes\mathrm{diag}(\mu_{q(1/\eta_1^2)},\ldots,\mu_{q(1/\eta_{p_k}^2)})$, and $\mathbf{\Lambda}_k$ stacks $p_k$ dimensional vectors defined as $\boldsymbol{\lambda}_{kt}={\mu}_{q(1/\sigma_t^2)}\mathrm{diag}(\boldsymbol{\mu}_{q(\gamma_{kt})})\mathbf{x}_{k,t-1}(y_t - \sum_{m=1,m\neq k}^K \mathbf{x}_{m,t-1}^\prime\mathrm{diag}(\boldsymbol{\mu}_{q(\gamma_{mt})}){\boldsymbol{\mu}}_{q(\mathbf{b}_{mt})})$.
\end{proposition}

Recall from Eq.\eqref{eq:betas} that the dynamic of the regression coefficient $\beta_{jt}$ is a compounding process of the latent state $b_{jt}$ and the indicator variable $\gamma_{jt}$. Proposition \ref{prop:q_beta_tilde_main} builds upon Propositions \ref{prop:up_gamma_main} and \ref{prop:up_beta_main} and provides the optimal variational density for the time-varying regression coefficients for variables $j=1,\ldots,p_k$ in group $k=1,\ldots,K$.

\begin{proposition}\label{prop:q_beta_tilde_main}
    Let $q^\ast(\mathbf{b}_k)$ and $q^\ast(\gamma_{jt})$, for $j=1,\ldots,p_k$, be the optimal variational densities presented in Propositions \ref{prop:up_gamma_main} and \ref{prop:up_beta_main}. Define $\boldsymbol{\beta}_k=\mathbf{\Gamma}_k\mathbf{b}_k$ with $\mathbf{\Gamma}_k=\mathrm{diag}(\boldsymbol{\iota}_{p_k}^\prime,\boldsymbol{\gamma}_{k1}^\prime,\ldots,\boldsymbol{\gamma}_{kn}^\prime)$, where $\boldsymbol{\gamma}_{kt}=\left(\gamma_{1t},\ldots,\gamma_{p_kt}\right)^\prime$ and $\boldsymbol{\iota}_{p_k}$ is a $p_k$-dimensional vector of ones. The optimal variational density of $\boldsymbol{\beta}_k$ is defined as:
	\begin{equation}
q^\ast(\boldsymbol{\beta}_k)=\sum_{\mathbf{s}\in\mathcal{S}}\,w_s\,\mathsf{N}_{p_k(n+1)}(\mathbf{D}_{s}\boldsymbol{\mu}_{q(\mathbf{b}_k)},\mathbf{D}_{s}^{1/2}\mathbf{\Sigma}_{q(\mathbf{b}_k)}\mathbf{D}_{s}^{1/2}),\label{eq:beta}
	\end{equation}
	where $\mathcal{S}$ is a sequence of $\{0,1\}$ of length $p_k n$ with cardinality $|\mathcal{S}|=2^{p_k n}$, the diagonal matrix $\mathbf{D}_{s}=\mathrm{diag}(1,s_{11},\ldots,s_{p_k,n})$, and mixing weights:
	\begin{equation}
		w_s = \prod_{j=1}^{p_k}\prod_{t=1}^n\mu_{q(\gamma_{jt})}^{s_{jt}}(1-\mu_{q(\gamma_{jt})})^{1-s_{jt}},
	\end{equation}
	where $\mathbf{s}=(s_{11},\ldots,s_{jt},\ldots,s_{p_k,n})\in\mathcal{S}$ is an element in $\mathcal{S}$. The mean and variance can be computed analytically as:
	\begin{align}
		\boldsymbol{\mu}_{q(\beta_k)}&=\boldsymbol{\mu}_{q(\Gamma_k)}\boldsymbol{\mu}_{q(\mathbf{b}_k)},\\ \mathbf{\Sigma}_{q(\beta_k)}&=(\boldsymbol{\mu}_{q(\gamma_k)}\boldsymbol{\mu}_{q(\gamma_k)}^\prime+\mathbf{W}_{\mu_{q(\gamma_k)}})\odot\mathbf{\Sigma}_{q(\mathbf{b}_k)}+ \mathbf{W}_{\mu_{q(\gamma_k)}}\odot\boldsymbol{\mu}_{q(\mathbf{b}_k)}\boldsymbol{\mu}_{q(\mathbf{b}_k)}^\prime,
	\end{align}
	where $\mathbf{W}_{\mu_{q(\gamma_k)}}$ is a diagonal matrix with elements $\mu_{q(\gamma_{jt})}(1-\mu_{q(\gamma_{jt})})$, and $\boldsymbol{\mu}_{q(\gamma_k)}$ the collection of $\boldsymbol{\mu}_{q(\gamma_{jt})}$ for $j=1,\ldots,p_k$ and $t=1,\ldots,n$. 
\end{proposition}

To approximate the posterior of the stochastic volatility process $\mathbf{h}$, we leverage on a Gaussian variational approximation and assume that the corresponding variational density is $q\left(\mathbf{h}\right)=\mathsf{N}_{n+1}(\boldsymbol{\mu}_{q(h)},\boldsymbol{\Sigma}_{q(h)})$. Proposition \ref{prop:up_logsima} builds upon \cite{rhode_wand2016} to derive a convenient iterative updating scheme to estimate $\boldsymbol{\mu}_{q(h)},\boldsymbol{\Sigma}_{q(h)}$.

\begin{proposition}\label{prop:up_logsima}
	\it Let $\boldsymbol{\varepsilon}^2=\boldsymbol{\varepsilon}\odot\boldsymbol{\varepsilon}$ with components $[\boldsymbol{\varepsilon}^2]_t=(y_t-\boldsymbol{\beta}_t^\prime\mathbf{x}_t)^2$. Assuming a Gaussian approximation $\mathbf{h}\sim\mathsf{N}_{n+1}(\boldsymbol{\mu}_{q(h)},\boldsymbol{\Sigma}_{q(h)})$, an iterative optimization algorithm can be set as:
	\begin{align*}
		\boldsymbol{\Sigma}_{q(h)}^{new} &= \left[\nabla_{\boldsymbol{\mu}_{q(h)},\boldsymbol{\mu}_{q(h)}}^2 S(\boldsymbol{\mu}_{q(h)}^{old},\boldsymbol{\Sigma}_{q(h)}^{old})\right]^{-1},\qquad
		\boldsymbol{\mu}_{q(h)}^{new} = \boldsymbol{\mu}_{q(h)}^{old} + \boldsymbol{\Sigma}_{q(h)}^{new}\nabla_{\boldsymbol{\mu}_{q(h)}} S(\boldsymbol{\mu}_{q(h)}^{old},\boldsymbol{\Sigma}_{q(h)}^{old}),\vspace{-1em}
	\end{align*}\vspace{-1em}
	where 
	\begin{align*}
		\nabla_{\boldsymbol{\mu}_{q(h)}} S(\boldsymbol{\mu}_{q(h)}^{old},\boldsymbol{\Sigma}_{q(h)}^{old}) &= -\frac{\boldsymbol{\iota}_n}{2}+\frac{1}{2}\mathbb{E}_q(\boldsymbol{\varepsilon}^2)\odot\mathrm{e}^{-\boldsymbol{\mu}_{q(h)}^{old}+\boldsymbol{\sigma}^{2\,old}_{q(h)}/2} -\mu_{q(1/\nu^2)}\mathbf{Q}\boldsymbol{\mu}_{q(h)}^{old},\vspace{-1em}
	\end{align*}\vspace{-1em}
	and
	\begin{align*}
		\nabla_{\boldsymbol{\mu}_{q(h)},\boldsymbol{\mu}_{q(h)}}^2 S(\boldsymbol{\mu}_{q(h)}^{old},\boldsymbol{\Sigma}_{q(h)}^{old}) &= -\frac{1}{2}\mathsf{Diag}(\mathbb{E}_q(\boldsymbol{\varepsilon}^2)\odot\mathrm{e}^{-\boldsymbol{\mu}_{q(h)}^{old}+\boldsymbol{\sigma}^{2\,old}_{q(h)}/2}) -\mu_{q(1/\nu^2)}\mathbf{Q},
	\end{align*}
	denote the first and second derivative of $S(\boldsymbol{\mu}_{q(h)},\boldsymbol{\Sigma}_{q(h)})$ with respect to $\boldsymbol{\mu}_{q(h)}$ and evaluated at $(\boldsymbol{\mu}_{q(h)}^{old},\boldsymbol{\Sigma}_{q(h)}^{old})$, and  $\boldsymbol{\sigma}^2_{q(h)}=\mathrm{diag}(\boldsymbol{\Sigma}_{q(h)})$.  
\end{proposition}

Section B in Supplementary Material reports the optimal variational densities $q^\ast(z_{jt})$, $q^\ast(\boldsymbol{\omega}_j)$, $q^\ast(\eta_j^2)$, $q^\ast(\xi_j^2)$, and $q^\ast(\nu^2)$. For completeness, it also reports the optimal variational density for the homoskedastic case in which the prior for the variance of the residuals in Eq.\eqref{eq:tvp} is an inverse-gamma $\sigma^2\sim\mathsf{IG}(A_\sigma,B_\sigma)$. 


\paragraph{Smoothing the inclusion probabilities.}
\label{subsec:smoothing}

The optimal variational density in Proposition \ref{prop:up_gamma_main} shows that $q(\gamma_{jt})=\mathsf{Bern}(\mathrm{expit}(\omega_{q(\gamma_{jt})}))$, such that the time trajectory of the posterior inclusion probabilities can be obtained as the mean vector
$\mathbb{E}_{q}(\boldsymbol{\gamma}_j|\mathbf{y})=\mathrm{expit}(\boldsymbol{\omega}_{q(\gamma_{j})})$. This entirely data-driven approach could produce erratic posterior inclusion probabilities, especially with noisy dependent variables and predictors. The left panel of Figure \ref{fig:non_smooth} shows this case in point based on a simple simulation example. The inclusion probability $\mathbb{P}(\gamma_{jt}|\mathbf{y})$ can be quite noisy. This is practically inconvenient since $\beta_{jt}=b_{jt}\gamma_{jt}$ for $j=1,\ldots,p_k$, such that the dynamic of $\beta_{jt}$ inherits any eventual erratic behaviour of the inclusion probability.

\begin{figure}[ht]
	\centering
\hspace{-2em}\subfigure[Non-smooth $\boldsymbol{\mu}_{q(\beta)}$ and $\boldsymbol{\mu}_{q(\gamma)}$]{\includegraphics[width=.52\textwidth]{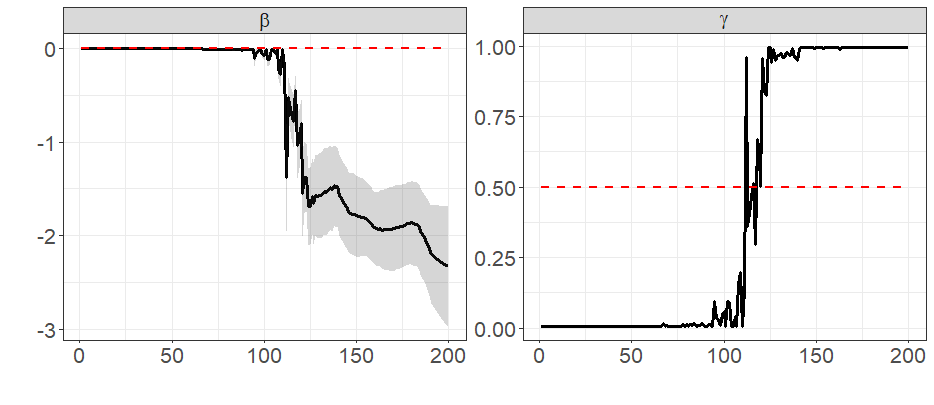}\label{fig:non_smooth}}
	\subfigure[Smooth $\boldsymbol{\mu}_{q(\beta)}$ and $\boldsymbol{\mu}_{q(\gamma)}$]{\includegraphics[width=.52\textwidth]{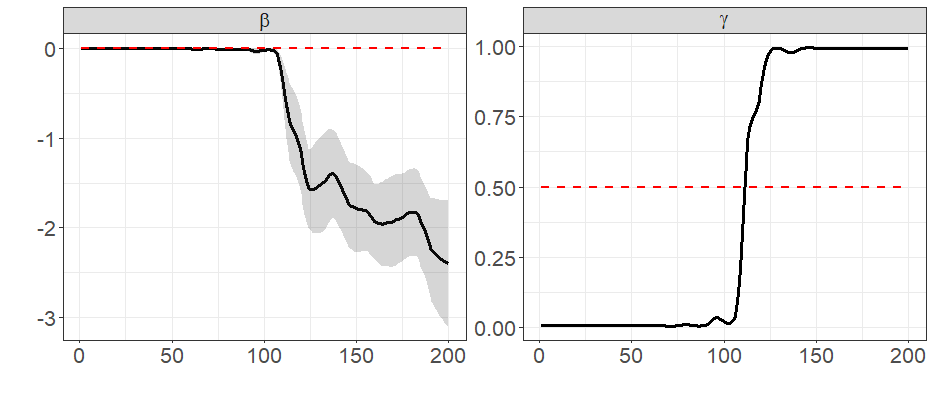}\label{fig:smooth}}
	\caption{\small Smoothing $\boldsymbol{\mu}_{q(\beta_j)}$ and the posterior probability of inclusion $\mathbb{P}(\gamma_{jt}=1|\mathbf{y})$ for $t=1,\ldots,n$.}\hspace{-3em}\label{fig:non_smooth_to_smooth}
\end{figure}\vspace{-1.5em}
To address this issue, we propose to regularise the estimates of the time trajectory of $\gamma_{jt}$ for $t=1,\ldots,n$ based on a parametric approximation. In particular, we define $\{q(\gamma_{jt})\}_{t=1}^n$ as the closest approximation $\{\widetilde{q}(\gamma_{jt})\}_{t=1}^n$ in terms of KL divergence which leads to a smooth sequence of inclusion probabilities, whose expected values coincide with the non-smooth estimates. Proposition \ref{prop:up_gamma_sm_main} formalises our smoothing approach.

\begin{proposition}\label{prop:up_gamma_sm_main}
    A smooth estimate for the trajectory of the inclusion probabilities can be achieved assuming $\widetilde{q}(\gamma_{jt})=\mathsf{Bern}(\mathrm{expit}(\mathbf{w}_t^\prime\mathbf{f}_j))$ such that the expectation of the joint vector $\boldsymbol{\gamma}_j=(\gamma_{j1},\ldots,\gamma_{jn})^\prime$ is equal to $\mathbb{E}_{\widetilde{q}}(\boldsymbol{\gamma}_j)=\mathbf{W}\mathbf{f}_j$, where $\mathbf{W}$ is a $n\times k$ B-spline basis matrix. The optimal value of $\mathbf{f}_j$ is the solution of the optimization problem $\widehat{\mathbf{f}}_j = \arg\max_{\mathbf{f}_j\in\mathbb{R}^k} \psi(\mathbf{f}_j)$ where $\psi(\mathbf{f}_j) = \sum_{t=1}^{n}\left[(\omega_{q(\gamma_{jt})}-\mathbf{w}^\prime_t\mathbf{f}_j)\mathrm{expit}(\mathbf{w}^\prime_t\mathbf{f}_j) +\log(1+\exp(\mathbf{w}^\prime_t\mathbf{f}_j))\right]$, such that the gradient is equal to
$\nabla_{\mathbf{f}} \psi(\mathbf{f}) = \sum_{t=1}^{n}\mathbf{w}_t(\omega_{q(\gamma_{jt})}-\mathbf{w}^\prime_t\mathbf{f})\frac{\mathrm{expit}(\mathbf{w}^\prime_t\mathbf{f})}{1+\exp(\mathbf{w}^\prime_t\mathbf{f})}$.
\end{proposition}

The proof is available in the Supplementary Material. 
Note that Proposition \ref{prop:up_gamma_sm_main} holds the same for correlated predictors since the original $q(\gamma_{jt})$ is obtained by taking into account the correlation structure of $\mathbf{b}_k$ (see Proposition \ref{prop:up_gamma_main}). Figure \ref{fig:smooth} shows that as a by-product of a smoother inclusion probability, the dynamic of the regression coefficient is also regularised. 

For the interested reader, Algorithm 2 reported in the Supplementary Material illustrates the complete iterative optimization scheme to perform the model estimation. The corresponding R code is available online at \url{https://github.com/whitenoise8/BGdvs}.



\paragraph{Convergence properties and computational efficiency}

In this section, we discuss the convergence properties of our variational Bayes algorithm, which can then be exploited to improve computational efficiency further. Consider the variables $j=1,\ldots,p_k$ in group $k$. Proposition \ref{theo:main} extends \cite{ormerod2017VS} to the case of dynamic variable selection in a time-varying parameter regression model.

\begin{proposition}\label{theo:main}
Assume that the inclusion probabilities for a given variable $j$ in group $k$ at the $i$-th iteration of the algorithm is such that $\max_{t\in\{1,\ldots,n\}}\mu^{(i)}_{q(\gamma_{jt})}=\epsilon$, for $\epsilon\approx 0$. Let the $\boldsymbol{\omega}_j$ covariance update be $\boldsymbol{\Sigma}^{(i)}_{q(\omega_j)}-\boldsymbol{\Sigma}^{(i-1)}_{q(\omega_j)}\geq 0$. Then we obtain:
	\begin{center}
		\begin{enumerate}
			\item $\mu^{(i+1)}_{q(\gamma_{jt})} = \mathrm{expit}\left\{\mu^{(i+1)}_{q(\omega_{jt})}-\frac{1}{2}\mu^{(i+1)}_{q(1/\sigma_t^2)}x_{jt-1}^2\mu_{q(1/\eta_j^2)}^{-1(i+1)}q_{tt}+O(\epsilon)\right\}$, $q_{tt}=[\mathbf{Q}^{-1}]_{tt}\geq 0$;
			\item $\mu_{q(\omega_{jt})}^{(i+1)} = -1/2\sum_{k=1}^n s_{tk}+O(\epsilon)$, $s_{tk}=[\boldsymbol{\Sigma}_{q(\omega_{j})}]_{tk}\geq 0$;
			\item $\mu_{q(\omega_{jt})}^{(i+1)}\leq\mu_{q(\omega_{jt})}^{(i)}$ decreases after each iteration.
		\end{enumerate}
	\end{center}  
\end{proposition}
The proof is available in Section C in the Supplementary Material. From Proposition \ref{theo:main} and Lemma 1 in \cite{ormerod2017VS}
two useful numerical results arise: first, for $\epsilon\approx 0$, the following approximation for the update of the inclusion probabilities holds:
	\begin{align}
		\mu^{(i+1)}_{q(\gamma_{jt})} \approx \mathrm{expit}\left\{\mu^{(i+1)}_{q(\omega_{jt})}-\frac{1}{2}\mu^{(i+1)}_{q(1/\sigma_t^2)}x_{jt-1}^2\mu_{q(1/\eta_j^2)}^{-1(i+1)}q_{tt}\right\}.\label{eq:sparsity}
	\end{align}
This implies that for $\text{M}^{(i)}=\arg\min_{t\in\{1,\ldots,n\}}\mu^{(i)}_{q(\omega_{jt})}\ll0$, after $i$ iterations, the sequence $\{\mu^{(i)}_{q(\gamma_{jt})}\}_{t=1}^n$ is indistinguishable from zero. As a result, the variational densities are concentrated to a point mass at zero $\forall t=1,\ldots,n$. Second, if $\mu^{(i)}_{q(\gamma_{jt})}\approx 0$, $\forall t$, then all successive updates $i+k\geq i$  imply $\mu^{(i+k)}_{q(\gamma_{jt})}\approx0$ since ${\mu}_{q(\omega_{jt})}^{(i+k)} \leq {\mu}_{q(\omega_{jt})}^{(i)}$ and therefore $\text{M}^{(i+k)}\leq \text{M}^{(i)}$.

Proposition \ref{theo:main} provides a useful ``dimension reduction'' strategy whereby we can remove the $j$-th variable from the set of predictors within the coordinate ascent algorithm. Such an exclusion strategy improves the computational efficiency when $p$ increases but the signal $\bar{p}\leq p$ remains constant, where $\bar{p}=\mathsf{card}(\mathcal{J})$ and $\mathcal{J}=\{j:\sum_{t=1}^n\gamma_{jt}>0\}$ collects the indexes of regression coefficients that are included in the model at least for one $t$. The extended algorithm, which includes Proposition \ref{theo:main}, is summarized in Algorithm 3 in the Supplementary Material. Section C.1 in the Supplementary Material provides a comprehensive visual discussion of the implications of Proposition \ref{theo:main} for a simple simulation set-up. The approximation is exact after less than 30 iterations, meaning the algorithm induces sparsity in the full trajectory of time-varying parameters as highlighted in Eq.\eqref{eq:sparsity}.

\paragraph{Prior elicitation.} We follow \cite{koop_korobilis_2020} and set $\mu_{q(\gamma_{jt})}^{(0)}=1/2$, $\forall t,j$. Next, we assume inverse-gamma priors for the variances parameters $\nu^2\sim\mathsf{IG}(A_\nu,B_\nu)$, $\eta_j^2\sim\mathsf{IG}(A_\eta,B_\eta)$, and $\xi_j^2\sim\mathsf{IG}(A_\xi,B_\xi)$.  We follow \cite{ormerod2017VS} and set $A_\sigma=B_\sigma=A_\eta=B_\eta=0.01$ to maintain non-informativeness. These are all standard conjugate priors. An important feature of our dynamic variable selection method is that the time-variation of $\gamma_{jt}$ depends on $\omega_{jt}$, where the dynamic of the latter is governed by the variance $\xi^2_j\sim\mathsf{IG}(A_\xi,B_\xi)$. 

To shed light on the impact of $A_\xi,B_\xi$, in the Supplementary Material we provide a deeper analysis on the estimated variational mean $\{\mu_{q(\omega_{jt})}\}_{t=1}^n$ and variance $\mathbf{\Sigma}_{q(\omega_j)}$, and the resulting $\{\mu_{q(\omega_{jt})}\}_{t=1}^n$, for three alternative limit cases. Specifically, we investigate a series of comparative statics based on the optimal variational densities. This provides a transparent strategy to select values for $A_\xi,B_\xi$.

\section{Simulation study}
\label{sec:sim}
We evaluate our approach through $M=100$ Monte Carlo replications, sampling from the data-generating process $y_t = \sum_{j=1}^p \beta_{jt}x_{jt-1} + \varepsilon_t$ with $\varepsilon_t \sim \mathsf{N}(0,0.25)$. Mirroring our empirical application, we set $n=200$ observations and consider $p\in\{10,50,100,200\}$ predictors. The coefficients follow three distinct patterns: $\beta_{1t}$ serves as a persistent time-varying intercept ($\gamma_{1t}=1$ for all $t$), $\beta_{2:7,t}$ exhibit various sparsity regimes, while $\beta_{8:p,t}$ remain insignificant throughout ($\gamma_{8:p,t}=0$ for all $t$). Complete specifications for $\beta_{1:7,t}$ appear in Section E of the Supplementary Material.

We examine six variants of our framework: The baseline independent predictor model ({\tt BG}); a version with smoothing on inclusion probabilities ({\tt BGS}); two alternative hyperprior specifications for $\xi^2_j\sim\mathsf{IG}(A_\xi,B_\xi)$ with $B_\xi=1$ ({\tt BGS(1)}) and $B_\xi=10$ ({\tt BGS(10)}); a constant volatility restriction ({\tt BGH}); a fully correlated predictor version ({\tt BG joint}); and a grouped structure variant ({\tt BG group}). These specifications differ primarily in their variational factorization, as detailed in Section \ref{sec:variational}.

For comparison, we benchmark against established variable selection methods including the static spike-and-slab {\tt EMVS} \citep{rockova_george.2014}, lasso {\tt LASSO} \citep{tibshirani1996regression}, dynamic model averaging {\tt DMA} \citep{raftery_etal2010}, fused lasso {\tt fLASSO} \citep{tibshirani2005sparsity}, and dynamic spike-and-slab variants {\tt DVS} \citep{koop_korobilis_2020} and {\tt DSS} \citep{rockova_mcalinn_2021}. We also consider non-sparse shrinkage approaches: the horseshoe prior {\tt BHS} \citep{carvalho_etal.2010} and dynamic shrinkage process {\tt DSP} \citep{kowal2019dynamic}. All static methods employ rolling-window estimation with 100 observations.

We assess variable selection accuracy using the F1 score, which balances precision and recall as: $\text{F1} = \frac{2 \cdot \text{TP}}{2 \cdot \text{TP} + \text{FP} + \text{FN}}$, where $\text{TP}$, $\text{FP}$, and $\text{FN}$ denote counts of true positives, false positives, and false negatives, respectively. The metric ranges from 0 (no true positives) to 1 (perfect classification with zero false positives/negatives), providing a rigorous evaluation metric \citep[e.g.,][]{bernardi2024variational}.

\paragraph{Independent predictors.} We first consider the scenario with independently generated predictors $\{x_{jt}\}_{j=1}^p \sim \mathsf{N}(0,1)$ at each time point $t=1,\ldots,n$. In this setting, we implement the group factorization for {\tt BG group} using randomly assigned groups of ten variables. Figure~\ref{fig:res_tvp_sim} displays the aggregate F1 scores for dynamically sparse parameters ($\beta_{2:7,t}$) with $p\in\{10,100\}$. Section E.1 in the Supplemental Material reports the results for $p\in\{50,200\}$.

\begin{figure}[!ht]
\centering
\subfigure[F1-score ($p=10$)]{\includegraphics[width=0.48\textwidth]{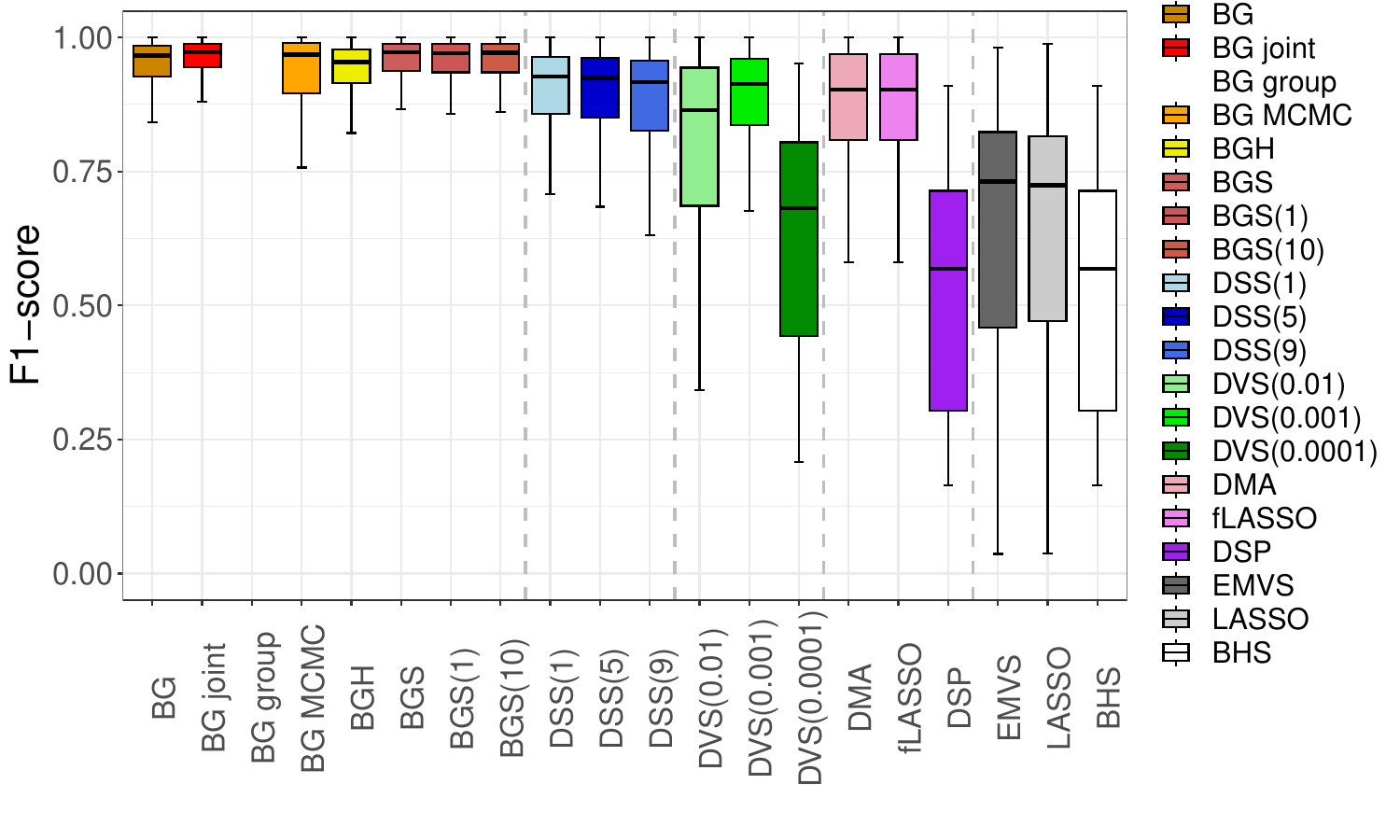}\label{fig:f1_tvp_10}}
\subfigure[F1-score ($p=100$)]{\includegraphics[width=0.48\textwidth]{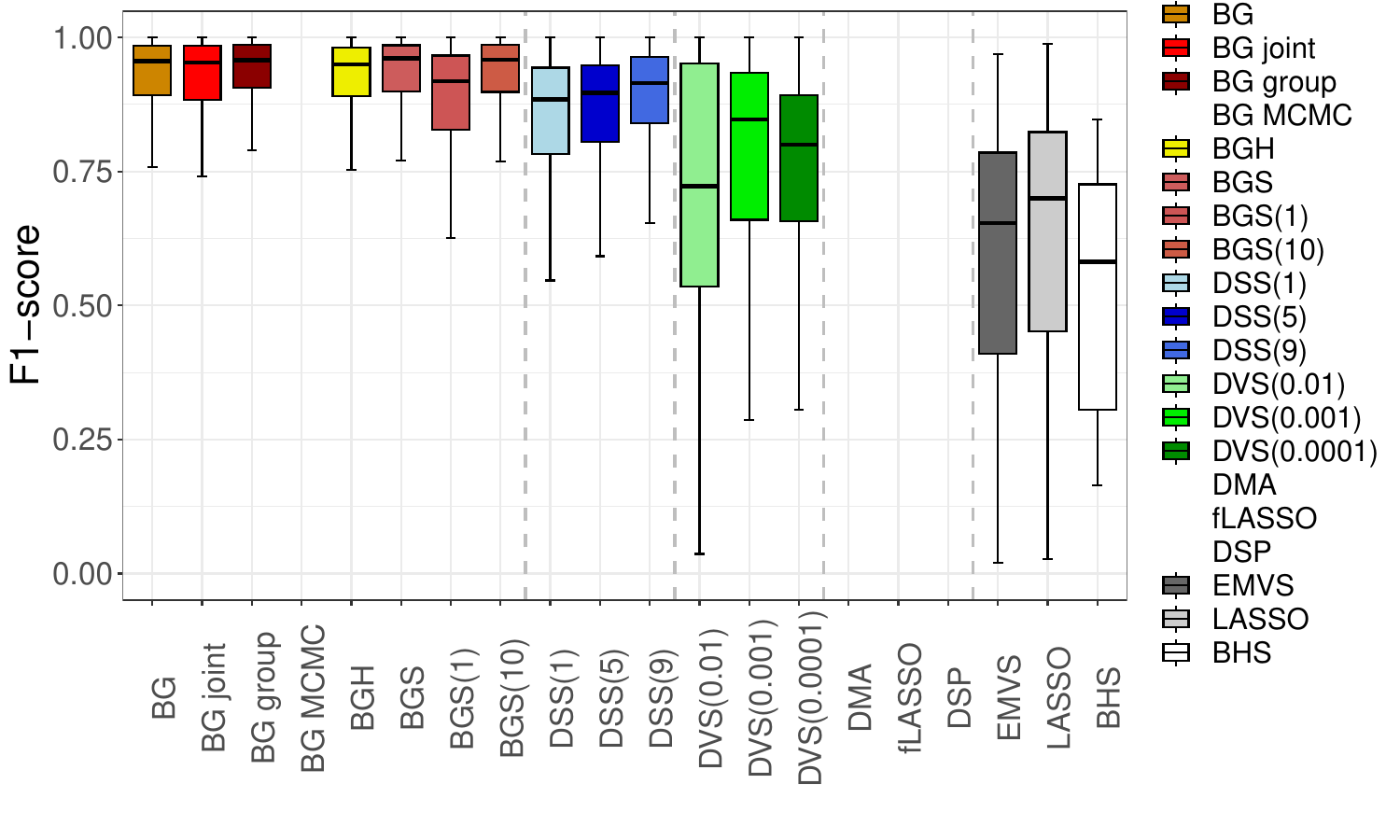}\label{fig:f1_tvp_100}}
\caption{Aggregate F1 scores for parameters $\beta_{2:7,t}$ exhibiting dynamic sparsity patterns. Results shown for independent predictors across dimensions $p=10,100$.}\label{fig:res_tvp_sim}
\end{figure}

The results reveal three key findings. First, static methods ({\tt LASSO}, {\tt EMVS}, {\tt fLASSO}) prove inadequate for time-varying sparsity. Second, the dynamic spike-and-slab prior \citep{koop_korobilis_2020} shows significant sensitivity to the variance ratio $\underline{c}\in\{0.01,0.001,0.0001\}$. In contrast, \cite{rockova_mcalinn_2021}'s method remains robust to varying marginal importance weights $\Theta\in\{0.1,0.5,0.9\}$, though both {\tt DVS} and {\tt DSS} degrade with increasing $p$. Third, while dynamic approaches {\tt DMA} and {\tt DSP} match newer methods ({\tt DSS}, {\tt DVS}) in low dimensions ($p\leq20$), they can no longer be estimated for $p>20$. 

Our variational Bayes approach outperforms all competitors, with advantages growing in higher dimensions. As expected in this uncorrelated setting, {\tt BG group} performs similarly to {\tt BG}.\footnote{For $p=10$, {\tt BG joint} and {\tt BG group} coincide; we report only {\tt BG group}.} Additional results in Section E.1 in the Supplementary Material confirm all methods handle consistently (non-)zero parameters adequately, highlighting our framework's superior detection of complex time-varying sparsity patterns.


\paragraph{Correlated predictors.} We now examine a scenario where predictors follow $\mathbf{x}_t \sim \mathsf{N}_p(\mathbf{0}_p,\mathbf{\Sigma}_x)$, with $\mathbf{\Sigma}_x$ estimated from $p\in\{10,50,100,200\}$ randomly selected macroeconomic series in our empirical application (Section~\ref{sec:appl}). This design replicates the correlation structures of the used macroeconomic data. For {\tt BG group}, we maintain random assignment to groups of ten variables, deliberately avoiding correlation-informed grouping to test robustness. Figure~\ref{fig:res_sim_corr} reports the results for $p\in\{10,100\}$ while the results for $p\in\{50,200\}$ are reported in Section E.1 in the Supplemental Material. 

\begin{figure}[!ht]
\subfigure[F1-score for $p=10$]{\includegraphics[width=.48\textwidth]{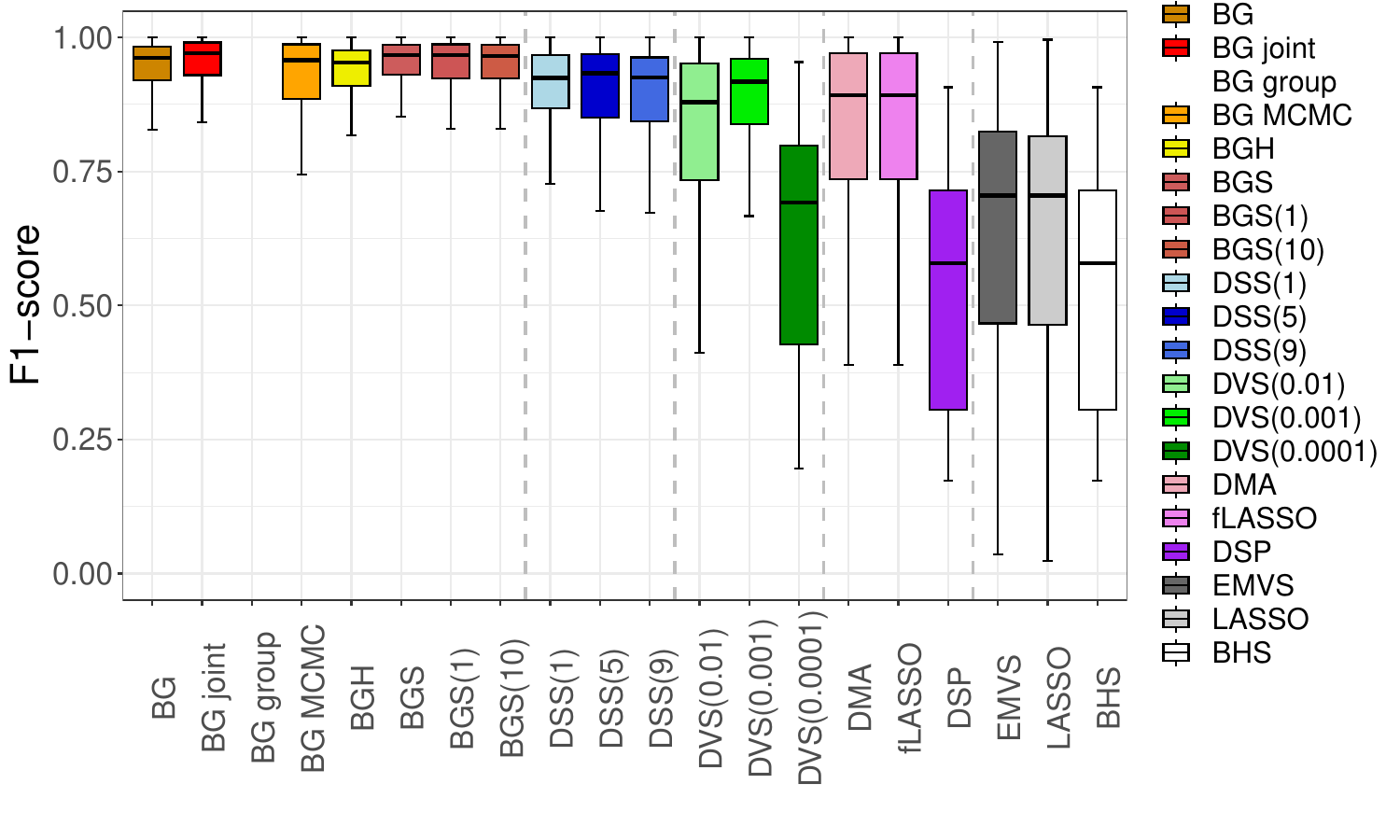}}
\subfigure[F1-score for $p=100$]{\includegraphics[width=.48\textwidth]{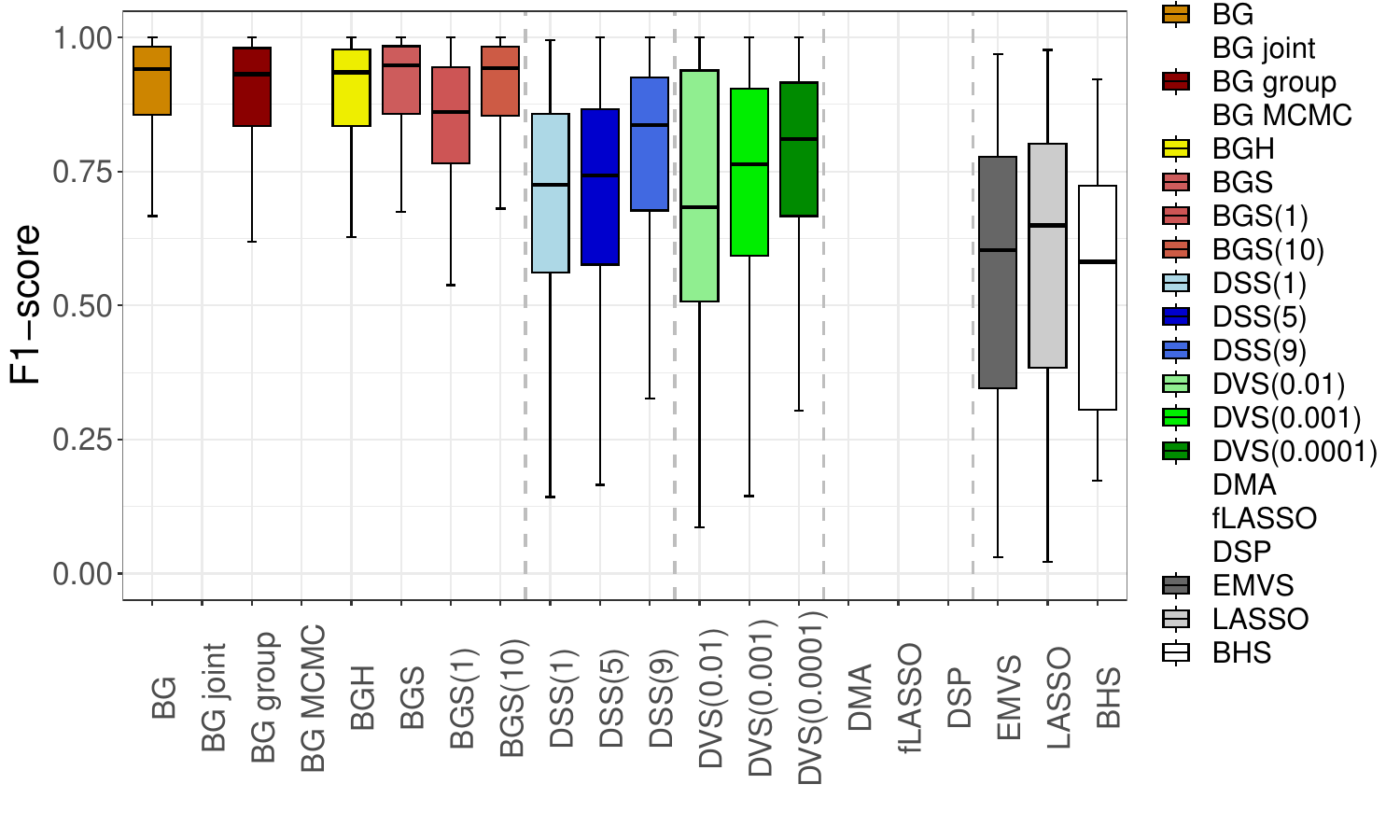}}
\caption{\small Aggregate F1 score for parameters $\beta_{2:7,t}$ exhibiting dynamic sparsity patterns over the sample period with correlated predictors. Results are reported for dimensions $p=10,100$.}\label{fig:res_sim_corr}
\end{figure}

The results reveal two key insights: First, correlation-induced selection bias reduces F1 scores compared to the independent case, with deterioration accelerating at $p>50$. Second, {\tt BG group} maintains comparable performance to {\tt BG} despite its agnostic grouping, suggesting robustness to imperfect group specifications. These results align with theoretical expectations about correlated covariates while demonstrating our method's practical applicability to macroeconomic data structures. 

\paragraph{Computational efficiency.} We assess computational efficiency for $p\in\{5,10,20,50,100,200\}$ with $n=200$ observations. Figure~\ref{fig:time} compares log-scale computation times (minutes) for key methods: {\tt BG}, {\tt BG group}, {\tt BG joint}, {\tt BG MCMC}, {\tt DSS(1)}, and {\tt DVS(0.01)}. Table E.2 in the Supplementary Material provides computational times for all remaining models.\footnote{We reimplemented all methods in {\bf Rcpp} when they were originally coded in different languages, and we executed all implementations on the same 2.5 GHz Intel Xeon W-2175 processor with 32GB of RAM.}

\begin{figure}[!ht]
\centering
\subfigure[Independent predictors]{\includegraphics[width=0.48\textwidth]{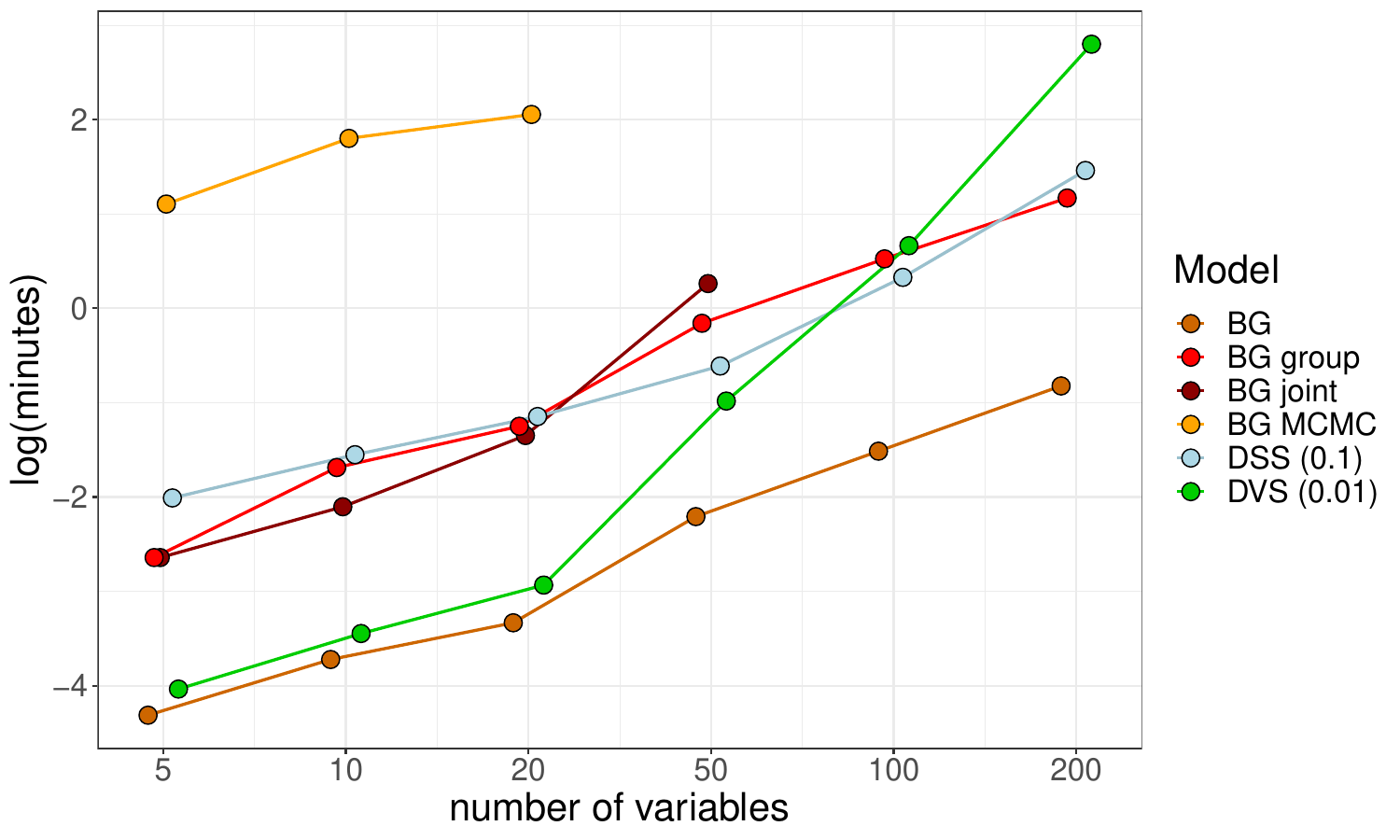}\label{fig:time_ind}}
\subfigure[Correlated predictors]{\includegraphics[width=0.48\textwidth]{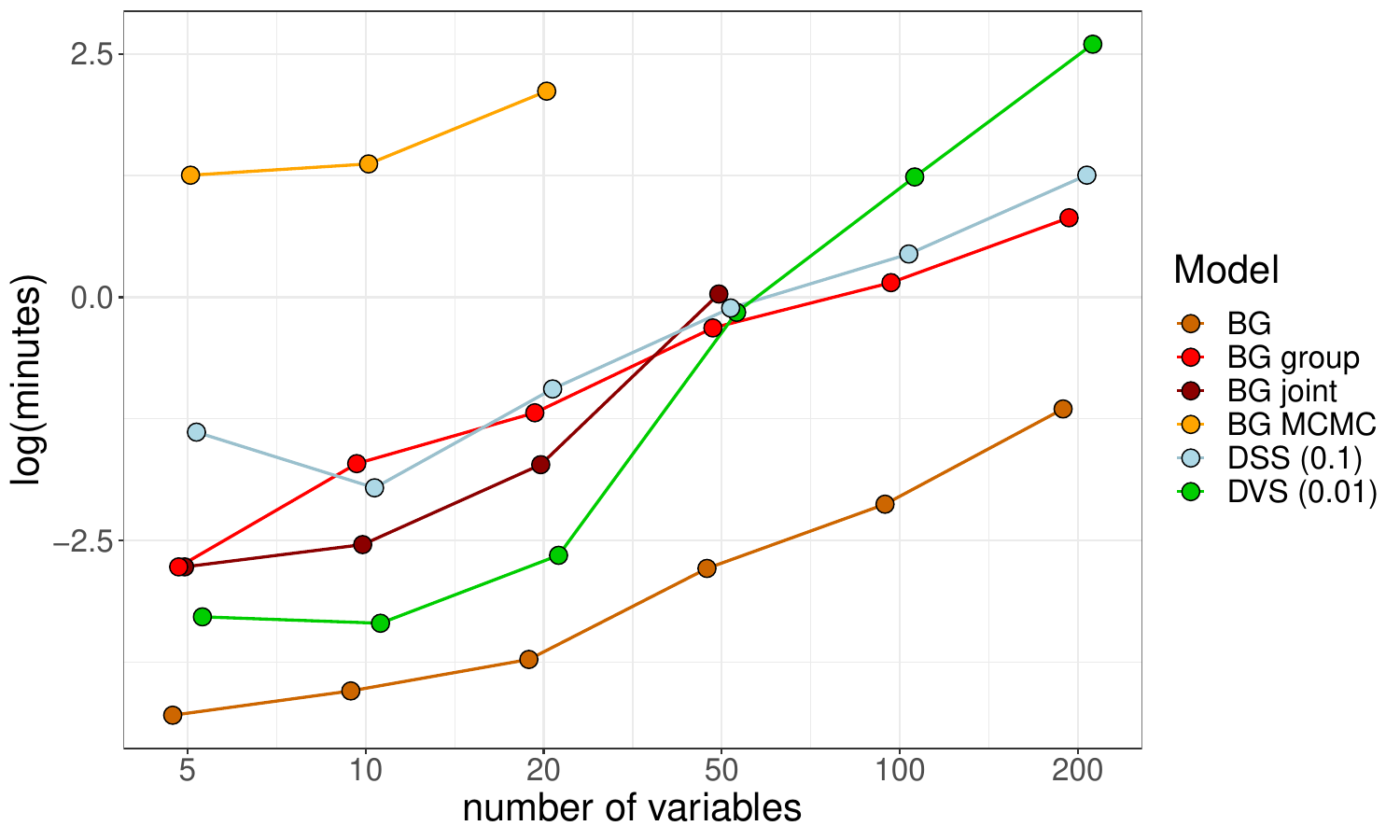}\label{fig:time_corr}}
\caption{Computation time (log minutes) for dynamic variable selection methods. Dashed lines indicate practical feasibility thresholds.}\label{fig:time}
\end{figure}

Three key findings emerge: First, while {\tt DVS} shows efficiency for $p<50$ due to its variational implementation, our {\tt BG} method achieves 5--10$\times$ speed gains at higher dimensions with independent predictors, and 100$\times$ gains with correlated predictors. Second, the fully correlated {\tt BG joint} becomes computationally prohibitive for $p>50$, whereas the group-structured {\tt BG group} maintains feasible computation costs, outperforming {\tt DVS} at $p=200$. Third, MCMC estimation requires 2.5--3 orders of magnitude more time than variational Bayes, becoming impractical for $p>20$. The results demonstrate our method's superior scalability, particularly in data-rich environments where correlation structures and high dimensionality pose dual challenges. 

\paragraph{Variational Bayes vs. MCMC.} 
The computational efficiency of variational Bayes (VB) comes with inherent tradeoffs in posterior approximation quality that merit careful consideration. Where MCMC samples from the exact joint posterior $p(\boldsymbol{\vartheta}|\mathbf{y})$, VB approximates it through an optimized factorized density $q(\boldsymbol{\vartheta})$. This formulation, particularly under mean-field assumptions, systematically produces under-dispersed posterior approximations \citep{Blei.2017}, a well-documented limitation of variational methods.

We evaluate approximation accuracy, we follow \cite{wand2011mean} and use an $L_1$-based metric $\mathcal{ACC}(\boldsymbol{\vartheta}) = [1 - 0.5 \int |q(\boldsymbol{\vartheta}) - p(\boldsymbol{\vartheta}|\mathbf{y})| d\boldsymbol{\vartheta}] \times100\%$, where $p(\boldsymbol{\vartheta}|\mathbf{y})$ represents the ground-truth MCMC posterior. Table \ref{tab:vs MCMC} reports the accuracy of $q^\ast\left(\beta_{jt}\right)$ with respect to its MCMC counterpart $p\left(\beta_{jt}|\mathbf{y}\right)$, aggregated by sparsity pattern type: $\beta_{1,t}$ (intercept), $\beta_{2,t}$ (one switch), $\beta_{3,t}$ (two switches), and $\beta_{4,t}$ (short-lived significance). The results show that VB achieves 84-92\% approximation accuracy under full factorization ($K=p$), with the joint specification ($K=1$) improving to 89-96\% across different coefficient types. This accuracy gradient confirms that preserving dependency structures in the variational family enhances posterior approximation quality.

\begin{table}[!ht]
\centering
\renewcommand{\arraystretch}{0.6}
\resizebox{1\textwidth}{!}{
    \begin{tabular}{lccccccccccccc}
    \toprule 
          & \multicolumn{6}{c}{Independent variables}               &       & \multicolumn{6}{c}{Correlated variables} \\
          \cmidrule{2-7}\cmidrule{9-14}
          & \multicolumn{2}{l}{$p=5$} & \multicolumn{2}{l}{$p=10$} & \multicolumn{2}{l}{$p=20$} &       & \multicolumn{2}{l}{$p=5$} & \multicolumn{2}{l}{$p=10$} & \multicolumn{2}{l}{$p=20$} \\
          \midrule 
          Parameter & {\tt BG} & {\tt BG joint} & {\tt BG} & {\tt BG joint}  & {\tt BG} & {\tt BG joint}  &       & {\tt BG} & {\tt BG joint}  & {\tt BG} & {\tt BG joint}  & {\tt BG} & {\tt BG joint}  \\
          \midrule 
    $\beta_{1,t}$ & 88.15 & 91.84 & 88.96 & 91.94 & 88.99 & 91.53 &       & 87.17 & 91.4  & 87.53 & 91.13 & 87.78 & 90.4 \\
    $\beta_{2,t}$ & 88.68 & 92.53 & 90.12 & 93.25 & 89.81 & 92.65 &       & 87.12 & 91.67 & 87.84 & 92.08 & 87.16 & 90.56 \\
    $\beta_{3,t}$ & 84.74 & 89.01 & 85.73 & 89.51 & 84.93 & 87.83 &       & 83.78 & 88.23 & 84.34 & 88.61 & 84.11 & 87.23 \\
    $\beta_{4,t}$ & 94.67 & 94.84 & 95.05 & 95.19 & 95.44 & 95.64 &       & 94.43 & 94.59 & 95.05 & 95.12 & 93.78 & 93.93 \\
    \bottomrule 
    \end{tabular}}
  \caption{\small The table shows the accuracy of $q^\ast\left(\beta_{jt}\right)$ with respect to its MCMC $p\left(\beta_{jt}|\mathbf{y}\right)$ counterpart. This is aggregated based on the type of sparsity dynamic; that is, $\beta_{1,t}$ (intercept), $\beta_{2,t}$ (one switch), $\beta_{3,t}$ (two switches), and $\beta_{4,t}$ (short-lived significance).}
    \label{tab:vs MCMC}%
\end{table}%

The practical consequences of these approximations become evident in Figure~\ref{fig:acc_tvp}, which compares posterior densities for a transitioning coefficient $\beta_{2t}$. The factorized VB approximation ($K=p$) produces markedly sharper transitions between inclusion/exclusion states while underestimating posterior variance during regime switches. In contrast, the joint specification ($K=1$) better preserves the uncertainty characteristics of the MCMC ground truth. This variance underestimation stems directly from Proposition~\ref{prop:q_beta_tilde_main}, where the factorized approximation tends to collapse toward a single mixture component.

\begin{figure}[h!]
\centering
\subfigure[Factorized ($K=p$)]{\includegraphics[width=0.45\textwidth]{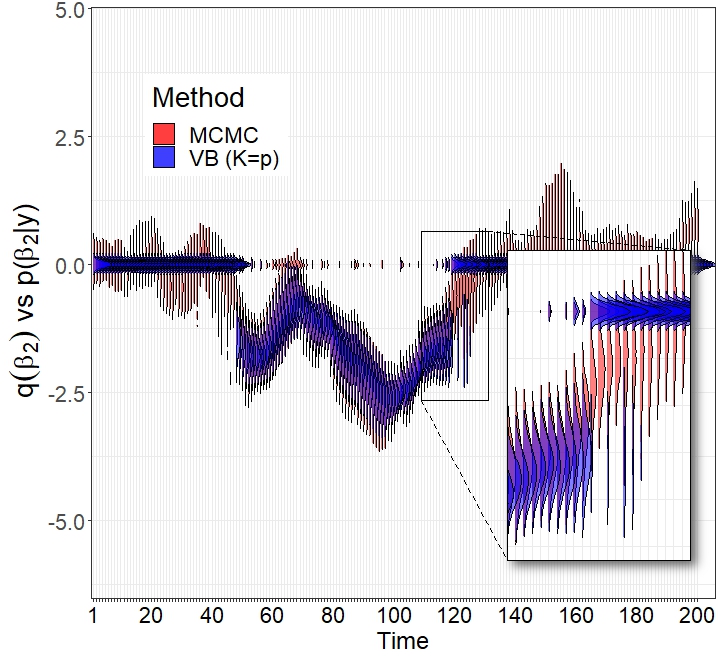}}
\subfigure[Joint ($K=1$)]{\includegraphics[width=0.45\textwidth]{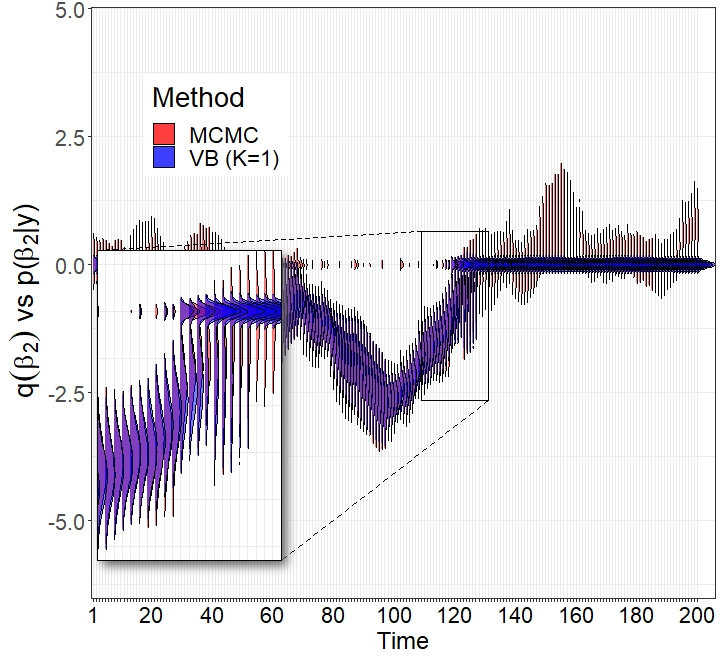}}
\caption{Posterior comparison (VB: blue, MCMC: red) for $\beta_{2t}$ during regime transitions, demonstrating VB's variance underestimation in (a) and improved uncertainty characterization in (b).}\label{fig:acc_tvp}
\end{figure}

Notably, these theoretical limitations have minimal operational impact on variable selection performance. The VB approximation preserves decision-critical information—particularly the inclusion probabilities $\mathbb{P}(\gamma_{jt}=1)$—explaining its comparable selection accuracy to MCMC in Figures~\ref{fig:res_tvp_sim}--\ref{fig:res_sim_corr}. This practical robustness, combined with VB's orders-of-magnitude speed advantage (see Figure \ref{fig:time}), establishes its clear superiority for high-dimensional applications. The method's scalability enables implementations where MCMC proves computationally prohibitive, particularly in recursive forecasting contexts requiring frequent model re-estimation. To summarize, while theoretical accuracy tradeoffs exist, VB's preserved decision quality and computational tractability make it the preferred approach for variable selection in data-rich environments.

\section{Empirical analysis of inflation predictability}
\label{sec:appl}

We assess the empirical performance of our dynamic variable selection approach through an inflation forecasting application, using a comprehensive set of macroeconomic predictors \citep{stock2006forecasting}. This context is particularly appropriate given the well-documented parameter instability in inflation dynamics, especially for out-of-sample forecasting \citep{huber2021inducing}.

Our dataset comprises 229 quarterly macroeconomic variables from the FRED-QD database \citep{mccracken2020fred}, spanning 1967:Q3 to 2022:Q2. Following standard practice, we transform all predictors to ensure stationarity. The variables are organized into 14 economically meaningful categories: (1) National Income and Product Accounts, (2) Industrial Production, (3) Employment and Unemployment, (4) Housing, (5) Inventories, Orders and Sales, (6) Prices, (7) Earnings and Productivity, (8) Interest Rates, (9) Money and Credit, (10) Household Balance Sheets, (11) Exchange Rates, (12) Other Indicators, (13) Stock Markets, and (14) Non-Household Balance Sheets. We supplement these with two lags of each response variable. The correlation structure of these predictors is analyzed in Section F.1 of the Supplementary Material. The categories inform our variational factorization from Equation~\eqref{eq:q_fact}, allowing us to incorporate economic structure into the estimation while maintaining computational efficiency. 

We examine four distinct inflation measures as response variables: total Consumer Price Index (CPIAUCSL), core Consumer Price Index (CPILFESL), GDP deflator (GDPCTPI), and Personal Consumption Expenditures deflator (PCECTPI). The parenthetical codes correspond to the original variable identifiers in the FRED-QD database. For each price series $P_t$, we compute $h$-quarter ahead annualized inflation rates as $y_{t+h} = (400/h)\ln(P_t/P_{t-1})$.

\paragraph{Out-of-sample forecasting performance.}We evaluate $h$-quarter ahead point and density forecasts using four specifications of our approach: the grouped dynamic Bernoulli-Gaussian model ({\tt BG group}), baseline independent coefficients ({\tt BG}), constant volatility variant ({\tt BGH}), and smoothed inclusion probability specification ({\tt BGS}). All models employ uninformative hyperparameters ($A_\nu=0.01,B_\nu=0.01$; $A_\eta=0.01,B_\eta=0.01$; $A_\xi=2,B_\xi=5$) following sensitivity analyses in Supplementary Material.

We compare against three classes of alternative approaches. Benchmark models include the local-level specification {\tt TVI} \citep{stock2007has}, standard {\tt AR(2)}, and time-varying AR(2) ({\tt TVAR(2)}) \citep{koop_korobilis_2020}. Both {\tt TVI} and {\tt TVAR(2)} account for stochastic volatility. Static selection methods comprise the horseshoe prior {\tt BHS}, spike-and-slab {\tt EMVS} \citep{rockova_george.2014}, Dirac specification {\tt GLP} \citep{Giannone:Lenza:Primiceri:2021}, and five-factor regression {\tt F5}. Dynamic alternatives consist of {\tt DVS} \citep{koop_korobilis_2020} with $\underline{c}=0.0001$ -- which simulation shows works comparatively well (see Figure \ref{fig:res_tvp_sim}) -- and {\tt DSS} \citep{rockova_mcalinn_2021} with $\Theta\in\{0.5,0.9\}$.

\begin{table}[!ht]
\centering
\renewcommand{\arraystretch}{0.6}
\resizebox{1\textwidth}{!}{
   \begin{tabular}{clrrrrrrrrrrrrrrrr}
   \toprule 
   & & \multicolumn{3}{c}{Benchmarks} & & \multicolumn{4}{c}{Static selection (rolling)} & & \multicolumn{7}{c}{Dynamic selection}\\
   \cmidrule{3-5}\cmidrule{7-10}\cmidrule{12-18}
    \multicolumn{1}{l}{Horizon} & Inflation & \multicolumn{1}{l}{{\tt TVI}} & \multicolumn{1}{l}{{\tt AR(2)}} & \multicolumn{1}{l}{{\tt TVAR(2)}} &       & \multicolumn{1}{l}{{\tt PCA(5)}} & \multicolumn{1}{l}{{\tt HS}} & \multicolumn{1}{l}{{\tt EMVS}} & \multicolumn{1}{l}{{\tt GLP}} &       & \multicolumn{1}{l}{{\tt BGS}} & \multicolumn{1}{l}{{\tt BG group}} & \multicolumn{1}{l}{{\tt BGH}} & \multicolumn{1}{l}{{\tt BG}} &     \multicolumn{1}{l}{{\tt DSS(5)}} & \multicolumn{1}{l}{{\tt DSS(9)}} & \multicolumn{1}{l}{{\tt DVS}} \\
    \midrule
  1     & CPIAUCSL & \cellcolor[rgb]{ .988,  .988,  1}2.38 & \cellcolor[rgb]{ .988,  .957,  .965}2.54 & \cellcolor[rgb]{ .988,  .969,  .98}2.48 &       & \cellcolor[rgb]{ .984,  .769,  .78}3.38 & \cellcolor[rgb]{ .984,  .788,  .8}3.29 & \cellcolor[rgb]{ .976,  .482,  .49}4.67 & \cellcolor[rgb]{ .984,  .824,  .835}3.13 &       & \cellcolor[rgb]{ .988,  .965,  .976}2.49 & \cellcolor[rgb]{ .988,  .98,  .992}2.43 & \cellcolor[rgb]{ .988,  .961,  .973}2.51 & \cellcolor[rgb]{ .988,  .91,  .918}2.75 & \cellcolor[rgb]{ .984,  .82,  .831}3.15 & \cellcolor[rgb]{ .984,  .82,  .831}3.15 & \cellcolor[rgb]{ .973,  .412,  .42}4.99 \\
    1     & CPILFESL & \cellcolor[rgb]{ .988,  .953,  .965}1.08 & \cellcolor[rgb]{ .988,  .973,  .984}1.02 & \cellcolor[rgb]{ .988,  .953,  .965}1.08 &       & \cellcolor[rgb]{ .973,  .412,  .42}2.75 & \cellcolor[rgb]{ .988,  .894,  .906}1.26 & \cellcolor[rgb]{ .988,  .902,  .914}1.23 & \cellcolor[rgb]{ .98,  .682,  .69}1.92 &       & \cellcolor[rgb]{ .988,  .988,  1}0.96 & \cellcolor[rgb]{ .988,  .988,  1}0.97 & \cellcolor[rgb]{ .988,  .98,  .992}0.99 & \cellcolor[rgb]{ .988,  .973,  .984}1.02 & \cellcolor[rgb]{ .98,  .604,  .612}2.16 & \cellcolor[rgb]{ .98,  .631,  .639}2.07 & \cellcolor[rgb]{ .988,  .91,  .922}1.21 \\
    1     & GDPCTPI & \cellcolor[rgb]{ .988,  .922,  .933}1.13 & \cellcolor[rgb]{ .988,  .965,  .976}1.06 & \cellcolor[rgb]{ .988,  .929,  .941}1.12 &       & \cellcolor[rgb]{ .976,  .475,  .482}1.91 & \cellcolor[rgb]{ .98,  .675,  .686}1.56 & \cellcolor[rgb]{ .984,  .835,  .843}1.29 & \cellcolor[rgb]{ .984,  .722,  .733}1.48 &       & \cellcolor[rgb]{ .988,  .984,  .996}1.03 & \cellcolor[rgb]{ .988,  .988,  1}1.01 & \cellcolor[rgb]{ .988,  .976,  .988}1.04 & \cellcolor[rgb]{ .988,  .988,  1}1.02 & \cellcolor[rgb]{ .973,  .412,  .42}2.01 & \cellcolor[rgb]{ .976,  .427,  .435}1.99 & \cellcolor[rgb]{ .98,  .612,  .62}1.67 \\
    1     & PCECTPI & \cellcolor[rgb]{ .988,  .984,  .996}1.75 & \cellcolor[rgb]{ .988,  .965,  .976}1.82 & \cellcolor[rgb]{ .988,  .973,  .984}1.78 &       & \cellcolor[rgb]{ .984,  .792,  .804}2.40 & \cellcolor[rgb]{ .984,  .792,  .8}2.41 & \cellcolor[rgb]{ .988,  .863,  .871}2.17 & \cellcolor[rgb]{ .988,  .871,  .882}2.14 &       & \cellcolor[rgb]{ .988,  .98,  .992}1.77 & \cellcolor[rgb]{ .988,  .988,  1}1.73 & \cellcolor[rgb]{ .988,  .957,  .969}1.84 & \cellcolor[rgb]{ .988,  .969,  .98}1.80 & \cellcolor[rgb]{ .984,  .776,  .784}2.47 & \cellcolor[rgb]{ .984,  .776,  .788}2.46 & \cellcolor[rgb]{ .973,  .412,  .42}3.71 \\
          \midrule
     2     & CPIAUCSL & \cellcolor[rgb]{ .988,  .969,  .98}1.65 & \cellcolor[rgb]{ .984,  .788,  .8}2.28 & \cellcolor[rgb]{ .988,  .929,  .941}1.79 &       & \cellcolor[rgb]{ .98,  .565,  .573}3.04 & \cellcolor[rgb]{ .976,  .49,  .498}3.30 & \cellcolor[rgb]{ .973,  .412,  .42}3.56 & \cellcolor[rgb]{ .98,  .69,  .702}2.61 &       & \cellcolor[rgb]{ .988,  .988,  1}1.59 & \cellcolor[rgb]{ .988,  .976,  .988}1.63 & \cellcolor[rgb]{ .988,  .957,  .969}1.71 & \cellcolor[rgb]{ .988,  .984,  .996}1.61 & \cellcolor[rgb]{ .98,  .698,  .71}2.58 & \cellcolor[rgb]{ .98,  .694,  .706}2.60 & \cellcolor[rgb]{ .98,  .647,  .659}2.76 \\
    2     & CPILFESL & \cellcolor[rgb]{ .988,  .925,  .933}0.78 & \cellcolor[rgb]{ .988,  .898,  .91}0.86 & \cellcolor[rgb]{ .988,  .878,  .89}0.92 &       & \cellcolor[rgb]{ .973,  .412,  .42}2.37 & \cellcolor[rgb]{ .984,  .761,  .773}1.28 & \cellcolor[rgb]{ .984,  .784,  .796}1.21 & \cellcolor[rgb]{ .984,  .71,  .718}1.45 &       & \cellcolor[rgb]{ .988,  .988,  1}0.57 & \cellcolor[rgb]{ .988,  .988,  1}0.58 & \cellcolor[rgb]{ .988,  .984,  .996}0.59 & \cellcolor[rgb]{ .988,  .988,  1}0.57 & \cellcolor[rgb]{ .98,  .58,  .588}1.86 & \cellcolor[rgb]{ .98,  .6,  .612}1.79 & \cellcolor[rgb]{ .988,  .906,  .918}0.84 \\
    2     & GDPCTPI & \cellcolor[rgb]{ .988,  .953,  .965}0.69 & \cellcolor[rgb]{ .984,  .8,  .812}0.99 & \cellcolor[rgb]{ .984,  .839,  .847}0.92 &       & \cellcolor[rgb]{ .976,  .416,  .424}1.76 & \cellcolor[rgb]{ .98,  .612,  .62}1.37 & \cellcolor[rgb]{ .98,  .698,  .71}1.19 & \cellcolor[rgb]{ .98,  .659,  .667}1.28 &       & \cellcolor[rgb]{ .988,  .988,  1}0.62 & \cellcolor[rgb]{ .988,  .969,  .98}0.66 & \cellcolor[rgb]{ .988,  .988,  1}0.61 & \cellcolor[rgb]{ .988,  .988,  1}0.62 & \cellcolor[rgb]{ .973,  .412,  .42}1.76 & \cellcolor[rgb]{ .976,  .475,  .482}1.64 & \cellcolor[rgb]{ .984,  .761,  .773}1.07 \\
    2     & PCECTPI & \cellcolor[rgb]{ .988,  .988,  1}1.18 & \cellcolor[rgb]{ .988,  .882,  .894}1.73 & \cellcolor[rgb]{ .988,  .945,  .957}1.41 &       & \cellcolor[rgb]{ .984,  .784,  .796}2.23 & \cellcolor[rgb]{ .984,  .737,  .745}2.48 & \cellcolor[rgb]{ .984,  .765,  .773}2.34 & \cellcolor[rgb]{ .988,  .859,  .871}1.85 &       & \cellcolor[rgb]{ .988,  .984,  .996}1.22 & \cellcolor[rgb]{ .988,  .984,  .996}1.22 & \cellcolor[rgb]{ .988,  .969,  .98}1.29 & \cellcolor[rgb]{ .988,  .984,  .996}1.21 & \cellcolor[rgb]{ .984,  .796,  .808}2.17 & \cellcolor[rgb]{ .984,  .8,  .808}2.16 & \cellcolor[rgb]{ .973,  .412,  .42}4.13 \\
                    \midrule
    4     & CPIAUCSL & \cellcolor[rgb]{ .988,  .988,  1}0.92 & \cellcolor[rgb]{ .984,  .761,  .773}1.83 & \cellcolor[rgb]{ .988,  .988,  1}0.92 &       & \cellcolor[rgb]{ .98,  .584,  .592}2.53 & \cellcolor[rgb]{ .973,  .412,  .42}3.21 & \cellcolor[rgb]{ .984,  .718,  .729}2.00 & \cellcolor[rgb]{ .984,  .702,  .714}2.06 &       & \cellcolor[rgb]{ .988,  .973,  .984}0.99 & \cellcolor[rgb]{ .988,  .984,  .996}0.94 & \cellcolor[rgb]{ .988,  .976,  .988}0.97 & \cellcolor[rgb]{ .988,  .953,  .965}1.07 & \cellcolor[rgb]{ .98,  .663,  .675}2.21 & \cellcolor[rgb]{ .98,  .604,  .612}2.46 & \cellcolor[rgb]{ .984,  .831,  .843}1.55 \\
    4     & CPILFESL & \cellcolor[rgb]{ .988,  .898,  .91}0.68 & \cellcolor[rgb]{ .984,  .812,  .824}0.96 & \cellcolor[rgb]{ .988,  .902,  .914}0.66 &       & \cellcolor[rgb]{ .973,  .412,  .42}2.29 & \cellcolor[rgb]{ .984,  .71,  .718}1.31 & \cellcolor[rgb]{ .984,  .78,  .792}1.07 & \cellcolor[rgb]{ .98,  .675,  .686}1.41 &       & \cellcolor[rgb]{ .988,  .988,  1}0.37 & \cellcolor[rgb]{ .988,  .988,  1}0.38 & \cellcolor[rgb]{ .988,  .969,  .98}0.44 & \cellcolor[rgb]{ .988,  .988,  1}0.37 & \cellcolor[rgb]{ .98,  .6,  .612}1.66 & \cellcolor[rgb]{ .98,  .667,  .675}1.44 & \cellcolor[rgb]{ .988,  .859,  .871}0.80 \\
    4     & GDPCTPI & \cellcolor[rgb]{ .988,  .929,  .941}0.53 & \cellcolor[rgb]{ .984,  .718,  .725}1.03 & \cellcolor[rgb]{ .984,  .808,  .82}0.82 &       & \cellcolor[rgb]{ .973,  .412,  .42}1.74 & \cellcolor[rgb]{ .976,  .545,  .557}1.43 & \cellcolor[rgb]{ .98,  .576,  .584}1.36 & \cellcolor[rgb]{ .976,  .545,  .553}1.43 &       & \cellcolor[rgb]{ .988,  .988,  1}0.39 & \cellcolor[rgb]{ .988,  .988,  1}0.39 & \cellcolor[rgb]{ .988,  .988,  1}0.39 & \cellcolor[rgb]{ .988,  .988,  1}0.39 & \cellcolor[rgb]{ .976,  .427,  .435}1.71 & \cellcolor[rgb]{ .976,  .459,  .467}1.63 & \cellcolor[rgb]{ .984,  .808,  .816}0.82 \\
    4     & PCECTPI & \cellcolor[rgb]{ .988,  .988,  1}0.69 & \cellcolor[rgb]{ .984,  .792,  .804}1.46 & \cellcolor[rgb]{ .988,  .988,  1}0.70 &       & \cellcolor[rgb]{ .98,  .655,  .663}2.00 & \cellcolor[rgb]{ .973,  .412,  .42}2.92 & \cellcolor[rgb]{ .98,  .686,  .694}1.87 & \cellcolor[rgb]{ .984,  .765,  .776}1.56 &       & \cellcolor[rgb]{ .988,  .98,  .992}0.74 & \cellcolor[rgb]{ .988,  .984,  .996}0.71 & \cellcolor[rgb]{ .988,  .984,  .996}0.72 & \cellcolor[rgb]{ .988,  .965,  .976}0.79 & \cellcolor[rgb]{ .98,  .671,  .678}1.93 & \cellcolor[rgb]{ .98,  .631,  .643}2.07 & \cellcolor[rgb]{ .988,  .914,  .925}0.99 \\
                    \midrule
    8     & CPIAUCSL & \cellcolor[rgb]{ .988,  .988,  1}0.43 & \cellcolor[rgb]{ .98,  .667,  .675}1.34 & \cellcolor[rgb]{ .988,  .984,  .996}0.45 &       & \cellcolor[rgb]{ .976,  .42,  .427}2.04 & \cellcolor[rgb]{ .976,  .455,  .463}1.94 & \cellcolor[rgb]{ .976,  .463,  .471}1.92 & \cellcolor[rgb]{ .973,  .412,  .42}2.05 &       & \cellcolor[rgb]{ .988,  .976,  .988}0.47 & \cellcolor[rgb]{ .988,  .976,  .988}0.47 & \cellcolor[rgb]{ .988,  .98,  .992}0.46 & \cellcolor[rgb]{ .988,  .976,  .988}0.47 & \cellcolor[rgb]{ .976,  .455,  .463}1.94 & \cellcolor[rgb]{ .976,  .459,  .467}1.93 & \cellcolor[rgb]{ .988,  .902,  .91}0.69 \\
    8     & CPILFESL & \cellcolor[rgb]{ .988,  .988,  1}0.19 & \cellcolor[rgb]{ .984,  .773,  .784}0.87 & \cellcolor[rgb]{ .988,  .961,  .969}0.28 &       & \cellcolor[rgb]{ .973,  .412,  .42}1.99 & \cellcolor[rgb]{ .98,  .678,  .686}1.17 & \cellcolor[rgb]{ .98,  .627,  .639}1.32 & \cellcolor[rgb]{ .98,  .682,  .69}1.15 &       & \cellcolor[rgb]{ .988,  .984,  .996}0.21 & \cellcolor[rgb]{ .988,  .988,  1}0.19 & \cellcolor[rgb]{ .988,  .98,  .992}0.21 & \cellcolor[rgb]{ .988,  .98,  .992}0.21 & \cellcolor[rgb]{ .976,  .549,  .557}1.57 & \cellcolor[rgb]{ .98,  .639,  .651}1.28 & \cellcolor[rgb]{ .988,  .945,  .957}0.33 \\
    8     & GDPCTPI & \cellcolor[rgb]{ .988,  .988,  1}0.26 & \cellcolor[rgb]{ .98,  .702,  .71}0.90 & \cellcolor[rgb]{ .988,  .949,  .961}0.34 &       & \cellcolor[rgb]{ .976,  .447,  .455}1.48 & \cellcolor[rgb]{ .98,  .612,  .62}1.11 & \cellcolor[rgb]{ .976,  .51,  .518}1.33 & \cellcolor[rgb]{ .976,  .541,  .549}1.26 &       & \cellcolor[rgb]{ .988,  .988,  1}0.26 & \cellcolor[rgb]{ .988,  .988,  1}0.26 & \cellcolor[rgb]{ .988,  .988,  1}0.26 & \cellcolor[rgb]{ .988,  .988,  1}0.26 & \cellcolor[rgb]{ .973,  .412,  .42}1.55 & \cellcolor[rgb]{ .976,  .529,  .537}1.29 & \cellcolor[rgb]{ .988,  .882,  .89}0.50 \\
    8     & PCECTPI & \cellcolor[rgb]{ .988,  .988,  1}0.32 & \cellcolor[rgb]{ .984,  .71,  .718}1.02 & \cellcolor[rgb]{ .988,  .969,  .98}0.38 &       & \cellcolor[rgb]{ .976,  .459,  .467}1.64 & \cellcolor[rgb]{ .976,  .522,  .529}1.49 & \cellcolor[rgb]{ .976,  .518,  .525}1.49 & \cellcolor[rgb]{ .976,  .537,  .549}1.44 &       & \cellcolor[rgb]{ .988,  .976,  .988}0.36 & \cellcolor[rgb]{ .988,  .98,  .992}0.35 & \cellcolor[rgb]{ .988,  .969,  .98}0.38 & \cellcolor[rgb]{ .988,  .969,  .98}0.37 & \cellcolor[rgb]{ .973,  .412,  .42}1.75 & \cellcolor[rgb]{ .976,  .447,  .455}1.67 & \cellcolor[rgb]{ .984,  .839,  .851}0.70 \\
    \bottomrule 
    \end{tabular}%
}
\caption{\small This table reports the out-of-sample root mean squared prediction error across models. The sample period is from 1967Q3 to 2022Q3. The first prediction is generated in 1997Q3.}
\label{tab:mse}
\end{table}

Table~\ref{tab:mse} reports root mean squared errors $\text{RMSE}_i = \sqrt{n^{-1}\sum_{t=1}^n \widehat{e}_{i,t}^2}$, where $\widehat{e}_{i,t}$ denotes the forecast error for model $i$ at time $t$. Our recursive evaluation spans 1997:Q3 to 2022:Q3 using 30-year rolling windows, with results color-coded from white (best) to dark red (worst) for visual comparison.

The results reveal several important patterns. First, the strong performance of {\tt TVI} confirms \cite{stock2007has}'s insight regarding persistent inflation components - particularly during the Great Moderation (1985-2007) when well-anchored inflation expectations diminished the direct influence of macroeconomic predictors. Second, all static selection methods underperform, suggesting inflation-predictor relationships significantly evolve across different economic regimes. Third, while dynamic alternatives {\tt DSS} and {\tt DVS} improve upon static approaches, our {\tt BGS} specification achieves superior forecasting accuracy with average RMSEs of 1.56, 1.00, 0.62, and 0.33 for horizons $h=1,2,4,8$ respectively. The minimal improvement from economic groupings ({\tt BG group} yields RMSEs of 1.65, 1.00, 0.65, and 0.33) suggests relevant inflation signals frequently cut across economic categories, often emerging simultaneously from diverse sectors during specific periods.



In Section F.2 of the  Supplementary Material, we assess the statistical significance of our variational Bayes approach's forecasting gains using Diebold-Mariano (DM) tests \citep{diebold2002}. For brevity, we report the results for the Total CPI (CPIAUCSL), the GDP deflator (GDPCTPI), and the PCE deflator (PCECTPI) for $h=1,2,4$ quarter ahead. The tests evaluate $\mathcal{H}_0: \text{MSE}^C = \text{MSE}^R$ against $\mathcal{H}_a: \text{MSE}^C > \text{MSE}^R$, where $\text{MSE}^C$ and $\text{MSE}^R$ represent the mean squared errors of column (comparison) and row (reference) models respectively. The results demonstrate that our {\tt BG}, {\tt BGS}, and {\tt BG group} specifications significantly outperform all other macroeconomic predictor-based approaches at conventional significance levels, while achieving point forecasts statistically indistinguishable from the {\tt TVI} specification. 

This finding substantiates two key insights: First, our dynamic variable selection framework successfully distills the predictive content from a large macroeconomic dataset without sacrificing the robustness of simpler benchmarks. Second, the maintained accuracy relative to {\tt TVI} confirms our method's ability to adaptively discard irrelevant predictors while preserving the core inflation dynamics captured by the local-level model.

\begin{table}[h!]
\centering
\renewcommand{\arraystretch}{0.6}
\resizebox{1\textwidth}{!}{
   \begin{tabular}{clrrrrrrrrrrrrrrrr}
   \toprule 
   & & \multicolumn{3}{c}{Benchmarks} & & \multicolumn{4}{c}{Static selection (rolling)} & & \multicolumn{7}{c}{Dynamic selection}\\
   \cmidrule{3-5}\cmidrule{7-10}\cmidrule{12-18}
    \multicolumn{1}{l}{Horizon} & Inflation & \multicolumn{1}{l}{{\tt TVI}} & \multicolumn{1}{l}{{\tt AR(2)}} & \multicolumn{1}{l}{{\tt TVAR(2)}} &       & \multicolumn{1}{l}{{\tt PCA(5)}} & \multicolumn{1}{l}{{\tt HS}} & \multicolumn{1}{l}{{\tt EMVS}} & \multicolumn{1}{l}{{\tt GLP}} &       & \multicolumn{1}{l}{{\tt BGS}} & \multicolumn{1}{l}{{\tt BG group}} & \multicolumn{1}{l}{{\tt BGH}} & \multicolumn{1}{l}{{\tt BG}} &     \multicolumn{1}{l}{{\tt DSS(5)}} & \multicolumn{1}{l}{{\tt DSS(9)}} & \multicolumn{1}{l}{{\tt DVS}} \\
    \midrule
    1     & CPIAUCSL & \cellcolor[rgb]{ .984,  .961,  .973}-2.5 & \cellcolor[rgb]{ .988,  .988,  1}-2.3 & \cellcolor[rgb]{ .984,  .98,  .992}-2.3 &       & \cellcolor[rgb]{ .984,  .941,  .953}-2.8 & \cellcolor[rgb]{ .984,  .918,  .929}-3.0 & \cellcolor[rgb]{ .98,  .804,  .816}-4.2 & \cellcolor[rgb]{ .984,  .914,  .925}-3.0 &       & \cellcolor[rgb]{ .984,  .953,  .965}-2.6 & \cellcolor[rgb]{ .984,  .953,  .965}-2.6 & \cellcolor[rgb]{ .984,  .976,  .988}-2.4 & \cellcolor[rgb]{ .984,  .953,  .965}-2.6 & \cellcolor[rgb]{ .98,  .804,  .812}-4.2 & \cellcolor[rgb]{ .98,  .804,  .816}-4.2 & \cellcolor[rgb]{ .973,  .412,  .42}-8.3 \\
    1     & CPILFESL & \cellcolor[rgb]{ .984,  .875,  .882}-1.6 & \cellcolor[rgb]{ .984,  .863,  .875}-1.7 & \cellcolor[rgb]{ .988,  .988,  1}-1.1 &       & \cellcolor[rgb]{ .976,  .686,  .694}-2.5 & \cellcolor[rgb]{ .984,  .882,  .89}-1.6 & \cellcolor[rgb]{ .98,  .702,  .714}-2.4 & \cellcolor[rgb]{ .98,  .745,  .757}-2.2 &       & \cellcolor[rgb]{ .984,  .933,  .941}-1.4 & \cellcolor[rgb]{ .984,  .933,  .945}-1.4 & \cellcolor[rgb]{ .984,  .937,  .949}-1.3 & \cellcolor[rgb]{ .984,  .875,  .886}-1.6 & \cellcolor[rgb]{ .973,  .412,  .42}-3.7 & \cellcolor[rgb]{ .973,  .435,  .443}-3.6 & \cellcolor[rgb]{ .976,  .667,  .675}-2.5 \\
    1     & GDPCTPI & \cellcolor[rgb]{ .984,  .878,  .89}-2.2 & \cellcolor[rgb]{ .984,  .965,  .976}-1.7 & \cellcolor[rgb]{ .988,  .988,  1}-1.5 &       & \cellcolor[rgb]{ .98,  .824,  .835}-2.6 & \cellcolor[rgb]{ .984,  .902,  .914}-2.1 & \cellcolor[rgb]{ .98,  .835,  .847}-2.5 & \cellcolor[rgb]{ .984,  .961,  .973}-1.7 &       & \cellcolor[rgb]{ .984,  .953,  .965}-1.7 & \cellcolor[rgb]{ .984,  .937,  .949}-1.9 & \cellcolor[rgb]{ .984,  .969,  .98}-1.6 & \cellcolor[rgb]{ .984,  .953,  .965}-1.7 & \cellcolor[rgb]{ .976,  .616,  .624}-4.0 & \cellcolor[rgb]{ .976,  .647,  .659}-3.8 & \cellcolor[rgb]{ .973,  .412,  .42}-5.4 \\
    1     & PCECTPI & \cellcolor[rgb]{ .984,  .933,  .945}-2.4 & \cellcolor[rgb]{ .988,  .988,  1}-2.0 & \cellcolor[rgb]{ .984,  .98,  .992}-2.0 &       & \cellcolor[rgb]{ .984,  .925,  .937}-2.5 & \cellcolor[rgb]{ .984,  .925,  .937}-2.5 & \cellcolor[rgb]{ .984,  .89,  .902}-2.8 & \cellcolor[rgb]{ .984,  .922,  .933}-2.5 &       & \cellcolor[rgb]{ .984,  .965,  .976}-2.1 & \cellcolor[rgb]{ .984,  .961,  .973}-2.2 & \cellcolor[rgb]{ .984,  .965,  .976}-2.1 & \cellcolor[rgb]{ .984,  .945,  .957}-2.3 & \cellcolor[rgb]{ .98,  .733,  .745}-4.1 & \cellcolor[rgb]{ .98,  .725,  .733}-4.2 & \cellcolor[rgb]{ .973,  .412,  .42}-6.8 \\
          \midrule 
    2     & CPIAUCSL & \cellcolor[rgb]{ .984,  .851,  .863}-3.1 & \cellcolor[rgb]{ .984,  .973,  .984}-2.3 & \cellcolor[rgb]{ .984,  .98,  .992}-2.3 &       & \cellcolor[rgb]{ .984,  .894,  .902}-2.8 & \cellcolor[rgb]{ .984,  .875,  .886}-3.0 & \cellcolor[rgb]{ .98,  .757,  .765}-3.7 & \cellcolor[rgb]{ .984,  .859,  .871}-3.0 &       & \cellcolor[rgb]{ .988,  .988,  1}-2.2 & \cellcolor[rgb]{ .984,  .98,  .992}-2.2 & \cellcolor[rgb]{ .984,  .98,  .992}-2.3 & \cellcolor[rgb]{ .984,  .98,  .992}-2.3 & \cellcolor[rgb]{ .976,  .69,  .702}-4.2 & \cellcolor[rgb]{ .976,  .675,  .686}-4.3 & \cellcolor[rgb]{ .973,  .412,  .42}-6.0 \\
    2     & CPILFESL & \cellcolor[rgb]{ .973,  .412,  .42}-6.2 & \cellcolor[rgb]{ .984,  .894,  .906}-1.7 & \cellcolor[rgb]{ .988,  .988,  1}-0.8 &       & \cellcolor[rgb]{ .98,  .824,  .835}-2.3 & \cellcolor[rgb]{ .984,  .871,  .878}-1.9 & \cellcolor[rgb]{ .98,  .835,  .847}-2.2 & \cellcolor[rgb]{ .984,  .867,  .875}-1.9 &       & \cellcolor[rgb]{ .984,  .98,  .992}-0.9 & \cellcolor[rgb]{ .984,  .957,  .969}-1.1 & \cellcolor[rgb]{ .984,  .937,  .949}-1.3 & \cellcolor[rgb]{ .984,  .973,  .984}-0.9 & \cellcolor[rgb]{ .98,  .722,  .733}-3.3 & \cellcolor[rgb]{ .98,  .737,  .745}-3.2 & \cellcolor[rgb]{ .984,  .867,  .878}-1.9 \\
    2     & GDPCTPI & \cellcolor[rgb]{ .973,  .412,  .42}-5.3 & \cellcolor[rgb]{ .984,  .922,  .933}-1.6 & \cellcolor[rgb]{ .984,  .969,  .98}-1.3 &       & \cellcolor[rgb]{ .98,  .796,  .804}-2.6 & \cellcolor[rgb]{ .98,  .816,  .824}-2.4 & \cellcolor[rgb]{ .98,  .8,  .808}-2.5 & \cellcolor[rgb]{ .984,  .898,  .91}-1.8 &       & \cellcolor[rgb]{ .988,  .988,  1}-1.2 & \cellcolor[rgb]{ .984,  .961,  .973}-1.4 & \cellcolor[rgb]{ .984,  .976,  .988}-1.3 & \cellcolor[rgb]{ .984,  .98,  .992}-1.2 & \cellcolor[rgb]{ .976,  .667,  .678}-3.5 & \cellcolor[rgb]{ .98,  .733,  .745}-3.0 & \cellcolor[rgb]{ .976,  .655,  .663}-3.6 \\
    2     & PCECTPI & \cellcolor[rgb]{ .98,  .729,  .737}-3.7 & \cellcolor[rgb]{ .984,  .957,  .969}-2.1 & \cellcolor[rgb]{ .984,  .984,  .996}-1.9 &       & \cellcolor[rgb]{ .984,  .878,  .89}-2.6 & \cellcolor[rgb]{ .984,  .867,  .878}-2.7 & \cellcolor[rgb]{ .98,  .816,  .827}-3.0 & \cellcolor[rgb]{ .984,  .902,  .914}-2.5 &       & \cellcolor[rgb]{ .988,  .988,  1}-1.9 & \cellcolor[rgb]{ .984,  .98,  .992}-1.9 & \cellcolor[rgb]{ .984,  .953,  .965}-2.1 & \cellcolor[rgb]{ .984,  .969,  .98}-2.0 & \cellcolor[rgb]{ .976,  .671,  .682}-4.1 & \cellcolor[rgb]{ .976,  .671,  .682}-4.1 & \cellcolor[rgb]{ .973,  .412,  .42}-5.9 \\
                    \midrule 
    4     & CPIAUCSL & \cellcolor[rgb]{ .973,  .467,  .478}-4.0 & \cellcolor[rgb]{ .984,  .898,  .91}-2.2 & \cellcolor[rgb]{ .984,  .925,  .937}-2.1 &       & \cellcolor[rgb]{ .98,  .804,  .816}-2.6 & \cellcolor[rgb]{ .98,  .729,  .741}-2.9 & \cellcolor[rgb]{ .976,  .667,  .675}-3.2 & \cellcolor[rgb]{ .98,  .784,  .796}-2.7 &       & \cellcolor[rgb]{ .984,  .984,  .996}-1.9 & \cellcolor[rgb]{ .984,  .973,  .984}-1.9 & \cellcolor[rgb]{ .988,  .988,  1}-1.9 & \cellcolor[rgb]{ .984,  .976,  .988}-1.9 & \cellcolor[rgb]{ .973,  .459,  .467}-4.0 & \cellcolor[rgb]{ .973,  .412,  .42}-4.2 & \cellcolor[rgb]{ .973,  .455,  .463}-4.1 \\
    4     & CPILFESL & \cellcolor[rgb]{ .973,  .412,  .42}-5.0 & \cellcolor[rgb]{ .98,  .824,  .835}-1.8 & \cellcolor[rgb]{ .988,  .988,  1}-0.5 &       & \cellcolor[rgb]{ .98,  .769,  .78}-2.2 & \cellcolor[rgb]{ .98,  .792,  .804}-2.0 & \cellcolor[rgb]{ .98,  .835,  .847}-1.7 & \cellcolor[rgb]{ .98,  .784,  .796}-2.1 &       & \cellcolor[rgb]{ .984,  .98,  .992}-0.5 & \cellcolor[rgb]{ .984,  .941,  .953}-0.9 & \cellcolor[rgb]{ .984,  .902,  .914}-1.2 & \cellcolor[rgb]{ .984,  .98,  .992}-0.6 & \cellcolor[rgb]{ .976,  .686,  .698}-2.9 & \cellcolor[rgb]{ .98,  .753,  .761}-2.4 & \cellcolor[rgb]{ .984,  .859,  .871}-1.5 \\
    4     & GDPCTPI & \cellcolor[rgb]{ .973,  .412,  .42}-7.3 & \cellcolor[rgb]{ .984,  .91,  .922}-1.7 & \cellcolor[rgb]{ .988,  .988,  1}-0.8 &       & \cellcolor[rgb]{ .984,  .851,  .863}-2.3 & \cellcolor[rgb]{ .984,  .867,  .878}-2.2 & \cellcolor[rgb]{ .98,  .827,  .839}-2.6 & \cellcolor[rgb]{ .984,  .875,  .882}-2.1 &       & \cellcolor[rgb]{ .984,  .984,  .996}-0.8 & \cellcolor[rgb]{ .984,  .969,  .98}-1.0 & \cellcolor[rgb]{ .984,  .957,  .969}-1.2 & \cellcolor[rgb]{ .984,  .984,  .996}-0.9 & \cellcolor[rgb]{ .98,  .765,  .776}-3.3 & \cellcolor[rgb]{ .98,  .78,  .788}-3.1 & \cellcolor[rgb]{ .984,  .871,  .878}-2.1 \\
    4     & PCECTPI & \cellcolor[rgb]{ .973,  .412,  .42}-4.9 & \cellcolor[rgb]{ .984,  .898,  .906}-2.0 & \cellcolor[rgb]{ .988,  .988,  1}-1.5 &       & \cellcolor[rgb]{ .98,  .82,  .831}-2.5 & \cellcolor[rgb]{ .98,  .8,  .812}-2.6 & \cellcolor[rgb]{ .98,  .741,  .749}-2.9 & \cellcolor[rgb]{ .98,  .839,  .851}-2.4 &       & \cellcolor[rgb]{ .984,  .961,  .973}-1.7 & \cellcolor[rgb]{ .984,  .969,  .98}-1.6 & \cellcolor[rgb]{ .984,  .957,  .969}-1.7 & \cellcolor[rgb]{ .984,  .98,  .992}-1.5 & \cellcolor[rgb]{ .976,  .573,  .58}-3.9 & \cellcolor[rgb]{ .973,  .518,  .525}-4.2 & \cellcolor[rgb]{ .98,  .71,  .718}-3.1 \\
                    \midrule 
    8     & CPIAUCSL & \cellcolor[rgb]{ .973,  .412,  .42}-7.5 & \cellcolor[rgb]{ .984,  .906,  .918}-1.9 & \cellcolor[rgb]{ .988,  .988,  1}-0.9 &       & \cellcolor[rgb]{ .984,  .875,  .886}-2.2 & \cellcolor[rgb]{ .984,  .855,  .867}-2.4 & \cellcolor[rgb]{ .98,  .757,  .769}-3.6 & \cellcolor[rgb]{ .984,  .875,  .886}-2.2 &       & \cellcolor[rgb]{ .984,  .965,  .973}-1.2 & \cellcolor[rgb]{ .984,  .957,  .969}-1.3 & \cellcolor[rgb]{ .984,  .945,  .957}-1.4 & \cellcolor[rgb]{ .984,  .961,  .973}-1.2 & \cellcolor[rgb]{ .98,  .761,  .773}-3.5 & \cellcolor[rgb]{ .98,  .757,  .765}-3.6 & \cellcolor[rgb]{ .984,  .906,  .918}-1.9 \\
    8     & CPILFESL & \cellcolor[rgb]{ .973,  .42,  .427}-2.9 & \cellcolor[rgb]{ .98,  .714,  .722}-1.4 & \cellcolor[rgb]{ .988,  .988,  1}-0.1 &       & \cellcolor[rgb]{ .976,  .635,  .647}-1.8 & \cellcolor[rgb]{ .976,  .639,  .647}-1.8 & \cellcolor[rgb]{ .976,  .612,  .62}-1.9 & \cellcolor[rgb]{ .976,  .678,  .686}-1.6 &       & \cellcolor[rgb]{ .98,  .824,  .835}-0.9 & \cellcolor[rgb]{ .984,  .863,  .871}-0.7 & \cellcolor[rgb]{ .984,  .847,  .859}-0.8 & \cellcolor[rgb]{ .984,  .882,  .894}-0.6 & \cellcolor[rgb]{ .973,  .412,  .42}-2.9 & \cellcolor[rgb]{ .976,  .569,  .58}-2.1 & \cellcolor[rgb]{ .984,  .914,  .925}-0.4 \\
    8     & GDPCTPI & \cellcolor[rgb]{ .973,  .412,  .42}-7.7 & \cellcolor[rgb]{ .984,  .902,  .914}-1.6 & \cellcolor[rgb]{ .988,  .988,  1}-0.5 &       & \cellcolor[rgb]{ .984,  .882,  .89}-1.8 & \cellcolor[rgb]{ .984,  .878,  .886}-1.9 & \cellcolor[rgb]{ .98,  .816,  .827}-2.6 & \cellcolor[rgb]{ .984,  .855,  .867}-2.1 &       & \cellcolor[rgb]{ .984,  .941,  .953}-1.1 & \cellcolor[rgb]{ .984,  .941,  .953}-1.1 & \cellcolor[rgb]{ .984,  .914,  .925}-1.4 & \cellcolor[rgb]{ .984,  .933,  .945}-1.2 & \cellcolor[rgb]{ .98,  .792,  .804}-2.9 & \cellcolor[rgb]{ .984,  .851,  .859}-2.2 & \cellcolor[rgb]{ .984,  .918,  .929}-1.4 \\
    8     & PCECTPI & \cellcolor[rgb]{ .973,  .412,  .42}-7.3 & \cellcolor[rgb]{ .984,  .906,  .918}-1.6 & \cellcolor[rgb]{ .988,  .988,  1}-0.7 &       & \cellcolor[rgb]{ .984,  .867,  .878}-2.0 & \cellcolor[rgb]{ .984,  .859,  .867}-2.2 & \cellcolor[rgb]{ .98,  .769,  .78}-3.2 & \cellcolor[rgb]{ .984,  .871,  .878}-2.0 &       & \cellcolor[rgb]{ .984,  .961,  .973}-1.0 & \cellcolor[rgb]{ .984,  .957,  .969}-1.0 & \cellcolor[rgb]{ .984,  .933,  .945}-1.3 & \cellcolor[rgb]{ .984,  .953,  .965}-1.1 & \cellcolor[rgb]{ .98,  .741,  .749}-3.5 & \cellcolor[rgb]{ .98,  .761,  .773}-3.2 & \cellcolor[rgb]{ .984,  .871,  .882}-2.0 \\
    \bottomrule 
    \end{tabular}%
}
\caption{\small This table reports the out-of-sample log-predictive score calculated across models. The sample period is from 1967Q3 to 2022Q3. The first prediction is generated in 1997Q3.}
\label{tab:lscore}
\end{table}

Table~\ref{tab:lscore} presents log-predictive scores $LS_i = n^{-1}\sum_{t=1}^n \log(S_{i,t})$, where $\log(S_{i,t})$ denotes the time-$t$ predictive likelihood for model $i$. Higher scores indicate better density forecasting performance. Three key patterns emerge from these results.

First, the strong performance of {\tt TVAR(2)} for density forecasting, particularly at longer horizons, highlights the critical role of properly characterizing inflation uncertainty. While {\tt TVI} produces accurate point forecasts, its density forecasts prove excessively concentrated around the conditional mean, systematically underestimating true inflation uncertainty.

Second, rolling-window implementations of static shrinkage and selection methods ({\tt BHS}, {\tt EMVS}, {\tt GLP}) outperform dynamic alternatives {\tt DSS} and {\tt DVS} at shorter horizons ($h\leq4$), suggesting that adaptive methods may initially overfit when estimating predictive densities.

Most significantly, our {\tt BG}, {\tt BGS}, and {\tt BG group} specifications demonstrate superior density forecasting performance across all horizons, with the advantage growing substantially at longer forecasts ($h=8$). This robust performance indicates our approach successfully captures both gradual inflation dynamics and abrupt volatility regime changes. The consistent outperformance across all four inflation measures further confirms the method's ability to adapt to different conditions while maintaining well-calibrated uncertainty estimates.

\paragraph{Retrospective analysis of inflation drivers}
Our dynamic variable selection framework provides both accurate forecasts and interpretable estimates of time-varying regression parameters $\beta_{j,t}$ that reveal the macroeconomic drivers of inflation. Focusing on one-quarter ahead forecasts, Figure~\ref{fig:beta_appl1} presents coefficient estimates $\mu_{q(\beta_{jt})}$ and inclusion probabilities $\mu_{q(\gamma_{jt})}$ from the {\tt BGS} model, showing only predictors active for significant periods. For the sake of brevity, we report the results for Total CPI and GDP deflator. The remaining plots for PCE deflator and Core CPI are available in Section F3 of the Supplementary Material. 

\begin{figure}[ht]
\centering
\hspace{-1em}\subfigure[Total CPI (CPIAUCSL)]{\includegraphics[width=0.48\textwidth]{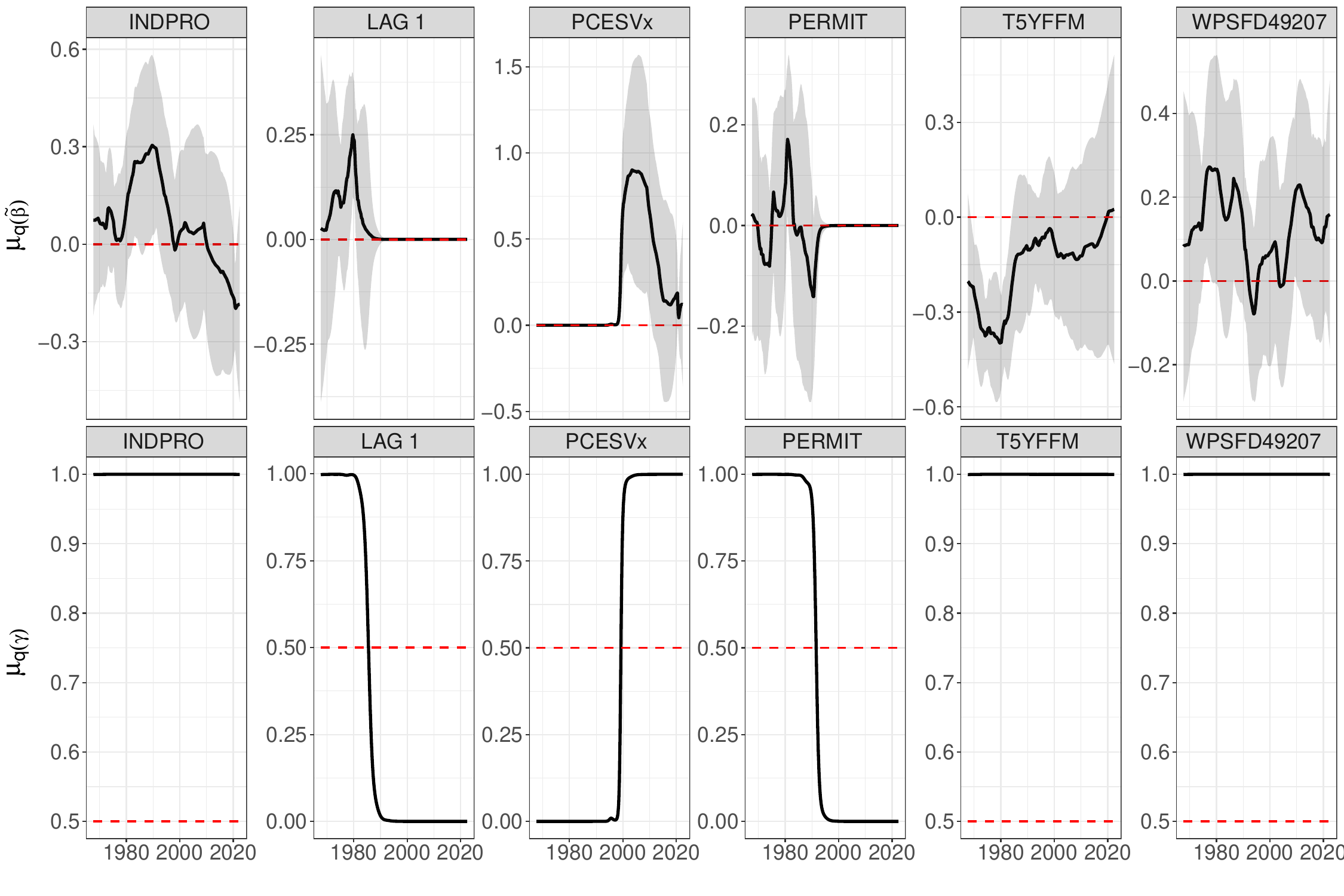}}\hspace{1em}\subfigure[GDP deflator (GDPCTPI)]{\includegraphics[width=0.48\textwidth]{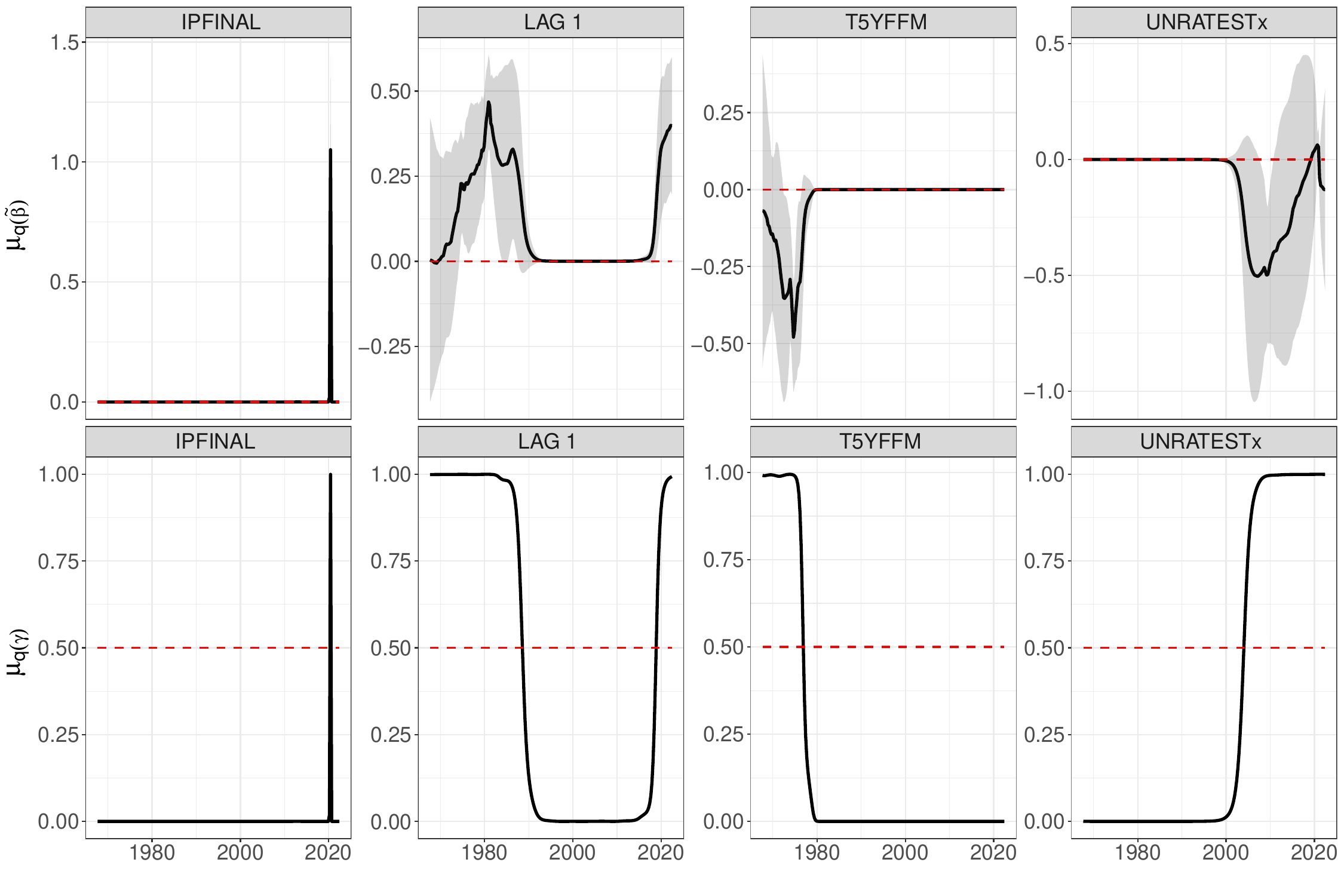}}
\caption{Time-varying coefficient estimates and inclusion probabilities for four inflation measures.}\label{fig:beta_appl1}
\end{figure}

The analysis yields several important findings about inflation dynamics. First, the sparse sets of active predictors support the unobserved components view advanced by \cite{stock2007has} and \cite{harvey2007trends}, where inflation variation primarily reflects latent local trends rather than observable economic drivers. This sparsity is particularly evident during the Great Moderation period from the mid-1980s through 2007, when inflation expectations became well-anchored.

Second, while certain predictors like industrial production (INDPRO), 5-year Treasury rates (T5YFFM), and producer prices (WPSFD49207) appear consistently across different inflation measures, their timing and magnitude of influence show important variations. Lagged inflation emerges as particularly significant during the 1970s oil shocks and Volcker disinflation period.

Third, the estimates reveal distinct regimes in monetary policy transmission. While 5-year Treasury rates maintain negative predictive relationships with CPI and PCE inflation throughout the sample, they lose predictive power for GDP deflator inflation following the financial deregulation of the early 1980s. Simultaneously, supply-side factors like industrial production show persistent positive effects, supporting cost-push theories of inflation.

Fourth, the results provide nuanced evidence about traditional inflation relationships. Short-term unemployment gains predictive power for GDP deflator following the 2008 financial crisis, offering conditional support for the Phillips curve that aligns with recent evidence of its nonlinear, slack-dependent nature \citep{blanchard2016phillips}. The framework also successfully captures transient shocks, such as the 2020 spikes in industrial production (IPFINAL) and consumption (PCECC96) that anticipated the subsequent inflation surge.



\section{Concluding remarks}
\label{sec:concl}

We develop a variational Bayes method for dynamic variable selection in high-dimensional time-varying parameter regressions. The approach delivers significant computational efficiency gains over MCMC while maintaining accuracy. In inflation forecasting applications with 220+ macroeconomic predictors, our method outperforms existing alternatives and provides interpretable time-varying estimates of inflation drivers. The results highlight the importance of dynamic selection when modeling economic time series.

\bibliographystyle{chicago}
\spacingset{1}
\bibliography{BGTVP_biblio}

\begin{thebibliography}{}

\bibitem[\protect\citeauthoryear{Bernardi, Bianchi, and Bianco}{Bernardi et~al.}{2024}]{bernardi2024variational}
Bernardi, M., D.~Bianchi, and N.~Bianco (2024).
\newblock Variational inference for large {B}ayesian vector autoregressions.
\newblock {\em Journal of Business \& Economic Statistics\/}~{\em 42\/}(3), 1066--1082.

\bibitem[\protect\citeauthoryear{Blanchard}{Blanchard}{2016}]{blanchard2016phillips}
Blanchard, O. (2016).
\newblock The phillips curve: back to the'60s?
\newblock {\em American Economic Review\/}~{\em 106\/}(5), 31--34.

\bibitem[\protect\citeauthoryear{Blei, Kucukelbir, and McAuliffe}{Blei et~al.}{2017}]{Blei.2017}
Blei, D.~M., A.~Kucukelbir, and J.~D. McAuliffe (2017).
\newblock Variational inference: A review for statisticians.
\newblock {\em Journal of the American Statistical Association\/}~{\em 112\/}(518), 859--877.

\bibitem[\protect\citeauthoryear{Carvalho, Polson, and Scott}{Carvalho et~al.}{2010}]{carvalho_etal.2010}
Carvalho, C.~M., N.~G. Polson, and J.~G. Scott (2010).
\newblock The horseshoe estimator for sparse signals.
\newblock {\em Biometrika\/}~{\em 97\/}(2), 465--480.

\bibitem[\protect\citeauthoryear{Diebold and Mariano}{Diebold and Mariano}{1995}]{diebold2002}
Diebold, F.~X. and R.~S. Mariano (1995).
\newblock Comparing predictive accuracy.
\newblock {\em Journal of Business \& economic statistics\/}~{\em 20}, 134--144.

\bibitem[\protect\citeauthoryear{Giannone, Lenza, and Primiceri}{Giannone et~al.}{2021}]{Giannone:Lenza:Primiceri:2021}
Giannone, D., M.~Lenza, and G.~E. Primiceri (2021).
\newblock Economic predictions with big data: The illusion of sparsity.
\newblock {\em Econometrica\/}~{\em 89\/}(5), 2409--2437.

\bibitem[\protect\citeauthoryear{Harvey, Trimbur, and Van~Dijk}{Harvey et~al.}{2007}]{harvey2007trends}
Harvey, A.~C., T.~M. Trimbur, and H.~K. Van~Dijk (2007).
\newblock Trends and cycles in economic time series: A {B}ayesian approach.
\newblock {\em Journal of Econometrics\/}~{\em 140\/}(2), 618--649.

\bibitem[\protect\citeauthoryear{Huber, Koop, and Onorante}{Huber et~al.}{2021}]{huber2021inducing}
Huber, F., G.~Koop, and L.~Onorante (2021).
\newblock Inducing sparsity and shrinkage in time-varying parameter models.
\newblock {\em Journal of Business \& Economic Statistics\/}~{\em 39\/}(3), 669--683.

\bibitem[\protect\citeauthoryear{Kastner and Fr{\"u}hwirth-Schnatter}{Kastner and Fr{\"u}hwirth-Schnatter}{2014}]{kastner2014ancillarity}
Kastner, G. and S.~Fr{\"u}hwirth-Schnatter (2014).
\newblock Ancillarity-sufficiency interweaving strategy {(ASIS)} for boosting {MCMC} estimation of stochastic volatility models.
\newblock {\em Computational Statistics \& Data Analysis\/}~{\em 76}, 408--423.

\bibitem[\protect\citeauthoryear{Koop and Korobilis}{Koop and Korobilis}{2023}]{koop_korobilis_2020}
Koop, G. and D.~Korobilis (2023).
\newblock {B}ayesian dynamic variable selection in high dimensions.
\newblock {\em International Economic Review\/}~{\em 64\/}(3), 1047--1074.

\bibitem[\protect\citeauthoryear{Kowal, Matteson, and Ruppert}{Kowal et~al.}{2019}]{kowal2019dynamic}
Kowal, D.~R., D.~S. Matteson, and D.~Ruppert (2019).
\newblock Dynamic shrinkage processes.
\newblock {\em Journal of the Royal Statistical Society Series B: Statistical Methodology\/}~{\em 81\/}(4), 781--804.

\bibitem[\protect\citeauthoryear{McCracken and Ng}{McCracken and Ng}{2020}]{mccracken2020fred}
McCracken, M. and S.~Ng (2020).
\newblock {FRED-QD}: A quarterly database for macroeconomic research.
\newblock Technical report, National Bureau of Economic Research.

\bibitem[\protect\citeauthoryear{Mogliani and Simoni}{Mogliani and Simoni}{2024}]{mogliani2024bayesian}
Mogliani, M. and A.~Simoni (2024).
\newblock {B}ayesian bi-level sparse group regressions for macroeconomic forecasting.
\newblock {\em arXiv preprint arXiv:2404.02671\/}.

\bibitem[\protect\citeauthoryear{Ormerod and Wand}{Ormerod and Wand}{2010}]{ormerod_wand.2010}
Ormerod, J.~T. and M.~P. Wand (2010).
\newblock Explaining variational approximations.
\newblock {\em The American Statistician\/}~{\em 64\/}(2), 140--153.

\bibitem[\protect\citeauthoryear{Ormerod, You, and Müller}{Ormerod et~al.}{2017}]{ormerod2017VS}
Ormerod, J.~T., C.~You, and S.~Müller (2017).
\newblock {A variational Bayes approach to variable selection}.
\newblock {\em Electronic Journal of Statistics\/}~{\em 11\/}(2), 3549 -- 3594.

\bibitem[\protect\citeauthoryear{Polson, Scott, and Windle}{Polson et~al.}{2013}]{polson2013bayesian}
Polson, N.~G., J.~G. Scott, and J.~Windle (2013).
\newblock {B}ayesian inference for logistic models using p{\'o}lya--gamma latent variables.
\newblock {\em Journal of the American statistical Association\/}~{\em 108\/}(504), 1339--1349.

\bibitem[\protect\citeauthoryear{Raftery, K\'{a}rn\'{y}, and Ettler}{Raftery et~al.}{2010}]{raftery_etal2010}
Raftery, A.~E., M.~K\'{a}rn\'{y}, and P.~Ettler (2010).
\newblock Online prediction under model uncertainty via dynamic model averaging: Application to a cold rolling mill.
\newblock {\em Technometrics\/}~{\em 52\/}(1), 52--66.

\bibitem[\protect\citeauthoryear{Ray and Szab{\'o}}{Ray and Szab{\'o}}{2022}]{ray2022variational}
Ray, K. and B.~Szab{\'o} (2022).
\newblock Variational {B}ayes for high-dimensional linear regression with sparse priors.
\newblock {\em Journal of the American Statistical Association\/}~{\em 117\/}(539), 1270--1281.

\bibitem[\protect\citeauthoryear{Rohde and Wand}{Rohde and Wand}{2016}]{rhode_wand2016}
Rohde, D. and M.~P. Wand (2016).
\newblock Semiparametric mean field variational bayes: General principles and numerical issues.
\newblock {\em Journal of Machine Learning Research\/}~{\em 17\/}(172), 1--47.

\bibitem[\protect\citeauthoryear{Ro\v{c}kov\'{a} and George}{Ro\v{c}kov\'{a} and George}{2014}]{rockova_george.2014}
Ro\v{c}kov\'{a}, V. and E.~I. George (2014).
\newblock {EMVS}: The em approach to {B}ayesian variable selection.
\newblock {\em Journal of the American Statistical Association\/}~{\em 109\/}(506), 828--846.

\bibitem[\protect\citeauthoryear{Ro\v{c}kov\'{a} and McAlinn}{Ro\v{c}kov\'{a} and McAlinn}{2021}]{rockova_mcalinn_2021}
Ro\v{c}kov\'{a}, V. and K.~McAlinn (2021).
\newblock {Dynamic variable selection with spike-and-slab process priors}.
\newblock {\em Bayesian Analysis\/}~{\em 16\/}(1), 233 -- 269.

\bibitem[\protect\citeauthoryear{Stock and Watson}{Stock and Watson}{2006}]{stock2006forecasting}
Stock, J.~H. and M.~W. Watson (2006).
\newblock Forecasting with many predictors.
\newblock {\em Handbook of economic forecasting\/}~{\em 1}, 515--554.

\bibitem[\protect\citeauthoryear{Stock and Watson}{Stock and Watson}{2007}]{stock2007has}
Stock, J.~H. and M.~W. Watson (2007).
\newblock Why has us inflation become harder to forecast?
\newblock {\em Journal of Money, Credit and banking\/}~{\em 39}, 3--33.

\bibitem[\protect\citeauthoryear{Stock and Watson}{Stock and Watson}{2012}]{stock2012disentangling}
Stock, J.~H. and M.~W. Watson (2012).
\newblock Disentangling the channels of the 2007-09 recession.
\newblock {\em Brookings Papers on Economic Activity\/}~(1), 81--135.

\bibitem[\protect\citeauthoryear{Tibshirani}{Tibshirani}{1996}]{tibshirani1996regression}
Tibshirani, R. (1996).
\newblock Regression shrinkage and selection via the lasso.
\newblock {\em Journal of the Royal Statistical Society Series B: Statistical Methodology\/}~{\em 58\/}(1), 267--288.

\bibitem[\protect\citeauthoryear{Tibshirani, Saunders, Rosset, Zhu, and Knight}{Tibshirani et~al.}{2005}]{tibshirani2005sparsity}
Tibshirani, R., M.~Saunders, S.~Rosset, J.~Zhu, and K.~Knight (2005).
\newblock Sparsity and smoothness via the fused lasso.
\newblock {\em Journal of the Royal Statistical Society Series B: Statistical Methodology\/}~{\em 67\/}(1), 91--108.

\bibitem[\protect\citeauthoryear{Uribe and Lopes}{Uribe and Lopes}{2020}]{uribe2020dynamic}
Uribe, P.~W. and H.~F. Lopes (2020).
\newblock Dynamic sparsity on dynamic regression models.
\newblock {\em arXiv preprint arXiv:2009.14131\/}.

\bibitem[\protect\citeauthoryear{Wand, Ormerod, Padoan, and Fr{\"u}hwirth}{Wand et~al.}{2011}]{wand2011mean}
Wand, M.~P., J.~T. Ormerod, S.~A. Padoan, and R.~Fr{\"u}hwirth (2011).
\newblock Mean field variational {B}ayes for elaborate distributions.
\newblock {\em Bayesian Analysis\/}~{\em 6\/}(4), 847--900.

\end{thebibliography}

\newpage
\appendix
\renewcommand\thefigure{\thesection.\arabic{figure}}    
\renewcommand\thetable{\thesection.\arabic{table}}    
\numberwithin{equation}{section}
\numberwithin{proposition}{section}

\def\spacingset#1{\renewcommand{\baselinestretch}%
{#1}\small\normalsize} \spacingset{1}


\begin{center}
    \LARGE{\bf Supplementary Material for: \\ Scalable Variational Bayes Inference for Dynamic Variable Selection}
\end{center}

\spacingset{1.75} 

\vspace{1cm}
\noindent This Supplementary Material provides a set of theoretical, simulation, and empirical results complementary to what is discussed in the main paper. Section \ref{app:mcmc_algo} shows the derivation of an MCMC algorithm equivalent to our variational Bayes approach. This allows us to draw a more equitable comparison between estimation methods based on the same model. Section \ref{app:vb_densities} provides the full proofs and derivations of the optimal densities used to implement our variational Bayes estimation. Section \ref{app:theo} provides the formal proofs of our variational Bayes algorithm's sparsity-inducing and convergence properties. Section \ref{app:hyperparam} discusses the choice of the hyperparameters. Finally, Section \ref{app:TVP_more_sim} and Section \ref{app:additional_empirical} provide additional simulation and empirical results, respectively.

\section{An equivalent MCMC sampling scheme}\setcounter{figure}{0}\setcounter{table}{0}

\label{app:mcmc_algo}

We provide the full conditional distributions, which can be used to implement a conventional MCMC algorithm equivalent to our variational Bayes algorithm. It is worth noticing that the MCMC implementation lacks two important properties. First, it is not possible to smooth the posterior inclusion probabilities using the strategy in Proposition 5. Second, perhaps more importantly, the results described in Section 2.2 are no longer valid, and therefore, an efficient version of the MCMC that drops the unimportant variables online is not available.

\subsection{Full conditional distributions}

In this Section, we provide the set of full conditional distributions that can be used to implement a conventional Gibbs sampler assuming either stochastic or constant volatility in the residuals. 

\paragraph{Full conditional of $p(\sigma^2|\mbox{rest})$.} Recall that in the case of a homoschedastic error term, the prior assumption on $\sigma^2$ is $\sigma^2\sim\mathsf{IGa}(A_\sigma,B_\sigma)$.
The full conditional distribution of $\sigma^2$ given the rest $p(\sigma^2|\mbox{rest})\propto p(\mathbf{y}|\sigma^2,\mathbf{b},\boldsymbol{\gamma})p(\sigma^2)$ is proportional to:
\begin{align}\label{eq:full_s2}
		\log p(\sigma^2|\mathrm{rest})&\propto -\frac{n}{2}\log\sigma^2-\frac{1}{2\sigma^2}\boldsymbol{\varepsilon}^\prime\boldsymbol{\varepsilon} -(A_\sigma+1)\log\sigma^2-\frac{B_\sigma}{\sigma^2},
\end{align}
with $\boldsymbol{\varepsilon}=\mathbf{y}-\sum_{j=1}^p\mathbf{X}_j\mathbf{\Gamma}_j\mathbf{b}_j$, where $\mathbf{X}_k$ and $\mathbf{\Gamma}_k$ are diagonal matrices with elements $x_{kt-1}$ and $\gamma_{kt}$ respectively. Therefore, the full conditional distribution of the constant variance $\sigma^2$ is an inverse-gamma $\sigma^2|\mbox{rest}\sim\mathsf{IG}\left(A_\sigma+\frac{n}{2},B_\sigma+\frac{1}{2}\boldsymbol{\varepsilon}^\prime\boldsymbol{\varepsilon}\right)$.
 
\paragraph{Sampling $\mathbf{h}$.} In order to get posterior samples from $p(\mathbf{h}|\mathbf{y})$, we exploit the methodology described in \cite{kastner2014ancillarity} where the data are transformed as $\varepsilon_t=y_t-\sum_{j=1}^px_{jt-1}\gamma_{jt}b_{jt}$, for $t=1,\ldots,n$. 

\paragraph{Full conditional of $p(\mathbf{b}_j|\mbox{rest})$.} For simplicity, we provide the full conditionals for the independence case where variables are not grouped based on their correlation. Let the prior distribution on $\mathbf{b}_j$ be $\mathbf{b}_j\sim\mathsf{N}_{n+1}(0,\eta^2_j\mathbf{Q}^{-1})$.
The full conditional distribution of $\mathbf{b}_j$ given the rest $p(\mathbf{b}_j|\mbox{rest})\propto p(\mathbf{y}|\sigma^2,\mathbf{b},\boldsymbol{\gamma})p(\mathbf{b}_j|\eta^2_j)$ is proportional to:
	\begin{align*}
		\log p(\mathbf{b}_j|\mathrm{rest})\propto -\frac{1}{2}\left(\mathbf{y}-\sum_{k=1}^p\mathbf{X}_k\mathbf{\Gamma}_k\mathbf{b}_k\right)^\prime\mathbf{H}\left(\mathbf{y}-\sum_{k=1}^p\mathbf{X}_k\mathbf{\Gamma}_k\mathbf{b}_k\right) -\frac{1}{2\eta_j^2}{\mathbf{b}}_j^\prime\mathbf{Q}{\mathbf{b}}_j
	\end{align*}
	where $\mathbf{H}$, $\mathbf{X}_k$, and $\mathbf{\Gamma}_k$ are diagonal matrices with elements $1/\sigma^2_t$, $x_{kt-1}$, and $\gamma_{kt}$ for $t=1,\ldots,n$, respectively.
	Define $\boldsymbol{\varepsilon}_{-j}=\mathbf{y}-\sum_{k=1,k\neq j}^p\mathbf{X}_k\mathbf{\Gamma}_k\mathbf{b}_k$, then
	\begin{equation}\label{eq:full_b}
		\begin{aligned}
			\log p(\mathbf{b}_j|\mathrm{rest})&\propto -\frac{1}{2}\left(\boldsymbol{\varepsilon}_{-j}-\mathbf{X}_j\mathbf{\Gamma}_j\mathbf{b}_j\right)^\prime\mathbf{H}\left(\boldsymbol{\varepsilon}_{-j}-\mathbf{X}_j\mathbf{\Gamma}_j\mathbf{b}_j\right) -\frac{1}{2\eta_j^2}{\mathbf{b}}_j^\prime\mathbf{Q}{\mathbf{b}}_j \\
			&\propto -\frac{1}{2}\left(\mathbf{b}_j^\prime\mathbf{\Gamma}_j\mathbf{X}_j\mathbf{H}\mathbf{X}_j\mathbf{\Gamma}_j\mathbf{b}_j-2\mathbf{b}_j^\prime\mathbf{\Gamma}_j\mathbf{X}_j\mathbf{H}\boldsymbol{\varepsilon}_{-j}\right) -\frac{1}{2\eta_j^2}{\mathbf{b}}_j^\prime\mathbf{Q}{\mathbf{b}}_j \\
            &\propto -\frac{1}{2}\left(\mathbf{b}_j^\prime(\mathbf{\Gamma}_j\mathbf{X}_j\mathbf{H}\mathbf{X}_j\mathbf{\Gamma}_j+1/\eta_j^2\mathbf{Q})\mathbf{b}_j-2\mathbf{b}_j^\prime\mathbf{\Gamma}_j\mathbf{X}_j\mathbf{H}\boldsymbol{\varepsilon}_{-j}\right).
		\end{aligned}
	\end{equation}
Therefore, the full conditional distribution of $\mathbf{b}_j$ is a multivariate Gaussian distribution $\mathbf{b}_j|\mbox{rest}\sim\mathsf{N}_{n+1}\left(\boldsymbol{\mu}_{\mathbf{b}_j|\mbox{rest}},\mathbf{\Sigma}_{\mathbf{b}_j|\mbox{rest}}\right)$, with variance-covariance $\mathbf{\Sigma}_{\mathbf{b}_j|\mbox{rest}} = (\mathbf{\Gamma}_j\mathbf{X}_j\mathbf{H}\mathbf{X}_j\mathbf{\Gamma}_j+1/\eta_j^2\mathbf{Q})^{-1}$ and mean $\boldsymbol{\mu}_{\mathbf{b}_j|\mbox{rest}} = (\mathbf{\Gamma}_j\mathbf{X}_j\mathbf{H}\mathbf{X}_j\mathbf{\Gamma}_j+1/\eta_j^2\mathbf{Q})^{-1}\mathbf{\Gamma}_j\mathbf{X}_j\mathbf{H}\boldsymbol{\varepsilon}_{-j}$.

\paragraph{Full conditional of $p(\gamma_{jt}|\mbox{rest})$.} Recall that the prior assumption on $\gamma_{jt}$ is $\gamma_{jt}\sim\mathsf{Bern}(\mathrm{expit}(\omega_{jt}))$. The full conditional distribution of $\gamma_{jt}$, namely $p(\gamma_{jt}|\mbox{rest})\propto p(\mathbf{y}|\sigma^2,\mathbf{b},\boldsymbol{\gamma})p(\gamma_{jt}|\omega_{jt})$ is proportional to:
	\begin{align}\label{eq:full_gamma}
		\log p(\gamma_{jt}|\mathrm{rest})&\propto -\frac{1}{2\sigma^2_t}\left(y_t-\sum_{k=1}^p \gamma_{kt}b_{kt}x_{kt-1}\right)^2+\gamma_{jt}\omega_{jt} \nonumber\\
		&\propto -\frac{1}{2\sigma_t^2}(\gamma^2_{jt}b^2_{jt}x^2_{jt-1}-2\gamma_{jt}b_{jt}x_{jt-1} \varepsilon_{-j,t})+\gamma_{jt}\omega_{jt}\\
		&\propto \gamma_{jt}\left\{\omega_{jt}-\frac{1}{2\sigma_t^2}(b^2_{jt}x^2_{jt-1}-2b_{jt}x_{jt-1} \varepsilon_{-j,t})\right\}\nonumber.
	\end{align}
Therefore, the full conditional distribution of the indicator variable $\gamma_{jt}$ is a Bernoulli distribution $\gamma_{jt}|\mbox{rest}\sim\mathsf{Bern}\left(\mathrm{expit}\left\{\omega_{jt}-\frac{1}{2\sigma_t^2}(b^2_{jt}x^2_{jt-1}-2b_{jt}x_{jt-1} \varepsilon_{-j,t})\right\}\right)$.

\paragraph{Full conditional of $p(\boldsymbol{\omega}_j|\mbox{rest})$.} Let the prior distribution on $\boldsymbol{\omega}_j$ be $\boldsymbol{\omega}_j\sim\mathsf{N}_{n+1}(0,\xi^2_j\mathbf{Q}^{-1})$. The full conditional distribution $p(\boldsymbol{\omega}_j|\mbox{rest})\propto \left[\prod_{t=1}^np(\gamma_{jt}|\omega_{jt},z_{jt})p(z_{jt}|\omega_{jt})\right]p(\boldsymbol{\omega}_j|\xi^2_j)$ is proportional to:
\begin{align}\label{eq:full_om}
		\log p(\boldsymbol{\omega}_j|\mathrm{rest})&\propto \boldsymbol{\omega}_j^\prime(\boldsymbol{\gamma}_j-1/2\boldsymbol{\iota}_{n})-\frac{1}{2}\boldsymbol{\omega}_j^\prime\mathsf{Diag}(\mathbf{z}_j)\boldsymbol{\omega}_j -\frac{1}{2\xi_j^2}{\boldsymbol{\omega}}_j^\prime\mathbf{Q}{\boldsymbol{\omega}}_j \nonumber\\
        &\propto -\frac{1}{2}\left(\boldsymbol{\omega}_j^\prime(\mathsf{Diag}(\mathbf{z}_j)+1/\xi_j^2\mathbf{Q})\boldsymbol{\omega}_j -2\boldsymbol{\omega}_j^\prime(\boldsymbol{\gamma}_j-1/2\boldsymbol{\iota}_{n}) \right).
\end{align}
Therefore, the full conditional distribution of $\boldsymbol{\omega}_j$ is a multivariate Gaussian distribution $\boldsymbol{\omega}_j|\mbox{rest}\sim\mathsf{N}_{n+1}\left(\boldsymbol{\mu}_{\omega_j|\mbox{rest}},\mathbf{\Sigma}_{\omega_j|\mbox{rest}}\right)$, with variance-covariance $\mathbf{\Sigma}_{\omega_j|\mbox{rest}}=(\mathsf{Diag}(\mathbf{z}_j)+1/\xi_j^2\mathbf{Q})^{-1}$ and mean $\boldsymbol{\mu}_{\omega_j|\mbox{rest}}=(\mathsf{Diag}(\mathbf{z}_j)+1/\xi_j^2\mathbf{Q})^{-1}(\boldsymbol{\gamma}_j-1/2\boldsymbol{\iota}_{n})$.

\paragraph{Full conditional of $p(z_{jt}|\mbox{rest})$.} Recall the Polya-Gamma representation. Then, the full conditional distribution of $z_{jt}$, namely $p(z_{jt}|\mbox{rest})\propto p(\gamma_{jt}|z_{jt},\omega_{jt})p(z_{jt}|\omega_{jt})$ is proportional to:
\begin{equation}\label{eq:full_z}
		\begin{aligned}
			\log p(z_{jt}|\mbox{rest}) &\propto -z_{jt}\omega_{jt}^2 +\log p(z_{jt}),
		\end{aligned}
\end{equation}
where $p(z_{jt})$ is the density function of a Polya-Gamma random variable $\mathsf{PG}(1,0)$. Hence, $z_{jt}|\mbox{rest}\sim\mathsf{PG}(1,\omega_{jt})$.

\paragraph{Full conditional of $p(\eta^2_j|\mbox{rest})$.} Assume that a prior $\eta^2_j\sim\mathsf{IGa}(A_\eta,B_\eta)$. Then, the full conditional distribution of $\eta^2_j$ given the rest $p(\eta^2_j|\mbox{rest})\propto p(\mathbf{b}_j|\eta^2_j)p(\eta^2_j)$ is proportional to:
\begin{align}\label{eq:full_eta2}
    p(\eta^2_j|\mathrm{rest})&\propto -\frac{n+1}{2}\log\eta_j^2-\frac{1}{2\eta_j^2}{\mathbf{b}}_j^\prime\mathbf{Q}{\mathbf{b}}_j -(A_\eta+1)\log\eta_j^2-\frac{B_\eta}{\eta_j^2} \nonumber\\
    &\propto -(A_\eta+\frac{n+1}{2}+1)\log\eta_j^2-\frac{1}{\eta_j^2}\left(B_\eta+\frac{1}{2}{\mathbf{b}}_j^\prime\mathbf{Q}{\mathbf{b}}_j\right).
\end{align}
Therefore, the full conditional distribution of the conditional variance $\eta_j^2$ is an inverse-gamma $\eta_j^2|\mbox{rest}\sim\mathsf{IG}\left(A_\eta+\frac{n+1}{2},B_\eta+\frac{1}{2}{\mathbf{b}}_j^\prime\mathbf{Q}{\mathbf{b}}_j\right)$.

\paragraph{Full conditional of $p(\xi^2_j|\mbox{rest})$.} Recall that a priori $\xi^2_j\sim\mathsf{IGa}(A_\xi,B_\xi)$. The full conditional distribution of $\xi^2_j$ given the rest $p(\xi^2_j|\mbox{rest})\propto p(\boldsymbol{\omega}_j|\xi^2_j)p(\xi^2_j)$ is proportional to:
\begin{align}\label{eq:full_xi2}
    p(\xi^2_j|\mathrm{rest})&\propto -\frac{n+1}{2}\log\xi_j^2-\frac{1}{2\xi_j^2}{\boldsymbol{\omega}}_j^\prime\mathbf{Q}{\boldsymbol{\omega}}_j -(A_\xi+1)\log\xi_j^2-\frac{B_\xi}{\xi_j^2} \nonumber\\
    &\propto -(A_\xi+\frac{n+1}{2}+1)\log\xi_j^2-\frac{1}{\xi_j^2}\left(B_\xi+\frac{1}{2}{\boldsymbol{\omega}}_j^\prime\mathbf{Q}{\boldsymbol{\omega}}_j\right).
\end{align}
Hence, the full conditional distribution of the conditional variance $\xi_j^2$ is an inverse-gamma $\xi_j^2|\mbox{rest}\sim\mathsf{IG}\left(A_\xi+\frac{n+1}{2},B_\xi+\frac{1}{2}{\boldsymbol{\omega}}_j^\prime\mathbf{Q}{\boldsymbol{\omega}}_j\right)$.

\paragraph{Full conditional of $p(\nu^2|\mbox{rest})$.} Assume that a priori $\nu^2\sim\mathsf{IGa}(A_\nu,B_\nu)$. The full conditional distribution of $\nu^2$ given the rest $p(\nu^2|\mbox{rest})\propto p(\mathbf{h}|\nu^2)p(\nu^2)$ is proportional to:
\begin{align}\label{eq:full_nu2}
    p(\nu^2|\mathrm{rest})&\propto -\frac{n+1}{2}\log\nu^2-\frac{1}{2\nu^2}{\mathbf{h}}_j^\prime\mathbf{Q}{\mathbf{h}}_j -(A_\nu+1)\log\nu^2-\frac{B_\nu}{\nu^2} \nonumber\\
    &\propto -(A_\nu+\frac{n+1}{2}+1)\log\nu^2-\frac{1}{\nu^2}\left(B_\nu+\frac{1}{2}{\mathbf{h}}_j^\prime\mathbf{Q}{\mathbf{h}}_j\right).
\end{align}
Therefore, the full conditional distribution of the conditional variance $\nu^2$ is an inverse-gamma $\nu^2|\mbox{rest}\sim\mathsf{IG}\left(A_\nu+\frac{n+1}{2},B_\nu+\frac{1}{2}{\mathbf{h}}_j^\prime\mathbf{Q}{\mathbf{h}}_j\right)$.

Algorithm \ref{algo:algo_mcmc} summarises the Gibbs-sampling algorithm leveraging the full conditional distributions provided above. 

	\begin{algorithm}
		\SetAlgoLined
		\kwInit{$\boldsymbol{\vartheta}^{(0)}$, $\mathsf{ndraws}$, $A_\nu$, $B_\nu$, $A_\eta$, $B_\eta$, $A_\xi$, $B_\xi$}
		\For{$r=1,\ldots,\mathsf{ndraws}$}{
			\For{$j=1,\ldots,p$}{
                Compute $\mathbf{\Sigma}_{\mathbf{b}_j|\mbox{rest}} = (\mathbf{\Gamma}^{(r-1)}_j\mathbf{X}_j\mathbf{H}^{(r-1)}\mathbf{X}_j\mathbf{\Gamma}^{(r-1)}_j+1/\eta_{j}^{2{(r-1)}}\mathbf{Q})^{-1}$;\\
                Compute $\boldsymbol{\mu}_{\mathbf{b}_j|\mbox{rest}} = \mathbf{\Sigma}_{\mathbf{b}_j|\mbox{rest}}\mathbf{\Gamma}^{(r-1)}_j\mathbf{X}_j\mathbf{H}^{(r-1)}\boldsymbol{\varepsilon}^{(r-1)}_{-j}$;\\
				Sample $\mathbf{b}_j^{(r)}\sim\mathsf{N}_{n+1}\left(\boldsymbol{\mu}_{\mathbf{b}_j|\mbox{rest}},\mathbf{\Sigma}_{\mathbf{b}_j|\mbox{rest}}\right)$; \\
				Sample $\eta_j^{2(r)}\sim\mathsf{IG}\left(A_\eta+\frac{n+1}{2},B_\eta+\frac{1}{2}{\mathbf{b}}_j^{\prime(r)}\mathbf{Q}{\mathbf{b}}_j^{(r)}\right)$;\\
                Compute $\mathbf{\Sigma}_{\omega_j|\mbox{rest}} = (\mathsf{Diag}(\mathbf{z}^{(r-1)}_j)+1/\xi_j^{2{(r-1)}}\mathbf{Q})^{-1}$;\\
                Compute $\boldsymbol{\mu}_{\omega_j|\mbox{rest}} = \mathbf{\Sigma}_{\omega_j|\mbox{rest}}(\boldsymbol{\gamma}^{(r-1)}_j-1/2\boldsymbol{\iota}_{n})$;\\
				Sample $\boldsymbol{\omega}_j^{(r)}\sim\mathsf{N}_{n+1}\left(\boldsymbol{\mu}_{\omega_j|\mbox{rest}},\mathbf{\Sigma}_{\omega_j|\mbox{rest}}\right)$;\\
                Sample $\xi_j^{2(r)}\sim\mathsf{IG}\left(A_\xi+\frac{n+1}{2},B_\xi+\frac{1}{2}{\boldsymbol{\omega}}_j^{\prime(r)}\mathbf{Q}{\boldsymbol{\omega}}_j^{(r)}\right)$;\\
				\For{$t=1,\ldots,n$}{
					Sample $z^{(r)}_{jt}\sim\mathsf{PG}(1,\omega^{{(r)}}_{jt})$; \\
					Sample $\gamma^{(r)}_{jt}\sim\mathsf{Bern}\left(\mathrm{expit}\left\{\omega^{(r)}_{jt}-\frac{1}{2\sigma^{2(r-1)}_t}(b^{2(r)}_{jt}x^2_{jt}-2b^{(r)}_{jt}x_{jt} \varepsilon^{(r)}_{-j,t})\right\}\right)$;
				}
			}
			Sample $\mathbf{h}$ with $\varepsilon^{(r)}_t=y_t-\sum_{j=1}^px_{jt}\gamma^{(r)}_{jt} b^{(r)}_{jt}$ (heteroskedastic);\\
            Sample $\nu^{2(r)}\sim\mathsf{IG}\left(A_\nu+\frac{n+1}{2},B_\nu+\frac{1}{2}{\mathbf{h}}_j^{\prime(r)}\mathbf{Q}{\mathbf{h}}_j^{(r)}\right)$;\\
                Sample $\sigma^{2(r)}\sim\mathsf{IG}\left(A_\sigma+\frac{n}{2},B_\sigma+\frac{1}{2}\boldsymbol{\varepsilon}^{\prime(r)}\boldsymbol{\varepsilon}^{(r)}\right)$ (homoskedastic); \\
		}
		\caption{Gibbs-sampling scheme for dynamic variable selection.}
		\label{algo:algo_mcmc}
	\end{algorithm}

\section{Optimal variational densities}\setcounter{figure}{0}\setcounter{table}{0}

\label{app:vb_densities}
Proposition \ref{prop:up_beta} shows that the variational mean and variance of the time-varying parameters $\mathbf{b}_k$ depend on the entire trajectory $t=1,\ldots,n$ of the variational estimates $\boldsymbol{\mu}_{q(\gamma_{jt})}$ for all $j=1,\ldots,p_k$ in group $k$.  

\begin{proposition}\label{prop:up_beta}
The optimal variational density for $\mathbf{b}_k=\left(\mathbf{b}_{k0},\ldots,\mathbf{b}_{kn}\right)^\prime$, with $\mathbf{b}_{kt}=(b_{1t},\ldots,b_{p_k t})^\prime$, is a multivariate Gaussian  $q^\ast(\mathbf{b}_k)=\mathsf{N}_{p_k(n+1)}(\boldsymbol{\mu}_{q(\mathbf{b}_k)},\boldsymbol{\Sigma}_{q(\mathbf{b}_k)})$, where:
	\begin{align}
		\boldsymbol{\Sigma}_{q(\mathbf{b}_k)} &= (\mathbf{D}_k+\mathbf{Q}_k)^{-1},\qquad \boldsymbol{\mu}_{q(\mathbf{b}_k)}=\boldsymbol{\Sigma}_{q(\mathbf{b}_k)}{\mathbf{\Lambda}_k},
	\end{align}
where $\mathbf{D}_k$ is a block-diagonal matrix with $n+1$ blocks of size $p_k$ having generic block equal to $[\mathbf{D}_k]_t = \mu_{q(1/\sigma^2_t)}\mathbf{x}_{k,t-1}\mathbf{x}_{k,t-1}^\prime \odot \{\boldsymbol{\mu}_{q(\gamma_{kt})}\boldsymbol{\mu}_{q(\gamma_{kt})}^\prime+\mathrm{diag}(\boldsymbol{\mu}_{q(\gamma_{kt})}(1-\boldsymbol{\mu}_{q(\gamma_{kt})})\}$, $\mathbf{Q}_k$ is a tridiagonal block matrix $\mathbf{Q}_k=\mathbf{Q}\otimes\mathrm{diag}(\mu_{q(1/\eta_1^2)},\ldots,\mu_{q(1/\eta_{p_k}^2)})$, and $\mathbf{\Lambda}_k$ stacks $p_k$ dimensional vectors defined as $\boldsymbol{\lambda}_{kt}={\mu}_{q(1/\sigma_t^2)}\mathrm{diag}(\boldsymbol{\mu}_{q(\gamma_{kt})})\mathbf{x}_{k,t-1}(y_t - \sum_{m=1,m\neq k}^K \mathbf{x}_{m,t-1}^\prime\mathrm{diag}(\boldsymbol{\mu}_{q(\gamma_{mt})}){\boldsymbol{\mu}}_{q(\mathbf{b}_{mt})})$.
\end{proposition}

\begin{proof}
Let the prior distribution on $\mathbf{b}_k$ be $\mathbf{b}_k\sim\mathsf{N}_{p_k(n+1)}(0,\tilde{\mathbf{Q}}_k^{-1})$ with $\tilde{\mathbf{Q}}_k=\mathbf{Q}\otimes\mathrm{diag}(\eta^2_1,\ldots,\eta^2_{p_k})$, then the full conditional distribution of $\mathbf{b}_k$ given the rest $p(\mathbf{b}_k|\mbox{rest})\propto p(\mathbf{y}|\sigma^2,\mathbf{b},\boldsymbol{\gamma})p(\mathbf{b}_k|\eta^2_1,\ldots,\eta^2_{p_k})$ is proportional to:
	\begin{align*}
		\log p(\mathbf{b}_k|\mathrm{rest})\propto -\frac{1}{2}\left(\mathbf{y}-\sum_{m=1}^K\mathbf{X}_m\mathbf{\Gamma}_m\mathbf{b}_m\right)^\prime\mathbf{H}\left(\mathbf{y}-\sum_{m=1}^K\mathbf{X}_m\mathbf{\Gamma}_m\mathbf{b}_m\right) -\frac{1}{2}{\mathbf{b}}_k^\prime\tilde{\mathbf{Q}}_k{\mathbf{b}}_k
	\end{align*}
	where $\mathbf{H}$ and $\mathbf{\Gamma}_m$ are diagonal matrices with elements $1/\sigma^2_t$ and $\gamma_{jt}$, respectively, and $\mathbf{X}_m=[\mathbf{X}_1,\ldots,\mathbf{X}_{p_m}]$ with $\mathbf{X}_j$ diagonal matrix with elements $x_{jt-1}$, for $j=1,\ldots,p_m$ and $t=1,\ldots,n$.
	Define $\boldsymbol{\varepsilon}_{-k}=\mathbf{y}-\sum_{m=1,m\neq k}^K\mathbf{X}_m\mathbf{\Gamma}_m\mathbf{b}_m$, then
	\begin{equation}
		\begin{aligned}
			\log p(\mathbf{b}_k|\mathrm{rest})&\propto -\frac{1}{2}\left(\boldsymbol{\varepsilon}_{-k}-\mathbf{X}_k\mathbf{\Gamma}_k\mathbf{b}_k\right)^\prime\mathbf{H}\left(\boldsymbol{\varepsilon}_{-k}-\mathbf{X}_k\mathbf{\Gamma}_k\mathbf{b}_k\right) -\frac{1}{2}{\mathbf{b}}_k^\prime \tilde{\mathbf{Q}}_k\mathbf{b}_k \\
			&\propto -\frac{1}{2}\left(\mathbf{b}_k^\prime\mathbf{\Gamma}_k\mathbf{X}_k\mathbf{H}\mathbf{X}_k\mathbf{\Gamma}_k\mathbf{b}_k-2\mathbf{b}_k^\prime\mathbf{\Gamma}_k\mathbf{X}_k\mathbf{H}\boldsymbol{\varepsilon}_{-k}\right) -\frac{1}{2}{\mathbf{b}}_k^\prime\tilde{\mathbf{Q}}_k{\mathbf{b}}_k\\
            &\propto -\frac{1}{2}\left(\mathbf{b}_k^\prime(\mathbf{\Gamma}_k\mathbf{X}_k\mathbf{H}\mathbf{X}_k\mathbf{\Gamma}_k+\tilde{\mathbf{Q}}_k)\mathbf{b}_k-2\mathbf{b}_k^\prime\mathbf{\Gamma}_k\mathbf{X}_k\mathbf{H}\boldsymbol{\varepsilon}_{-k}\right).
		\end{aligned}
	\end{equation} 
 According to \cite{ormerod_wand.2010}, the optimal variational density is given by:
	\begin{equation}
		\begin{aligned}\label{eq:q_beta}
			\log q^\ast({\mathbf{b}}_k) &\propto \mathbb{E}_{-\mathbf{b}_k}[\log p(\mathbf{b}_k|\mathrm{rest})] \\ 
			&\propto-\frac{1}{2}\left({\mathbf{b}}_k^\prime(\mathbf{D}_k+\mathbf{Q}_k){\mathbf{b}}_k-2{\mathbf{b}}_k^\prime\mathrm{diag}(\boldsymbol{\mu}_{q(\gamma_{k})})\mathbf{X}_k\mathbb{E}_q[\mathbf{H}]\mathbb{E}_q[\boldsymbol{\varepsilon}_{-j}]\right),
		\end{aligned}
	\end{equation}
	where $\mathbf{D}_k = \mathbb{E}_q[\mathbf{\Gamma}_k\mathbf{X}_k\mathbf{H}\mathbf{X}_k\mathbf{\Gamma}_k]$ is a block-diagonal matrix with $n+1$ blocks of size $p_k$ having generic block equal to $[\mathbf{D}_k]_t = \mu_{q(1/\sigma^2_t)}\mathbf{x}_{k,t-1}\mathbf{x}_{k,t-1}^\prime \odot \{\boldsymbol{\mu}_{q(\gamma_{kt})}\boldsymbol{\mu}_{q(\gamma_{kt})}^\prime+\mathrm{diag}(\boldsymbol{\mu}_{q(\gamma_{kt})}(1-\boldsymbol{\mu}_{q(\gamma_{kt})}))\}$, and $\mathbf{Q}_k$ is a tridiagonal block matrix $\mathbf{Q}_k=\mathbf{Q}\otimes\mathrm{diag}(\mu_{q(1/\eta_1^2)},\ldots,\mu_{q(1/\eta_{p_k}^2)})$.
    Define $\mathbf{\Lambda}_k=\mathrm{diag}(\boldsymbol{\mu}_{q(\gamma_{k})})\mathbf{X}_k\mathbb{E}_q[\mathbf{H}]\mathbb{E}_q[\boldsymbol{\varepsilon}_{-j}])$ with $\mathbb{E}_q[\mathbf{H}]=\mathrm{diag}(\mu_{q(1/\sigma^2_1)},\ldots,\mu_{q(1/\sigma^2_n)})$ and notice that equation \ref{eq:q_beta} represents the kernel of a multivariate Gaussian distribution as in Proposition \ref{prop:up_beta}.
\end{proof}

Proposition \ref{prop:up_gamma} provides the vector of optimal variational estimates of $\gamma_{jt}$ for all variables $j=1,\ldots,p_k$ that belong to a group $k=1,\ldots,K$ of size $p_k$.

\begin{proposition}\label{prop:up_gamma}
The optimal variational density of $\gamma_{jt}$ is a Bernoulli random variable $q^\ast(\gamma_{jt})=\mathsf{Bern}(\mathrm{expit}(\omega_{q(\gamma_{jt})}))$, where $\mathrm{expit}(\cdot)$ is the inverse of the logit function and 
    \begin{align}
	\omega_{q(\gamma_{jt})}&=\mu_{q(\omega_{jt})}-\frac{1}{2}\mu_{q(1/\sigma_t^2)}(x^2_{jt-1}\mathbb{E}_q[b^2_{jt}]-2\mu_{q(b_{jt})} x_{jt-1}\mu_{q(\varepsilon_{-j,t})}) \nonumber\\
    &\qquad -\mu_{q(1/\sigma_t^2)}x_{jt-1}\mathbf{x}^{-j}_{kt-1}\mathrm{diag}(\boldsymbol{\mu}^{-j}_{q(\gamma_{kt})})[\boldsymbol{\Sigma}_{q(\mathbf{b}_{kt})}]_{-j,j},\label{eq:up_gamma}
	\end{align}
 where $[\boldsymbol{\Sigma}_{q(\mathbf{b}_{kt})}]_{-j,j}$ is the $j$-th column withouth the $j$-th element of the $t$-th diagonal block of $\boldsymbol{\Sigma}_{q(\mathbf{b}_{k})}$, and $\mu_{q(\varepsilon_{-j,t})}=y_t - \mathbf{x}^{-j\,\prime}_{k,t-1}\mathrm{diag}(\boldsymbol{\mu}_{q(\gamma^{-j}_{kt})}){\boldsymbol{\mu}}_{q(\mathbf{b}^{-j}_{kt})}- \sum_{m=1,m\neq k}^K \mathbf{x}_{m,t-1}^\prime\mathrm{diag}(\boldsymbol{\mu}_{q(\gamma_{mt})}){\boldsymbol{\mu}}_{q(\mathbf{b}_{mt})}$ with $\boldsymbol{\mu}_{q(\gamma_{kt})}$ being the collection of $\boldsymbol{\mu}_{q(\gamma_{jt})}$ for the predictors in group $k$.
\end{proposition}

\begin{proof}
	The full conditional distribution of $\gamma_{jt}\sim p(\gamma_{jt}|\mbox{rest})$ is derived in Eq.\eqref{eq:full_gamma}. Thus, the optimal variational density is given by:
	\begin{equation}
		\begin{aligned}\label{eq:q_gamma}
			\log q^\ast(\gamma_{jt}) &\propto \mathbb{E}_{-\gamma_{jt}}\left[\log p(\gamma_{jt}|\mathrm{rest})\right] \\
			&\propto \gamma_{jt}\{\mu_{q(\omega_{jt})}-\frac{1}{2}\mu_{q(1/\sigma^2)}\big(x^2_{jt-1}\mathbb{E}_q[b^2_{jt}]-2\mathbb{E}_{-\gamma_{jt}}\left[b_{jt}x_{jt-1}\varepsilon_{-j,t}\right] \big)\},
		\end{aligned}
	\end{equation}
    where $\mathbb{E}_q[b^2_{jt}]=\mu^2_{q(b_{jt})}+\sigma^2_{q(b_{jt})}$ and $\varepsilon_{-j,t}=y_t - \mathbf{x}^{-j\,\prime}_{k,t-1}\mathrm{diag}(\boldsymbol{\gamma}^{-j}_{kt})\mathbf{b}^{-j}_{kt}- \sum_{m=1,m\neq k}^K \mathbf{x}_{m,t-1}^\prime\mathrm{diag}(\boldsymbol{\gamma}_{mt})\mathbf{b}_{mt}$.
    Hence, we need to account for the within-group correlation between $b_{jt}$ and $\mathbf{b}^{-j}_{kt}$:
    \begin{align*}
        \mathbb{E}_{-\gamma_{jt}}\left[b_{jt}x_{jt-1}\varepsilon_{-j,t}\right] &= \mathbb{E}_{-\gamma_{jt}}\left[b_{jt}x_{jt-1}\bigg(y_t - \sum_{m=1,m\neq k}^K \mathbf{x}_{m,t-1}^\prime\mathrm{diag}(\boldsymbol{\gamma}_{mt})\mathbf{b}_{mt}\bigg)\right] \\
        &\qquad - \mathbb{E}_{-\gamma_{jt}}\big[b_{jt}x_{jt-1}\mathbf{x}^{-j\,\prime}_{k,t-1}\mathrm{diag}(\boldsymbol{\gamma}^{-j}_{kt})\mathbf{b}^{-j}_{kt}\big] \\
        &= \mu_{q(b_{jt})}x_{jt-1}\bigg(y_t - \sum_{m=1,m\neq k}^K \mathbf{x}_{m,t-1}^\prime\mathrm{diag}(\boldsymbol{\mu}_{q(\gamma_{mt})})\boldsymbol{\mu}_{q(\mathbf{b}_{mt})}\bigg) \\
        &\qquad - x_{jt-1}\mathbf{x}^{-j\,\prime}_{k,t-1}\mathrm{diag}(\boldsymbol{\mu}_{q(\gamma^{-j}_{kt})})\big(\mu_{q(b_{jt})}\boldsymbol{\mu}_{q(\mathbf{b}^{-j}_{kt})}+\mathsf{Cov}(b_{jt},\mathbf{b}^{-j}_{kt})\big).
    \end{align*}
Define $\mu_{q(\varepsilon_{-jt})}=y_t - \mathbf{x}^{-j\,\prime}_{k,t-1}\mathrm{diag}(\boldsymbol{\mu}_{q(\gamma^{-j}_{kt})}){\boldsymbol{\mu}}_{q(\mathbf{b}^{-j}_{kt})}- \sum_{m=1,m\neq k}^K \mathbf{x}_{m,t-1}^\prime\mathrm{diag}(\boldsymbol{\mu}_{q(\gamma_{mt})}){\boldsymbol{\mu}}_{q(\mathbf{b}_{mt})}$, the latter equation reduces to:
\begin{align*}
        \mathbb{E}_{-\gamma_{jt}}\left[b_{jt}x_{jt-1}\varepsilon_{-j,t}\right] &=  \mu_{q(b_{jt})}x_{jt-1}\mu_{q(\varepsilon_{-jt})} - x_{jt-1}\mathbf{x}^{-j\,\prime}_{k,t-1}\mathrm{diag}(\boldsymbol{\mu}_{q(\gamma^{-j}_{kt})})[\boldsymbol{\Sigma}_{q(\mathbf{b}_{kt})}]_{-j,j}.
    \end{align*}
To conclude, Eq.\ref{eq:q_gamma} is the kernel of a Bernoulli distribution as in Proposition \ref{prop:up_gamma}.
\end{proof}

\begin{proposition}\label{prop:up_omega_app}
	\it The optimal variational density for the parameter ${\boldsymbol{\omega}}_j$ is a multivariate Gaussian $q^\ast({\boldsymbol{\omega}}_j)=\mathsf{N}_{n+1}(\boldsymbol{\mu}_{q(\omega_j)},\boldsymbol{\Sigma}_{q(\omega_j)})$, where:
	\begin{align}
		\boldsymbol{\Sigma}_{q(\omega_j)} &= (\mathsf{Diag}(0,\boldsymbol{\mu}_{q(z_j)})+\mu_{q(1/\xi_j^2)}\mathbf{Q})^{-1} ,\qquad \boldsymbol{\mu}_{q(\omega_j)}=\boldsymbol{\Sigma}_{q(\omega_j)}(0,\boldsymbol{\mu}^\intercal_{q(\bar{\gamma}_j)})^\intercal,
	\end{align}
	with $\boldsymbol{\mu}_{q(\bar{\gamma}_j)}=\boldsymbol{\mu}_{q(\gamma_j)}-1/2\boldsymbol{\iota}_{n}$.
\end{proposition}

\begin{proof}
	The full conditional distribution of $\boldsymbol{\omega}_j$ is defined in Eq.\eqref{eq:full_om}. Then, the optimal variational density is given by:
	\begin{equation}
		\begin{aligned}\label{eq:q_omega}
			\log q^\ast(\boldsymbol{\omega}_j) &\propto \mathbb{E}_{-\omega_j}[\log p(\boldsymbol{\omega}_j|\mathrm{rest})] \\ &\propto\boldsymbol{\omega}_j^\prime\boldsymbol{\mu}_{q(\bar{\gamma}_j)}-\frac{1}{2}\boldsymbol{\omega}_j^\prime\mathsf{Diag}(\boldsymbol{\mu}_{q(z_j)})\boldsymbol{\omega}_j -\frac{1}{2}\mu_{q(1/\xi_j^2)}{\boldsymbol{\omega}}_j^\prime\mathbf{Q}{\boldsymbol{\omega}}_j\\
			&\propto -\frac{1}{2}\left({\boldsymbol{\omega}}_j^\prime(\mathsf{Diag}(0,\boldsymbol{\mu}_{q(z_j)})+\mu_{q(1/\xi_j^2)}\mathbf{Q}){\boldsymbol{\omega}}_j -2{\boldsymbol{\omega}}_j^\prime(0,\boldsymbol{\mu}^\prime_{q(\bar{\gamma}_j)})^\prime\right),
		\end{aligned}
	\end{equation}
	where $\boldsymbol{\mu}_{q(\bar{\gamma}_j)}=\boldsymbol{\mu}_{q(\gamma_j)}-1/2\boldsymbol{\iota}_{n}$. Equation \ref{eq:q_omega} is the kernel of a multivariate Gaussian distribution as in Proposition \ref{prop:up_omega_app}.
\end{proof}

Recall from Eq.(2.2) that the dynamic of the regression coefficient $\beta_{jt}$ is a compounding process of the latent state $b_{jt}$ and the indicator variable $\gamma_{jt}$. Proposition \ref{prop:q_beta_tilde_app} builds upon Propositions \ref{prop:up_beta} and \ref{prop:up_gamma} and provides the optimal variational density for the time-varying regression coefficients for variables $j=1,\ldots,p_k$ in group $k$.

\begin{proposition}\label{prop:q_beta_tilde_app}
    Let $q^\ast(\mathbf{b}_k)$ and $q^\ast(\gamma_{jt})$, for $j=1,\ldots,p_k$, be the optimal variational densities presented in Propositions \ref{prop:up_beta} and \ref{prop:up_gamma}. Define $\boldsymbol{\beta}_k=\mathbf{\Gamma}_k\mathbf{b}_k$ with $\mathbf{\Gamma}_k=\mathsf{diag}(\boldsymbol{\iota}_{p_k}^\prime,\boldsymbol{\gamma}_{k1}^\prime,\ldots,\boldsymbol{\gamma}_{kn}^\prime)$, where $\boldsymbol{\gamma}_{kt}=\left(\gamma_{1t},\ldots,\gamma_{p_kt}\right)^\prime$ and $\boldsymbol{\iota}_{p_k}$ is a $p_k$-dimensional vector of ones. The optimal variational density of $\boldsymbol{\beta}_k$ is given by a mixture of multivariate Gaussian distributions:
	\begin{equation}
		q^\ast(\boldsymbol{\beta}_k)=\sum_{\mathbf{s}\in\mathcal{S}}\,w_s\,\mathsf{N}_{p_k(n+1)}(\mathbf{D}_{s}\boldsymbol{\mu}_{q(\mathbf{b}_k)},\mathbf{D}_{s}^{1/2}\mathbf{\Sigma}_{q(\mathbf{b}_k)}\mathbf{D}_{s}^{1/2}),
	\end{equation}
	where $\mathcal{S}$ is a sequence of $\{0,1\}$ of length $p_k n$ with cardinality $|\mathcal{S}|=2^{p_k n}$, the diagonal matrix $\mathbf{D}_{s}=\mathsf{diag}(1,s_{11},\ldots,s_{p_k,n})$, and mixing weights:
	\begin{equation}
		w_s = \prod_{j=1}^{p_k}\prod_{t=1}^n\mu_{q(\gamma_{jt})}^{s_{jt}}(1-\mu_{q(\gamma_{jt})})^{1-s_{jt}},
	\end{equation}
	where $\mathbf{s}=(s_{11},\ldots,s_{jt},\ldots,s_{p_k,n})\in\mathcal{S}$ is an element in $\mathcal{S}$. Moreover, the mean and variance can be computed analytically:
	\begin{align}
		\boldsymbol{\mu}_{q(\beta_k)}&=\boldsymbol{\mu}_{q(\Gamma_k)}\boldsymbol{\mu}_{q(\mathbf{b}_k)},\\ \mathbf{\Sigma}_{q(\beta_k)}&=(\boldsymbol{\mu}_{q(\gamma_k)}\boldsymbol{\mu}_{q(\gamma_k)}^\prime+\mathbf{W}_{\mu_{q(\gamma_k)}})\odot\mathbf{\Sigma}_{q(\mathbf{b}_k)}+ \mathbf{W}_{\mu_{q(\gamma_k)}}\odot\boldsymbol{\mu}_{q(\mathbf{b}_k)}\boldsymbol{\mu}_{q(\mathbf{b}_k)}^\prime,
	\end{align}
	where $\mathbf{W}_{\mu_{q(\gamma_k)}}$ is a diagonal matrix with elements $\mu_{q(\gamma_{jt})}(1-\mu_{q(\gamma_{jt})})$, and $\boldsymbol{\mu}_{q(\gamma_k)}$ the collection of $\boldsymbol{\mu}_{q(\gamma_{jt})}$ for $j=1,\ldots,p_k$ and $t=1,\ldots,n$.
\end{proposition}

\begin{proof}
    Without loss of generality, assume that the $k$-th group has only one element $\mathbf{b}_j$ and $p_k=1$. Recall that under the mean-field variational Bayes setting we have that \begin{equation}\label{eq:q_fact_beta_gamma}
		q(\mathbf{b}_j,\gamma_{j1},\ldots,\gamma_{jn})=q(\mathbf{b}_j)\prod_{t=1}^nq(\gamma_{jt}).
	\end{equation}
	For the sake of simplicity, in what follows, we drop the index $j$ and define $\boldsymbol{\gamma}=\mathsf{diag}(\mathbf{\Gamma})$ the diagonal elements in $\mathbf{\Gamma}$. Consider the following transformation of random variables $(\boldsymbol{\gamma} =\boldsymbol{\gamma}, \boldsymbol{\beta}=\mathbf{\Gamma}\mathbf{b})$, so that $\mathbf{b}=\mathbf{\Gamma}^{-1}\boldsymbol{\beta}$. Hence, it follows that:
	\begin{equation}
		\mathbf{J}=\left[\begin{array}{cc}
			\nabla_\gamma (\gamma_1,\ldots,\gamma_n)^\prime & \nabla_{\mathbf{b}} (\gamma_1,\ldots,\gamma_n)^\prime \\
			\nabla_\gamma \mathbf{\Gamma}^{-1}\boldsymbol{\beta} & \nabla_{\boldsymbol{\beta}} \mathbf{\Gamma}^{-1}\boldsymbol{\beta}
		\end{array}\right] =\left[\begin{array}{cc}
		\mathbf{I}_n & \mathbf{0} \\
		\nabla_\gamma \mathbf{\Gamma}^{-1}\boldsymbol{\beta} & \mathbf{\Gamma}^{-1}
	\end{array}\right],
	\end{equation}
	and so $|\mathbf{J}|=|\mathbf{\Gamma}^{-1}|$. The joint distribution of $(\boldsymbol{\beta},\gamma_1,\ldots,\gamma_n)$ can be written as:
	\begin{equation}\label{eq:q_joint_betatilde_gamma}
		q(\boldsymbol{\beta},\gamma_1,\ldots,\gamma_n)=|\mathbf{\Gamma}^{-1}|q(\mathbf{\Gamma}^{-1}\boldsymbol{\beta})\prod_{t=1}^nq(\gamma_{jt}) = f(\boldsymbol{\beta}|\gamma_1,\ldots,\gamma_n)f(\gamma_1,\ldots,\gamma_n),
	\end{equation}
	where $q$ are then replaced by the optimal elements $q^\ast$. For the conditional distribution in \eqref{eq:q_joint_betatilde_gamma}, we have that:
	\begin{equation}
		f(\boldsymbol{\beta}|\boldsymbol{\gamma}) = |\mathbf{\Gamma}^{-1}|\phi_{n+1}(\mathbf{\Gamma}^{-1}\boldsymbol{\beta}|\boldsymbol{\mu}_{q(\mathbf{b})},\mathbf{\Sigma}_{q(\mathbf{b})}),
	\end{equation}
	where $\phi_{n+1}(\cdot|\boldsymbol{\mu},\mathbf{\Sigma})$ is the density function of a multivariate Gaussian. After some computations we have that $f(\boldsymbol{\beta}|\boldsymbol{\gamma}) = \phi_{n+1}(\boldsymbol{\beta}|\boldsymbol{\mu}(\boldsymbol{\gamma}),\mathbf{\Sigma}(\boldsymbol{\gamma}))$ with mean vector $\boldsymbol{\mu}(\boldsymbol{\gamma})=\mathbf{\Gamma}\boldsymbol{\mu}_{q(\mathbf{b})}$ and covariance matrix $\mathbf{\Sigma}(\boldsymbol{\gamma})=\mathbf{\Gamma}^{1/2}\mathbf{\Sigma}_{q(\mathbf{b})}\mathbf{\Gamma}^{1/2}$. The marginal for $\boldsymbol{\beta}$ can be found as:
	\begin{equation}\label{eq:q_marg_betatilde}
		q(\boldsymbol{\beta})=\sum_{\mathbf{s}\in\mathcal{S}}\phi_{n+1}(\boldsymbol{\beta}|\boldsymbol{\mu}(\boldsymbol{\gamma}=\mathbf{s}),\mathbf{\Sigma}(\boldsymbol{\gamma}=\mathbf{s}))\prod_{t=1}^nq(\gamma_{t}=s_t),
	\end{equation}
	where $\mathcal{S}$ denotes the domain of $\boldsymbol{\gamma}=(1,\gamma_1,\ldots,\gamma_n)$, and it is composed by all the possible sequences of $\{0,1\}$ of length n, since the first element is fixed to be 1. The latter set has cardinality $|\mathcal{S}|=2^n$. The distributional result concerning $\boldsymbol{\beta}$ is therefore proven.\par 
	Now compute the marginal mean recall that $\mathbb{E}_x(x)=\mathbb{E}_y(\mathbb{E}_x(x|y))$. Hence $\mathbb{E}_q(\boldsymbol{\beta})=\mathbb{E}_\gamma(\mathbf{\Gamma}\boldsymbol{\mu}_{q(\mathbf{b})})=\boldsymbol{\mu}_{q(\Gamma)}\boldsymbol{\mu}_{q(\mathbf{b})}$. The marginal variance-covariance matrix is then computed as $\mathsf{Var}_q(\boldsymbol{\beta}) =\mathbb{E}(\boldsymbol{\beta}\boldsymbol{\beta}^\prime)-\mathbb{E}(\boldsymbol{\beta})\mathbb{E}(\boldsymbol{\beta})^\prime$ where 
	\begin{align}
		\mathbb{E}(\boldsymbol{\beta}\boldsymbol{\beta}^\prime) &= \mathbb{E}(\mathbf{\Gamma}{\mathbf{b}}(\mathbf{\Gamma}{\mathbf{b}})^\prime) = \mathbb{E}(\mathbf{\Gamma}{\mathbf{b}}{\mathbf{b}}^\prime\mathbf{\Gamma}) = \mathbb{E}(\boldsymbol{\gamma}\boldsymbol{\gamma}^\prime\odot{\mathbf{b}}{\mathbf{b}}^\prime)= \mathbb{E}(\boldsymbol{\gamma}\boldsymbol{\gamma}^\prime)\odot\mathbb{E}({\mathbf{b}}{\mathbf{b}}^\prime) \nonumber\\
		&= (\boldsymbol{\mu}_{q(\gamma)}\boldsymbol{\mu}_{q(\gamma)}^\prime+\mathbf{W}_{\mu_{q(\gamma)}})\odot(\boldsymbol{\mu}_{q(\mathbf{b})}\boldsymbol{\mu}_{q(\mathbf{b})}^\prime+\mathbf{\Sigma}_{q(\mathbf{b})}),
	\end{align}
	where $\mathbf{W}_{\mu_{q(\gamma)}}$ is a diagonal matrix with elements $(1,\{\mu_{q(\gamma_{t})}(1-\mu_{q(\gamma_{t})})\}_{t=1}^n)$. Plug-in the latter in the formula for $\mathsf{Var}_q(\boldsymbol{\beta})$ and recall the analytical form of the mean $\mathbb{E}(\boldsymbol{\beta})$. After some simplification we end up with $\mathbf{\Sigma}_{q(\beta)}=(\boldsymbol{\mu}_{q(\gamma)}\boldsymbol{\mu}_{q(\gamma)}^\prime+\mathbf{W}_{\mu_{q(\gamma)}})\odot\mathbf{\Sigma}_{q(\mathbf{b})}+ \mathbf{W}_{\mu_{q(\gamma)}}\odot\boldsymbol{\mu}_{q(\mathbf{b})}\boldsymbol{\mu}_{q(\mathbf{b})}^\prime$, which concludes the proof. The result for a generic group $k$ of size $p_k$ can be proven following the same steps and adjusting for the array's dimensions.
\end{proof}

\begin{proposition}\label{prop:up_logsima_app}
	\it Let $\boldsymbol{\varepsilon}^2=\boldsymbol{\varepsilon}\odot\boldsymbol{\varepsilon}$ with components $[\boldsymbol{\varepsilon}^2]_t=(y_t-\boldsymbol{\beta}_t\mathbf{x}_{t-1})^2$. Assuming a Gaussian Markov Random Field (GMRF) approximation $q^\ast(\mathbf{h})=\mathsf{N}_{n+1}(\boldsymbol{\mu}_{q(h)},\mathbf{\Sigma}_{q(h)})$, with mean vector $\boldsymbol{\mu}_{q(h)}$ and variance-covariance matrix $\mathbf{\Sigma}_{q(h)}$, an iterative algorithm can be set as:
	\begin{align}
		\boldsymbol{\Sigma}_{q(h)}^{new} &= \left[\nabla_{\boldsymbol{\mu}_{q(h)},\boldsymbol{\mu}_{q(h)}}^2 S(\boldsymbol{\mu}_{q(h)}^{old},\boldsymbol{\Sigma}_{q(h)}^{old})\right]^{-1} \\
		\boldsymbol{\mu}_{q(h)}^{new} &= \boldsymbol{\mu}_{q(h)}^{old} + \boldsymbol{\Sigma}_{q(h)}^{new}\nabla_{\boldsymbol{\mu}_{q(h)}} S(\boldsymbol{\mu}_{q(h)}^{old},\boldsymbol{\Sigma}_{q(h)}^{old}).
	\end{align}
	where 
	\begin{align}
		\nabla_{\boldsymbol{\mu}_{q(h)}} S(\boldsymbol{\mu}_{q(h)}^{old},\boldsymbol{\Sigma}_{q(h)}^{old}) &= -\frac{\boldsymbol{\iota}_n}{2}+\frac{1}{2}\mathbb{E}_q(\boldsymbol{\varepsilon}^2)\odot\mathrm{e}^{-\boldsymbol{\mu}_{q(h)}^{old}+\boldsymbol{\sigma}^{2\,old}_{q(h)}/2} -\mu_{q(1/\nu^2)}\mathbf{Q}\boldsymbol{\mu}_{q(h)}^{old},
	\end{align}
	and
	\begin{align}
		\nabla_{\boldsymbol{\mu}_{q(h)},\boldsymbol{\mu}_{q(h)}}^2 S(\boldsymbol{\mu}_{q(h)}^{old},\boldsymbol{\Sigma}_{q(h)}^{old}) &= -\frac{1}{2}\mathsf{Diag}(\mathbb{E}_q(\boldsymbol{\varepsilon}^2)\odot\mathrm{e}^{-\boldsymbol{\mu}_{q(h)}^{old}+\boldsymbol{\sigma}^{2\,old}_{q(h)}/2}) -\mu_{q(1/\nu^2)}\mathbf{Q},
	\end{align}
	denote the first and second derivative of $S(\boldsymbol{\mu}_{q(h)},\boldsymbol{\Sigma}_{q(h)})$ with respect to $\boldsymbol{\mu}_{q(h)}$ and evaluated at $(\boldsymbol{\mu}_{q(h)}^{old},\boldsymbol{\Sigma}_{q(h)}^{old})$, and  $\boldsymbol{\sigma}^2_{q(h)}=\mathsf{diag}(\boldsymbol{\Sigma}_{q(h)})$.  
\end{proposition}

\begin{proof}
	The updating scheme follows the algorithm provided in \cite{rhode_wand2016} for Gaussian variational approximations. The function $S$ is called \textit{non-entropy function} and it is given by $S(\boldsymbol{\mu}_{q(h)},\boldsymbol{\Sigma}_{q(h)})=\mathbb{E}_q(\log p(\mathbf{y},\boldsymbol{\vartheta}))$:
	\begin{align}
		\begin{aligned}
			S(\boldsymbol{\mu}_{q(h)},\boldsymbol{\Sigma}_{q(h)}) &= -\frac{\boldsymbol{\iota}_n^\prime}{2}\boldsymbol{\mu}_{q(h)}-\frac{1}{2}\mathbb{E}^\prime_q(\boldsymbol{\varepsilon}^2)\mathrm{e}^{-\boldsymbol{\mu}_{q(h)}+\boldsymbol{\sigma}^2_{q(h)}/2} \\
			&\qquad -\frac{1}{2}\mu_{q(1/\nu^2)}\left(\boldsymbol{\mu}_{q(h)}^\prime\mathbf{Q}\boldsymbol{\mu}_{q(h)}+\mathsf{tr}\left\{\boldsymbol{\Sigma}_{q(h)}\mathbf{Q}\right\}\right),
		\end{aligned}
	\end{align}
	where $\boldsymbol{\varepsilon}^2=\boldsymbol{\varepsilon}\odot\boldsymbol{\varepsilon}$ with components $[\boldsymbol{\varepsilon}^2]_t=(y_t-\boldsymbol{\beta}_t\mathbf{x}_{t})^2$, and $\boldsymbol{\sigma}^2_{q(h)}=\mathsf{diag}(\boldsymbol{\Sigma}_{q(h)})$.
	Then, the first derivative with respect to the variational mean vector $\boldsymbol{\mu}_{q(h)}$ is given by
	\begin{align}
		\nabla_{\boldsymbol{\mu}_{q(h)}} S(\boldsymbol{\mu}_{q(h)},\boldsymbol{\Sigma}_{q(h)}) &= -\frac{\boldsymbol{\iota}_n}{2}+\frac{1}{2}\mathbb{E}_q(\boldsymbol{\varepsilon}^2)\odot\mathrm{e}^{-\boldsymbol{\mu}_{q(h)}+\boldsymbol{\sigma}^2_{q(h)}/2} -\mu_{q(1/\nu^2)}\mathbf{Q}\boldsymbol{\mu}_{q(h)}.
	\end{align}
	Moreover, derive $\nabla_{\boldsymbol{\mu}_{q(h)}} S(\boldsymbol{\mu}_{q(h)},\boldsymbol{\Sigma}_{q(h)})$ again with respect to $\boldsymbol{\mu}_{q(h)}$:
	\begin{align}
		\nabla_{\boldsymbol{\mu}_{q(h)},\boldsymbol{\mu}_{q(h)}}^2 S(\boldsymbol{\mu}_{q(h)},\boldsymbol{\Sigma}_{q(h)}) &= -\frac{1}{2}\mathsf{Diag}(\mathbb{E}_q(\boldsymbol{\varepsilon}^2)\odot\mathrm{e}^{-\boldsymbol{\mu}_{q(h)}+\boldsymbol{\sigma}^2_{q(h)}/2}) -\mu_{q(1/\nu^2)}\mathbf{Q}.
	\end{align}
\end{proof}

\begin{remark}\label{prop:up_sigmasq}
	Under the multivariate Gaussian approximation of $q(\mathbf{h})$ with mean vector $\boldsymbol{\mu}_{q(h)}$ and covariance matrix $\mathbf{\Sigma}_{q(h)}$, the optimal density of the vector $\boldsymbol{\sigma}^2=\exp\{\mathbf{h}\}$, namely $q^\ast(\boldsymbol{\sigma}^2)$, is a multivariate log-normal distribution such that:
	\begin{align}
		\mathbb{E}_q[\sigma_t^2] &= \exp\{\mu_{q(h_t)}+1/2\sigma^2_{q(h_t)}\},\\\label{eq:sigmat}
            \mathbb{E}_q[1/\sigma_t^2] &= \exp\{-\mu_{q(h_t)}+1/2\sigma^2_{q(h_t)}\},\\
		\mathsf{Var}_q[\sigma_t^2] &= \exp\{2\mu_{q(h_t)}+\sigma^2_{q(h_t)}\}(\exp\{\sigma^2_{q(h_t)}\}-1), \\
		\mathsf{Cov}_q[\sigma_t^2,\sigma_{t+1}^2] &= \exp\{\mu_{q(h_t)}+\mu_{q(h_{t+1})}+1/2(\sigma^2_{q(h_t)}+\sigma^2_{q(h_{t+1})})\}(\exp\{\mathsf{Cov}_q[h_t,h_{t+1}]\}-1). \nonumber
	\end{align}
\end{remark}

\begin{proposition}\label{prop:up_sigmasq_homo}
	\it The optimal variational density for the homoskedastic variance $\sigma^2$ is an inverse-gamma $q^\ast(\sigma^2)=\mathsf{IG}(A_{q(\sigma^2)},B_{q(\sigma^2)})$ where:
	\begin{align}
		A_{q(\sigma^2)} &= A_\sigma+\frac{n}{2},
		B_{q(\sigma^2)} = B_\sigma+\frac{1}{2}\mathbb{E}_q\left[\boldsymbol{\varepsilon}^\prime\boldsymbol{\varepsilon}\right],
	\end{align}
	where:
	\begin{align*}
		\mathbb{E}_{-\sigma^2}\left[\boldsymbol{\varepsilon}^\prime\boldsymbol{\varepsilon}\right] &=\mathbf{y}^\prime\mathbf{y} -2\left(\sum_{j=1}^p\mathbf{X}_j\boldsymbol{\mu}_{q(\Gamma_j)}\boldsymbol{\mu}_{q(\mathbf{b}_j)}\right)^\prime\mathbf{y} + \sum_{j=1}^p\mathsf{tr}\left\{\left(\boldsymbol{\mu}_{q(\mathbf{b}_j)}\boldsymbol{\mu}_{q(\mathbf{b}_j)}^\prime+\mathbf{\Sigma}_{q(\mathbf{b}_j)}\right)\boldsymbol{\mu}_{q(\Gamma_j)}\mathbf{X}^2_j\right\}\\
		&\qquad +\sum_{j=1}^p\boldsymbol{\mu}_{q(\mathbf{b}_j)}^\prime\boldsymbol{\mu}_{q(\Gamma_j)}\mathbf{X}_j\sum_{k=1,k\neq j}^p\mathbf{X}_k\boldsymbol{\mu}_{q(\Gamma_k)}\boldsymbol{\mu}_{q(\mathbf{b}_k)}.
	\end{align*}
\end{proposition}
\begin{proof}
	The full conditional distribution of $\sigma^2$ given the rest $p(\sigma^2|\mbox{rest})$ is derived in Eq.\eqref{eq:full_s2}. Thus, the optimal variational density is given by:
	\begin{equation}
		\begin{aligned}\label{eq:q_sigma2_homo}
			\log q^\ast(\sigma^2) &\propto \mathbb{E}_{-\sigma^2}[\log p(\sigma^2|\mathrm{rest})] \\ 
			&\propto -(A_\sigma+\frac{n}{2}+1)\log\sigma^2 -\frac{1}{\sigma^2}\left\{B_\sigma+\frac{1}{2}\mathbb{E}_{-\sigma^2}\left[\boldsymbol{\varepsilon}^\prime\boldsymbol{\varepsilon}\right]\right\},
		\end{aligned}
	\end{equation}
	where:
	\begin{align*}
		\mathbb{E}_{-\sigma^2}\left[\boldsymbol{\varepsilon}^\prime\boldsymbol{\varepsilon}\right] &= \mathbb{E}_{-\sigma^2}\left[\left(\mathbf{y}-\sum_{j=1}^p\mathbf{X}_j\mathbf{\Gamma}_j\mathbf{b}_j\right)^\prime\left(\mathbf{y}-\sum_{j=1}^p\mathbf{X}_j\mathbf{\Gamma}_j\mathbf{b}_j\right)\right]=\mathbf{y}^\prime\mathbf{y} -2\left(\sum_{j=1}^p\mathbb{E}_{-\sigma^2}\left[\mathbf{X}_j\mathbf{\Gamma}_j\mathbf{b}_j\right]\right)^\prime\mathbf{y} \\
		&\qquad +\sum_{j=1}^p\mathbb{E}_{-\sigma^2}\left[\mathbf{b}_j^\prime\mathbf{\Gamma}_j\mathbf{X}_j\mathbf{X}_j\mathbf{\Gamma}_j\mathbf{b}_j+\mathbf{b}_j^\prime\mathbf{\Gamma}_j\mathbf{X}_j\sum_{k=1,k\neq j}^p\mathbf{X}_k\mathbf{\Gamma}_k\mathbf{b}_k\right]\\
		&=\mathbf{y}^\prime\mathbf{y} -2\left(\sum_{j=1}^p\mathbf{X}_j\boldsymbol{\mu}_{q(\Gamma_j)}\boldsymbol{\mu}_{q(\mathbf{b}_j)}\right)^\prime\mathbf{y} + \sum_{j=1}^p\mathsf{tr}\left\{\mathbb{E}_{\mathbf{b}_j}\left[\mathbf{b}_j\mathbf{b}_j^\prime\right]\boldsymbol{\mu}_{q(\Gamma_j)}\mathbf{X}^2_j\right\}\\
		&\qquad +\sum_{j=1}^p\boldsymbol{\mu}_{q(\mathbf{b}_j)}^\prime\boldsymbol{\mu}_{q(\Gamma_j)}\mathbf{X}_j\sum_{k=1,k\neq j}^p\mathbf{X}_k\boldsymbol{\mu}_{q(\Gamma_k)}\boldsymbol{\mu}_{q(\mathbf{b}_k)}\\
		&=\mathbf{y}^\prime\mathbf{y} -2\left(\sum_{j=1}^p\mathbf{X}_j\boldsymbol{\mu}_{q(\Gamma_j)}\boldsymbol{\mu}_{q(\mathbf{b}_j)}\right)^\prime\mathbf{y} + \sum_{j=1}^p\mathsf{tr}\left\{\left(\boldsymbol{\mu}_{q(\mathbf{b}_j)}\boldsymbol{\mu}_{q(\mathbf{b}_j)}^\prime+\mathbf{\Sigma}_{q(\mathbf{b}_j)}\right)\boldsymbol{\mu}_{q(\Gamma_j)}\mathbf{X}^2_j\right\}\\
		&\qquad +\sum_{j=1}^p\boldsymbol{\mu}_{q(\mathbf{b}_j)}^\prime\boldsymbol{\mu}_{q(\Gamma_j)}\mathbf{X}_j\sum_{k=1,k\neq j}^p\mathbf{X}_k\boldsymbol{\mu}_{q(\Gamma_k)}\boldsymbol{\mu}_{q(\mathbf{b}_k)}.
	\end{align*}
	Equation \ref{eq:q_sigma2_homo} represents the kernel of a Inverse-Gamma distribution as in Proposition \ref{prop:up_sigmasq_homo}.
\end{proof}

\begin{proposition}\label{prop:up_z}
	\it The optimal variational density for the $z_{jt}$ parameters is a Polya-Gamma $q^\ast(z_{jt})=\mathsf{PG}(1,\sqrt{\mu_{q(\omega^2_{jt})}})$ and define
	\begin{equation}
		\mu_{q(z_{jt})} = \mathbb{E}_q\left[z_{jt}\right] = \frac{1}{2}\frac{1}{\sqrt{\mu_{q(\omega^2_{jt})}}}\tanh\left(\frac{\sqrt{\mu_{q(\omega^2_{jt})}}}{2}\right)
	\end{equation}
\end{proposition}
\begin{proof}
 The full conditional distribution of $z_{jt}$ given the rest is proportional to Eq.\eqref{eq:full_z}. Then the optimal variational density is such that
\begin{equation}
		\begin{aligned}\label{eq:q_z}
			\log q^\ast(z_{jt}) &\propto -z_{jt}\mu_{q(\omega_{jt}^2)} +\log p(z_{jt}).
		\end{aligned}
\end{equation}
Equation \ref{eq:q_z} represents the kernel of a Polya-Gamma distribution as in Proposition \ref{prop:up_z}.
\end{proof}

\begin{proposition}\label{prop:up_eta}
	\it The optimal variational density for the variance parameter $\eta_j^2$ is an inverse-gamma distribution $q^\ast(\eta_j^2)=\mathsf{IG}(A_{q(\eta_j^2)},B_{q(\eta_j^2)})$, where:
	\begin{align}
		A_{q(\eta_j^2)} &=  A_\eta + \frac{n+1}{2},\qquad B_{q(\eta_j^2)} =  B_\eta +\frac{1}{2}\left(\boldsymbol{\mu}_{q(\mathbf{b}_j)}^\prime\mathbf{Q}\boldsymbol{\mu}_{q(\mathbf{b}_j)} +\mathsf{tr}\left\{\boldsymbol{\Sigma}_{q(\mathbf{b}_j)}\mathbf{Q}\right\}\right).
	\end{align}
\end{proposition}
\begin{proof}
	The full conditional distribution of $\eta^2_j$ given the rest $p(\eta^2_j|\mbox{rest})$ is described in Eq.\eqref{eq:full_eta2}. Then, the optimal variational density is given by:
	\begin{equation}
		\begin{aligned}\label{eq:q_eta}
			\log q^\ast(\eta_j^2) &\propto \mathbb{E}_{-\eta^2_j}[\log p(\eta^2_j|\mathrm{rest})] \\	
			&\propto -\frac{n+1}{2}\log\eta_j^2-\frac{1}{2\eta_j^2}\mathbb{E}_{-\eta_j^2}\left[{\mathbf{b}}_j^\prime\mathbf{Q}{\mathbf{b}}_j\right]  -(A_\eta+1)\log\eta_j^2-\frac{B_\eta}{\eta_j^2} \\
			&\propto -\left(\frac{n}{2}+A_\eta+1\right)\log\eta_j^2-\frac{1}{\eta_j^2}\left(B_\eta+\frac{1}{2}\left(\boldsymbol{\mu}_{q(\mathbf{b}_j)}^\prime\mathbf{Q}\boldsymbol{\mu}_{q(\mathbf{b}_j)} +\mathsf{tr}\left\{\boldsymbol{\Sigma}_{q(\mathbf{b}_j)}\mathbf{Q}\right\}\right)\right).
		\end{aligned}
	\end{equation}
	Equation \ref{eq:q_eta} represents the kernel of an Inverse-Gaussian distribution as in Proposition \ref{prop:up_eta}.
\end{proof}

\begin{proposition}\label{prop:up_xi}
	\it The optimal variational density for the variance parameter $\xi_j^2$ is an inverse-gamma distribution $q^\ast(\xi_j^2)=\mathsf{IG}(A_{q(\xi_j^2)},B_{q(\xi_j^2)})$, where:
	\begin{align}
		A_{q(\xi_j^2)} &=  A_\xi + \frac{n+1}{2},\qquad B_{q(\xi_j^2)} =  B_\xi +\frac{1}{2}\left(\boldsymbol{\mu}_{q(\omega_j)}^\prime\mathbf{Q}\boldsymbol{\mu}_{q(\omega_j)} +\mathsf{tr}\left\{\boldsymbol{\Sigma}_{q(\omega_j)}\mathbf{Q}\right\}\right).
	\end{align}
\end{proposition}
\begin{proof}
	The full conditional distribution of $\xi^2_j$ given the rest $p(\xi^2_j|\mbox{rest})$ is described in Eq.\eqref{eq:full_xi2}. Thus, the optimal variational density is given by:
	\begin{equation}
		\begin{aligned}\label{eq:q_xi}
			\log q^\ast(\xi_j^2) &\propto \mathbb{E}_{-\xi^2_j}[\log p(\xi^2_j|\mathrm{rest})] \\
			&\propto-\frac{n+1}{2}\log\xi_j^2-\frac{1}{2\xi_j^2}\mathbb{E}_{-\xi_j^2}\left[{\boldsymbol{\omega}}_j^\prime\mathbf{Q}{\boldsymbol{\omega}}_j\right] -(A_\xi+1)\log\xi_j^2-\frac{B_\xi}{\xi_j^2} \\
			&\propto -(\frac{n}{2}+A_\xi+1)\log\xi_j^2-\frac{1}{\xi_j^2}\left(B_\xi+\frac{1}{2}\left(\boldsymbol{\mu}_{q(\omega_j)}^\prime\mathbf{Q}\boldsymbol{\mu}_{q(\omega_j)} +\mathsf{tr}\left\{\boldsymbol{\Sigma}_{q(\omega_j)}\mathbf{Q}\right\}\right)\right).
		\end{aligned}
	\end{equation}
	Equation \ref{eq:q_xi} represents the kernel of an Inverse-Gaussian distribution as in Proposition \ref{prop:up_xi}.
\end{proof}

\begin{proposition}\label{prop:up_nu}
	\it The optimal variational density for the variance parameter $\nu^2$ is an inverse-gamma distribution $q^\ast(\nu^2)=\mathsf{IG}(A_{q(\nu^2)},B_{q(\nu^2)})$, where:
	\begin{align}
		A_{q(\nu^2)} &=  A_\nu + \frac{n+1}{2},\qquad B_{q(\nu^2)} =  B_\nu +\frac{1}{2}\left(\boldsymbol{\mu}_{q(h)}^\prime\mathbf{Q}\boldsymbol{\mu}_{q(h)} +\mathsf{tr}\left\{\boldsymbol{\Sigma}_{q(h)}\mathbf{Q}\right\}\right).
	\end{align}
\end{proposition}
\begin{proof}
	The full conditional distribution of $\nu^2$ given the rest $p(\nu^2|\mbox{rest})$ is described in Eq.\eqref{eq:full_nu2}. Thus, the optimal variational density is given by:
	\begin{equation}
		\begin{aligned}\label{eq:q_nu}
			\log q^\ast(\nu^2) &\propto \mathbb{E}_{-\nu^2}[\log p(\nu^2|\mathrm{rest})] \\
			&\propto-\frac{n+1}{2}\log\nu^2-\frac{1}{2\nu^2}\mathbb{E}_{-\nu^2}\left[\mathbf{h}^\prime\mathbf{Q}\mathbf{h}\right] -(A_\nu+1)\log\nu^2-\frac{B_\nu}{\nu^2} \\
			&\propto -\left(\frac{n}{2}+A_\nu+1\right)\log\nu^2-\frac{1}{\nu^2}\left(B_\nu+\frac{1}{2}\left(\boldsymbol{\mu}_{q(h)}^\prime\mathbf{Q}\boldsymbol{\mu}_{q(h)} +\mathsf{tr}\left\{\boldsymbol{\Sigma}_{q(h)}\mathbf{Q}\right\}\right)\right).
		\end{aligned}
	\end{equation}
	Equation \ref{eq:q_nu} represents the kernel of an Inverse-Gaussian distribution as in Proposition \ref{prop:up_nu}.
\end{proof}

\subsection{Smoothing the sparsity dynamics}
\label{app:smoothing}

\begin{proposition}\label{prop:up_sm_gamma}
    A smooth estimate for the trajectory of the inclusion probabilities can be achieved assuming $\widetilde{q}(\boldsymbol{\gamma}_j)=\prod_{t=1}^n \widetilde{q}(\gamma_{jt})$ such that $\widetilde{q}(\gamma_{jt})=\mathsf{Bern}(\mathrm{expit}(\mathbf{w}_t^\prime\mathbf{f}_j))$ with constraints on the mean. Therefore, the expectation of the joint vector $\boldsymbol{\gamma}_j=(\gamma_{j1},\ldots,\gamma_{jn})^\prime$ is equal to $\mathbb{E}_{\widetilde{q}}(\boldsymbol{\gamma}_j)=\mathbf{W}\mathbf{f}_j$, where $\mathbf{W}$ is a $n\times k$ B-spline basis matrix. The optimal value of $\mathbf{f}_j$ is the solution of the optimization problem $\widehat{\mathbf{f}}_j = \arg\max_{\mathbf{f}_j\in\mathbb{R}^k} \psi(\mathbf{f}_j)$ where 
    \begin{equation*}
        \psi(\mathbf{f}_j) = \sum_{t=1}^{n}\left[(\omega_{q(\gamma_{jt})}-\mathbf{w}^\prime_t\mathbf{f}_j)\mathrm{expit}(\mathbf{w}^\prime_t\mathbf{f}_j) +\log(1+\exp(\mathbf{w}^\prime_t\mathbf{f}_j))\right],
    \end{equation*}
    such that the gradient is equal to
$\nabla_{\mathbf{f}} \psi(\mathbf{f}) = \sum_{t=1}^{n}\mathbf{w}_t(\omega_{q(\gamma_{jt})}-\mathbf{w}^\prime_t\mathbf{f})\frac{\mathrm{expit}(\mathbf{w}^\prime_t\mathbf{f})}{1+\exp(\mathbf{w}^\prime_t\mathbf{f})}$.
\end{proposition}

\begin{proof}
	To find the best $\widetilde{q}$ that approximates $q$, minimize the Kullback-Leibler divergence $\mathcal{KL}\left(\widetilde{q}\mid\mid q\right)$. This corresponds to maximize $\mathbb{E}_{\widetilde{q}}[\log q]-\mathbb{E}_{\widetilde{q}}[\log \widetilde{q}]$ over the parameters of the approximating density $\widetilde{q}$. In our case we obtain:
	\begin{align*}
		\hat{\mathbf{f}} = \arg\max_{\mathbf{f}\in\mathbb{R}^k}\psi(\mathbf{f}) &= \arg\max_{\mathbf{f}\in\mathbb{R}^k} \left\{\mathbb{E}_{\widetilde{q}}[\log q(\boldsymbol{\gamma}_j)]-\mathbb{E}_{\widetilde{q}}[\log \widetilde{q}(\boldsymbol{\gamma})]\right\} \\
		&= \arg\max_{\mathbf{f}\in\mathbb{R}^k} \sum_{t=1}^n\left\{\mathbb{E}_{\widetilde{q}}[\log q(\gamma_{jt})]-\mathbb{E}_{\widetilde{q}}[\log \widetilde{q}(\gamma_{jt})]\right\}
	\end{align*}
	and define $\psi_t(\mathbf{f}) = \mathbb{E}_{\widetilde{q}}[\log q(\gamma_{jt})]-\mathbb{E}_{\widetilde{q}}[\log \widetilde{q}(\gamma_{jt})]$. The first term is equal to:
	\begin{align*}
		\mathbb{E}_{\widetilde{q}}[\log q(\gamma_{jt})]=\mathbb{E}_{\widetilde{q}}[\gamma_{jt}\omega_{q(\gamma_{jt})}] = \omega_{q(\gamma_{jt})}\mathrm{expit}(\mathbf{w}^\prime_t\mathbf{f}),
	\end{align*}
	while the second one can be written as:
	\begin{align*}
		\mathbb{E}_{\widetilde{q}}[\log \widetilde{q}(\gamma_{jt})]&=\mathbb{E}_{\widetilde{q}}[\gamma_{j,t}\mathbf{w}^\prime_t\mathbf{f}-\log(1+\exp(\mathbf{w}^\prime_t\mathbf{f}))] \\
		&= \mathbf{w}^\prime_t\mathbf{f}\mathrm{expit}(\mathbf{w}^\prime_t\mathbf{f})-\log(1+\exp(\mathbf{w}^\prime_t\mathbf{f})).
	\end{align*}
	Group together and obtain:
	\begin{align*}
		\psi_t(\mathbf{f}) = (\omega_{q(\gamma_{jt})}-\mathbf{w}^\prime_t\mathbf{f})\mathrm{expit}(\mathbf{w}^\prime_t\mathbf{f})+\log(1+\exp(\mathbf{w}^\prime_t\mathbf{f})).
	\end{align*}
	which defines the $t$ component of the loss function in the Proposition. Now derive $\frac{\partial}{\partial \mathbf{f}} \psi(\mathbf{f})$:
	\begin{align*}
		\nabla_{\mathbf{f}} \psi(\mathbf{f}) = \frac{\partial}{\partial \mathbf{f}} \psi(\mathbf{f}) = \sum_{t=1}^n\frac{\partial}{\partial \mathbf{f}}\psi_t(\mathbf{f}).
	\end{align*}
	Compute the latter and get:
	\begin{align*}
		\frac{\partial}{\partial \mathbf{f}}\psi_t(\mathbf{f}) &= -\mathbf{w}_t\mathrm{expit}(\mathbf{w}^\prime_t\mathbf{f})+\mathbf{w}_t(\omega_{q(\gamma_{jt})}-\mathbf{w}^\prime_t\mathbf{f})\frac{\mathrm{expit}(\mathbf{w}^\prime_t\mathbf{f})}{1+\exp(\mathbf{w}^\prime_t\mathbf{f})}+\mathbf{w}_t\mathrm{expit}(\mathbf{w}^\prime_t\mathbf{f}) \\
		&=\mathbf{w}_t(\omega_{q(\gamma_{jt})}-\mathbf{w}^\prime_t\mathbf{f})\frac{\mathrm{expit}(\mathbf{w}^\prime_t\mathbf{f})}{1+\exp(\mathbf{w}^\prime_t\mathbf{f})},
	\end{align*}
	which completes the proof.
\end{proof}

The iterative optimization to perform approximate posterior inference based on the optimal variational densities outlined above is sketched in Algorithm \ref{algo:algo1}.

	\begin{algorithm}
		\SetAlgoLined
		\kwInit{$q(\boldsymbol{\vartheta})$, $\Delta_{\boldsymbol{\vartheta}}$, $A_\nu$, $B_\nu$, $A_\eta$, $B_\eta$, $A_\xi$, $B_\xi$}
		\While{$\big(\widehat{\Delta}_{\boldsymbol{\vartheta}}>\Delta_{\boldsymbol{\vartheta}}\big)$}{
                \For{$k=1,\ldots,K$}{
				Update $q(\mathbf{b}_k)$ as in \ref{prop:up_beta};\\
			\For{$j=1,\ldots,p_k$}{
				and $q(\eta_j)$ as in \ref{prop:up_eta}; \\
				Update $q(\boldsymbol{\omega}_j)$ as in \ref{prop:up_omega_app} and
				$q(\xi_j)$ as in \ref{prop:up_xi}; \\
				\For{$t=1,\ldots,n$}{
					Update $q(z_{jt})$ as in \ref{prop:up_z}; \\
					Update $q(\gamma_{jt})$ as in \ref{prop:up_gamma} (non-smooth) or \ref{prop:up_sm_gamma} (smooth);
				}
                }
			}
			Update $q(\mathbf{h})$ as in \ref{prop:up_logsima_app} (heteroskedastic) or    $q(\sigma^2)$ as in \ref{prop:up_sigmasq_homo} (homoskedastic); \\
			Update $q(\nu^2)$ as in \ref{prop:up_nu};\\
			Compute $\widehat{\Delta}_{\boldsymbol{\vartheta}} = q(\boldsymbol{\vartheta})^{(\mbox{iter})}-q(\boldsymbol{\vartheta})^{(\mbox{iter}-1)}$ ;
		}
		\caption{Variational Bayes for dynamic variable selection.}
		\label{algo:algo1}
	\end{algorithm}




\section{Sparsity-inducing and convergence properties}\setcounter{figure}{0}\setcounter{table}{0}

\label{app:theo}

Before discussing the theoretical properties of our algorithmic procedure, we need to provide definitions and lemmas that are instrumental in the proof.

\begin{definition}
	$\mathbf{A}$ is a Z-matrix if its off-diagonal elements satisfy $a_{i,j}\leq0$, for $i\neq j$.
\end{definition}
\vspace{0.2cm}
\begin{definition}
	$\mathbf{A}$ is a strictly diagonally dominant (SDD) matrix if, for each $i$ row of $\mathbf{A}$, $|a_{i,i}|>\sum_{j\neq i}|a_{i,j}|$.
\end{definition}
\vspace{0.2cm}
\begin{corollary}
	If a matrix $\mathbf{A}$ is SDD and all its diagonal elements $a_{i,i}$ are positive, then the real parts of its eigenvalues are positive.
\end{corollary}
\vspace{0.2cm}
\begin{definition}\label{def:mmat}
	A matrix $\mathbf{A}$ is considered an M-matrix if it is a strictly diagonally dominant Z-matrix and all its diagonal elements $a_{i,i}$ are positive.
\end{definition}
\vspace{0.2cm}
\begin{corollary}\label{coro:posmat}
	If a matrix $\mathbf{A}$ is an M-matrix, then it belongs to inverse-positive matrices, i.e all elements of the inverse are positive $[\mathbf{A}^{-1}]_{i,j}\geq0$, for all $(i,j)$.
\end{corollary}
\vspace{0.2cm}
\begin{lemma}\label{lemma:Qinv_pos}
	The matrix $\mathbf{Q}^{-1}$ is a positive matrix, i.e $[\mathbf{Q}^{-1}]_{i,j}\geq0$.
\end{lemma}
\begin{proof}
	Follows from the tridiagonal form of $\mathbf{Q}$ with $q_{1,1}=1+1/k_0$, and $k_0<+\infty$.
\end{proof}
\vspace{0.2cm}
\begin{lemma}\label{lemma:sigmapos}
	The matrix $\boldsymbol{\Sigma}_{q(\omega_j)}$ is a positive matrix, i.e $[\boldsymbol{\Sigma}_{q(\omega_j)}]_{i,j}\geq0$.
\end{lemma}
\begin{proof}
	Recall that 
	\begin{equation}
		\boldsymbol{\Sigma}_{q(\omega_j)}=\mathbf{W}^{-1}=\left(\mathsf{Diag}\left(0,\mathbb{E}_q\left[\mathbf{z}_j\right]\right)+\mu_{q(1/\xi_j^2)}\mathbf{Q}\right)^{-1},  
	\end{equation}
	is tridiagonal, where $\mathbb{E}_q\left[z_{jt}\right]>0$ and $\mu_{q(1/\xi_j^2)}>0$. Notice that $\mathbf{W}$ has off-diagonal elements equal to $-\mu_{q(1/\xi_j^2)}<0$ in the first sub/over-diagonal and $0$ elsewhere and therefore it is a Z-matrix. Moreover, $w_{t,t}>0$ for all $t$ and:
	\begin{align}
		&w_{1,1} = (1+k_0^{-1})\mu_{q(1/\xi_j^2)} > \mu_{q(1/\xi_j^2)} = |w_{1,2}| \\
		&w_{t,t} = 2\mu_{q(1/\xi_j^2)}+\mathbb{E}_q\left[z_{jt}\right] > 2\mu_{q(1/\xi_j^2)} = |w_{t,t-1}|+|w_{t,t+1}|, \quad t=2,\ldots,n\\
		&w_{n+1,n+1}= \mu_{q(1/\xi_j^2)}+\mathbb{E}_q\left[z_{jn}\right] > \mu_{q(1/\xi_j^2)} = |w_{n+1,n-1}|,
	\end{align}
	thus $\mathbf{W}$ is SDD with positive diagonal elements. Hence, by definition \ref{def:mmat} is an M-matrix and corollary \ref{coro:posmat} tells us that its inverse is a positive matrix.
\end{proof}
\begin{proposition}\label{prop:main}
	Assume that the maximum over time of the inclusion probabilities, for a given variable $j$, at the $i$-th iteration of the algorithm is such that $\max_{t\in\{1,\ldots,n\}}\mu^{(i)}_{q(\gamma_{jt})}=\epsilon$, and $\epsilon\ll 1$ is small enough. Moreover, let $\boldsymbol{\Sigma}^{(i)}_{q(\omega_j)}-\boldsymbol{\Sigma}^{(i-1)}_{q(\omega_j)}\geq 0$, then:
	\begin{center}
		\begin{enumerate}
			\item $\mu^{(i+1)}_{q(\gamma_{jt})} = \mathrm{expit}\left\{\mu^{(i+1)}_{q(\omega_{jt})}-\frac{1}{2}\mu^{(i+1)}_{q(1/\sigma_t^2)}x_{jt-1}^2\mu_{q(1/\eta_j^2)}^{-1(i+1)}q_{tt}+O(\epsilon)\right\}$, $q_{tt}=[\mathbf{Q}^{-1}]_{tt}\geq 0$;
			\item $\mu_{q(\omega_{jt})}^{(i+1)} = -1/2\sum_{k=1}^n s_{tk}+O(\epsilon)$, $s_{tk}=[\boldsymbol{\Sigma}_{q(\omega_{j})}]_{tk}\geq 0$;
			\item $\mu_{q(\omega_{jt})}^{(i+1)}\leq\mu_{q(\omega_{jt})}^{(i)}$ decreases after each iteration.
		\end{enumerate}
	\end{center}  
\end{proposition}

\begin{proof}
	Consider the $k$-th group with elements $j=1,\ldots,p_k$.
We start proving {\it 1)}. Recall that the update for $\mu^{(i+1)}_{q(\gamma_{jt})}$ is equal to:
	\begin{align}
		\mu^{(i+1)}_{q(\gamma_{jt})} &= \mathrm{expit}\bigg\{\mu^{(i)}_{q(\omega_{jt})}-1/2\mu^{(i)}_{q(1/\sigma_t^2)}\left(\mathbb{E}^{(i+1)}_q[b^2_{jt}]x_{jt-1}^2-2\mu^{(i+1)}_{q(b_{jt})}{x}_{jt-1}\mathbb{E}^{(i+1)}_q[\varepsilon_{jt}]\right) \nonumber\\
        &\qquad\qquad\qquad -\mu^{(i)}_{q(1/\sigma_t^2)}x_{jt-1}\mathbf{x}^{-j}_{kt-1}\mathrm{diag}\big(\boldsymbol{\mu}^{(i+1)}_{q(\gamma^{-j}_{kt})}\big)\big[\boldsymbol{\Sigma}^{(i+1)}_{q(\mathbf{b}_{kt})}\big]_{-j,j}\bigg\}.
	\end{align}
 The variance matrix $\boldsymbol{\Sigma}_{q(\mathbf{b}_k)}^{(i+1)}=\big(\mathbf{D}^{(i+1)}_k+\mathbf{Q}^{(i+1)}_k\big)^{-1}$ and exploiting the block matrix inversion formula, it holds:
	\begin{align*}
		[\boldsymbol{\Sigma}_{q(\mathbf{b}_k)}^{(i+1)}]_{j,j} &= \big(\mathbf{M}_{q(\gamma_j)}^{(i)}\odot\big(\mathbf{X}_j^\prime\mathbf{H}^{(i)}\mathbf{X}_j\big)+\mu^{(i)}_{q(1/\eta_j^2)}\mathbf{Q} -\mathbf{M}_{q(\gamma_j)}^{(i)}\mathbf{X}_j^\prime\mathbf{H}^{(i)}\mathbf{P}^{(i)}_j\mathbf{H}^{(i)}\mathbf{X}_j\mathbf{M}_{q(\gamma_j)}^{(i)}\big)^{-1}, \\
        [\boldsymbol{\Sigma}_{q(\mathbf{b}_k)}^{(i+1)}]_{-j,j} &= -\big(\mathbf{D}^{(i)}_{-j}+\mathbf{Q}^{(i)}_{-j}\big)^{-1}\mathbf{M}^{(i)}_{q(\gamma_{-j})}\mathbf{X}_{-j}^\prime\mathbf{H}^{(i)}\mathbf{X}_{j}\mathbf{M}^{(i)}_{q(\gamma_{j})}[\boldsymbol{\Sigma}_{q(\mathbf{b}_k)}^{(i+1)}]_{j,j},
	\end{align*}
 where $\mathbf{M}_{q(\gamma_j)}=\mathrm{diag}(\boldsymbol{\mu}_{q(\gamma_{j1})},\ldots,\boldsymbol{\mu}_{q(\gamma_{jn})})$, and $\mathbf{P}_j=\mathbf{X}_{-j}\mathbf{M}_{q(\gamma_{-j})}(\mathbf{D}_{-j}+\mathbf{Q}_{-j})^{-1}\mathbf{M}_{q(\gamma_{-j})}\mathbf{X}_{-j}^\prime$.
Notice that the latter doesn't depend on the elements in $\mathbf{M}_{q(\gamma_j)}$. Then, we can write each inclusion probability as $\mu^{(i)}_{q(\gamma_{jt})}=\alpha_t\epsilon$, with $0<\alpha_t\leq1$ and where $\epsilon=\max_t \mu^{(i)}_{q(\gamma_{jt})}$, and define $\mathbf{M}_{q(\gamma_j)} = \epsilon\mathbf{A}$ with ${\mathbf{A}}=\mathrm{diag}(\alpha_1,\ldots,\alpha_n)$, so that the equations above become:
\begin{align*}
		[\boldsymbol{\Sigma}_{q(\mathbf{b}_k)}^{(i+1)}]_{j,j} &= \big(\epsilon\mathbf{A}\odot\big(\mathbf{X}_j^\prime\mathbf{H}^{(i)}\mathbf{X}_j\big)+\mu^{(i)}_{q(1/\eta_j^2)}\mathbf{Q} -\epsilon^2\mathbf{A}\mathbf{X}_j^\prime\mathbf{H}^{(i)}\mathbf{P}^{(i)}_j\mathbf{H}^{(i)}\mathbf{X}_j\mathbf{A}\big)^{-1}, \\
        [\boldsymbol{\Sigma}_{q(\mathbf{b}_k)}^{(i+1)}]_{-j,j} &= -\epsilon\big(\mathbf{D}^{(i)}_{-j}+\mathbf{Q}^{(i)}_{-j}\big)^{-1}\mathbf{M}^{(i)}_{q(\gamma_{-j})}\mathbf{X}_{-j}^\prime\mathbf{H}^{(i)}\mathbf{X}_{j}\mathbf{A}[\boldsymbol{\Sigma}_{q(\mathbf{b}_k)}^{(i+1)}]_{j,j}.
	\end{align*}
Note that $\lim_{\epsilon\rightarrow0}\,[\boldsymbol{\Sigma}_{q(\mathbf{b}_k)}^{(i+1)}]_{j,j} =\mu^{-1\,(i)}_{q(1/\eta_j^2)}\mathbf{Q}^{-1}$
 and each element in $[\boldsymbol{\Sigma}_{q(\mathbf{b}_k)}^{(i+1)}]_{-j,j}$ is $O(\epsilon)$.
 Consider now the mean vector $\boldsymbol{\mu}_{q(\mathbf{b}_k)}^{(i+1)} =\boldsymbol{\Sigma}^{(i+1)}_{q(\mathbf{b}_k)}{\mathbf{\Lambda}^{(i+1)}_k}$. The sub-vector corresponding to the $j$-th element in the $k$-th group is defined as:
	\begin{align}
		\boldsymbol{\mu}_{q(\mathbf{b}_j)}^{(i+1)} = [\boldsymbol{\Sigma}_{q(\mathbf{b}_k)}^{(i+1)}]_{j,j}{\mathbf{\Lambda}^{(i+1)}_j} + [\boldsymbol{\Sigma}_{q(\mathbf{b}_k)}^{(i+1)}]_{j,-j}{\mathbf{\Lambda}^{(i+1)}_{-j}},
	\end{align}
 where 
 \begin{align*}
\mathbf{\Lambda}^{(i+1)}_j &=\mathbf{M}^{(i)}_{q(\gamma_j)}\mathbf{X}_j^\prime\mathbf{H}^{(i)}\bigg(\mathbf{y} - \sum_{m=1,m\neq k}^K \mathbf{X}_{m}\mathbf{M}^{(i)}_{q(\gamma_m)}{\boldsymbol{\mu}}_{q(\mathbf{b}_{m})}\bigg) \\
&=\epsilon\mathbf{A}\mathbf{X}_j^\prime\mathbf{H}^{(i)}\bigg(\mathbf{y} - \sum_{m=1,m\neq k}^K \mathbf{X}_{m}\mathbf{M}^{(i)}_{q(\gamma_m)}{\boldsymbol{\mu}}_{q(\mathbf{b}_{m})}\bigg),
 \end{align*}
 and $\mathbf{\Lambda}^{(i+1)}_{-j}$ doesn't depend on the elements in $\mathbf{M}_{q(\gamma_j)}$. Hence, each element in $\boldsymbol{\mu}_{q(\mathbf{b}_j)}^{(i+1)}$ is $\mathbb{E}^{(i+1)}_q[b_{jt}]=O(\epsilon)$. It follows that
	\begin{align}\label{eq:mu_beta2_th}
		\mathbb{E}^{(i+1)}_q[b^2_{jt}] = \big(\mu^{(i+1)}_{q(b_{jt})}\big)^2+\sigma^{2(i+1)}_{q(b_{jt})} = \left[\mu^{(i)}_{q(1/\eta_j^2)}\right]^{-1}q_{t,t}+O(\epsilon),
	\end{align}
which completes the proof.
	Similarly we prove {\it 2)}. Recall the function to jointly update $\boldsymbol{\mu}^{(i+1)}_{q(\omega_{j})}$:
	\begin{align}
		\boldsymbol{\mu}^{(i+1)}_{q(\omega_{j})}=\boldsymbol{\Sigma}^{(i+1)}_{q(\omega_j)}\left(0,\boldsymbol{\mu}^{(i)\prime}_{q({\gamma}_j)}-1/2\boldsymbol{\iota}_n^\prime\right)^\prime,
	\end{align}
	then the update of the $t$-th component is:
	\begin{align}\label{eq:up_omega_th}
		{\mu}^{(i+1)}_{q(\omega_{jt})}&=\mathbf{s}_t^\prime\left(0,\boldsymbol{\mu}^{(i)\prime}_{q({\gamma}_j)}-1/2\boldsymbol{\iota}_n^\prime\right)^\prime \nonumber\\
		&=-1/2\mathbf{s}_t^\prime\left(0,\boldsymbol{\iota}_n^\prime\right)^\prime+\mathbf{s}_t^\prime\left(0,\boldsymbol{\mu}^{(i)\prime}_{q({\gamma}_j)}\right)^\prime \nonumber\\
		&=-1/2\sum_{k=1}^n{s}_{tk}+\sum_{k=1}^n{s}_{tk}{\mu}^{(i)}_{q({\gamma}_{jk})},
	\end{align}
	where $\mathbf{s}_t$ denotes the $t$-th column in $\boldsymbol{\Sigma}^{(i+1)}_{q(\omega_j)}$. Notice that, since ${\mu}^{(i)}_{q({\gamma}_{jk})}\leq\epsilon$, for all $k$, we can write ${\mu}^{(i)}_{q({\gamma}_{jk})} = \alpha_k\epsilon$, where $0<\alpha_k\leq1$. If we plug-in the latter in Eq.\eqref{eq:up_omega_th} we get
	\begin{equation}
		{\mu}^{(i+1)}_{q(\omega_{j,t})}=-1/2\sum_{k=1}^n s_{tk}+\epsilon\sum_{k=1}^n\alpha_k s_{tk}=-1/2\sum_{k=1}^n s_{tk}+O(\epsilon).
	\end{equation}
	To prove the last statement {\it 3)}, assume that we observe $\boldsymbol{\Sigma}^{(i)}_{q(\omega_{j})}-\boldsymbol{\Sigma}^{(i-1)}_{q(\omega_{j})}$ positive matrix. Then we have that, for $\epsilon$ small:
	\begin{equation}\label{eq:start_proof}
		|\boldsymbol{\mu}^{(i)}_{q(\omega_{j})}|=\frac{1}{2}\boldsymbol{\Sigma}^{(i)}_{q(\omega_j)}\left(0,\boldsymbol{\iota}_n^\prime\right)^\prime\geq \frac{1}{2}\boldsymbol{\Sigma}^{(i-1)}_{q(\omega_j)}\left(0,\boldsymbol{\iota}_n^\prime\right)^\prime = |\boldsymbol{\mu}^{(i-1)}_{q(\omega_{j})}|,
	\end{equation}
	and therefore:
	\begin{equation}
		\mathbb{E}_q^{(i)}(\boldsymbol{\omega}_j\boldsymbol{\omega}_j^\prime) = \boldsymbol{\mu}^{(i)}_{q(\omega_{j})}(\boldsymbol{\mu}^{(i)}_{q(\omega_{j})})^\prime + \boldsymbol{\Sigma}^{(i)}_{q(\omega_j)} \geq \boldsymbol{\mu}^{(i-1)}_{q(\omega_{j})}(\boldsymbol{\mu}^{(i-1)}_{q(\omega_{j})})^\prime + \boldsymbol{\Sigma}^{(i-1)}_{q(\omega_j)} = \mathbb{E}_q^{(i-1)}(\boldsymbol{\omega}_j\boldsymbol{\omega}_j^\prime),
	\end{equation}
	which means that $\mathbb{E}_q^{(i)}(\boldsymbol{\omega}_j\boldsymbol{\omega}_j^\prime)-\mathbb{E}_q^{(i-1)}(\boldsymbol{\omega}_j\boldsymbol{\omega}_j^\prime)$ is a positive matrix. Consider now the update for the variable $z_{jt}$:
	\begin{equation}
		\mathbb{E}^{(i)}_q\left[z_{jt}\right] = \frac{1}{2}\frac{1}{\sqrt{\mathbb{E}_q^{(i)}({\omega}_{jt}^2)}}\tanh(\frac{\sqrt{\mathbb{E}_q^{(i)}({\omega}_{jt}^2)}}{2}) \leq \mathbb{E}^{(i-1)}_q\left[z_{jt}\right],
	\end{equation}
	since it is decreasing in $\mathbb{E}_q^{(i)}({\omega}_{jt}^2)$, for all $t$. And similarly for $\xi_j^2$:
	\begin{align}
		\mu^{(i)}_{q(1/\xi_j^2)} = \frac{A_\xi + \frac{n+1}{2}}{B_\xi +\frac{1}{2}\mathsf{tr}\left\{\mathbb{E}_q^{(i)}(\boldsymbol{\omega}_j\boldsymbol{\omega}_j^\prime)\mathbf{Q}\right\}} \leq \mu^{(i-1)}_{q(1/\xi_j^2)},
	\end{align}
	since it is decreasing in $\mathbb{E}_q^{(i)}(\boldsymbol{\omega}_j\boldsymbol{\omega}_j^\prime)$ and $\mathbb{E}_q^{(i)}(\boldsymbol{\omega}_j\boldsymbol{\omega}_j^\prime)-\mathbb{E}_q^{(i-1)}(\boldsymbol{\omega}_j\boldsymbol{\omega}_j^\prime)$ is a positive matrix. The next update of $\boldsymbol{\Sigma}_{q(\omega_j)}$ is equal to:
	\begin{equation}
		\boldsymbol{\Sigma}^{(i+1)}_{q(\omega_j)}=\left(\mathsf{Diag}\left(0,\mathbb{E}^{(i)}_q\left[\mathbf{z}_j\right]\right)+\mu^{(i)}_{q(1/\xi_j^2)}\mathbf{Q}\right)^{-1},
	\end{equation}
	which increases as both $\mathbb{E}^{(i)}_q\left[\mathbf{z}_j\right]$ and $\mu^{(i)}_{q(1/\xi_j^2)}$ decreases. Hence also $\boldsymbol{\Sigma}^{(i+1)}_{q(\omega_j)}-\boldsymbol{\Sigma}^{(i)}_{q(\omega_j)}$ is a positive matrix and therefore, for $\epsilon$ small:
	\begin{equation}
		|\boldsymbol{\mu}^{(i+1)}_{q(\omega_{j})}|\geq |\boldsymbol{\mu}^{(i)}_{q(\omega_{j})}|,
	\end{equation} 
	and from statement {\it 2)} we have that $\boldsymbol{\mu}^{(i+1)}_{q(\omega_{j})}\leq \boldsymbol{\mu}^{(i)}_{q(\omega_{j})}$. Set $i=i+1$ and repeat the procedure from Eq.\eqref{eq:start_proof}. We can see that $\boldsymbol{\mu}_{q(\omega_{j})}$ decreases after each iteration until convergence.
\end{proof}

The convergence results outlined in Proposition \ref{prop:main} provide a dimension reduction strategy whereby we can remove the $j$-th variable from the set of predictors during the estimation. Such an exclusion strategy improves the computational efficiency when $p$ increases but the signal $\bar{p}\leq p$ remains constant, where $\bar{p}=\mathsf{card}(\mathcal{J})$ and $\mathcal{J}=\{j:\sum_{t=1}^n\gamma_{jt}>0\}$ collects the indexes of regression coefficients that are included in the model at least for one $t$. The extended variational Bayes algorithm is summarised in Algorithm \ref{algo:algo2}.

	\begin{algorithm}
		\SetAlgoLined
		\kwInit{$q(\boldsymbol{\vartheta})$, $\Delta_{\boldsymbol{\vartheta}}$, $A_\nu$, $B_\nu$, $A_\eta$, $B_\eta$, $A_\xi$, $B_\xi$}		\While{$\big(\widehat{\Delta}_{\boldsymbol{\vartheta}}>\Delta_{\boldsymbol{\vartheta}}\big)$}{
			\For{$k=1,\ldots,K$}{
				Update $q(\mathbf{b}_k)$ as in \ref{prop:up_beta};\\
			\For{$j=1,\ldots,p_k$}{ 
				and $q(\eta_j)$ as in \ref{prop:up_eta}; \\
				Update $q(\boldsymbol{\omega}_j)$ as in \ref{prop:up_omega_app} and
				$q(\xi_j)$ as in \ref{prop:up_xi}; \\
				\For{$t=1,\ldots,n$}{
					Update $q(z_{jt})$ as in \ref{prop:up_z}; \\
					Update $q(\gamma_{jt})$ as in \ref{prop:up_gamma} (non-smooth) or \ref{prop:up_sm_gamma} (smooth);
				}
                }
			}
			Update $q(\boldsymbol{\sigma})$ as in  \ref{prop:up_logsima_app} (heteroskedastic) or \ref{prop:up_sigmasq_homo} (homoskedastic); \\
			Update $q(\nu^2)$ as in \ref{prop:up_nu};\\
			\If{assumptions in Proposition \ref{prop:main} hold}{
				\For{$k=1,\ldots,K$}{\For{$j=1,\ldots,p_k$}{
					\If{$\max_t\{\mu_{q(\gamma_{jt})}\}<\epsilon$}{Drop the $j$-th variable}
				}}
			}
			Compute $\widehat{\Delta}_{\boldsymbol{\vartheta}} = q(\boldsymbol{\vartheta})^{(\mbox{iter})}-q(\boldsymbol{\vartheta})^{(\mbox{iter}-1)}$ ;
		}
		\caption{Efficient variational Bayes for dynamic variable selection.}
		\label{algo:algo2}
	\end{algorithm}

\subsection{Additional convergence results}

Figure \ref{fig:dimensionreduction} provides a visual representation of Proposition \ref{prop:main} for a simple simulation set-up in which $\gamma_{jt}=0\ \forall t$. The dashed line identifies the iteration at which the conditions in Proposition \ref{prop:main} are satisfied for $\epsilon=0.01$. After few iterations $\mu_{q(\omega_{jt})}$ (right panel) becomes increasingly negative and $\mu_{q(\gamma_{jt})}$ (left panel) remains zero $\forall t$. Figure \ref{fig:convergence} provides a further comparison between the true variational update and the approximation proposed in Proposition \ref{prop:main} of $\mu_{q(\gamma_{jt})}$ (left panel) and $\mu_{q(\omega_{jt})}$ (right panel) for different time stamps over the simulated sample. Again, the vertical dashed line identifies the iteration at which the conditions of Proposition \ref{prop:main} hold for $\epsilon=0.01$. The approximation is exact after less than 30 iterations, meaning the algorithm induces sparsity in the full trajectory of time-varying parameters as highlighted in Eq.(2.18) in the main paper. 

\begin{figure}[!ht]
    \centering
\subfigure[Convergence of updates for $\mu_{q(\gamma_{jt})}$]{\includegraphics[width=.42\textwidth]{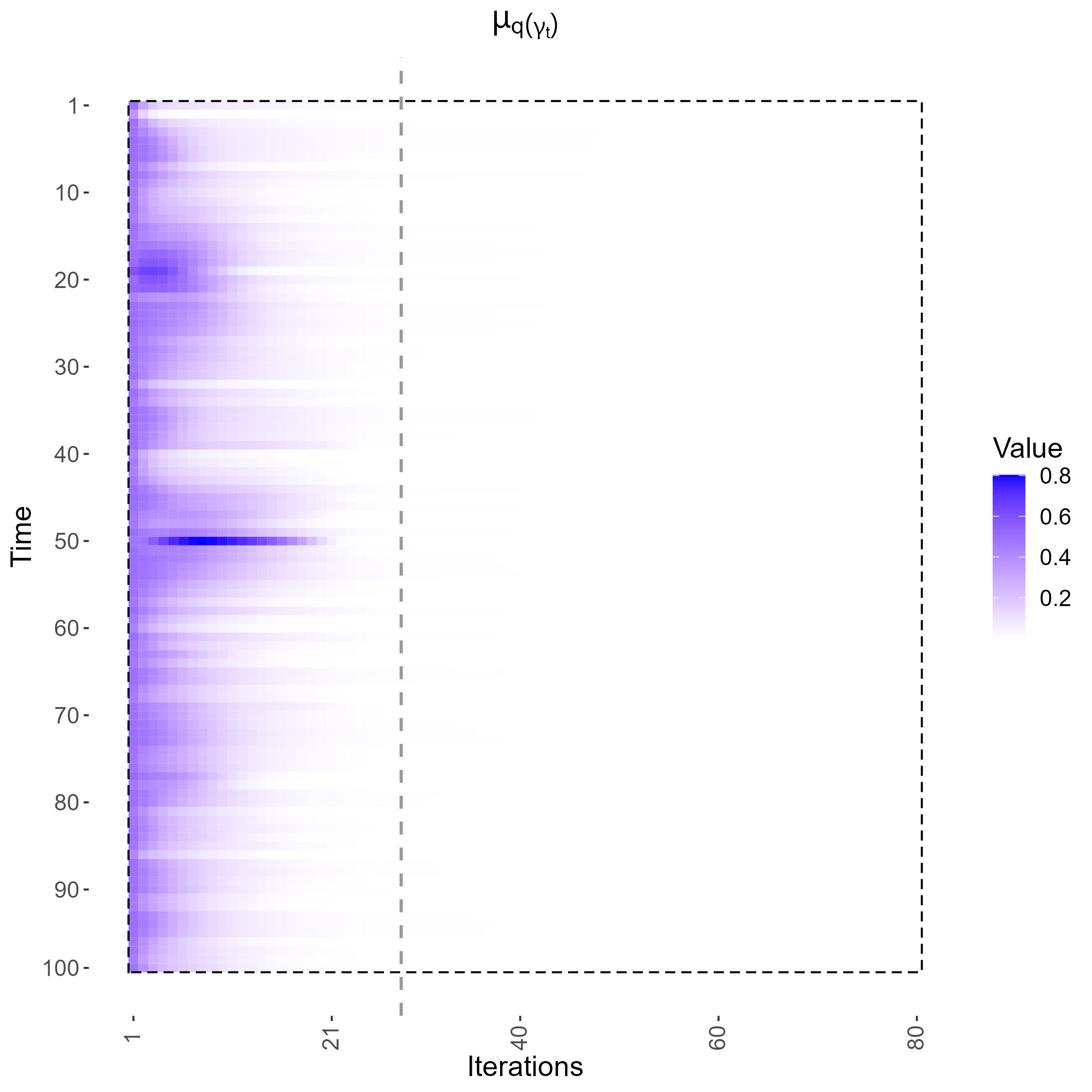}}\subfigure[Convergence of updates for $\mu_{q(\omega_{jt})}$]{\includegraphics[width=.42\textwidth]{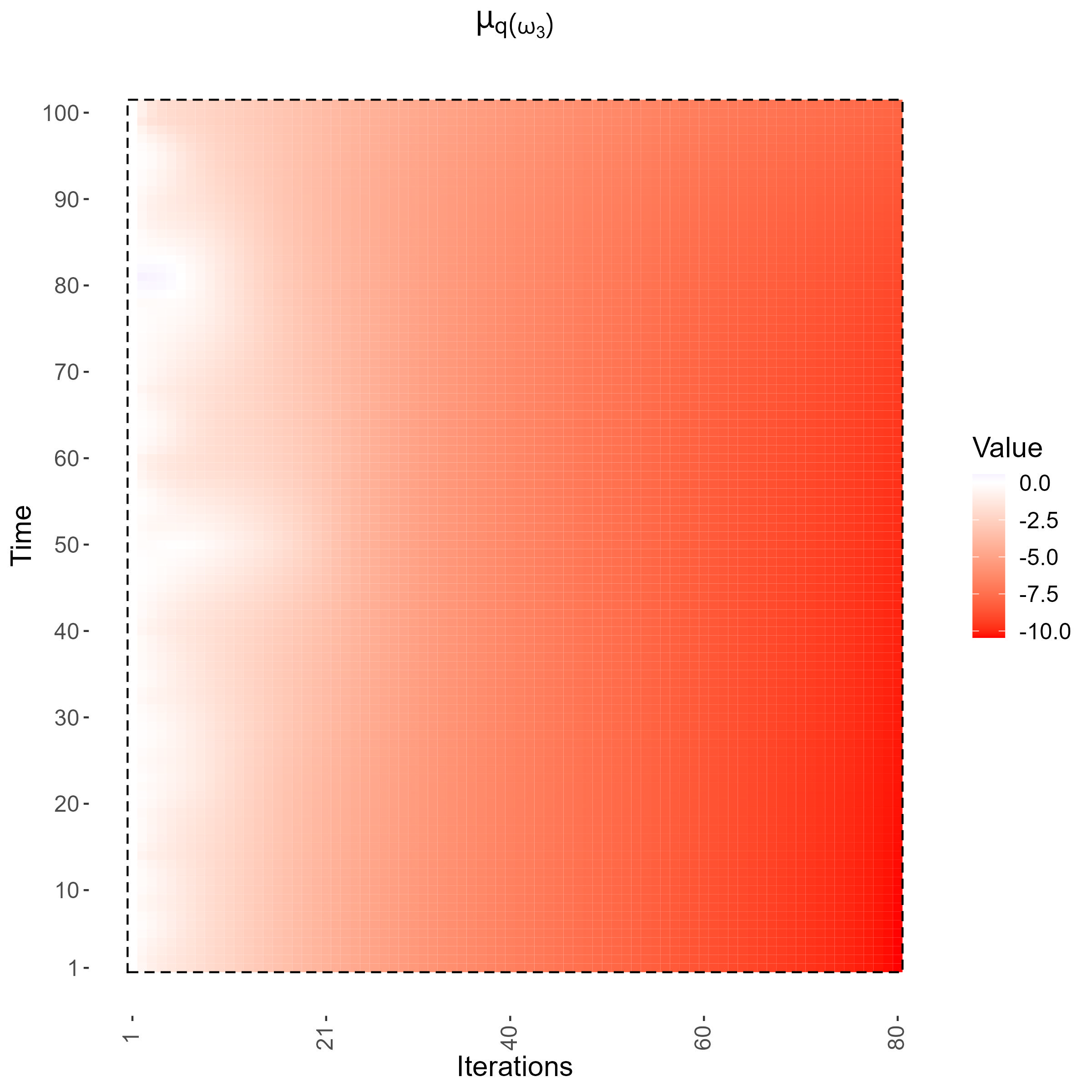}}
    \caption{\small Left panel shows variational update over iterations (x-axis) of the vector of posterior inclusion probabilities $(\mu_{q(\gamma_{j1})},\ldots,\mu_{q(\gamma_{jn})})$ (y-axis), for a parameter $j$ which is always zero $\forall t$. The dashed line identifies the iteration at which the conditions of Proposition \ref{prop:main} are satisfied for $\epsilon=0.01$. The right panel depicts the decreasing behaviour of $\mu_{q(\omega_{jt})}$, $\forall t$.}\label{fig:dimensionreduction} 
\end{figure}

\begin{figure}[!ht]
\centering
\subfigure[Convergence of updates for $\mu_{q(\gamma_{jt})}$]{\includegraphics[width=.45\textwidth]{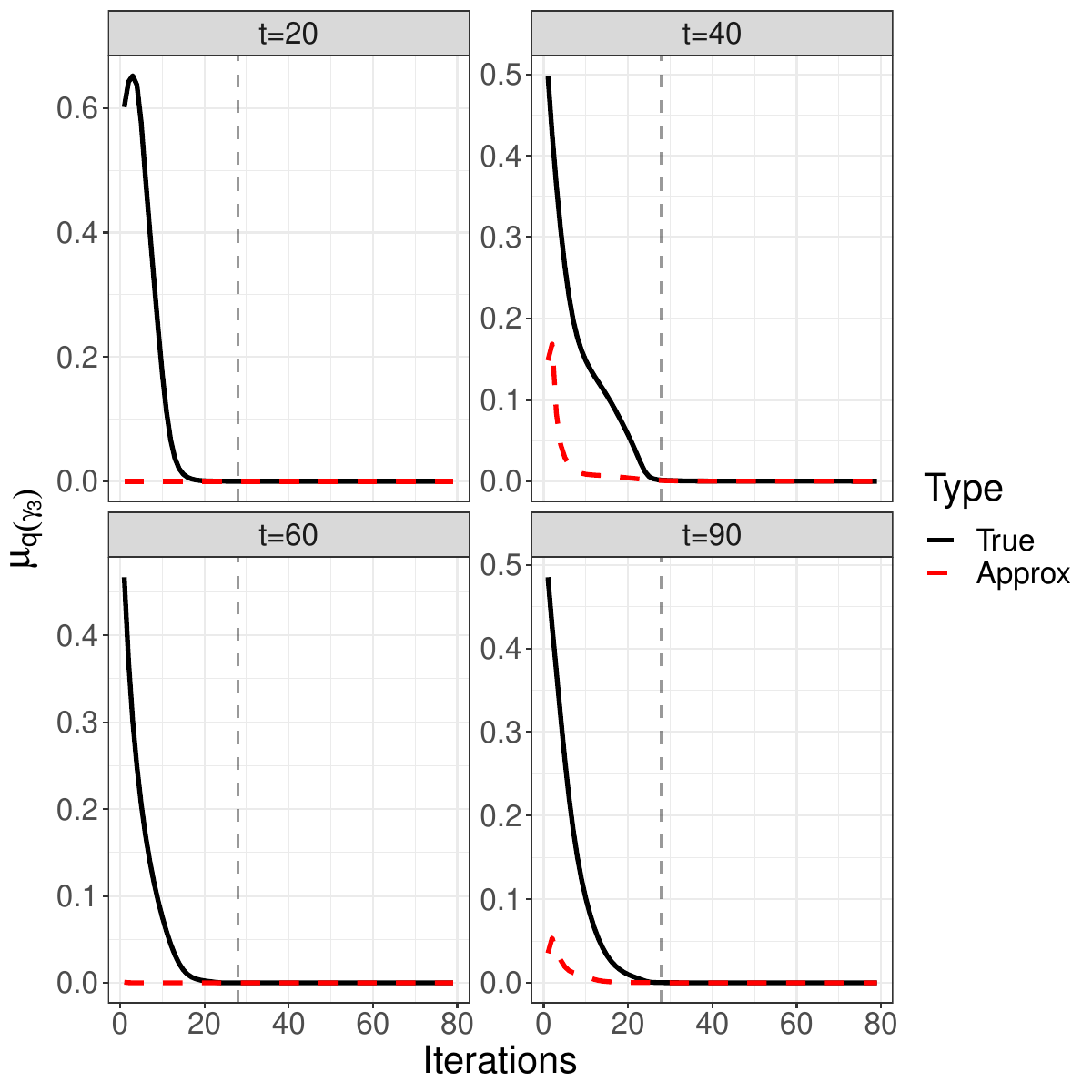}}\subfigure[Convergence of updates for $\mu_{q(\omega_{jt})}$]{\includegraphics[width=.45\textwidth]{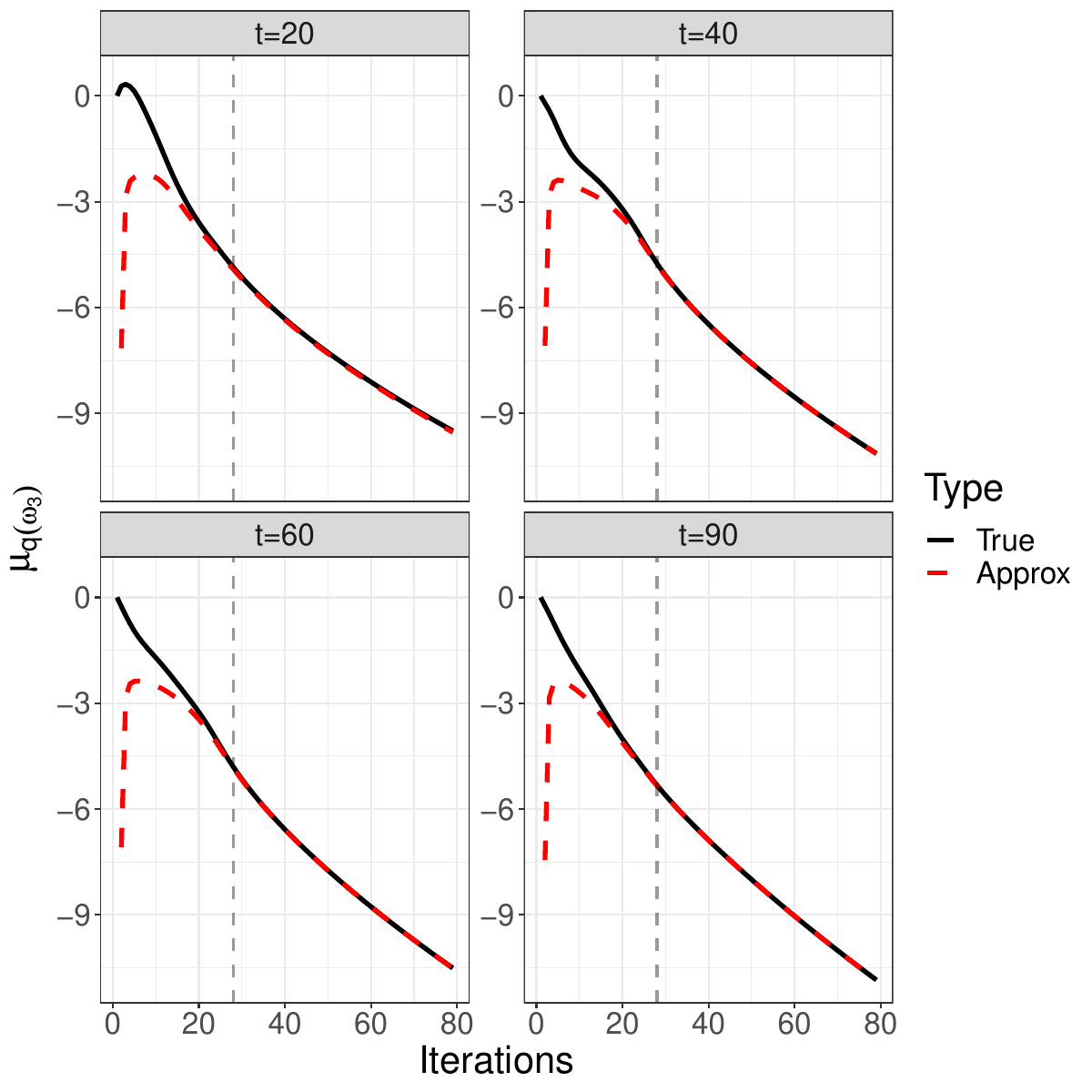}}
    \caption{\small Comparison between the true variational update and the approximation proposed in Proposition \ref{prop:main} of the posterior inclusion probability $\mu_{q(\gamma_{jt})}$ (left panel) and the auxiliary parameter $\mu_{q(\omega_{jt})}$ (right panel), for some times $t$ and for a parameter which is always zero $\forall t$.}\label{fig:convergence}
\end{figure}

Figure \ref{fig:convergence beta2} reports the posterior estimates of $\mu_{q(\gamma_{jt})}$ (left panel) and $\mu_{q(\omega_{jt})}$ (right panel) across $t=1,\ldots,n$ and for a parameter $j$ which is significant for only part of the sample. The colour intensity gives the value of the update. After less than 30 iterations, the posterior estimates of $\omega_{jt}$ and $\gamma_{jt}$ quickly converge to their true values. This threshold corresponds to the iteration at which the conditions of Proposition \ref{prop:main} are satisfied for $\epsilon=0.01$. In this respect, Figure \ref{fig:convergence beta2} complements Figure \ref{fig:dimensionreduction} in showing the convergence properties of our variational Bayes inference approach. 

\begin{figure}[!ht]
\hspace{-1.5em}\includegraphics[width=0.54\linewidth]{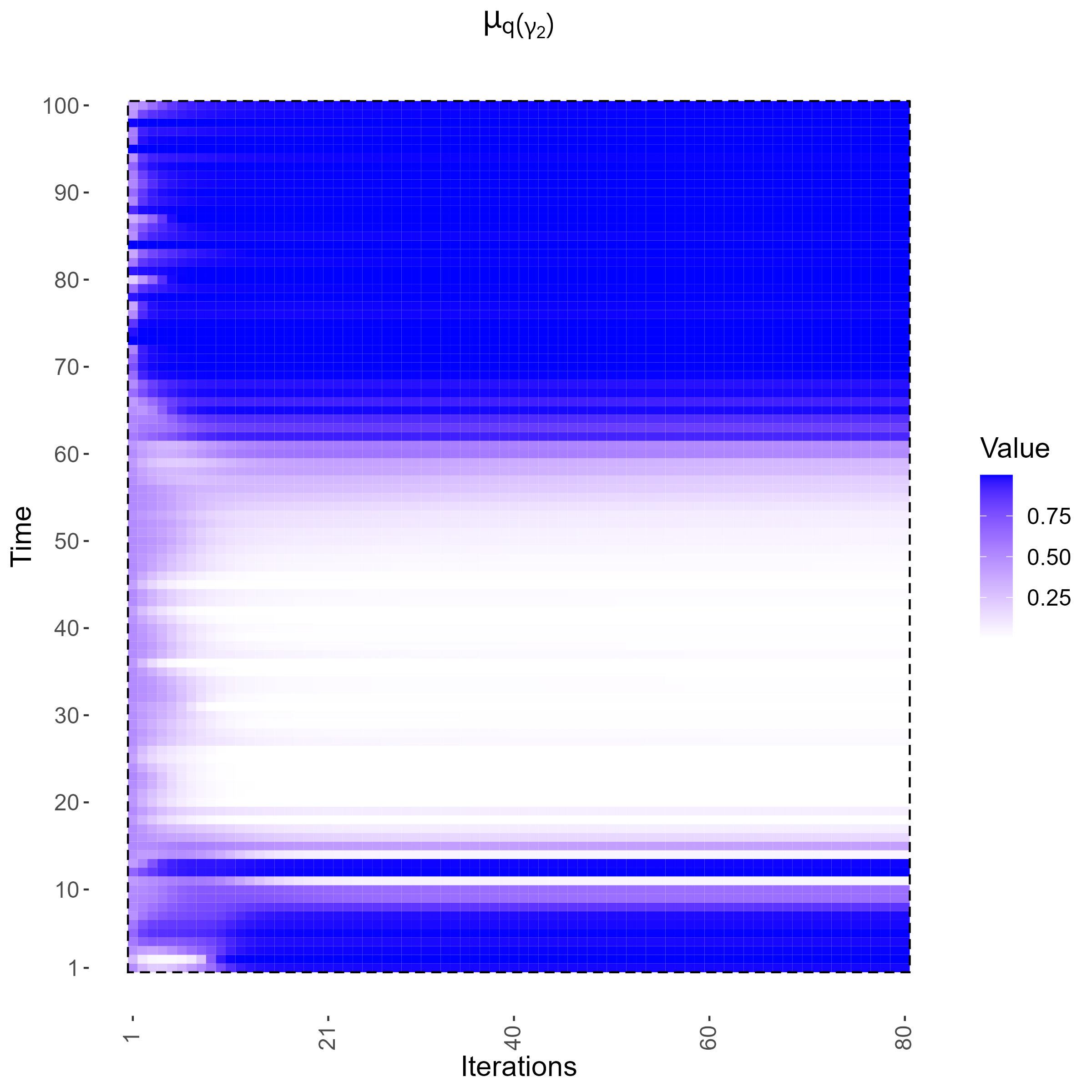}    \includegraphics[width=0.54\linewidth]{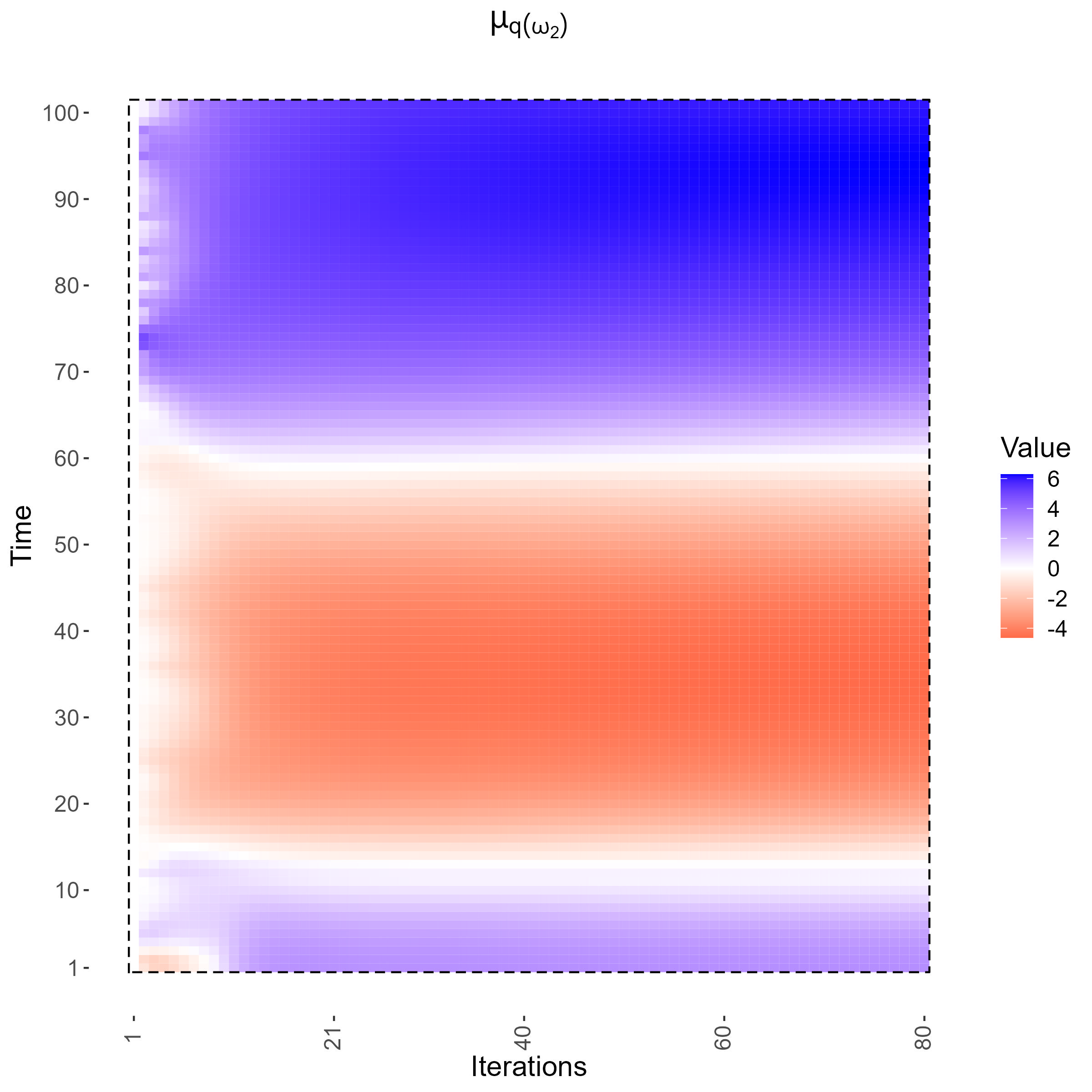}\hspace{-2em}
    \caption{\small Left panel shows the variational update over iterations (x-axis) until convergence of the vector of posterior inclusion probabilities $(\mu_{q(\gamma_{j1})},\ldots,\mu_{q(\gamma_{jn})})$ (y-axis), for a parameter $j$ which is significant for only part of the sample. The colour intensity gives the value of the update. The right panel depicts the behaviour of $\mu_{q(\omega_{jt})}$, $\forall t$.}
\label{fig:convergence beta2}
\end{figure}

\section{Prior hyper-parameters and algorithm initialization}\setcounter{figure}{0}\setcounter{table}{0}

\label{app:hyperparam}

We follow \cite{koop_korobilis_2020} and set $\mu_{q(\gamma_{jt})}^{(0)}=1/2$, $\forall t,j$. Next, we assume inverse-gamma priors for the variances parameters $\nu^2\sim\mathsf{IG}(A_\nu,B_\nu)$, $\eta_j^2\sim\mathsf{IG}(A_\eta,B_\eta)$, and $\xi_j^2\sim\mathsf{IG}(A_\xi,B_\xi)$.  We follow \cite{ormerod2017VS} and set $A_\sigma=B_\sigma=A_\eta=B_\eta=0.01$ to maintain non-informativeness. These are all standard conjugate priors. An important feature of our dynamic variable selection method is that the time-variation of $\gamma_{jt}$ depends on $\omega_{jt}$, where the dynamic of the latter is governed by the variance $\xi^2_j\sim\mathsf{IG}(A_\xi,B_\xi)$. For this reason, the couple $A_\xi,B_\xi$ of hyper-parameters deserve more scrutiny. 

To shed light on the impact of $A_\xi,B_\xi$ we study the estimated variational mean $\{\mu_{q(\omega_{jt})}\}_{t=1}^n$ and variance $\mathbf{\Sigma}_{q(\omega_j)}$, and the resulting $\{\mu_{q(\omega_{jt})}\}_{t=1}^n$, for three alternative limit cases. Specifically, we investigate a series of comparative statics based on Proposition \ref{prop:up_omega_app} and Proposition \ref{prop:up_xi}. This provides a transparent strategy to select values for $A_\xi,B_\xi$.

\begin{figure}[ht]
	\centering
	\subfigure[$\mathbf{\Sigma}_{q(\omega_j)}$]{\includegraphics[width=.31\textwidth]{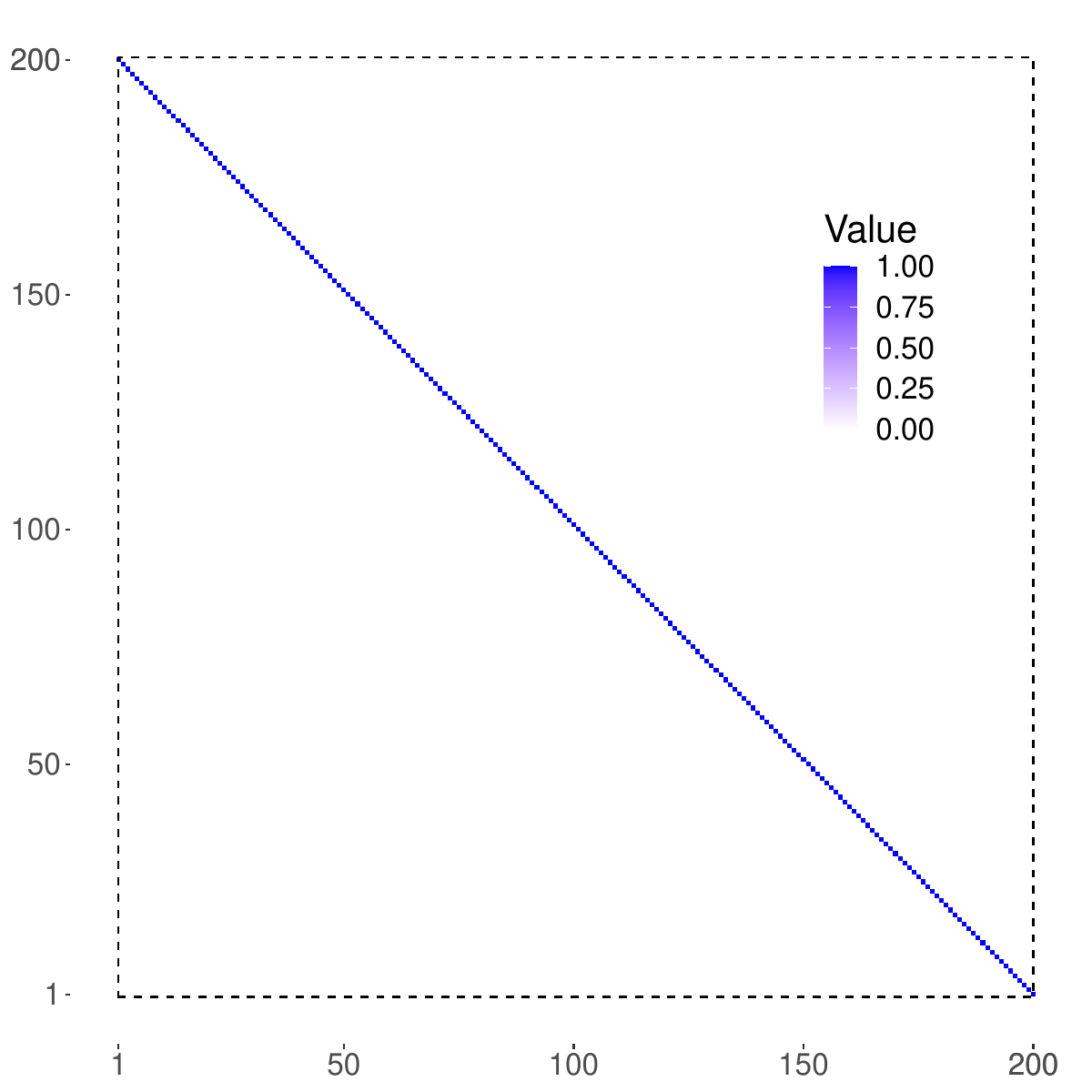}}
	\subfigure[$\{\mu_{q(\omega_{jt})}\}_{t=1}^n$]{\includegraphics[width=.30\textwidth]{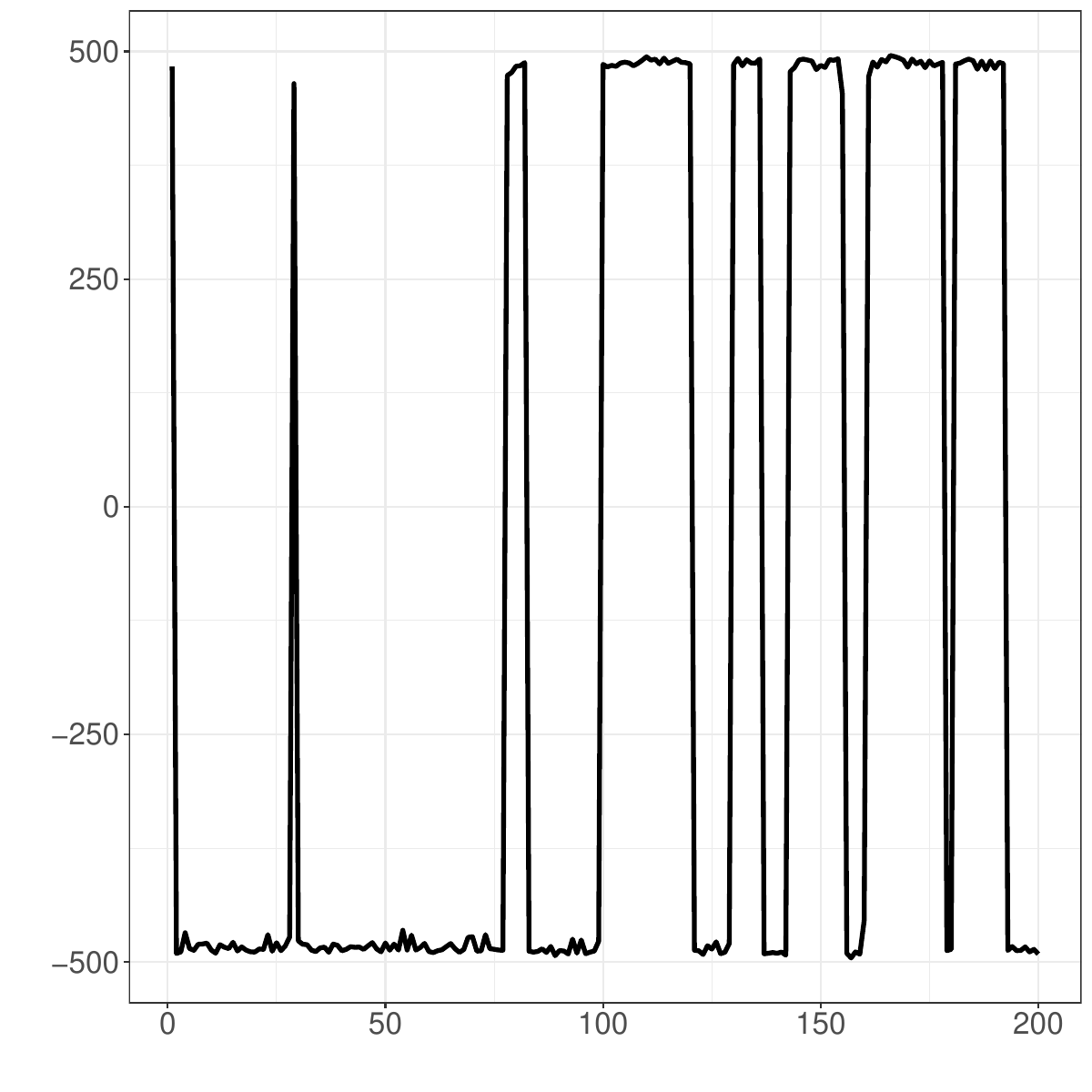}}
 	\subfigure[$\{\mu_{q(\gamma_{jt})}\}_{t=1}^n$]{\includegraphics[width=.30\textwidth]{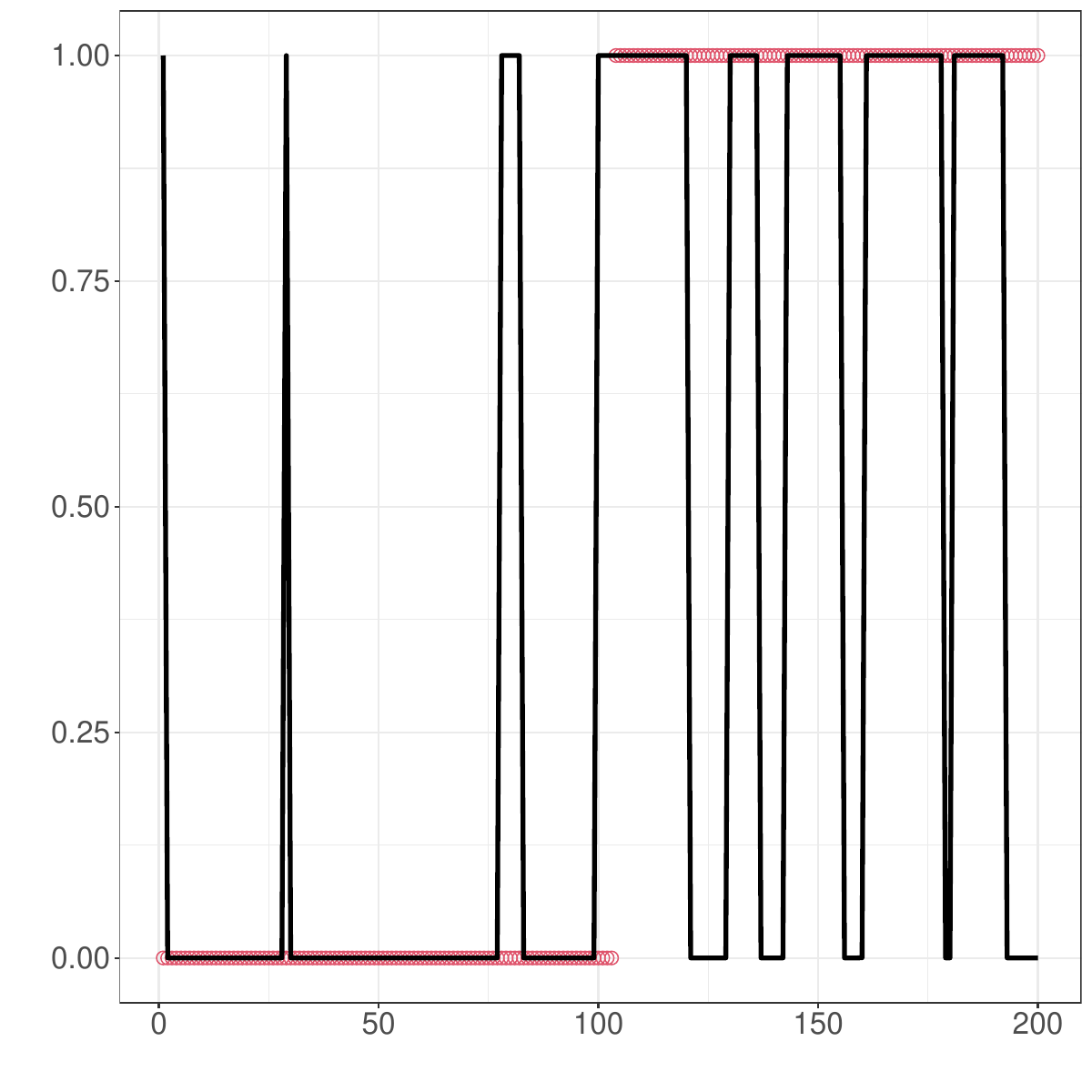}}
	\caption{\small Scenario A: $B_\xi\rightarrow\infty$. (a) Variational correlation matrix for the process $\{\omega_{jt}\}_{t=1}^n$ obtained from $\mathbf{\Sigma}_{q(\omega_j)}$. (b) Trajectory of $\{\mu_{q(\omega_{jt})}\}_{t=1}^n$. (c) The posterior inclusion probabilities $\{\mu_{q(\gamma_{jt})}\}_{t=1}^n$ compared to the simulated (red points).}\label{fig:B_to_inf}
\end{figure}

The first scenario considers $A_\xi$ constant and $B_\xi\rightarrow\infty$. Figure \ref{fig:B_to_inf} reports the resulting variational correlation matrix derived from $\mathbf{\Sigma}_{q(\omega_j)}$, the corresponding trajectory of $\{\mu_{q(\omega_{jt})}\}_{t=1}^n$, and the posterior estimates of the inclusion probabilities $\{\mu_{q(\gamma_{jt})}\}_{t=1}^n$. As $B_\xi\rightarrow+\infty$, the process $\{\omega_{jt}\}_{t=1}^n$ tends to be i.i.d -- $\mathbf{\Sigma}_{q(\omega_j)}$ is a diagonal matrix. This means that we lose the {\it a-priori} time dependence in the inclusion probability, which leads to an erratic dynamic of $\omega_{jt}$ (middle panel) and, thus, a highly irregular trajectory of $\{\mu_{q(\gamma_{jt})}\}_{t=1}^n$ (right panel).

\begin{figure}[ht]
	\centering
	\subfigure[$\mathbf{\Sigma}_{q(\omega_j)}$]{\includegraphics[width=.31\textwidth]{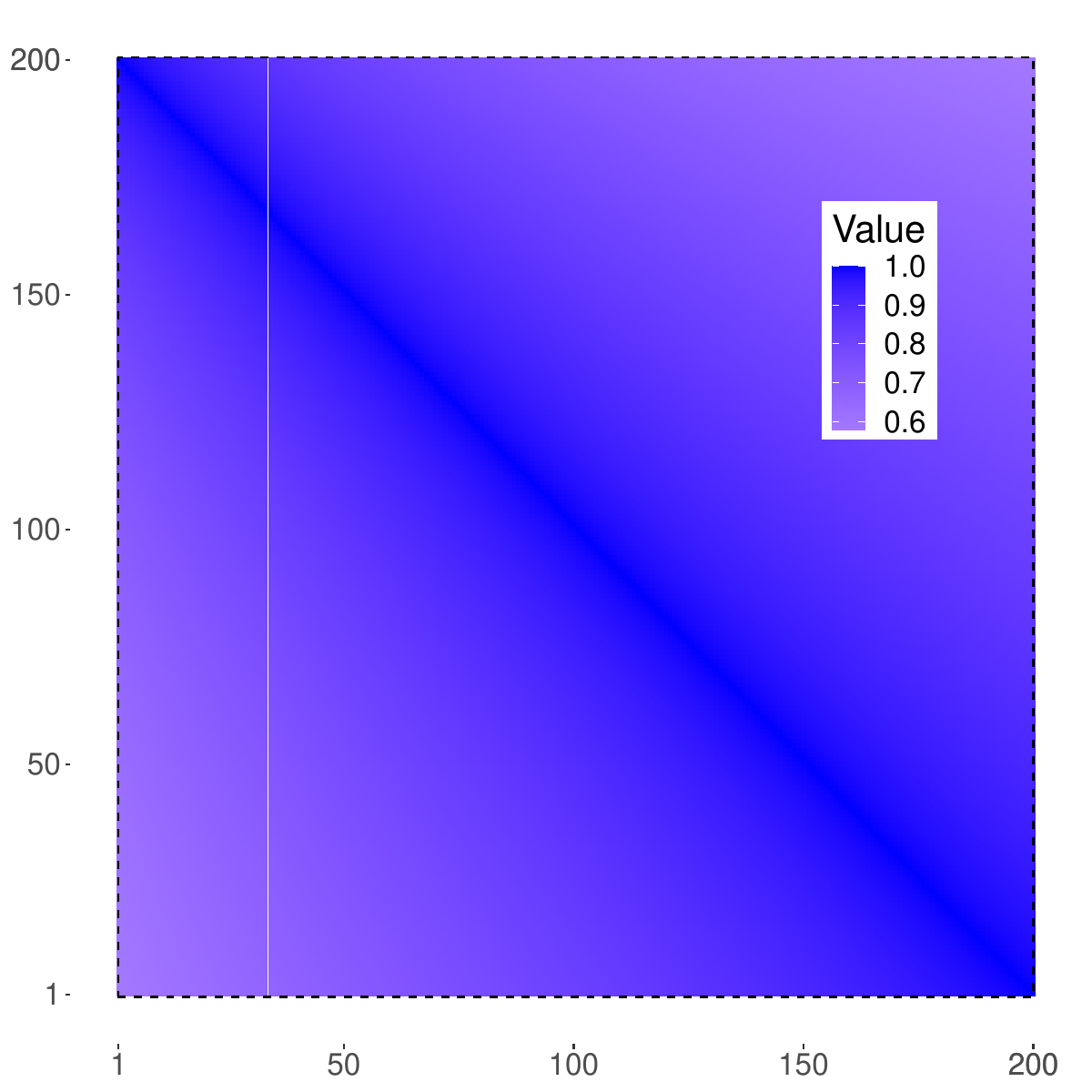}}
	\subfigure[$\{\mu_{q(\omega_{jt})}\}_{t=1}^n$]{\includegraphics[width=.30\textwidth]{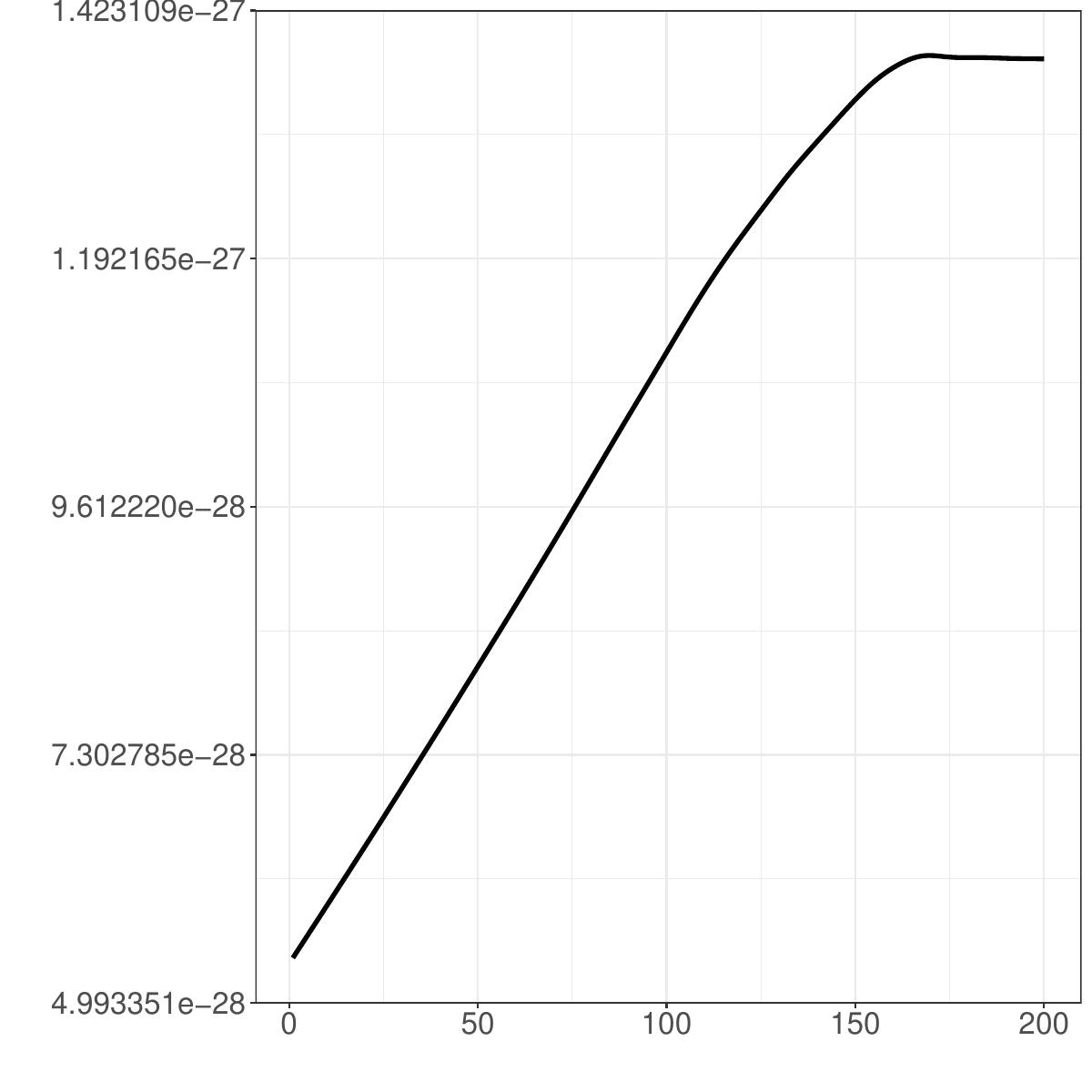}}
 	\subfigure[$\{\mu_{q(\gamma_{jt})}\}_{t=1}^n$]{\includegraphics[width=.30\textwidth]{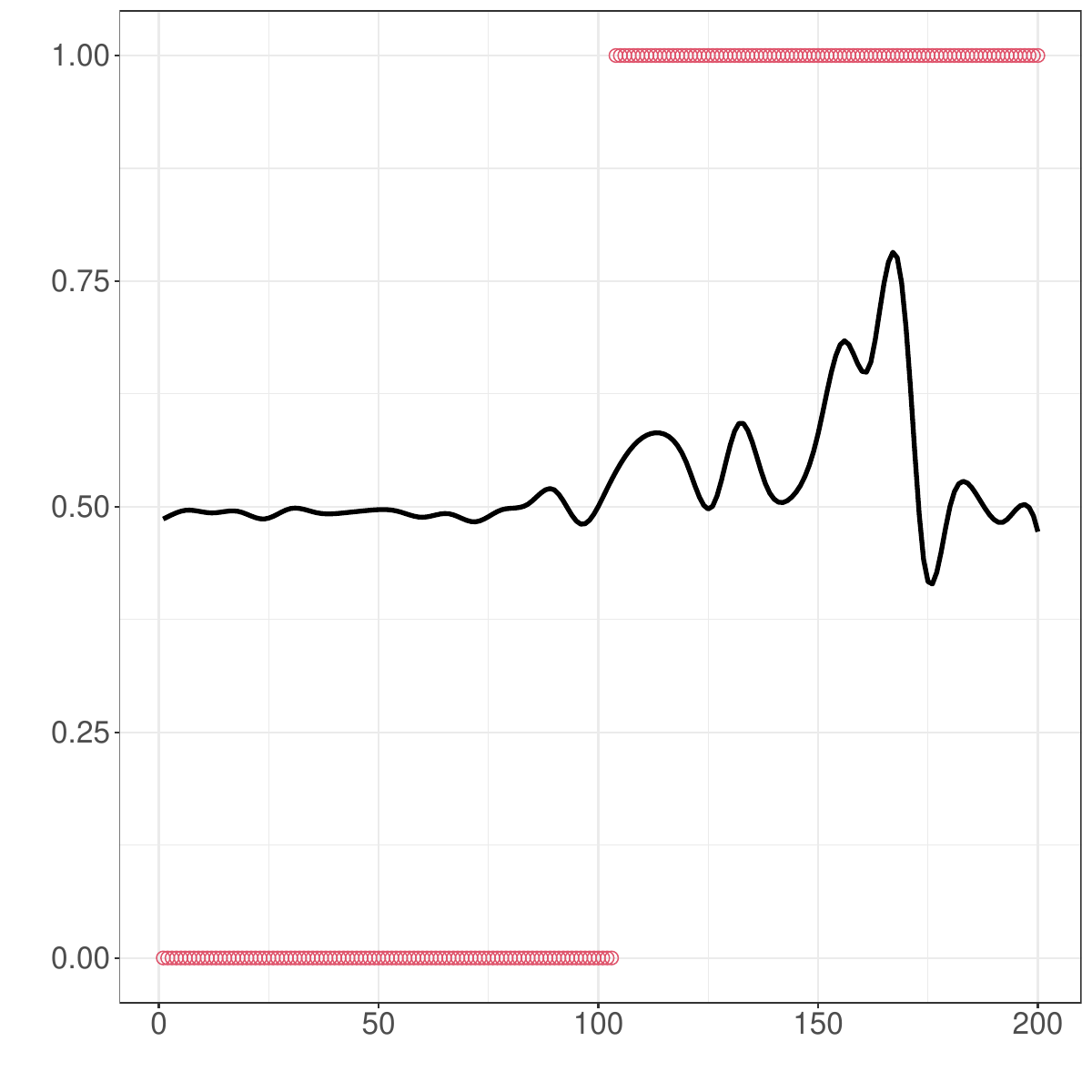}}
	\caption{\small Scenario B: $A_\xi\rightarrow\infty$. (a) Variational correlation matrix for the process $\{\omega_{jt}\}_{t=1}^n$ obtained from $\mathbf{\Sigma}_{q(\omega_j)}$. (b) Trajectory of $\{\mu_{q(\omega_{jt})}\}_{t=1}^n$. (c) The posterior inclusion probabilities $\{\mu_{q(\gamma_{jt})}\}_{t=1}^n$ compared to the simulated (red points).}\label{fig:A_to_inf}
\end{figure}

The second scenario considers $A_\xi\rightarrow\infty$ and fixed $B_\xi$. This implies that $\mu_{q(1/\xi^2_j)}\rightarrow \infty$ and, as a consequence, we give infinite weight to the matrix $\mathbf{Q}$ when computing $\mathbf{\Sigma}_{q(\omega_j)}$. Figure \ref{fig:A_to_inf} shows that such a strong a-priori time dependence leads to inclusion probabilities $\{\mu_{q(\gamma_{jt})}\}_{t=1}^n$ with low variability around the mean, i.e. $\mathrm{expit}(\mathbb{E}(\omega_{jt}))=0$. Thus, no meaningful time variation is captured in the variable importance.

\begin{figure}[ht]
	\centering
	\subfigure[$\mathbf{\Sigma}_{q(\omega_j)}$]{\includegraphics[width=.31\textwidth]{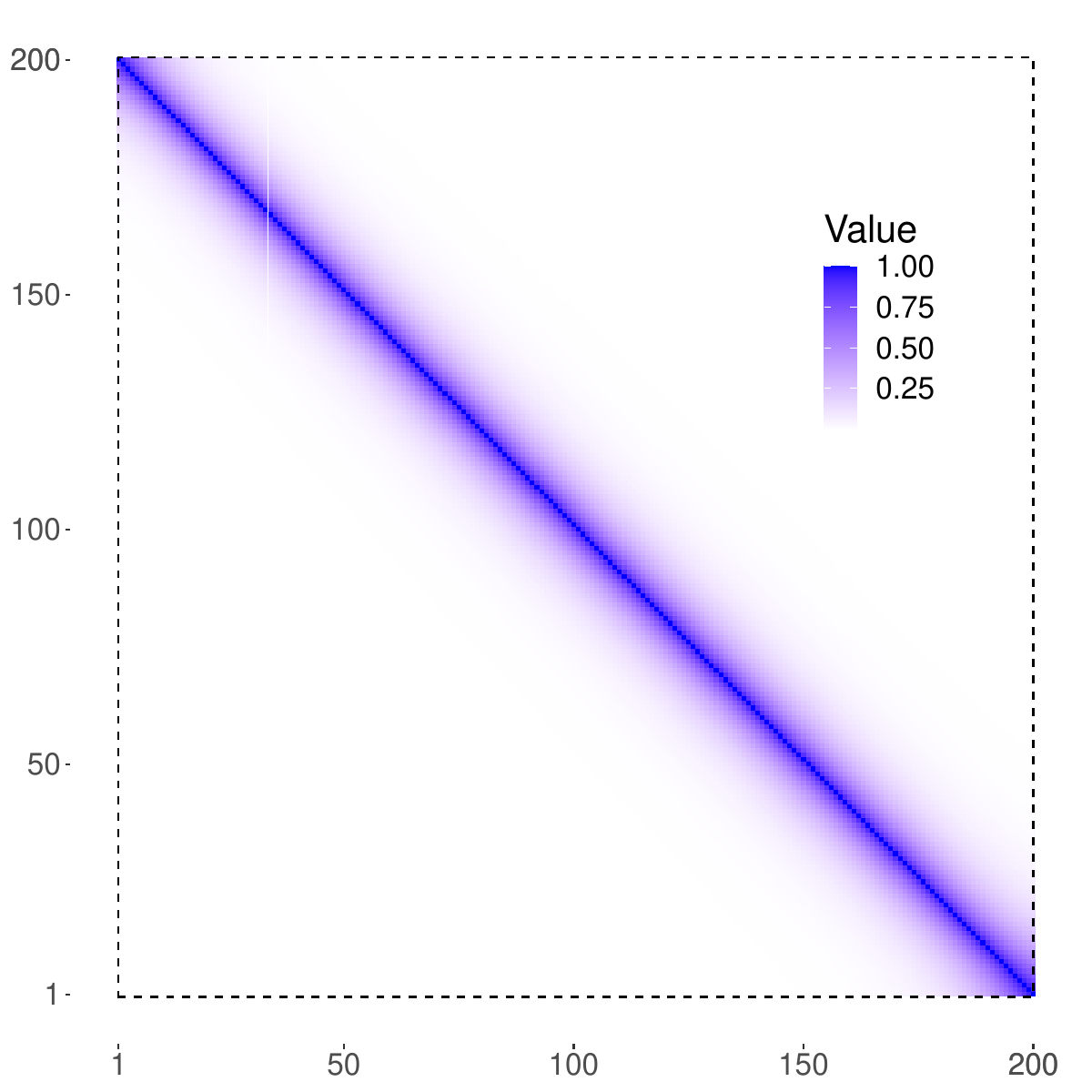}}
	\subfigure[$\{\mu_{q(\omega_{jt})}\}_{t=1}^n$]{\includegraphics[width=.30\textwidth]{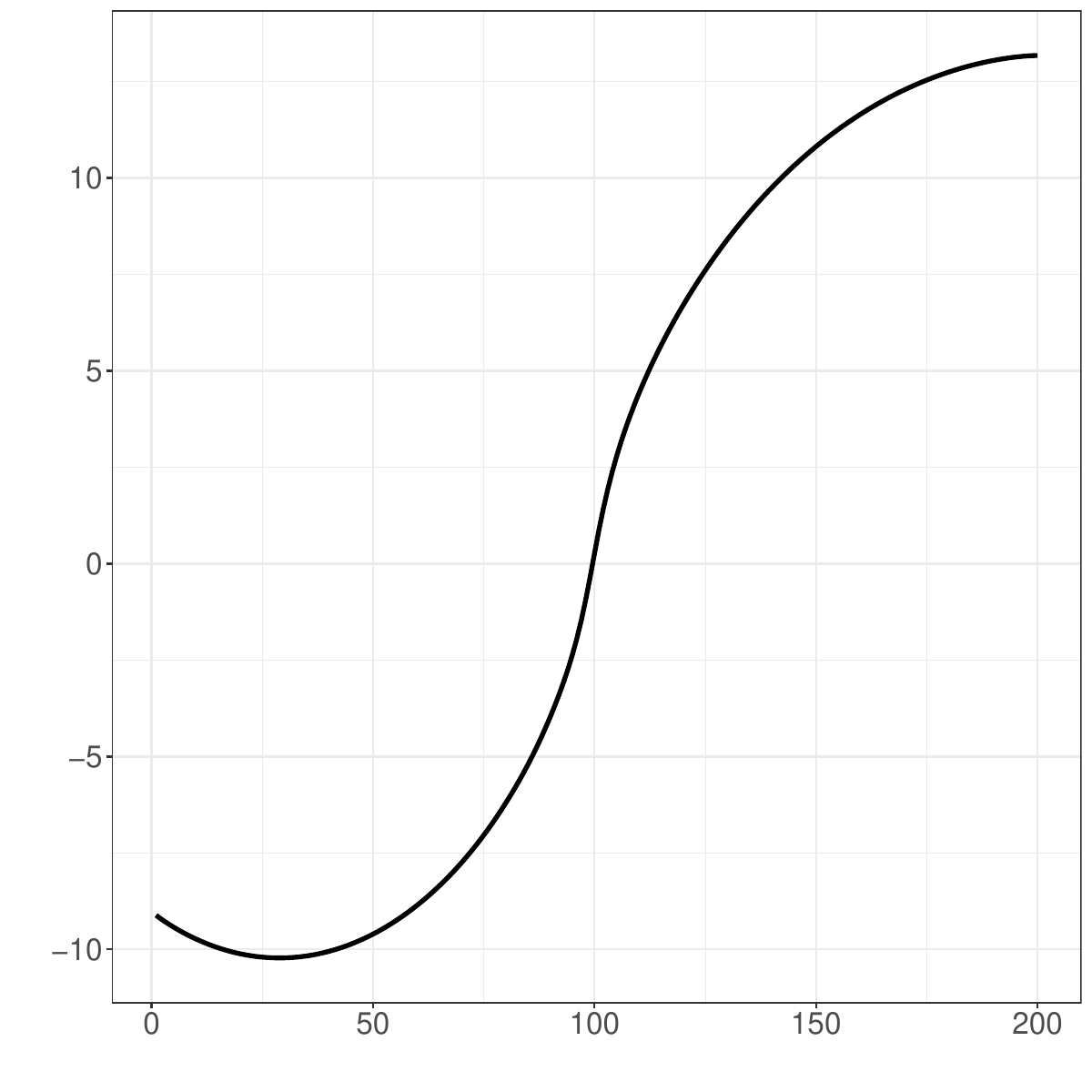}}
 	\subfigure[$\{\mu_{q(\gamma_{jt})}\}_{t=1}^n$]{\includegraphics[width=.30\textwidth]{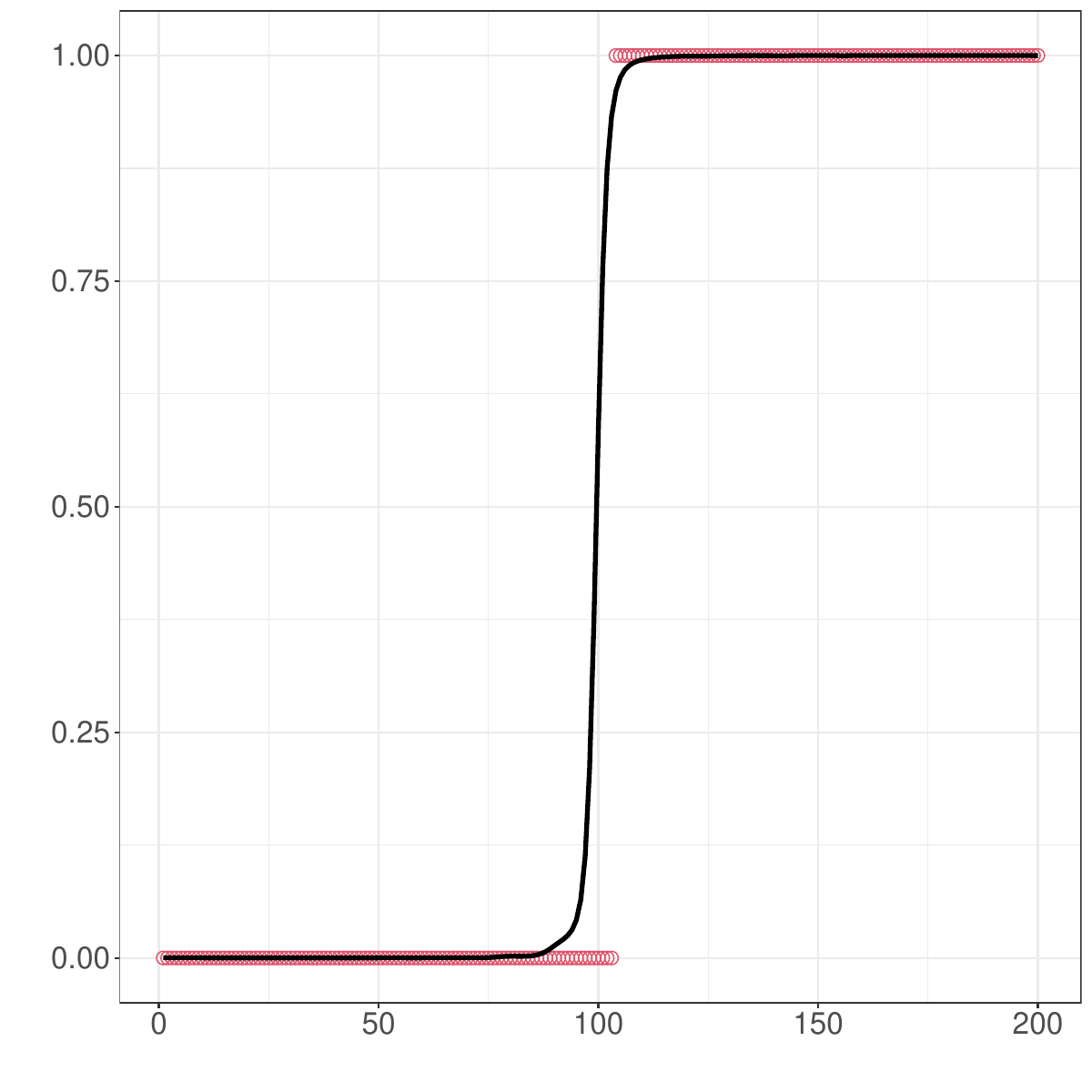}}
	\caption{\small Scenario C: $A_\xi/B_\xi\rightarrow c_1$, $0<c_1<\infty$. (a) Variational correlation matrix for the process $\{\omega_{jt}\}_{t=1}^n$ obtained from $\mathbf{\Sigma}_{q(\omega_j)}$. (b) Trajectory of $\{\mu_{q(\omega_{jt})}\}_{t=1}^n$. (c) The posterior inclusion probabilities $\{\mu_{q(\gamma_{jt})}\}_{t=1}^n$ compared to the simulated (red points).}\label{fig:AB_to_const}
\end{figure}

The third scenario considers $ A_\xi/B_\xi \rightarrow c_1$ as a bounded constant sufficiently greater than $0$. This implies that $\mu_{q(1/\xi^2_j)}\rightarrow c_2$, where $c_2\in\mathbb{R}^+$ and therefore a moderate weight to the matrix $\mathbf{Q}$ when computing $\mathbf{\Sigma}_{q(\omega_j)}$; that is, we account for a decreasing correlation as $|t_1-t_2|$, $t_1,t_2\in\{1,\ldots,n\}$ increases. Figure \ref{fig:AB_to_const} shows that such prior choice translates into a moderate variability in the trajectory of $\{\mu_{q(\omega_{jt})}\}_{t=1}^n$, which leads to posterior estimates $\{\mu_{q(\gamma_{jt})}\}_{t=1}^n$ that meaningfully track the variable importance over time. 

For the remaining simulation and empirical analysis, we fix $A_\xi=2$ so that $\mathsf{Var}(\xi^2_j)=\infty$ and $B_\xi$ can be interpreted as the mean of $\xi^2_j$. We set $B_\xi=5$, which satisfies $A_\xi/B_\xi\rightarrow c_1$. We also test in simulation $B_\xi=1$ or $B_\xi=10$. The model performance is comparable with $B_\xi=10$, while slightly deteriorates when choosing $B_\xi=1$ (see Section 3 in the main paper).

\section{Additional simulation details}\setcounter{figure}{0}\setcounter{table}{0}

\label{app:TVP_more_sim}

In this section, we report additional details on the simulation study reported in the main text. For the simulation studies, we consider $M=100$ replications from a linear data-generating process $y_t = \sum_{j=1}^p\beta_{jt}x_{jt-1}+\varepsilon_t$, with $\varepsilon_t\sim\mathsf{N}(0,0.25)$. We assume that different coefficients have different dynamics; for instance, $\beta_{1t}$ is a time-varying parameter always included in the model, i.e., $\gamma_{1t}=1$ $\forall t$. The $\beta_{1t}$ dynamic follows an AR(1) process with persistence equal to $0.98$ and conditional variance equal to $0.1$. The top-left panel in Figure \ref{fig:res_tvp_sim_app} provides an example of the simulated trajectory.  

 
\begin{figure}[!ht]
\hspace{-1.5em}\subfigure[Example of $\beta_{1t}$]{\includegraphics[width=.52\textwidth]{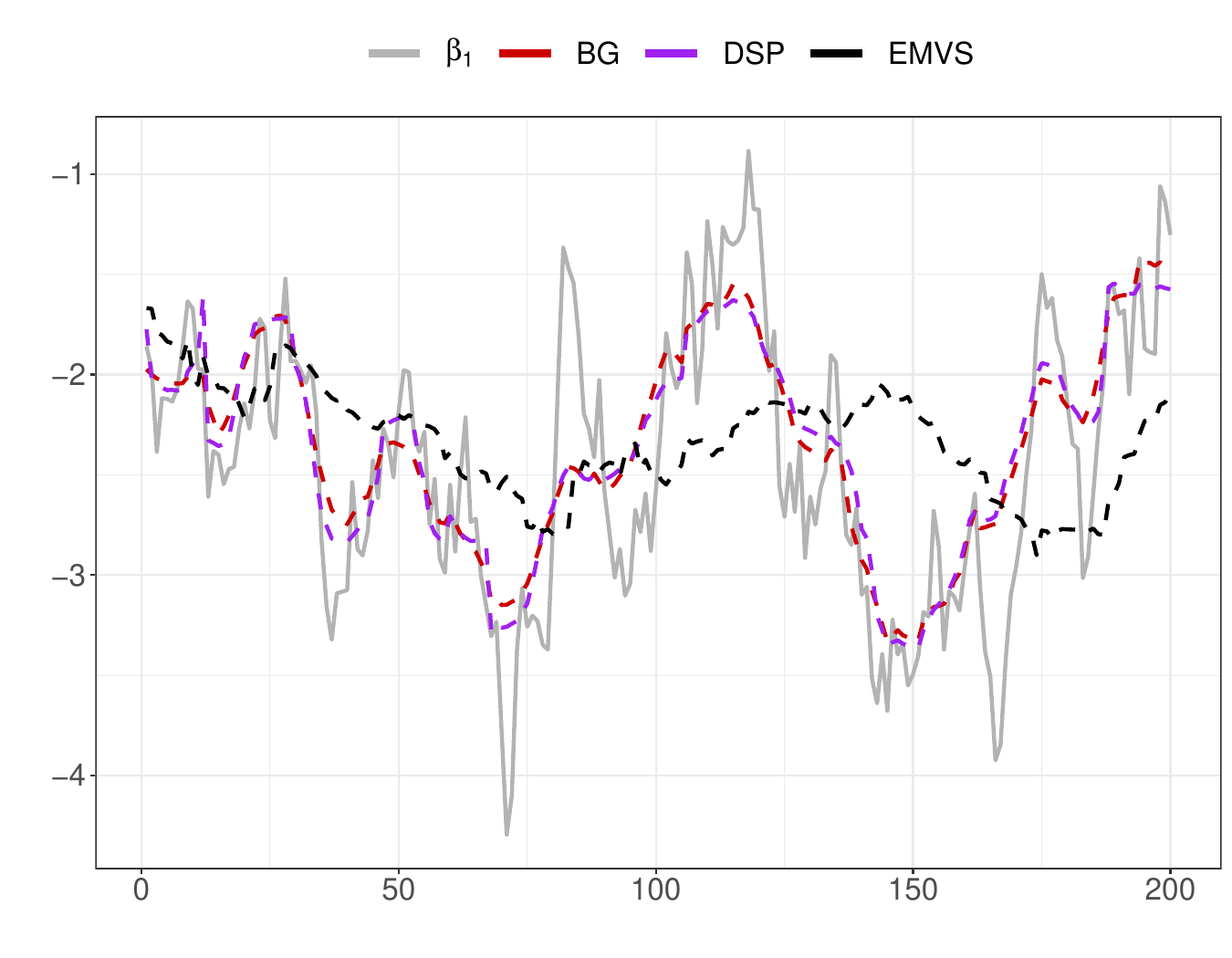}\label{fig:beta_sim1}}
\subfigure[Example of $\beta_{2:3,t}$]{\includegraphics[width=.52\textwidth]{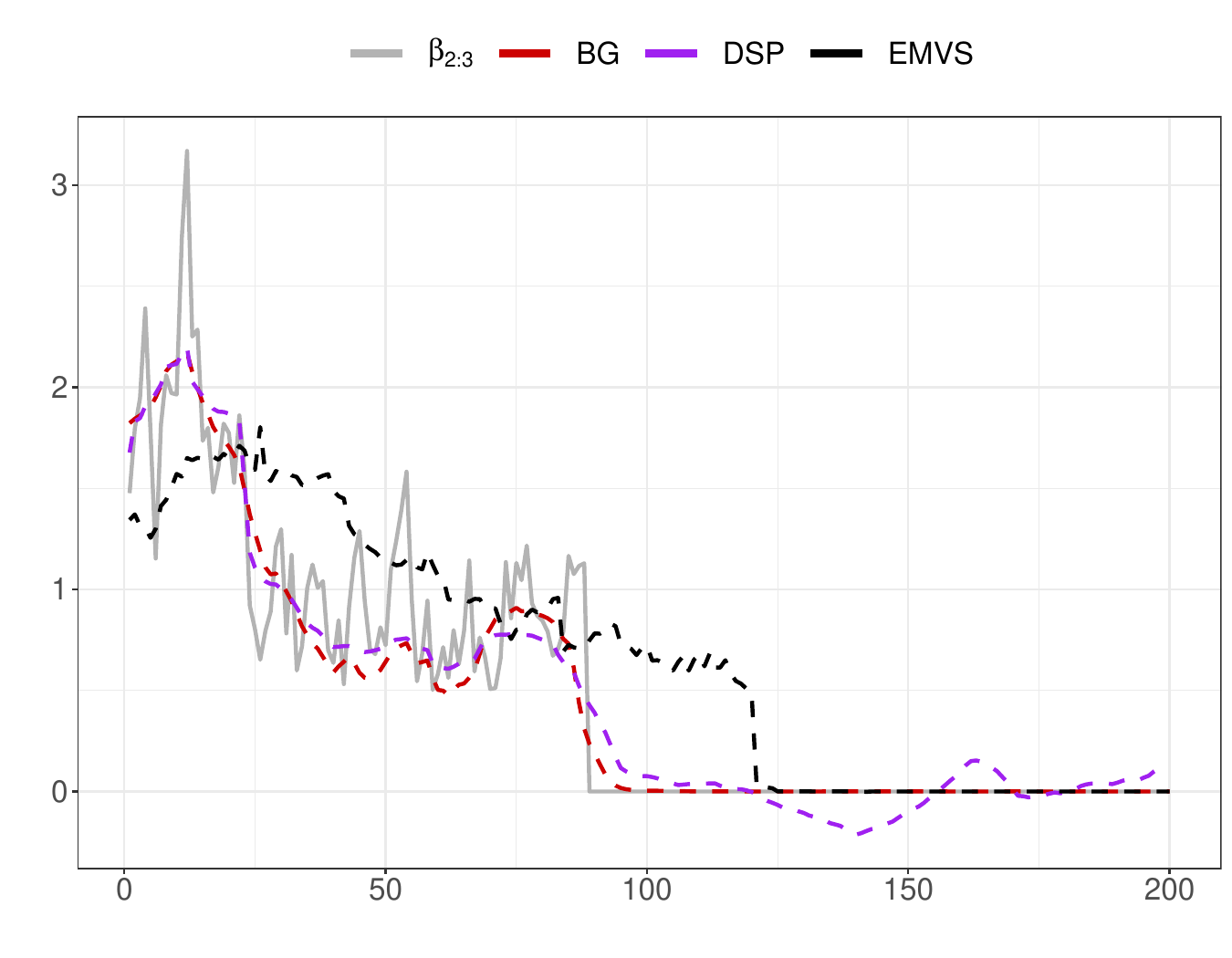}\label{fig:beta_sim2}}

\hspace{-1.5em}\subfigure[Example of $\beta_{4:5,t}$]{\includegraphics[width=.52\textwidth]{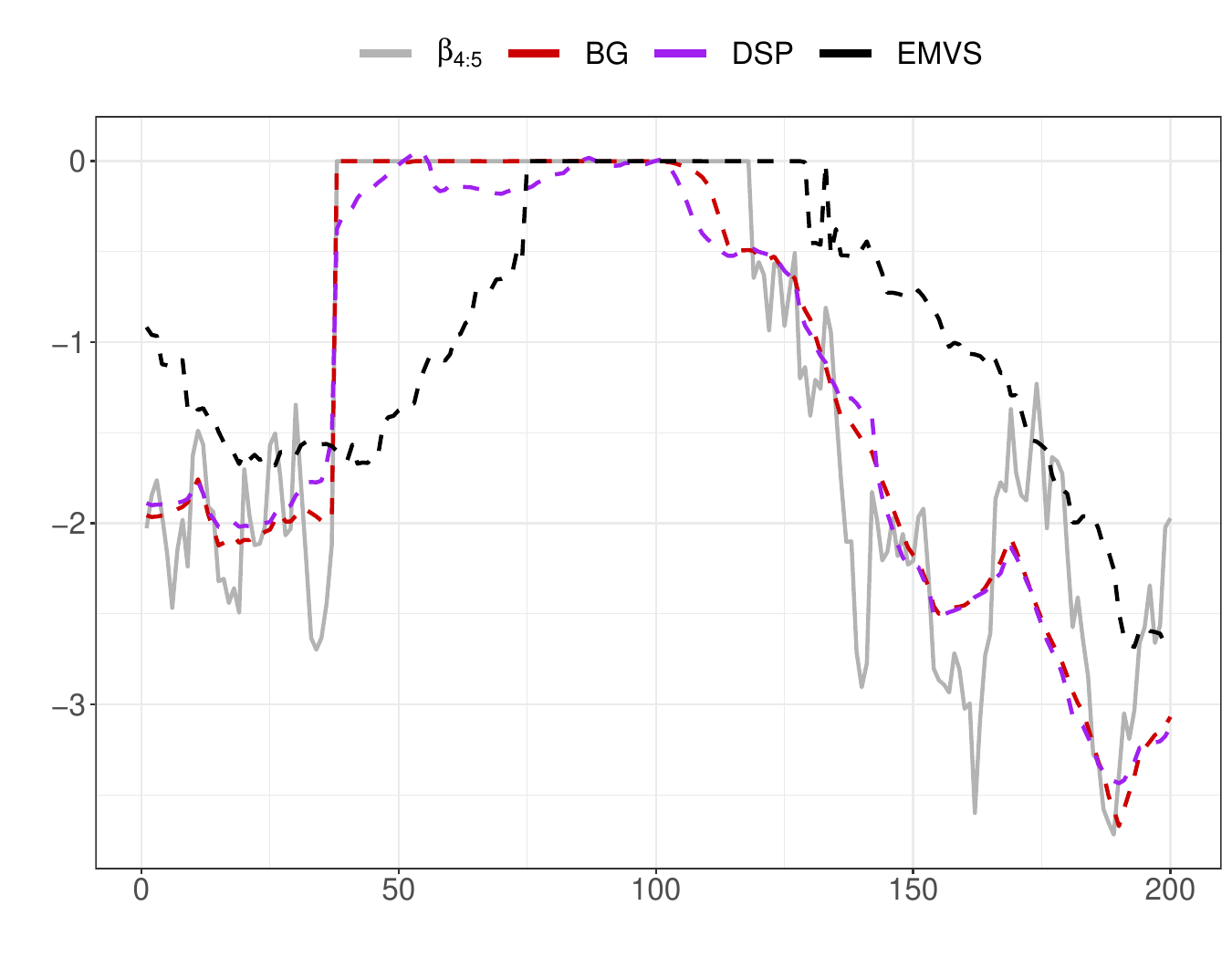}\label{fig:beta_sim4}}
\subfigure[Example of $\beta_{6:7,t}$]{\includegraphics[width=.52\textwidth]{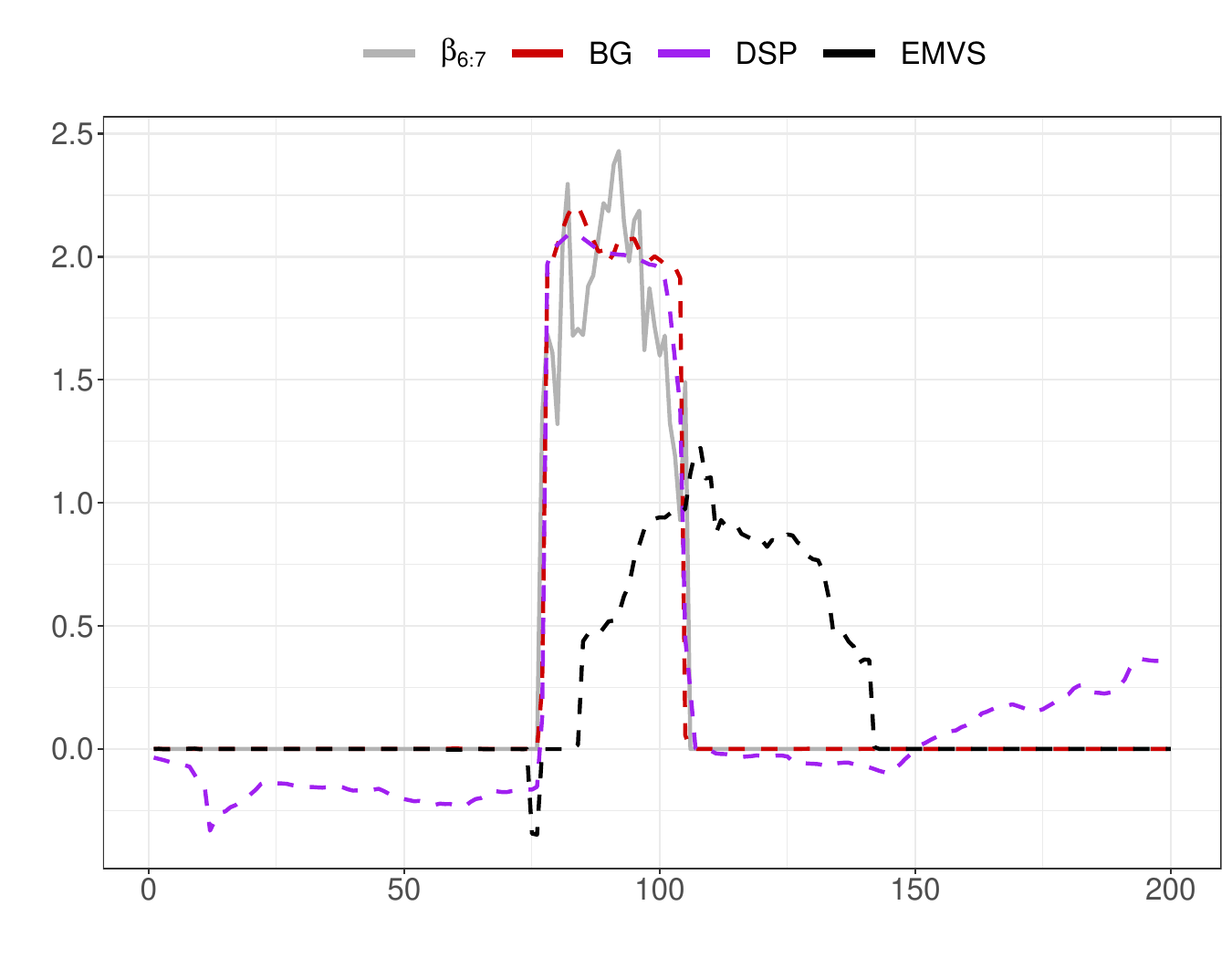}\label{fig:beta_sim6}} 
	\caption{\small Examples of simulated trajectories for $\beta_{jt}$ with $j=1,\ldots,7$ and $t=1,\ldots,200$.}\label{fig:res_tvp_sim_app}
\end{figure}

The coefficient $\beta_{2:3,t}$ is assumed to have a single switch from $\gamma_{jt}=0$ to $\gamma_{jt}=1$. Specifically, the parameter is generated by dividing the interval in sub-periods $[1,n] = [1,t_1] \cup [t_1+1,t_1+t_2] \cup ... \cup [t_1+\ldots+t_n+1,n]$, where $t_k\sim\mathsf{Pois}(n/2)$, so that the expected number of sub-periods is 2, and then randomly alternate periods where $\gamma_{jt}=0$ and $\gamma_{jt}=1$. For the intervals where $\gamma_{jt}=1$ we generate an AR(1) process as for $\beta_{1t}$. This represents a ``structural break'' scenario in which the estimation accuracy broadly deteriorates. The top-right panel in Figure \ref{fig:res_tvp_sim_app} provides an example of the simulated trajectory.  

The parameters $\beta_{4:5,t}$ have a more complex dynamic with two switches from $\gamma_{jt}=0$ to $\gamma_{jt}=1$. Specifically, the parameters are generated as follows by dividing the interval into sub-periods as for $\beta_{2:3,t}$, but set $t_k\sim\mathsf{Pois}(n/4)$, so that the expected number of sub-periods is 4, and then randomly alternate periods where $\gamma_{jt}=0$ and $\gamma_{jt}=1$. For the intervals where $\gamma_{jt}=1$ the process is an AR(1) as for $\beta_{1t}$. The bottom-left panel of Figure \ref{fig:res_tvp_sim_app} provides an example of the simulated trajectory. 

Finally, the parameters $\beta_{6:7,t}$ are modelled as a short-lived regression coefficient, which is significant only for a short fraction of the sample. An example of the simulated trajectory is reported in the bottom-right panel of Figure \ref{fig:res_tvp_sim_app}. The dynamic of the parameter is generated by sampling an interval length $\Delta_i\sim\mathsf{Pois}(n/10)$ and placing it at random on the timeline such that $\gamma_{jt}=1$ in that period, then generating a trajectory for the coefficient as for $\beta_{1t}$. This constitutes a rather extreme case in which a predictor is significant only for a very short period. This setup is used across all simulation studies whereby $\beta_{1:7,t}$ are significant and $\beta_{8:p,t}$ are non-significant. 

\subsection{Additional simulation results}
\label{app:BG vs others}

Figure 2 in the main paper reports the F1 score aggregated for all parameters with dynamic sparsity $\beta_{2:7,t}$ for $p=10$ and $p=100$. Figure \ref{fig:res_sim_corr app} reports the additional results for $p=50$ and $p=200$. The top panels report the results for the independent predictors, whereas the bottom panels report the results for the correlated predictors. 

\begin{figure}[!ht]
\begin{flushleft}
\bf Panel A\normalfont: Independent predictors.    
\end{flushleft}
\subfigure[F1-score ($p=50$)]{\includegraphics[width=0.48\textwidth]{Figures/F1_p10.pdf}\label{fig:f1_tvp_50}}
\subfigure[F1-score ($p=200$)]{\includegraphics[width=0.48\textwidth]{Figures/F1_p50.pdf}\label{fig:f1_tvp_200}}

\begin{flushleft}
\bf Panel B\normalfont: Correlated predictors.
\end{flushleft}
\hspace{-1em}\subfigure[F1-score for $p=50$]{\includegraphics[width=.48\textwidth]{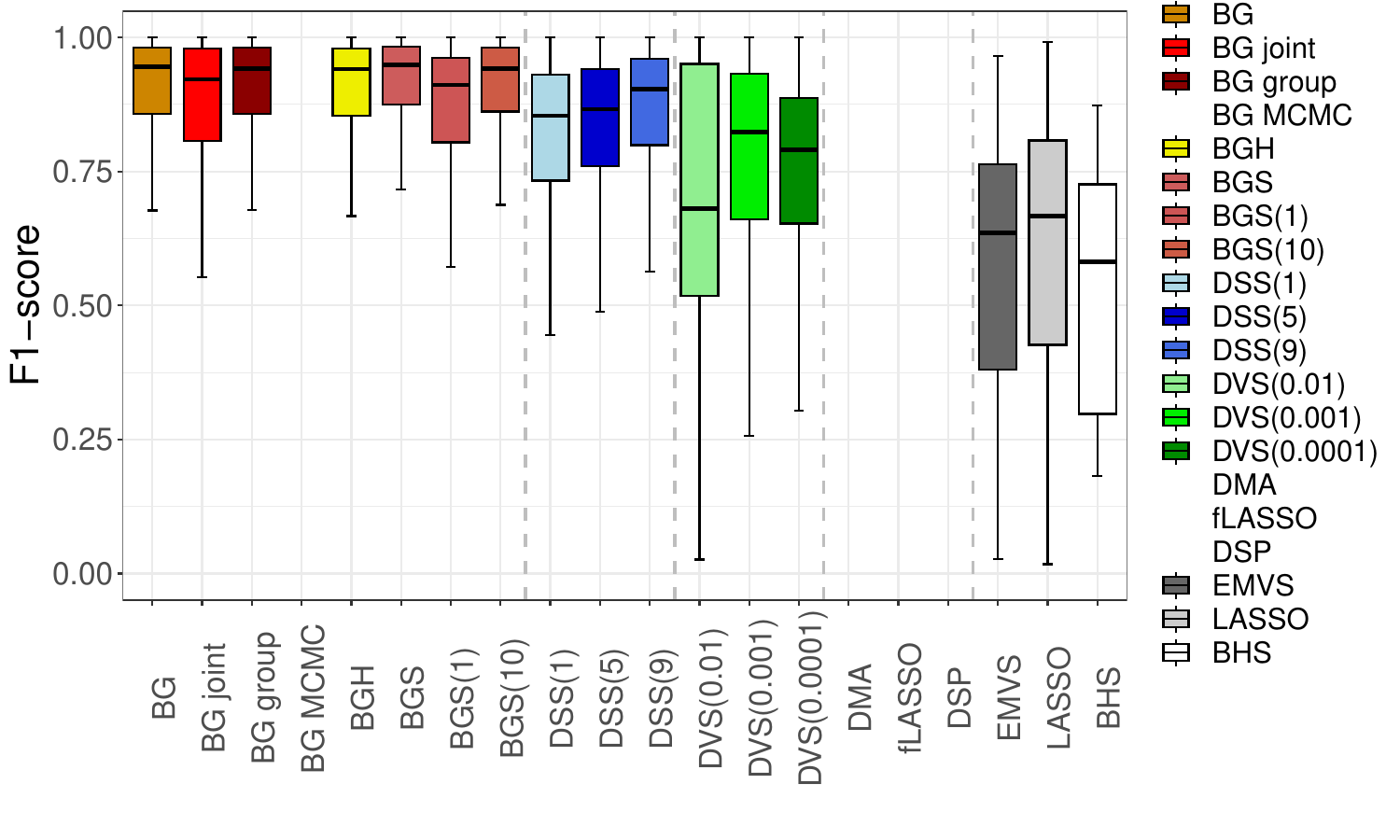}}
\subfigure[F1-score for $p=200$]{\includegraphics[width=.48\textwidth]{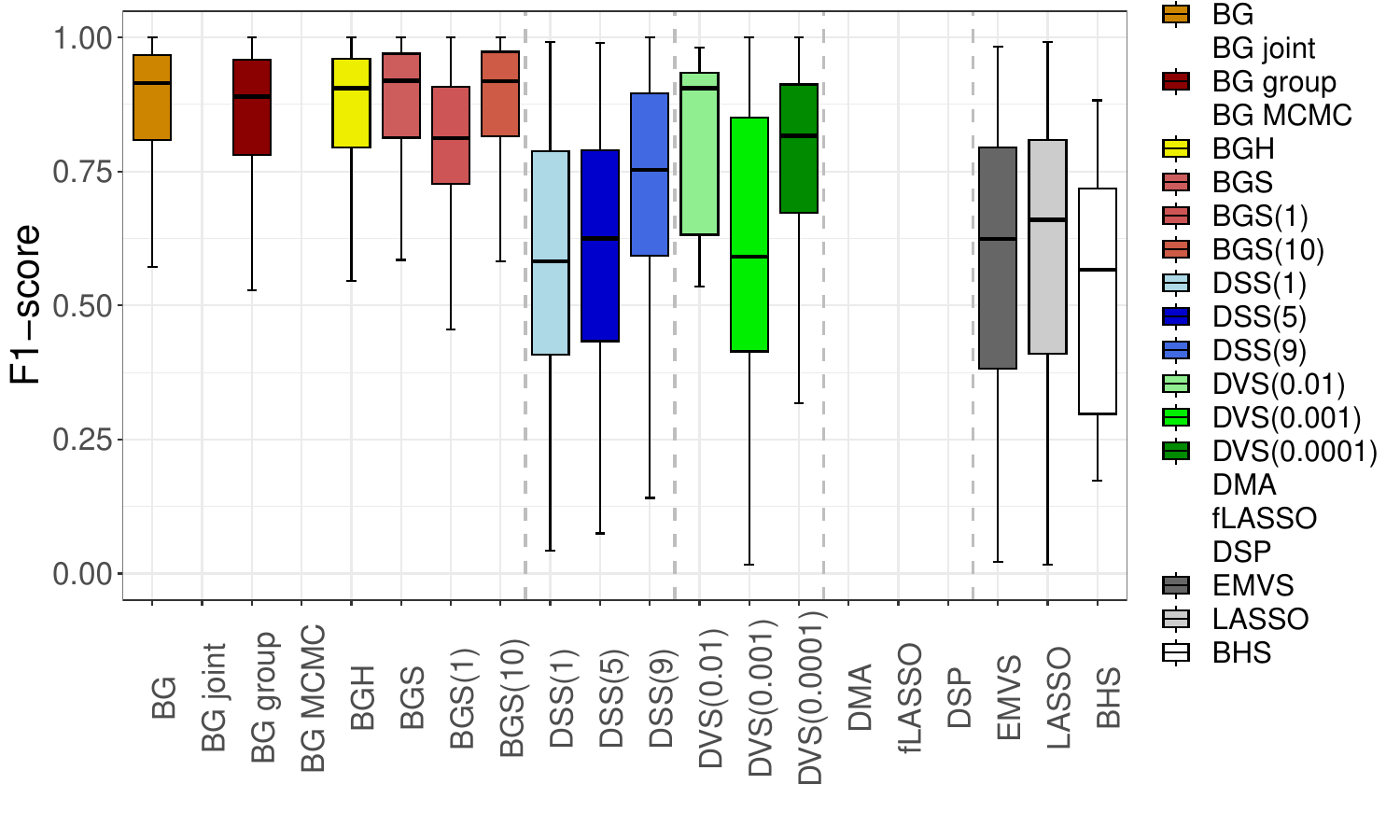}}

\caption{\small Aggregate F1 score for parameters $\beta_{2:7,t}$ exhibiting dynamic sparsity patterns over the sample period with correlated predictors. Results are reported for dimensions $p=50,200$.}\label{fig:res_sim_corr app}
\end{figure}

Table \ref{tab:f1_additional_indep} reports a set of additional results on the ability of each variable selection method to correctly identify the time-varying intercept $\beta_{1,t}$ and all non-significant predictors $\beta_{8:p,t}$. All variable selection methods can identify quite accurately $\gamma_{jt}$ when $\gamma_{jt}=1,\ \forall t$, especially for small to moderate dimensions. When the number of predictors increases, the performance of the dynamic spike-and-slab specifications of \cite{rockova_mcalinn_2021} and \cite{koop_korobilis_2020} tend to deteriorate, with an F1 score as low as 0.91. Overall, when a predictor is significant throughout the sample, there seems to be little disagreement across variable selection methods. 

\begin{table}[!ht]
\centering
\renewcommand{\arraystretch}{0.6}
\resizebox{1\textwidth}{!}{
   \begin{tabular}{llrrrrrrrrrrrrrr}
   \toprule 
    & \multicolumn{6}{c}{F1-score on $\beta_{1t}$} & & \multicolumn{6}{c}{Classification accuracy on $\beta_{8:p,t}$} \\
   \cmidrule{2-7}\cmidrule{9-14}
    Method & \multicolumn{1}{l}{$p=5$} & \multicolumn{1}{l}{$p=10$} & \multicolumn{1}{l}{$p=20$} & \multicolumn{1}{l}{$p=50$} & \multicolumn{1}{l}{$p=100$} & \multicolumn{1}{l}{$p=200$} &       & \multicolumn{1}{l}{$p=5$} & \multicolumn{1}{l}{$p=10$} & \multicolumn{1}{l}{$p=20$} & \multicolumn{1}{l}{$p=50$} & \multicolumn{1}{l}{$p=100$} & \multicolumn{1}{l}{$p=200$} \\
    \midrule
 {\tt BG}  & 0.988 & 0.994 & 0.999 & 0.999 & 0.999 & 1 &   & 0.998 & 0.999 & 1 & 1 & 1 & 1 \\
 {\tt BG joint}  & 0.996 & 0.998 & 0.999 & 0.998 &   &   &   & 1 & 1 & 1 & 1 &  & \\
 {\tt BG group}  & 0.996 & 0.997 & 0.999 & 0.999 & 0.999 & 1 &   & 1 & 0.999 & 1 & 1 & 1 & 1 \\
 {\tt BG MCMC}  & 0.999 & 1 & 0.999 &   &   &   &   & 1 & 0.996 & 0.988 &  &  & \\
 {\tt BGH}  & 0.978 & 0.988 & 0.998 & 0.999 & 0.999 & 1 &   & 0.994 & 0.998 & 1 & 1 & 1 & 1\\
 {\tt BGS}  & 0.989 & 0.994 & 0.999 & 0.999 & 0.999 & 1 &   & 0.999 & 0.999 & 1 & 1 & 1 & 1\\
 {\tt BGS(1)}  & 0.992 & 0.996 & 0.999 & 1 & 1 & 1 &   & 1 & 1 & 1 & 1 & 1 & 1\\
 {\tt BGS(10)}  & 0.986 & 0.992 & 0.998 & 0.998 & 0.998 & 1 &   & 0.998 & 0.999 & 1 & 1 & 1 & 1\\
 {\tt DSS(1)}  & 0.998 & 0.996 & 0.998 & 0.984 & 0.966 & 0.909 &   & 0.994 & 0.995 & 0.996 & 0.995 & 0.998 & 0.999\\
 {\tt DSS(5)}  & 0.998 & 0.996 & 0.999 & 0.99 & 0.971 & 0.932 &   & 0.992 & 0.992 & 0.996 & 0.994 & 0.997 & 0.998\\
 {\tt DSS(9)}  & 0.999 & 0.998 & 0.999 & 0.997 & 0.991 & 0.98 &   & 0.98 & 0.987 & 0.993 & 0.992 & 0.996 & 0.998\\
 {\tt DVS(0.01)}  & 0.995 & 0.95 & 0.945 & 0.888 & 0.937 & 0.979 &   & 0.957 & 0.996 & 0.998 & 1 & 1 & 1\\
 {\tt DVS(0.001)}  & 1 & 1 & 1 & 0.999 & 0.997 & 0.997 &   & 0.806 & 0.961 & 0.981 & 0.991 & 0.997 & 0.999\\
 {\tt DVS(0.0001)}  & 1 & 1 & 1 & 1 & 1 & 0.999 &   & 0.264 & 0.406 & 0.709 & 0.862 & 0.927 & 0.967\\
 {\tt DMA}  & 1 & 1 &   &   &   &   &   & 0.91 & 0.899 & & & & \\
 {\tt fLASSO}  & 1 & 1 & 1 &   &   &   &   & 0.91 & 0.899 & 0.891 & 1 & 1 & 1\\
 {\tt DSP}  & 1 & 1 & 1 &  &  &  &   & 0 & 0 & 0 &  &  & \\
 {\tt EMVS}  & 0.999 & 0.997 & 0.999 & 0.995 & 0.996 & 0.995 &   & 0.952 & 0.99 & 0.996 & 0.802 & 0.958 & 0.987\\
 {\tt LASSO}  & 0.999 & 0.997 & 0.998 & 0.995 & 0.991 & 0.987 &   & 0.917 & 0.94 & 0.951 & 0.944 & 0.964 & 0.981\\
    \bottomrule 
    \end{tabular}%
}
\caption{\small Results for $\beta_{1t}$ and $\beta_{8:p,t}$ in the independent design setting. This table reports the F1-score for the time-varying intercept $\beta_{1t}$ and the classification accuracy for the null coefficients $\beta_{8:p,t}$ in the independent design setting.}
\label{tab:f1_additional_indep}
\end{table}

Table \ref{tab:f1_additional_indep} suggests significant differences across methods in their ability to correctly identify irrelevant predictors throughout the sample. The aggregate F1 score for the parameters $\beta_{8:p,t}$ is often significantly lower for the {\tt DVS} and static variable selection methods. This result holds especially for lower dimensions, whereby the aggregate F1 score of the {\tt DVS} is between 0.26 and 0.96 from $p=5$ to $p=200$ when the hyper-parameter $\underline{c}=0.0001$. Nevertheless, one could argue that most models can correctly identify predictors which are just noise for the full sample. By coupling the results from Figure 2 in the main text and Table \ref{tab:f1_additional_indep}, one can argue that the key advantage of our dynamic variable selection method lies in its ability to detect non-trivial patterns of dynamic sparsity.

\subsection{Approximation accuracy of variational inference.} 
\label{sec:vs MCMC appendix}
We consider $p=\left\{5,10,20\right\}$ and generate $\{\boldsymbol{\beta}_{t}\}_{t=1}^{100}$ with such that $\beta_{1t}$ is a time-varying parameter always included in the model (see top-left panel in Figure \ref{fig:res_tvp_sim_app}), $\beta_{2t}$ is set to enter once over the sample period in the set of significant coefficients (see top-right panel in Figure \ref{fig:res_tvp_sim_app}), $\beta_{3t}$ enters twice over the sample period (see bottom-left panel in Figure \ref{fig:res_tvp_sim_app}), and $\beta_{4t}$ is a short-lived coefficient with a trajectory as shown in the bottom-right panel of Figure \ref{fig:res_tvp_sim_app}). Coefficients $\beta_{5:p,t}$ are set to zero over the sample period. 

In the main text, for each simulated parameter, we report the overlapping posterior densities for one selected replicate of $\beta_{2t}$ obtained via VB (blue) vs MCMC (red) under different assumptions on the factorization. Figure \ref{fig:acc_tvp_app} complements the results in the main text by showing the posterior densities of $\beta_{1t}$ obtained via VB (blue) and MCMC (red) for one selected replicate. The joint approximation provides better overlapping than the independent one, consistently with the results in the main text of the paper (see Section 3).

\begin{figure}[hp!]
\subfigure[Fully factorized variational density $K=p$]{\includegraphics[width=0.5\textwidth]{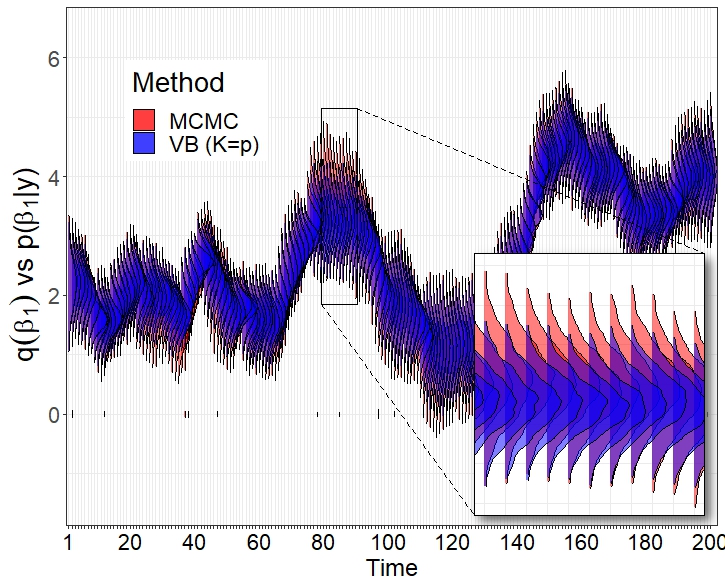}}
\subfigure[Joint variational density $K=1$]{\includegraphics[width=0.5\textwidth]{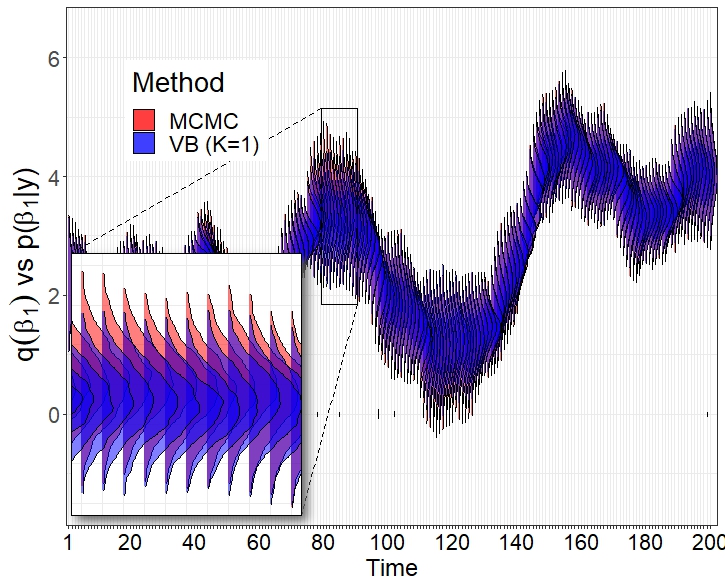}}
\caption{\small Comparison with MCMC. We show the posterior densities of $\beta_{1t}$ under different approximation assumptions (blue) and MCMC (red) for one selected replicate.}\label{fig:acc_tvp_app}
\end{figure}

\paragraph{Computational costs.}
In this Section, we expand the results in Figure 4 in the main text and report the computational cost -- expressed in minutes -- to estimate the full set of competing dynamic variable selection methods used in the main simulation study. 

\begin{table}[!ht]
\centering
\renewcommand{\arraystretch}{0.6}
\resizebox{1\textwidth}{!}{
\begin{tabular}{lrrrrrrrrrrrrrrrrr}
\multicolumn{18}{l}{Panel A: \normalfont Independent predictors }\\
\toprule
    \multicolumn{1}{l}{p} & \multicolumn{1}{l}{BG} & \multicolumn{1}{l}{BG MCMC} & \multicolumn{1}{l}{BG joint} & \multicolumn{1}{l}{BG group} & \multicolumn{1}{l}{BGH} & \multicolumn{1}{l}{BGS} & \multicolumn{1}{l}{BGS(1)} & \multicolumn{1}{l}{BGS(10)} & \multicolumn{1}{l}{DVS(0.01)} & \multicolumn{1}{l}{DVS(0.001)} & \multicolumn{1}{l}{DVS(0.0001)} & \multicolumn{1}{l}{DSS(1)} & \multicolumn{1}{l}{DSS(5)} & \multicolumn{1}{l}{DSS(9)} & \multicolumn{1}{l}{fLASSO} & \multicolumn{1}{l}{DMA} & \multicolumn{1}{l}{DSP} \\
    \midrule 
    5     & 0.01  & 3.02  & 0.07  & 0.07  & 0.01  & 0.02  & 0.02  & 0.01  & 0.02  & 0.03  & 0.03  & 0.13  & 0.13  & 0.1   & 0.1   & 0.1   & 0.41 \\
    10    & 0.02  & 6.06  & 0.12  & 0.18  & 0.02  & 0.03  & 0.03  & 0.03  & 0.03  & 0.03  & 0.04  & 0.21  & 0.21  & 0.18  & 0.45  & 0.23  & 1.19 \\
    20    & 0.04  & 7.81  & 0.26  & 0.29  & 0.03  & 0.04  & 0.04  & 0.04  & 0.05  & 0.04  & 0.06  & 0.32  & 0.31  & 0.24  & 2.15  &       & 3.4 \\
    50    & 0.11  &       & 1.3   & 0.85  & 0.11  & 0.13  & 0.14  & 0.12  & 0.37  & 0.42  & 0.44  & 0.54  & 0.54  & 0.5   &       &       &  \\
    100   & 0.22  &       &       & 1.69  & 0.22  & 0.24  & 0.3   & 0.22  & 1.94  & 1.95  & 1.97  & 1.39  & 1.39  & 1.29  &       &       &  \\
    200   & 0.44  &       &       & 3.22  & 0.44  & 0.46  & 0.7   & 0.41  & 16.45 & 15.93 & 16.04 & 4.31  & 4.3   & 4.12  &       &       &  \\
\bottomrule
\end{tabular}}
\vspace{1em}

\resizebox{1\textwidth}{!}{
\begin{tabular}{lrrrrrrrrrrrrrrrrr}
\multicolumn{8}{l}{Panel B: \normalfont Correlated predictors } \\
\toprule
     \multicolumn{1}{l}{p} & \multicolumn{1}{l}{BG} & \multicolumn{1}{l}{BG MCMC} & \multicolumn{1}{l}{BG joint} & \multicolumn{1}{l}{BG group} & \multicolumn{1}{l}{BGH} & \multicolumn{1}{l}{BGS} & \multicolumn{1}{l}{BGS(1)} & \multicolumn{1}{l}{BGS(10)} & \multicolumn{1}{l}{DVS(0.01)} & \multicolumn{1}{l}{DVS(0.001)} & \multicolumn{1}{l}{DVS(0.0001)} & \multicolumn{1}{l}{DSS(1)} & \multicolumn{1}{l}{DSS(5)} & \multicolumn{1}{l}{DSS(9)} & \multicolumn{1}{l}{fLASSO} & \multicolumn{1}{l}{DMA} & \multicolumn{1}{l}{DSP} \\
    \midrule 
    5     & 0.01  & 3.02  & 0.07  & 0.07  & 0.01  & 0.02  & 0.02  & 0.01  & 0.02  & 0.03  & 0.03  & 0.13  & 0.13  & 0.1   & 0.1   & 0.1   & 0.41 \\
    10    & 0.02  & 6.06  & 0.12  & 0.18  & 0.02  & 0.03  & 0.03  & 0.03  & 0.03  & 0.03  & 0.04  & 0.21  & 0.21  & 0.18  & 0.45  & 0.23  & 1.19 \\
    20    & 0.04  & 7.81  & 0.26  & 0.29  & 0.03  & 0.04  & 0.04  & 0.04  & 0.05  & 0.04  & 0.06  & 0.32  & 0.31  & 0.24  & 2.15  &       & 3.4 \\
    50    & 0.11  &       & 1.3   & 0.85  & 0.11  & 0.13  & 0.14  & 0.12  & 0.37  & 0.42  & 0.44  & 0.54  & 0.54  & 0.5   &       &       &  \\
    100   & 0.22  &       &       & 1.69  & 0.22  & 0.24  & 0.3   & 0.22  & 1.94  & 1.95  & 1.97  & 1.39  & 1.39  & 1.29  &       &       &  \\
    200   & 0.44  &       &       & 3.22  & 0.44  & 0.46  & 0.7   & 0.41  & 16.45 & 15.93 & 16.04 & 4.31  & 4.3   & 4.12  &       &       &  \\
    \bottomrule 
    \end{tabular}
}
\caption{\small Computational efficiency. This table reports the computational time expressed in minutes to estimate a given dynamic variable selection model used in the main empirical analysis. Panel A reports the results for the independent predictors' design, whereas Panel B reports the results for the correlated predictors' design.}
\label{tab:comptime}
\end{table} 

Panel A of Table \ref{tab:comptime} reports the computational cost of each model for the simulation design with independent predictors. The results show that popular dynamic variable selection methods such as the {\tt fLASSO} are not competitive for large dimensional predictive regressions. Similarly, the {\tt DMA} of \cite{raftery_etal2010} is ten times more computationally expensive than our {\tt BG}, {\tt BGS} specifications for a small-scale model with $p=10$. The {\tt DSP} of \cite{kowal2019dynamic} is almost a hundred times slower for $p=20$ and is not computationally viable for larger dimensions. The results for the simulation design with correlated predictors (see Panel B) largely confirm the insight from independent variables; that is, (1) the {\tt BG}, {\tt BGS} specifications are the most computationally efficient, (2) a group structure allows to consider a non-trivial correlation structure in the predictors while still maintaining computational efficiency, and (3) popular dynamic variable selection methods such as the fused lasso and dynamic spike-and-slab specifications are significantly less competitive from a computational perspective.

\section{Additional empirical results}\setcounter{figure}{0}\setcounter{table}{0}

\label{app:additional_empirical}

In this section, we will discuss some additional empirical results that have not been included in the main text. We first briefly discuss the correlation structure of the macroeconomic data. Next, we discuss the retrospective analysis of signal vs. uncertainty based on time-varying parameter estimates and the Diebold-Mariano tests for longer forecasting horizons. 

\subsection{Group correlations of macroeconomic predictors}
\label{app:data}

The panel of 229 predictors is retrieved from the FRED-QD database of \cite{mccracken2020fred}. The sample from 1967:Q3 to 2022:Q2. The series are classified into 14 groups: NIPA; Industrial Production; Employment and Unemployment; Housing; Inventories, Orders, and Sales; Prices; Earnings and Productivity; Interest Rates; Money and Credit; Household Balance Sheets; Exchange Rates; Other; Stock Markets; and Non-Household Balance Sheets. The predictors also include the first two lags of the response variable. 

The predictors are transformed following standard norms in the literature \citep[e.g.,][]{stock2012disentangling}. These often are simply a matter of making nominal series real using a deflator. Monthly frequency series are aggregated to a quarterly frequency using averages. Additional details on the data construction can be found in \cite{mccracken2020fred}. Figure \ref{fig:corr} shows the sample correlation structure of the data. The variables have been ordered based on their group. Although it is far from perfect, there is evidence of a block correlation structure in the panel of macroeconomic predictors; that is, the within-group correlation tends to be higher than the across-group correlation, at least on average, over the sample period. 

In the main empirical analysis, we leverage this group classification as prior information to implement the group factorisation $q\left(\mathbf{b}\right)=\prod_{k=1}^{K}q\left(\mathbf{b}_k\right)$.

\begin{figure}[!ht]
\centering
\includegraphics[width=.65\textwidth]{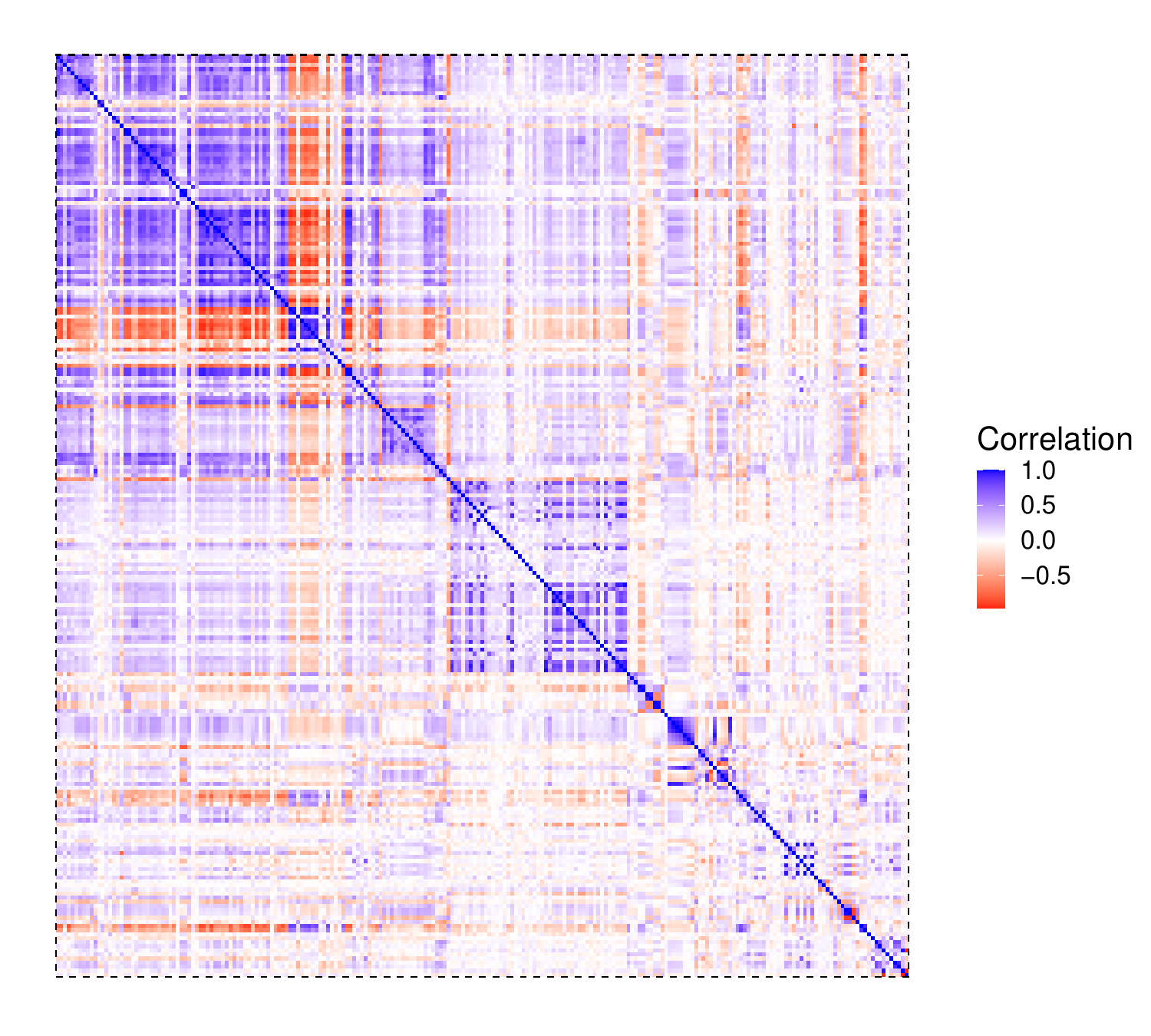}
	\caption{\small Sample correlation of the panel of macroeconomic predictors.}\label{fig:corr}
\end{figure}

\subsection{Out-of-sample forecasting}
\label{app:dm}

We test the statistical significance of the outperformance achieved by our variational Bayes dynamic variable selection based on pairwise \citep{diebold2002} (DM) tests. Figure \ref{fig:dm app} reports the probability of rejecting the null hypothesis $\mathcal{H}_0: MSE^C = MSE^R$ in favour of the alternative $\mathcal{H}_a: MSE^C > MSE^R$, where $MSE^C$ and $MSE^R$ denote the mean-squared error of the column and row model. For brevity, we report the results for the Total CPI (CPIAUCSL), the GDP deflator (GDPCTPI), and the PCE deflator (PCECTPI) for $h=1,2,4$ quarter ahead. For short-term forecasts, the results suggest that our {\tt BG}, {\tt BGS}, and {\tt BG group} models provide statistically comparable performance to the {\tt TVI} model. Yet, {\tt BG}, {\tt BGS}, and {\tt BG group} significantly outperform all of the other variable selection strategies, which make use of a large set of macroeconomic predictors. For longer-term predictions, the performance gap in favour of our dynamic variable selection becomes more statistically significant even compared with {\tt TVI}, {\tt AR(2)}, and {\tt TVAR(2)}.

\begin{figure}[!hp]
\centering
\begin{flushleft}
\bf Panel A: \normalfont Total CPI (CPIAUCSL)
\end{flushleft}
\hspace{-2em}\subfigure[$h=1$]{\includegraphics[width=.33\textwidth]{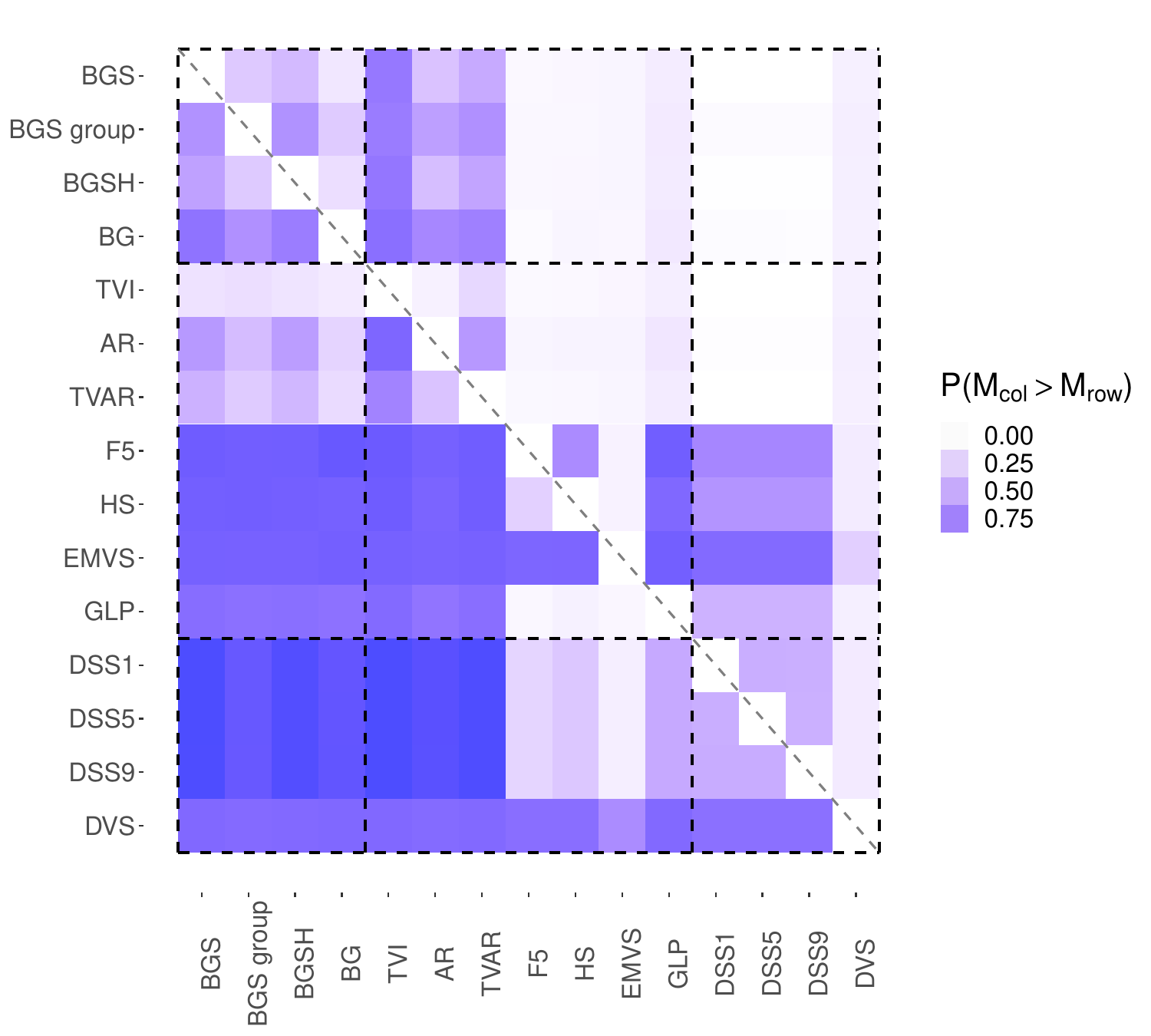}}\subfigure[$h=2$]{\includegraphics[width=.33\textwidth]
{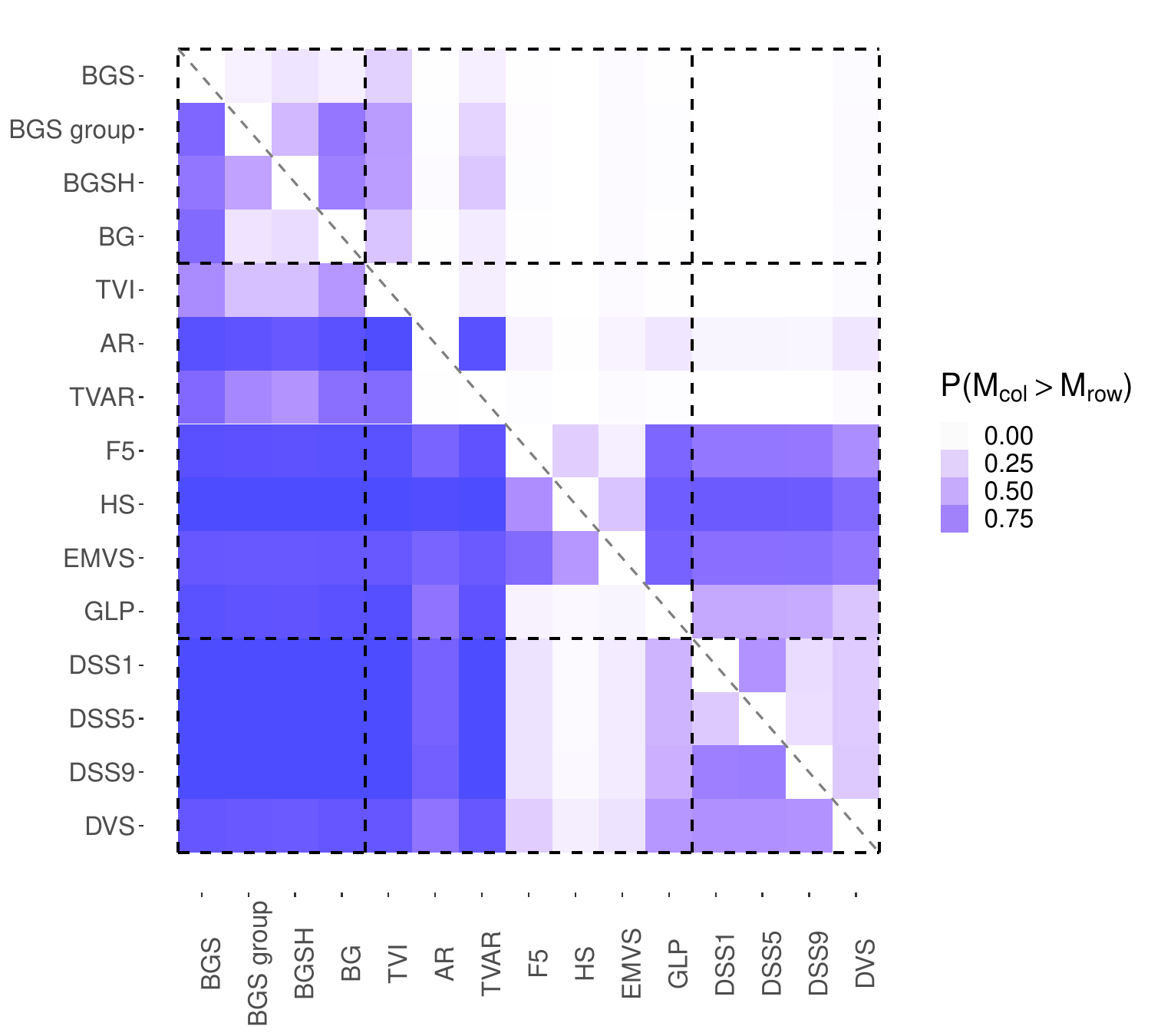}}\subfigure[$h=4$]{\includegraphics[width=.33\textwidth]
{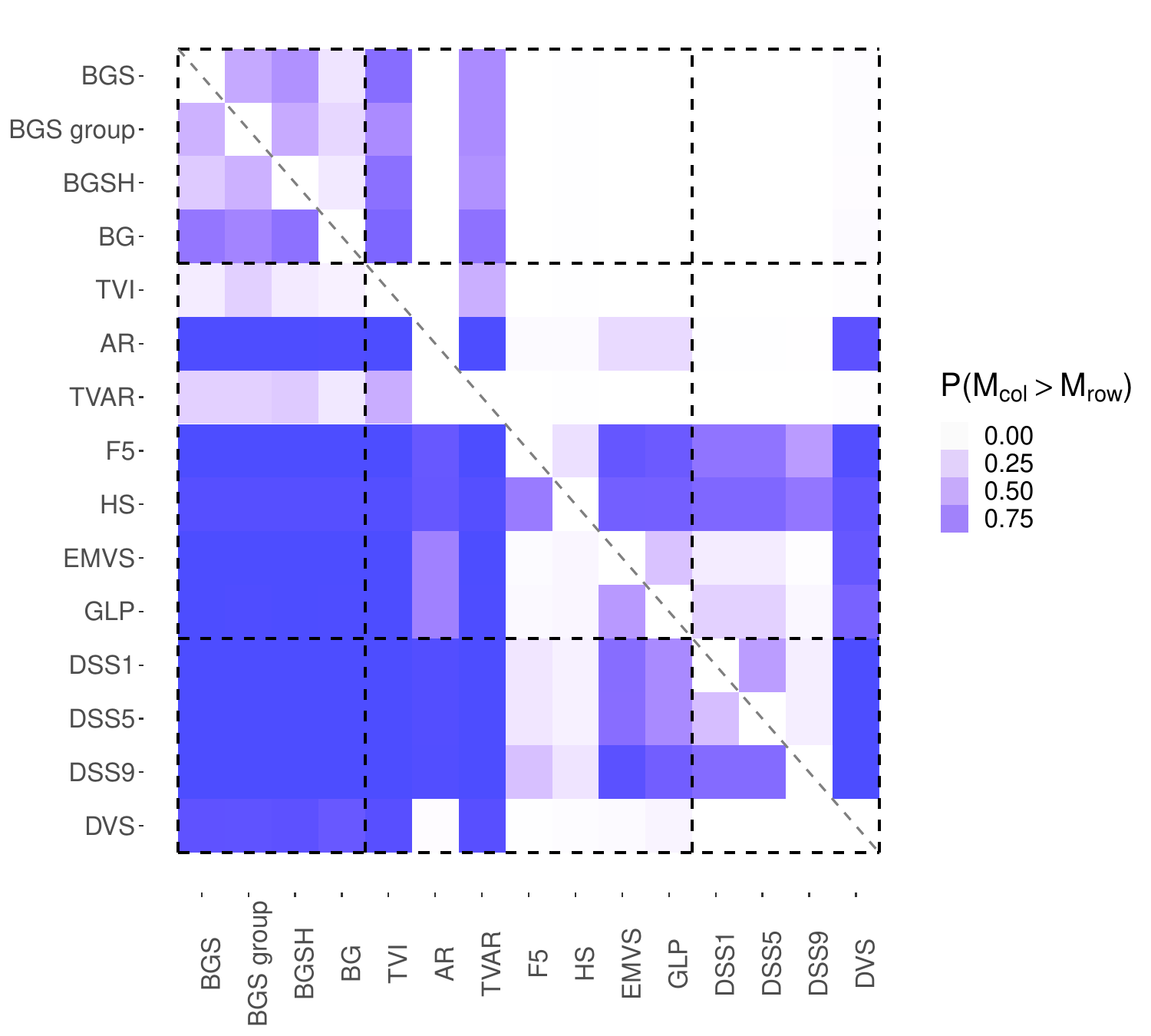}}\hspace{-2em}

\begin{flushleft}
\bf Panel B: \normalfont GDP deflator (GDPCTPI)
\end{flushleft}
\hspace{-2em}\subfigure[$h=1$]{\includegraphics[width=.33\textwidth]{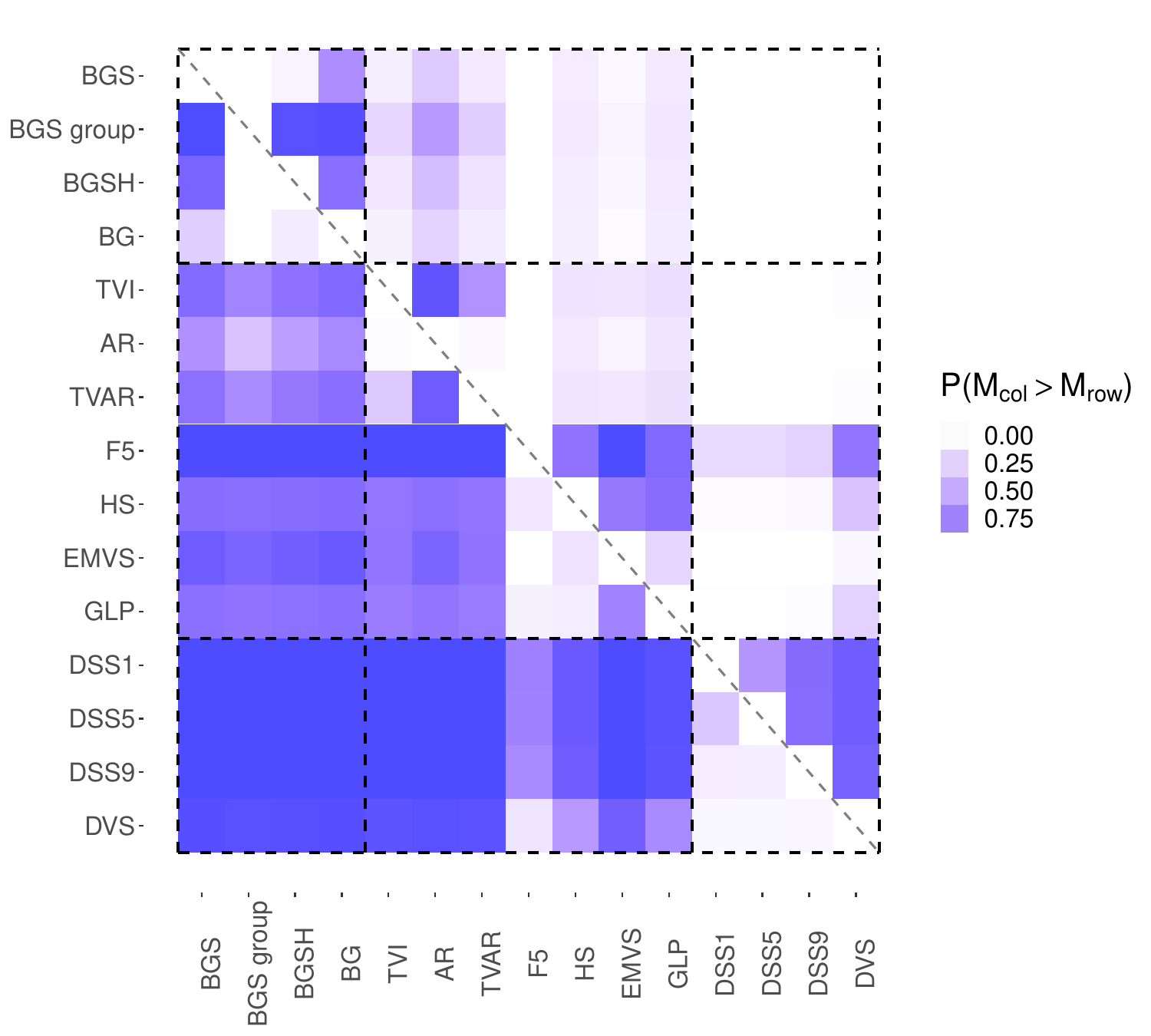}}\subfigure[$h=2$]{\includegraphics[width=.33\textwidth]
{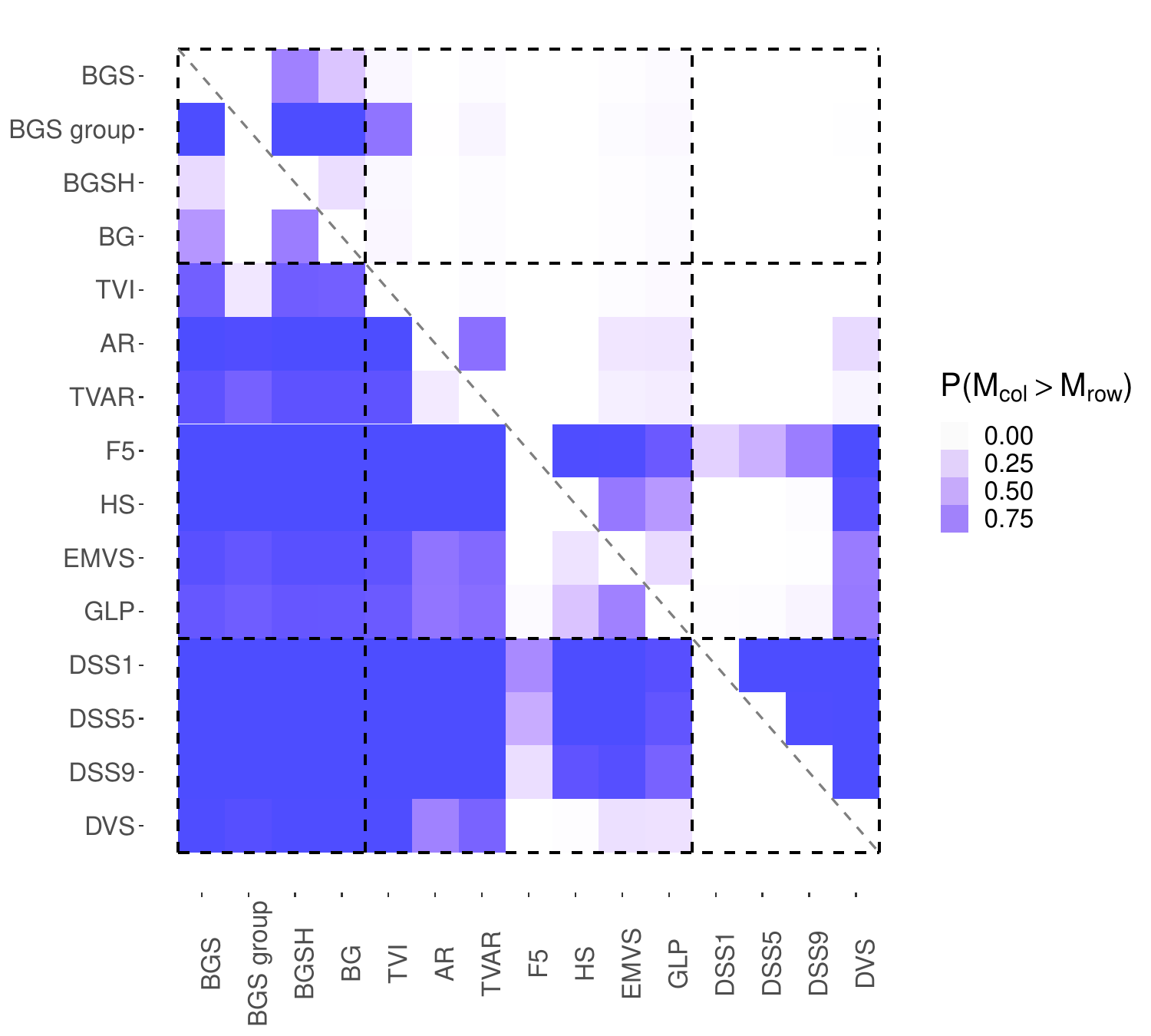}}\subfigure[$h=4$]{\includegraphics[width=.33\textwidth]
{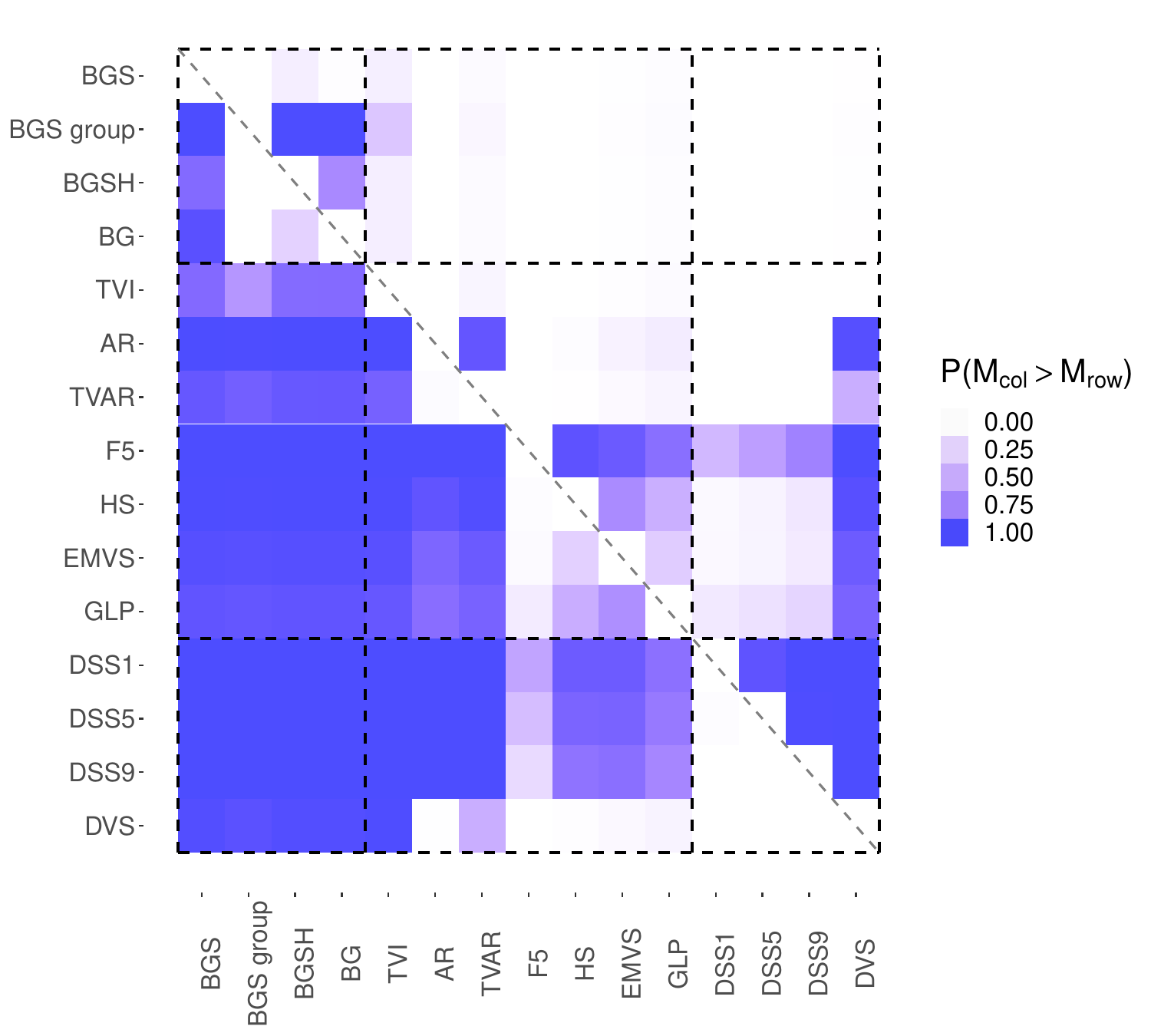}}\hspace{-2em}

\begin{flushleft}
\bf Panel C: \normalfont PCE deflator (PCECTPI)
\end{flushleft}
\hspace{-2em}\subfigure[$h=1$]{\includegraphics[width=.35\textwidth]{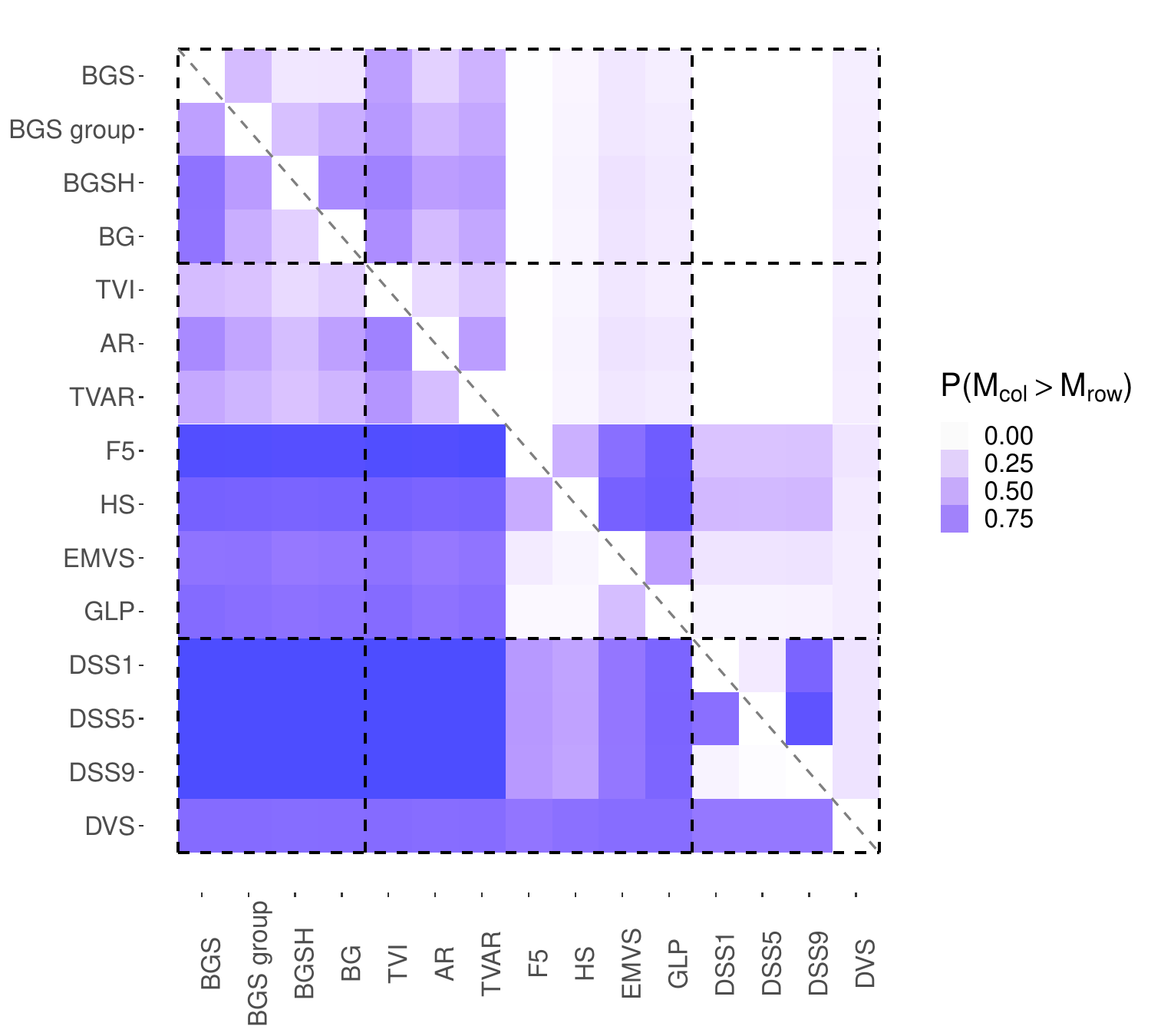}}\subfigure[$h=2$]{\includegraphics[width=.35\textwidth]
{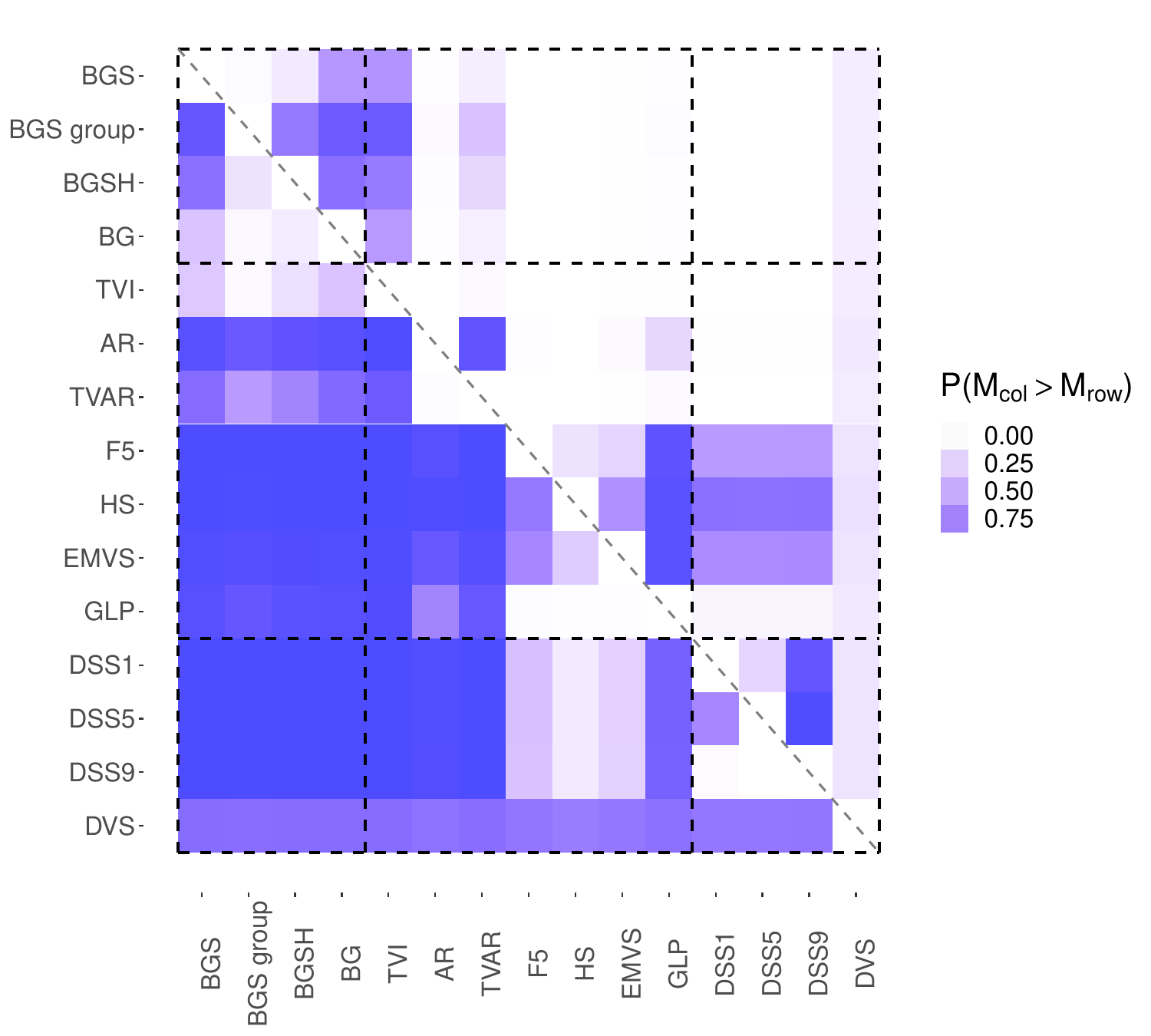}}\subfigure[$h=4$]{\includegraphics[width=.35\textwidth]
{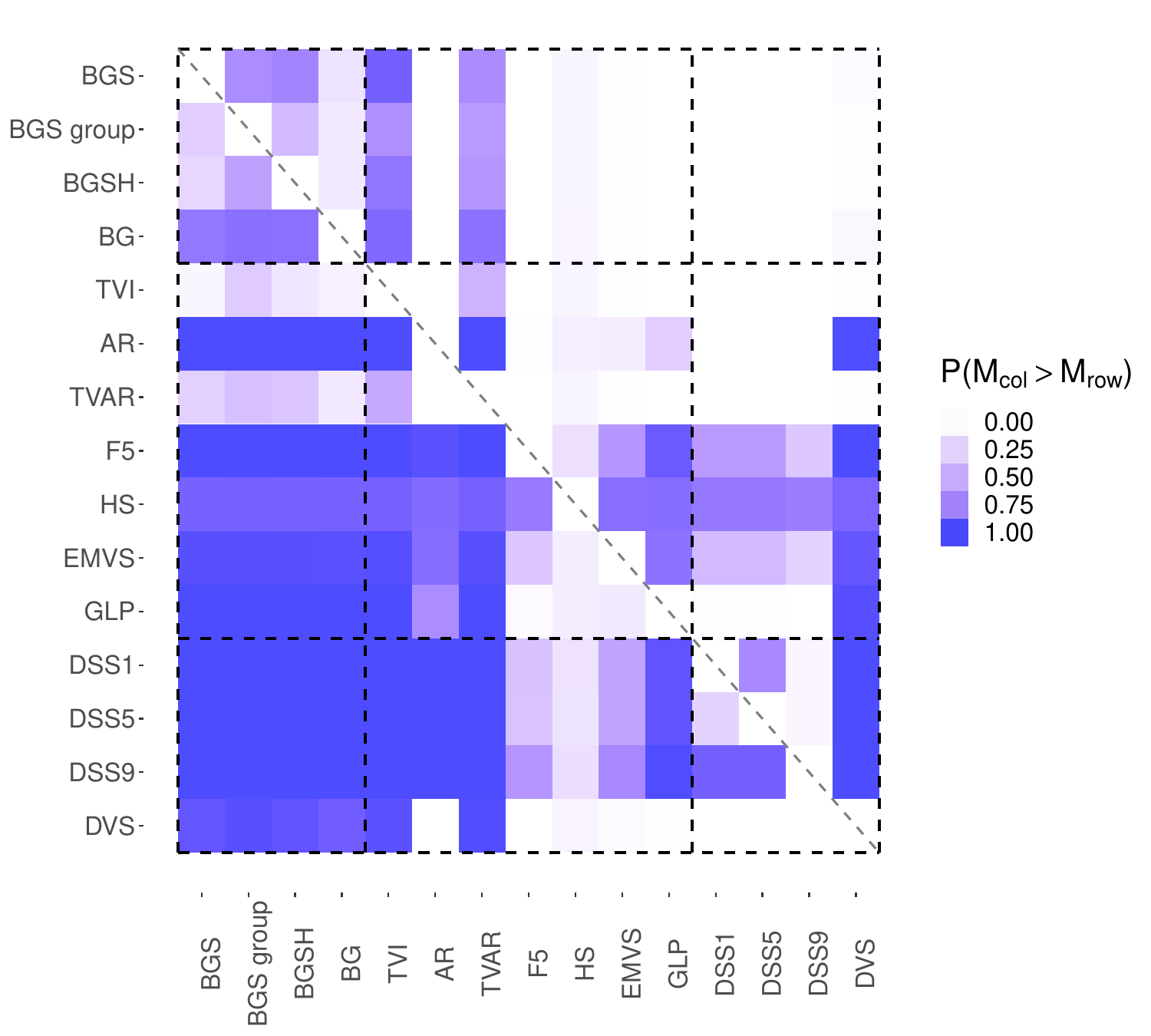}}\hspace{-2em}
	\caption{\small Diebold-Mariano test for the null hypothesis $\mathcal{H}_0: MSE^C \geq MSE^R$, where $MSE^C$ and $MSE^R$ denote the mean-squared error of the column and row model, respectively.}\label{fig:dm app}
\end{figure}

\subsection{Retrospective analysis of macroeconomic predictors}

Figure \ref{fig:beta_appl1 app} reports the coefficient estimates $\mu_{q(\beta_{jt})}$ and inclusion probabilities $\mu_{q(\gamma_{jt})}$ from the {\tt BGS} model, showing only predictors active for significant periods. The figure reports the results for PCE deflator and Core CPI. 

\begin{figure}[ht]
\centering
\hspace{-1em}\subfigure[PCE deflator (PCECTPI)]{\includegraphics[width=0.48\textwidth]{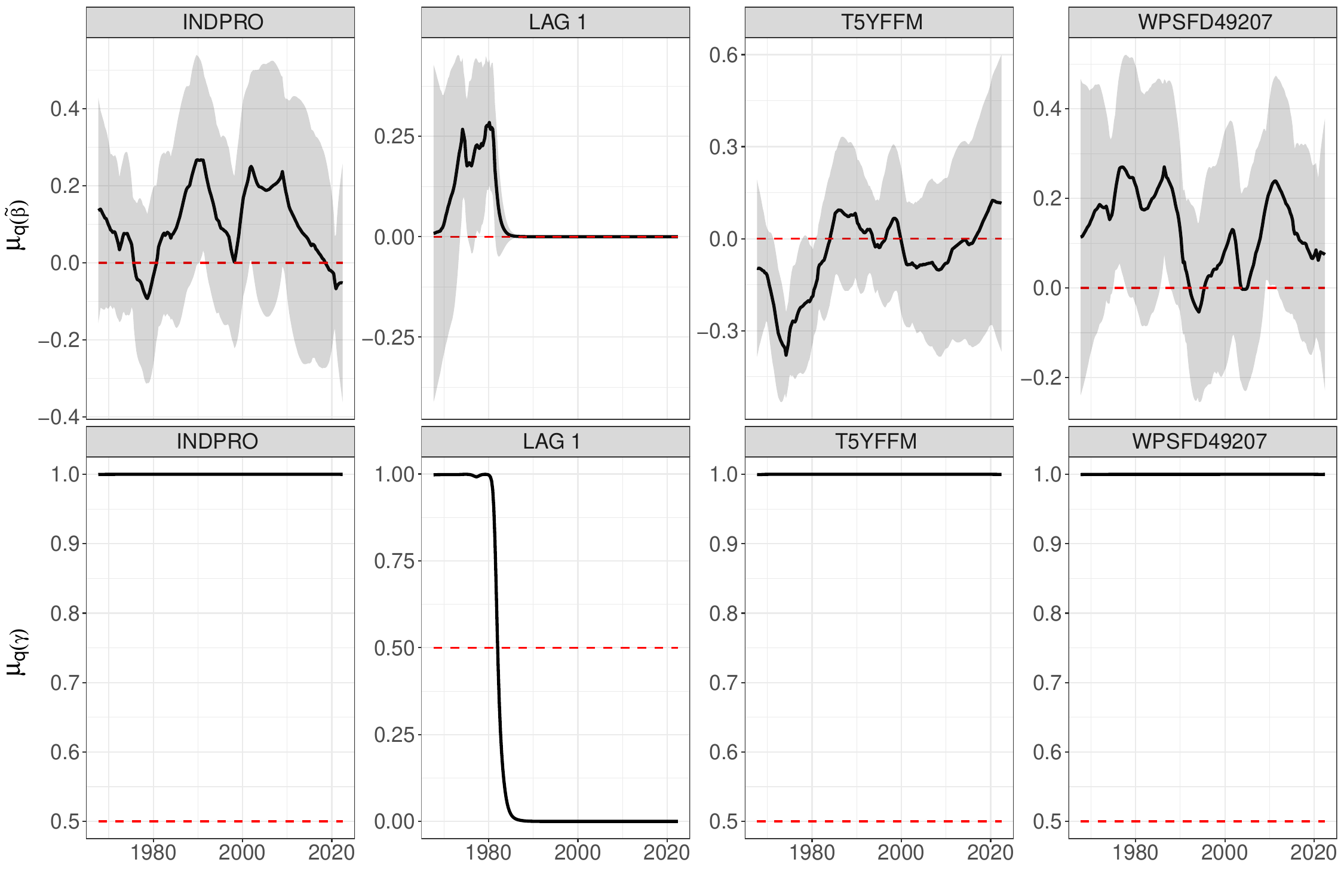}}\hspace{1em}\subfigure[Core CPI (CPILFESL)]{\includegraphics[width=0.48\textwidth]{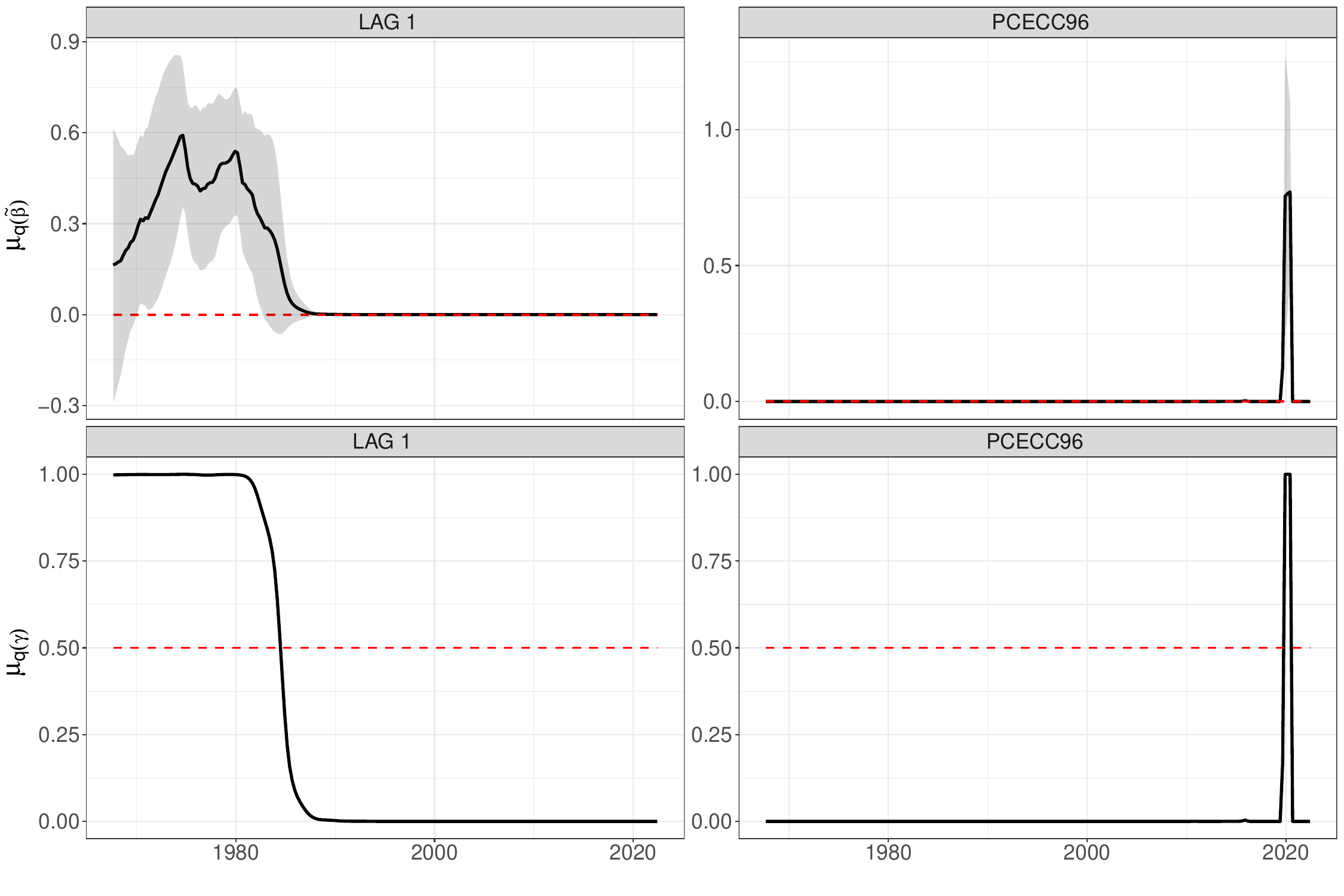}}
\caption{Time-varying coefficient estimates and inclusion probabilities for four inflation measures. Gray bands indicate NBER recessions.}\label{fig:beta_appl1 app}
\end{figure}

Figure \ref{fig:vol} reports the sum of absolute values of the variational mean of the active regression coefficients, i.e., $\sum_{j=1}^p|\mu_{q(\beta_{jt})}|$, which proxies the strength of the information available to predict inflation. The information from the predictors changed over time, decreasing in the middle part of the sample from the 90s to early 2000. In addition, a stronger signal from the predictors correlates with higher idiosyncratic volatility (dashed-red line); that is, a richer model is needed to predict inflation at times of higher uncertainty as proxied by the volatility in the residuals.  

\begin{figure}[!ht]
\hspace{-1em}\subfigure[Total CPI (CPIAUCSL)]{\includegraphics[width=.5\textwidth]{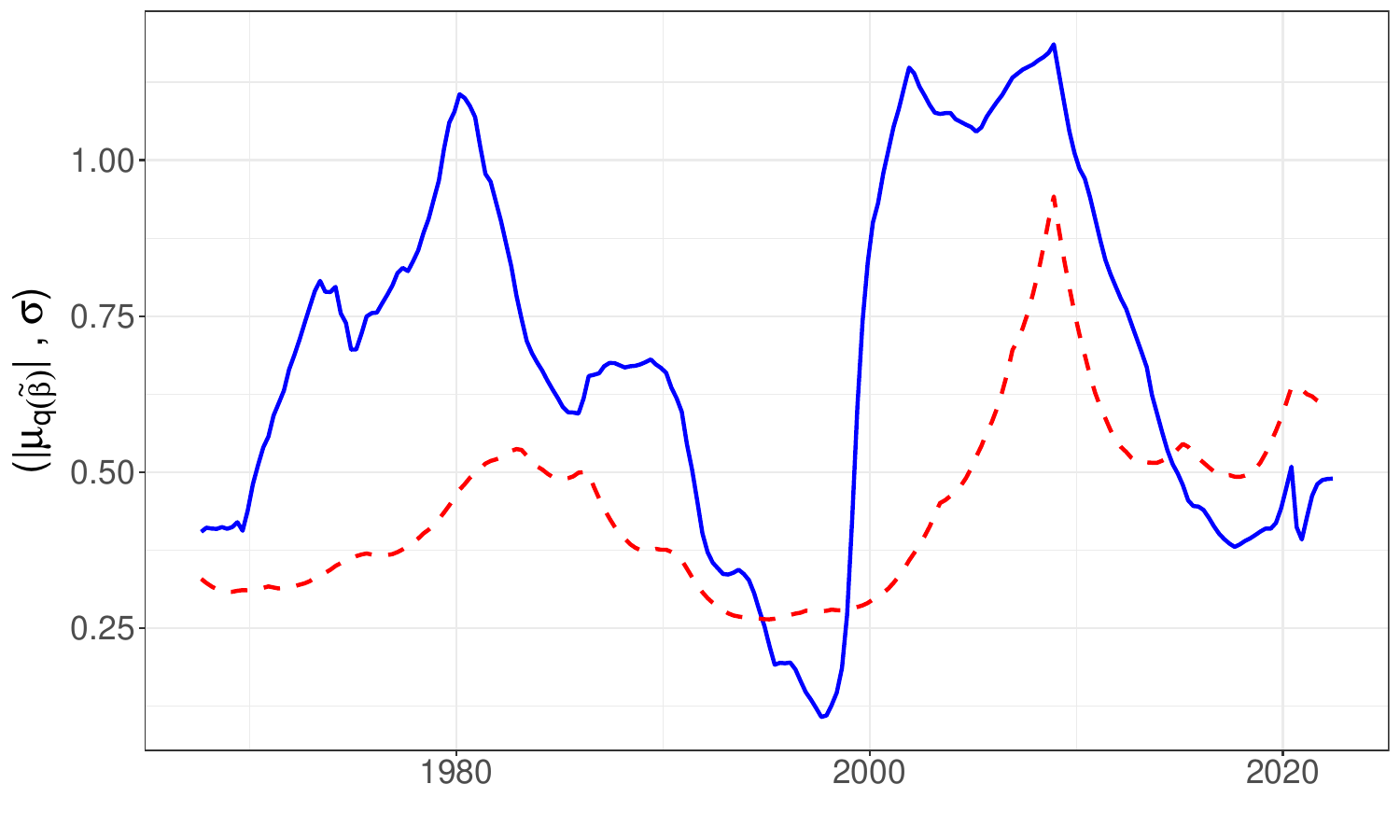}\label{fig:vol_cpiaucsl}}
	\subfigure[PCE deflator (PCECTPI)]{\includegraphics[width=.5\textwidth]{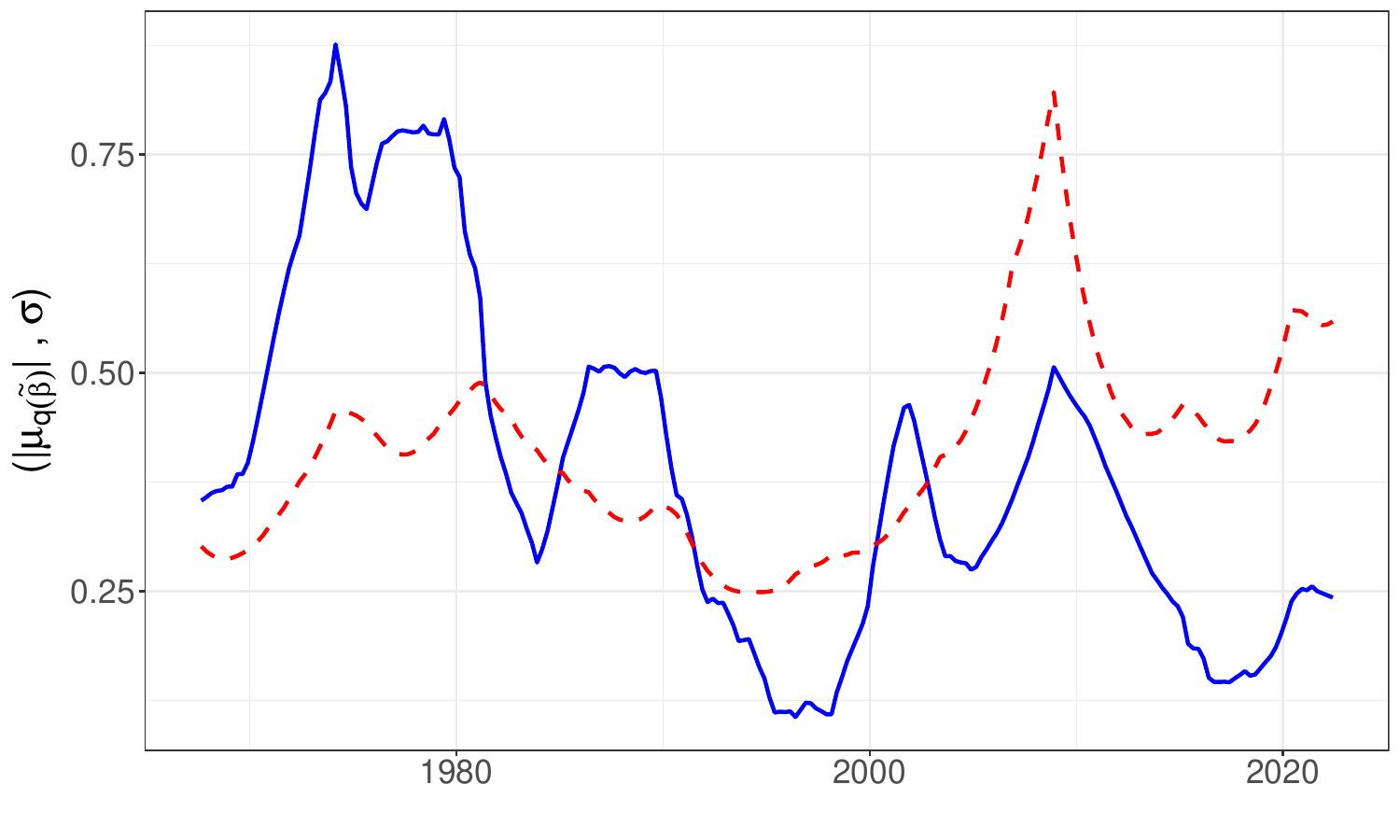}\label{fig:vol_pcectpi}}
 
 \hspace{-1em}\subfigure[Core CPI (CPILFESL)]{\includegraphics[width=.5\textwidth]{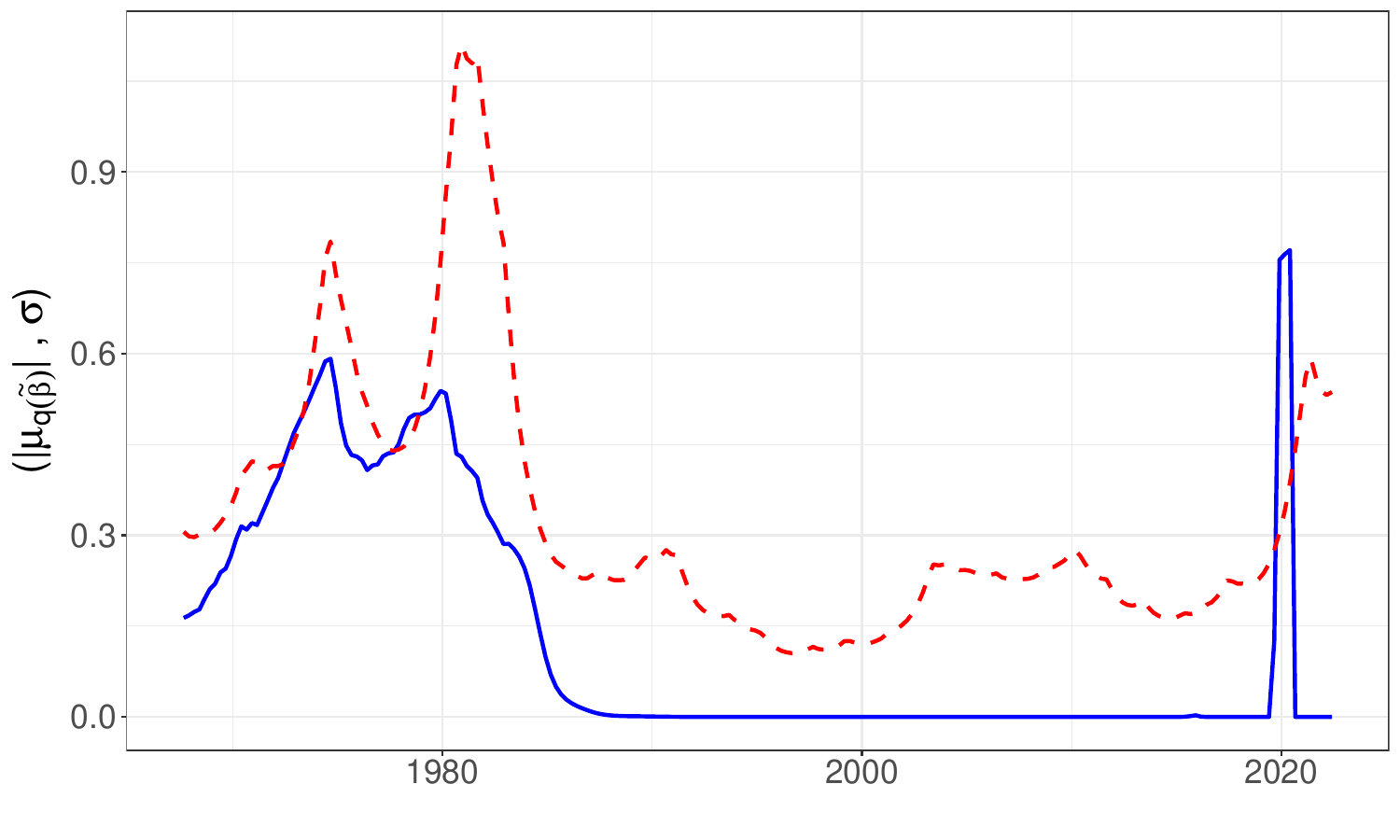}\label{fig:vol_cpilfesl}}
	\subfigure[GDP deflator (GDPCTPI)]{\includegraphics[width=.5\textwidth]{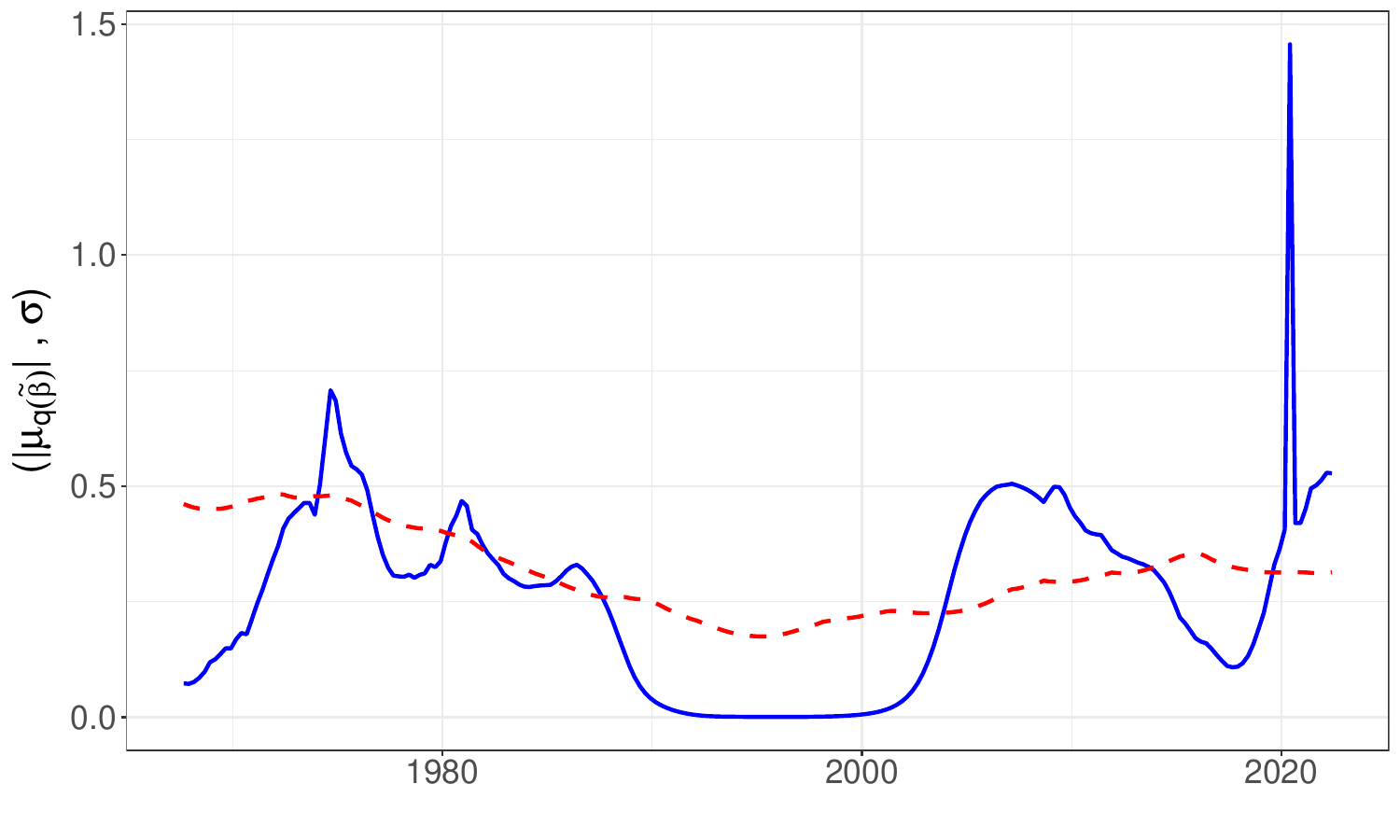}\label{fig:vol_gdpctpi}} 
	\caption{\small Signal (blue) computed as $\sum_{j=1}^p|\mu_{q(\beta_{jt})}|$, for $t=1,\ldots,n$, against the posterior estimates of stochastic volatility $\exp\left(h_t/2\right)$, for $t=1,\ldots,n$ (dashed-red).}\label{fig:vol}
\end{figure}

\end{document}